\documentclass[letterpaper]{article}
\pdfoutput=1
\usepackage[ascii]{inputenc}
\usepackage[dvipsnames,svgnames]{xcolor}
\usepackage[bookmarksnumbered,linktocpage,hypertexnames=false,colorlinks=true,allcolors=Blue,citecolor=Green,breaklinks=true,pdfusetitle]{hyperref}
\usepackage{bbm,braket,microtype,mathrsfs,amsmath,amssymb,xcolor,amsthm,amsfonts,latexsym,mathtools,graphicx,enumitem,booktabs,bm,xspace,float,mathdots,caption,subcaption,ellipsis,enumitem,overpic,thm-restate,nicefrac,fullpage,lmodern,authblk,mleftright}
\usepackage[T1]{fontenc}
\usepackage{textcomp}
\usepackage[english]{babel}
\usepackage[capitalize,nameinlink]{cleveref}
\newtheorem{thm}{Theorem}[section]\crefname{thm}{Theorem}{Theorems}
\newtheorem*{thm*}{Theorem}
\newtheorem{lem}[thm]{Lemma}\crefname{lem}{Lemma}{Lemmas}
\newtheorem{prop}[thm]{Proposition}\crefname{prop}{Proposition}{Propositions}
\newtheorem{cor}[thm]{Corollary}\crefname{cor}{Corollary}{Corollaries}
\crefname{figure}{Figure}{Figures}
\theoremstyle{definition}
\newtheorem{dfn}[thm]{Definition}\crefname{def}{Definition}{Definitions}
\theoremstyle{remark}
\newtheorem{rmk}[thm]{Remark}\crefname{rmk}{Remark}{Remarks}

\numberwithin{equation}{section}
\DeclareMathOperator{\tr}{tr}
\DeclareMathOperator{\rank}{rank}
\DeclareMathOperator{\Var}{Var}

\DeclareMathOperator{\poly}{poly}
\DeclareMathOperator{\spec}{spec}
\DeclareMathOperator{\supp}{supp}
\DeclareMathOperator{\MP}{MP}

\DeclareMathOperator{\JT}{JT}
\DeclareMathOperator{\EOW}{EOW}
\DeclareMathOperator{\rad}{Rad}

\DeclarePairedDelimiter{\abs}{\lvert}{\rvert}
\DeclarePairedDelimiter{\norm}{\lVert}{\rVert}
\DeclarePairedDelimiter\floor{\lfloor}{\rfloor}

\newcommand{\ot}{\otimes}
\newcommand{\id}{\mathrm{id}}
\newcommand{\CC}{\mathbbm{C}}
\newcommand{\RR}{\mathbbm{R}}
\newcommand{\NN}{\mathbbm{N}}
\newcommand{\EE}{\mathbbm{E}}
\newcommand{\HH}{\mathcal{H}}
\newcommand{\KK}{K}
\newcommand{\A}{\mathcal{A}}
\newcommand{\eps}{\varepsilon}
\newcommand{\bigO}{\mathcal O}

\newcommand{\ind}{\mathbbm{1}}
\newcommand{\minstar}{\mathrm{min}_*}
\newcommand*{\tran}{{\mkern-1.5mu\mathsf{T}}}

\newcommand{\proj}[1]{\mathinner{\lvert#1\rangle\langle#1\rvert}}
\newcommand{\PSD}{\mathcal P}
\newcommand{\Peq}{\mathcal P_{\scriptscriptstyle{=}}}
\newcommand{\Pleq}{\mathcal P_{\scriptscriptstyle{\leq}}}
\renewcommand{\d}{\ensuremath{\mathrm{d}}}

\allowdisplaybreaks[4]
\begin{document}

\title{Random Tensor Networks with Nontrivial Links}
\date{}
\author[1]{Newton Cheng}
\author[2]{C\'ecilia Lancien}
\author[1,3]{Geoff Penington}
\author[4]{Michael Walter}
\author[5,*]{Freek Witteveen}
\affil[1]{Center for Theoretical Physics and Department of Physics, University of California, Berkeley, USA}
\affil[2]{Institut Fourier \& CNRS, Universit\'e Grenoble Alpes, Gi\`eres, France}
\affil[3]{Institute for Advanced Study, Princeton, USA}
\affil[4]{Faculty of Computer Science, Ruhr University Bochum, Germany}
\affil[5]{Department of Mathematical Sciences and QMATH, University of Copenhagen, Denmark}

\maketitle
\begin{abstract}
  Random tensor networks are a powerful toy model for understanding the entanglement structure of holographic quantum gravity.
  However, unlike holographic quantum gravity, their entanglement spectra are flat.
  It has therefore been argued that a better model consists of random tensor networks with link states that are not maximally entangled, i.e., have nontrivial spectra.
  In this work, we initiate a systematic study of the entanglement properties of these networks.
  We employ tools from free probability, random matrix theory, and one-shot quantum information theory to study random tensor networks with bounded and unbounded variation in link spectra, and in cases where a subsystem has one or multiple minimal cuts.
  If the link states have bounded spectral variation, the limiting entanglement spectrum of a subsystem with two minimal cuts can be expressed as a free product of the entanglement spectra of each cut, along with a Marchenko-Pastur distribution.
  For a class of states with unbounded spectral variation, analogous to semiclassical states in quantum gravity, we relate the limiting entanglement spectrum of a subsystem with two minimal cuts to the distribution of the minimal entanglement across the two cuts.
  In doing so, we draw connections to previous work on split transfer protocols, entanglement negativity in random tensor networks, and Euclidean path integrals in quantum gravity.
\end{abstract}
\tableofcontents

\section{Introduction}
More than twenty years after its discovery, the AdS/CFT correspondence \cite{maldacena1999large} remains the only known example of a theory of quantum gravity.%
\footnote{Here, we are requiring any putative theory of quantum gravity to (a) be defined nonperturbatively and (b) have strong evidence for the existence of a semiclassical limit consisting of Einstein gravity coupled to quantum field theory.}
A crucial feature of this correspondence is that the emergence of a (classical) spacetime is closely related to the entanglement structure of the boundary theory.
Tensor networks appear to provide useful toy models for this aspect of AdS/CFT, mirroring many of its expected properties in a setting that can be made completely mathematically rigorous \cite{swingle2012entanglement,swingle2012constructing, pastawski2015holographic,hayden2016holographic}.
A particularly powerful model is given by \emph{random tensor networks}, which have the advantage of being highly analytically tractable, while exhibiting remarkably precise agreement with gravitational calculations (even including certain exponentially small corrections) \cite{hayden2016holographic, yang2016bidirectional, qi2017holographic, qi2018spacetime, penington2019replica, dong2021holographic,kudler2021negativity, qi2021holevo}.
Random tensors and tensor networks also arise in a number of other fields of physics, including quantum information, where they have been used to explore generic entanglement properties of quantum states \cite{hayden2006aspects, collins2010random, aubrun2012partial, aubrun2012phase, aubrun2012realigning, collins2012matrix, collins2013area, christandl2014eigenvalue, collins2016random, aubrun2017alice, hastings2017asymptotics, nezami2020multipartite, walter2021hypergraph, morgan2021classical, lancien2021correlation, akers2021reflected} and condensed matter physics, e.g. in the study of random circuits and measurements \cite{you2018machine, vasseur2019entanglement, lopez2020mean, nahum2021measurement, medina2021entanglement, yang2021entanglement, levy2021entanglement, li2021statistical}.

The most basic version of a random tensor network is characterized by a choice of bond dimension~$D$ and a graph~$G = (V,E)$, where the vertices $V = V_b \sqcup V_\partial$ of $G$ are partitioned into ``bulk'' vertices $V_b$ and ``boundary'' vertices $V_\partial$. To each edge $e \in E$, we associate a maximally entangled state
\begin{align}\label{eq:max ent}
  \frac{1}{\sqrt{D}} \sum_{i=1}^D \ket{ii}
\end{align}
on two $D$-dimensional Hilbert spaces, one of which is associated to each endpoint of $e$; each vertex $v \in V$ is therefore associated with a Hilbert space $\mathcal{H}_v$ of dimension~$D^{\deg(v)}$.
Finally, we project each \emph{bulk} vertex~$v_b \in V_b$ onto a Haar random state $\ket{\psi_{v_b}} \in \mathcal{H}_{v_b}$.
The resulting ``random tensor network state'' lives in the Hilbert space~$\mathcal{H}_\partial = \otimes_{v_\partial} \mathcal{H}_{v_\partial}$ associated to the boundary vertices $v_\partial \in V_\partial$, as shown in \cref{fig:peps}.
Such states, obtained by projecting maximally entangled edge states onto (not necessarily random) bulk vertex states, are also known as projected entangled pair states (PEPS) in the condensed matter literature \cite{verstraete2004valence,verstraete2008matrix,cirac2021matrix}.

\begin{figure}
  \centering
  \begin{subfigure}{.5\textwidth}
    \centering
    \includegraphics[width=.7\linewidth]{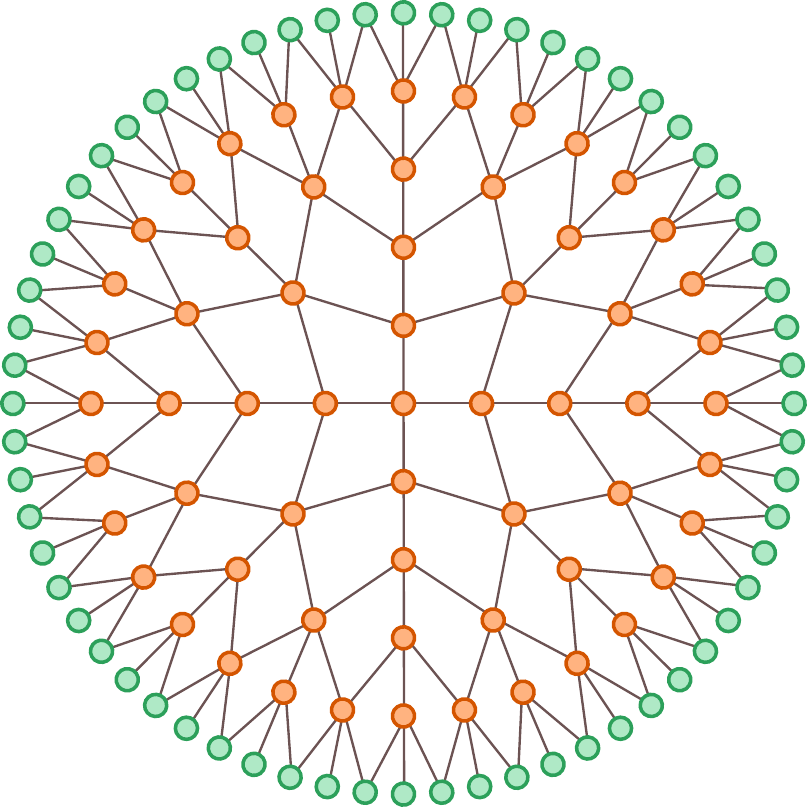}
    \caption{The skeleton of a random tensor network. We take the green-colored vertices to be the boundary vertices~$V_\partial$ and the remaining, orange-colored vertices to be the bulk vertices~$V_b$. We use this network for illustration throughout the remainder of this paper.}
    \label{fig:sub1}
  \end{subfigure}%
  \hspace*{0.5cm}
  \begin{subfigure}{.5\textwidth}
    \centering
    \begin{overpic}[width=.7\linewidth,grid=false]{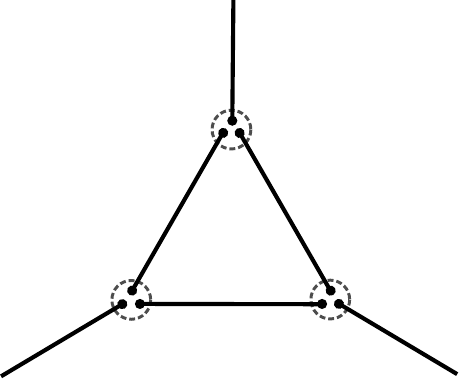}
      \put(17,25){$\bra{\psi_{v_b}}$}
      \put(66,40){\footnotesize{$\frac{1}{\sqrt{D}} \sum_{i=1}^D \ket{ii}$}}
    \end{overpic}
    \caption{A close-up picture of a random tensor network.
      We first associate each vertex $v\in V$ with a Hilbert space $\mathcal{H}_v$ of dimension $D^{\deg(v)}$ (here, $\deg(v)=3$).
      Then maximally entangled states are distributed according to the graph's edges, after which the state at each bulk vertex~$v_b\in V_b$ is projected onto a Haar random state~$\ket{\psi_{v_b}} \in \mathcal{H}_v$.}
    \label{fig:sub2}
  \end{subfigure}
  \caption{The basic structure of a random tensor network.}
  \label{fig:peps}
\end{figure}

To characterize the typical entanglement structure of random tensor network states, we can compute the von Neumann entropy $H(\rho_A)$ of the reduced density matrix $\rho_A$ on a subset $A \subset V_\partial$ of the boundary vertices.
In the limit where the bond dimension $D$ is very large, this entropy can be shown to converge with high probability to $\log(D) \abs{\gamma_A}$, where $\gamma_A$ is the set of edges crossing the minimal cut (for the moment, assumed to be the unique such cut) in the graph separating~$A$ from its boundary complement $V_\partial \setminus A$ (see \cref{subfig:one-cut}).
This formula is closely analogous to the Ryu-Takayanagi (RT) formula and its generalizations in AdS/CFT, in which entropies are given by the area of minimal surfaces homologous to a subregion of the conformal boundary \cite{ryu2006holographic,ryu2006aspects,hubeny2007covariant,lewkowycz2013generalized,faulkner2013quantum,engelhardt2015quantum}. Indeed, this connection is one of the primary reasons for studying tensor networks as a toy model of quantum gravity. However, if one goes beyond the von Neumann entropy and studies finer details of the entanglement spectrum of random tensor network states, significant divergences from holography begin to appear, as we will see shortly.

When studying entanglement in either random tensor networks or quantum gravity, or more generally in quantum field theory, it is often convenient to study $k$-th R\'enyi entropies $H_k(\rho_A) = (1-k)^{-1}\log \tr [\rho_A^k]$.
For integer $k > 1$, these are more amenable to direct computation than the von Neumann entropy, and one can extract the von Neumann entropy by analytic continuation to $k = 1$.
The computation of R\'enyi entropies in random tensor network models is very similar to holographic computations.
In both cases, the idea is to use the \emph{replica trick} -- essentially, this is the observation that $\tr[\rho_A^k] = \tr[\tau \rho_A^{\ot k}]$ where $\tau$ is an operator which permutes the $k$ copies of $A$ cyclically.
In the holographic computation, this can be written as a path integral, on $k$ copies of the theory, glued together in an appropriate way.
By the holographic dictionary, this path integral can then be computed by the action of a bulk geometry with certain boundary conditions \cite{lewkowycz2013generalized}.
For random tensor networks, one finds that $\tr [\rho_A^k]$ concentrates around its expectation, and can be computed as the partition function of a classical spin model on the bulk vertices, with boundary conditions dictated by the choice of boundary subsystem \cite{hayden2016holographic}.
This computation will be explained in detail in \cref{sec:replica trick}.
From these computations, one finds that holographic CFT states and random tensor network states behave quite differently when $k \neq 1$.
For random tensor network states, the R\'enyi entropies are approximately independent of $k$ in the large $D$ limit, meaning their entanglement spectrum is close to ``flat;'' the boundary state $\rho_A$ is approximately maximally-mixed within a certain subspace.
On the other hand, CFT states that are dual to semiclassical spacetime geometries have R\'enyi entropies that vary non-trivially with $k$, meaning their entanglement spectrum contains a wide range of eigenvalues that contribute significantly to the state.
Recently, it has been argued that the class of ``fixed-area states'' in AdS/CFT \textit{do} have flat spectra, and more generally have an entanglement structure that closely matches random tensor network states \cite{akers2019holographic,dong2019flat,bao2019beyond,marolf2020probing,dong2021holographic}.
Fixed-area states have a well-defined semiclassical geometry associated to a fixed spatial slice; however, thanks to the uncertainty principle, they cannot describe a single semiclassical \emph{spacetime} geometry \cite{bao2019beyond}.
We discuss the connection between these states and random tensor networks in more detail in \cref{sec:gravity context}.

In the random tensor network model, the flatness of the spectrum can be traced to the maximally-entangled states used as ``link states'' (see \cref{eq:max ent,fig:sub2}) on the edges of the graph, which themselves have flat entanglement spectra.
To take results about random tensor networks beyond the fixed-area state regime, it is natural -- see, e.g., discussion in \cite{hayden2016holographic, bao2019beyond} -- to replace the maximally entangled link states by general states
\begin{align}\label{eq:link state intro}
  \ket{\phi_e} = \sum_{i = 1}^D \sqrt{\lambda_{e,i}} \ket{ii}.
\end{align}
The variation in the entanglement spectrum $\lambda_{e,i}$ represents the fluctuations in area in semiclassical gravitational states.
We introduce this model in \cref{sec:rtn}.
The goal of this paper will be to understand the entanglement spectra of random tensor networks with such \emph{nontrivial link states} and spectra in a number of different regimes.
In fact, a number of our results apply in an even more general setting, where the product of link states~$\otimes_e \ket{\phi_e}$ is replaced by a completely general ``background state'' density matrix $\phi_V$.
From a quantum gravity perspective, random tensor networks with general background states are needed to model bulk quantum fields in AdS/CFT, which yield significant physical consequences when the bulk entropies are large \cite{akers2020leading,akers2021quantum}.
Beyond holography, they also play a central role in the quantum information processing tasks of multiparty state merging \emph{split transfer} \cite{dutil2010one}, a connection we elaborate on in \cref{sec:split_recovery}.

If one considers a random tensor network with non-trivial link states, if there is a \textit{single} minimal cut $\gamma_A$ for a subsystem~$A$, then the resulting density matrix $\rho_A$ will have an entanglement spectrum that converges to that of $\abs{\gamma_A}$ copies of the link state along the minimal cut as~$D \to \infty$; indeed, this was implicit in \cite{hayden2016holographic}.
A more complex question, and the main focus of this work, is the case where there are \emph{two} minimal cuts, as in \cref{subfig:two-cut}.
This situation is motivated by questions in holography: it can be used to study the phase transition at the point where there are two competing minimal surfaces \cite{marolf2020probing, akers2020leading}, which has been relevant to recent advances on the black hole information paradox \cite{penington2020entanglement,almheiri2019entropy,penington2019replica,almheiri2020replica,marolf2020transcending}.

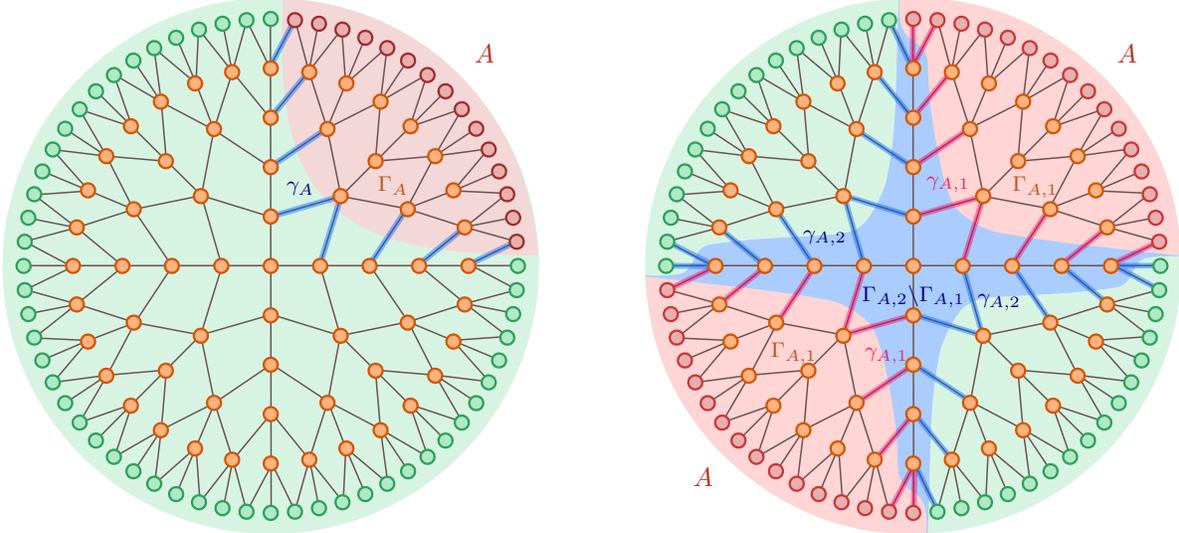
\begin{figure}
  \centering
  \begin{subfigure}[t]{.48\textwidth}
    \centering
    \begin{overpic}[width=0.9\textwidth,grid=false]{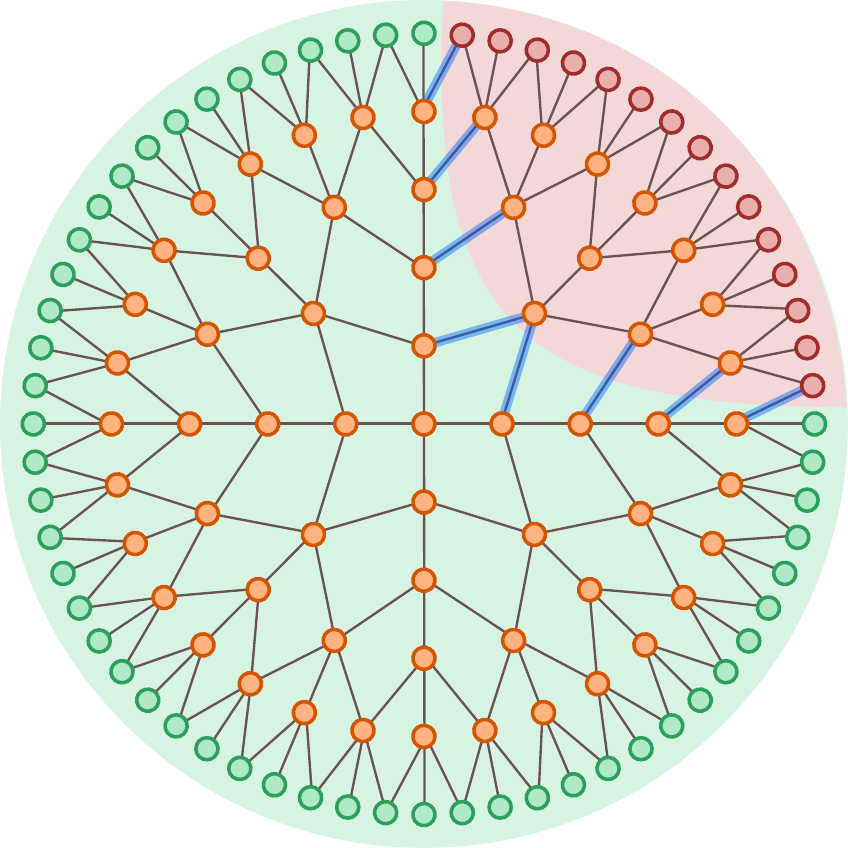}
      \put(53,64){\color{NavyBlue}{\footnotesize{$\gamma_A$}}}
      \put(88,88){\color{BrickRed}{$A$}}
      \put(70,64){\color{Bittersweet}{\footnotesize{$\Gamma_A$}}}
    \end{overpic}
    \caption{The boundary domain $A$ has a unique minimal cut~$\Gamma_A$. The set of edges crossing this minimal cut is denoted by~$\gamma_A$.}
    \label{subfig:one-cut}
  \end{subfigure}%
  \hspace*{0.5cm}
  \begin{subfigure}[t]{.48\textwidth}
    \centering
    \begin{overpic}[width=0.9\textwidth,grid=false]{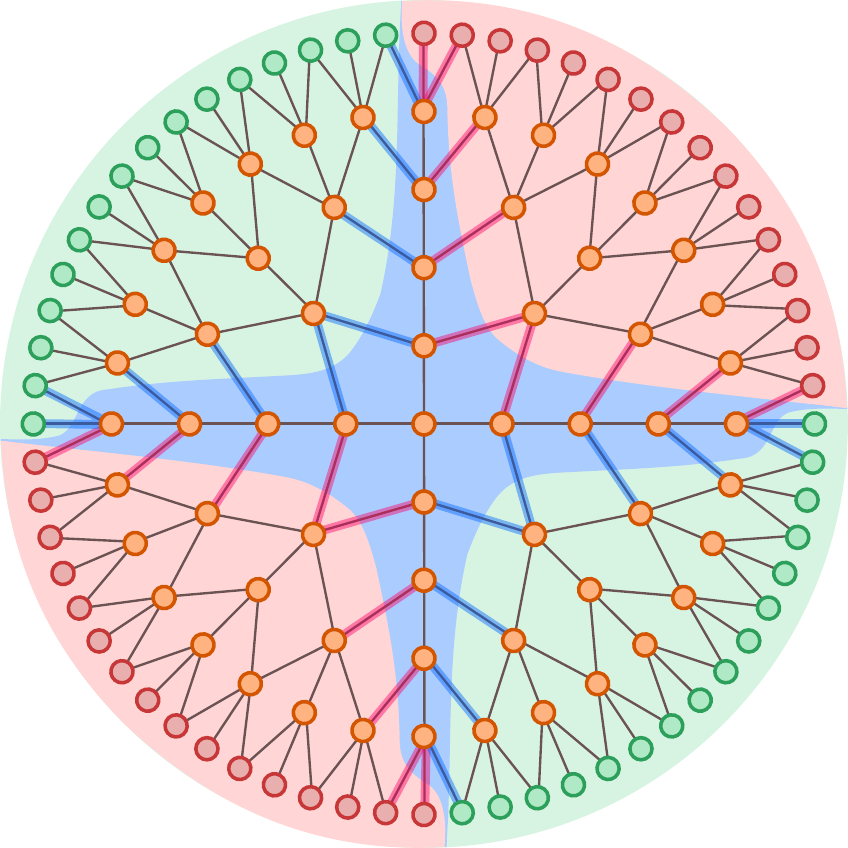}
      \put(52.5,65){\color{WildStrawberry}{\footnotesize{$\gamma_{A,1}$}}}
      \put(41,33){\color{WildStrawberry}{\footnotesize{$\gamma_{A,1}$}}}
      \put(29.5,55){\color{NavyBlue}{\footnotesize{$\gamma_{A,2}$}}}
      \put(62,43){\color{NavyBlue}{\footnotesize{$\gamma_{A,2}$}}}
      \put(88,88){\color{BrickRed}{$A$}}
      \put(9,9){\color{BrickRed}{$A$}}
      \put(68.5,64){\color{Bittersweet}{\footnotesize{$\Gamma_{A,1}$}}}
      \put(23.5,33){\color{Bittersweet}{\footnotesize{$\Gamma_{A,1}$}}}
      \put(40.5,43.5){\color{NavyBlue}{\footnotesize{$\Gamma_{A,2} \backslash \Gamma_{A,1}$}}}
    \end{overpic}
    \caption{
      The boundary domain~$A$ has two minimal cuts~$\Gamma_{A,i}$ with cut-sets~$\gamma_{A,i}$.
      In the notation of \cref{sec:ground state configurations}, the red region is $\Gamma_{A,1} = V_1$, the blue region is $\Gamma_{A,2}\setminus \Gamma_{A,1} = V_2$ and the green region is $V \setminus \Gamma_{A,2} = V_3$.
      Note that~$\Gamma_{A,1}$ has two connected components, while~$\Gamma_{A,2}$ is connected.
    }
    \label{subfig:two-cut}
  \end{subfigure}
  \caption{Tensor networks with one and two minimal cuts.}
  \label{fig:test}
\end{figure}

For our first main result, in \cref{sec:almost max entangled}, we consider a family of link states with increasing bond dimension~$D$ as in \cref{eq:link state intro}.
For each~$D$, the link state has an associated distribution
\begin{align*}
  \mu_e^{(D)} = \frac{1}{D}\sum_{i=1}^D \delta_{D\lambda_{e,i}},
\end{align*}
where $\delta_{x}$ is a $\delta$-distribution centered at $x$, so this is the discrete probability distribution given by a uniform distribution over the spectrum of the link state.
Note that $\lambda_{e,i}=\lambda_{e,i}^{(D)}$ also depends on $D$. We then require that the moments
\begin{align*}
  m_k^{(D)}
  = D^{k-1} \sum_{i=1}^D \lambda_{e,i}^k
\end{align*}
of the distributions $\mu_e^{(D)}$ converge to a finite limit as $D\to \infty$ for all positive integer $k$.
We refer to this as the \emph{bounded spectral variation} limit.
This means, in particular, that we must have $\lambda_{e,i} = \bigO(1/D)$ for all but a vanishing fraction of the eigenvalues $\lambda_{e,i}$.
If we let $\gamma_A$ denote a minimal cut for a boundary domain~$A$, then we may similarly define the associated~distribution
\begin{align*}
  \mu_{\gamma_A}^{(D)} =\frac{1}{D^{|\gamma_A|}} \sum_{\{i_e\}_{e \in \gamma_A}} \delta_{D^{|\gamma_A|} \prod_{e \in\gamma_A} \lambda_{e,i_e}}.
\end{align*}
By assumption, the moments of the distribution $\mu_{\gamma_A}^{(D)}$ converge to a finite limit (because those of each distribution $\mu_e^{(D)}$ do), implying that $\mu_{\gamma_A}^{(D)}$ converges weakly to some distribution $\mu_{\gamma_A}$.
Now consider the empirical distribution of the spectrum of reduced state $\rho_A$, which is the (random) distribution
\begin{align*}
  \mu_A^{(D)} = \frac{1}{D^{\abs{\gamma_A}}} \sum_{\lambda \in \spec(\rho_A)} \delta_{D^{\abs{\gamma_A}} \lambda}.
\end{align*}
In the case where there are two non-intersecting minimal cuts $\gamma_{A,1}$ and $\gamma_{A,2}$, we find that $\mu^{(D)}_A$ converges weakly, in probability, to a limiting distribution~$\mu_A$ given by a \emph{free product} $\MP(1) \boxtimes \mu_{\gamma_{A,1}} \boxtimes \mu_{\gamma_{A,2}}$, a notion from the theory of free probability.
Here, $\MP(1)$ is the Marchenko-Pastur distribution of parameter~1.
The situation is summarized by our first main result, which we state more precisely as \cref{thm:spectrum near max entangled}:

\begin{thm*}[Informal]
  Consider a family of link states in the bounded spectral variation limit.
  If there is a unique minimal cut~$\gamma_A$ for a boundary subsystem $A$, then $\mu_A^{(D)}$ converges weakly, in probability, to~$\mu_{\gamma_A}$, while if there are exactly two non-intersecting minimal cuts $\gamma_{A,1}$ and $\gamma_{A,2}$, it converges to~$\MP(1) \boxtimes \mu_{\gamma_{A,1}} \boxtimes \mu_{\gamma_{A,2}}$.
\end{thm*}

In \cref{sec:negativity} we briefly discuss the closely related problem of computing the entanglement negativity spectrum in the same regime.

For our second main result, in \cref{sec:far from max entangled}, we investigate a different regime, in which link states are allowed to have \emph{unbounded spectral variation} in the large $D$ limit.
This is the more relevant regime for holography, where fluctuations in the area of a surface (in Planck units) grow sublinearly but without bound in the semiclassical limit.
%
When the spectral variation is unbounded, there still exists a reasonable notion of a minimal cut that determines the entanglement spectrum of the boundary state, but the key difference is that minimality must now be defined entropically, rather than geometrically. In fact, the underlying graph essentially plays no role in this regime.
A sensible way to formalize this would be to use one-shot entropies:
we might say that a cut is ``minimal'' if the rank of the state along the cut is smaller than the inverse of the largest element in the entanglement spectrum along any other cut.
This condition, while intuitive, is a little too restrictive, and one can use smooth conditional entropies to get a weaker, but still meaningful, condition.
In \cref{dfn:gen mincut}, following \cite{akers2020leading} we introduce the notion of a \emph{unique $(\eps, K)$-minimal cut}~$\Gamma_A$, where $\Gamma_A\subseteq V$ is a subset of vertices such that the $\eps$-smooth conditional min-entropy of $\Gamma_A$ compared to competing cuts either contained in~$\Gamma_A$ or containing~$\Gamma_A$ is lower bounded by~$K$.
Intuitively, we want $K$ to be as large as possible, and indeed, we show in \cref{thm:rtn with arbitrary link states} that the spectrum of a reduced density matrix $\rho_A$ will be close to the spectrum along a unique $(\eps, K)$-minimal cut with an error exponentially small in~$K$.

The situation is more complicated if there are two non-intersecting  $(\eps, K)$-minimal cuts, defined in \cref{dfn:gen pair cut} as both cuts satisfying an $(\eps, K)$-minimality property for the same $\eps$ and $K$.
Here, we need to impose a regularity condition on the link states, motivated by the example of a link state that is a (large) number of copies of a fixed state:
\begin{align}\label{eq:many copies intro}
  \ket{\phi_e} = \ket{\phi_0}^{\ot n}
\end{align}
where $\ket{\phi_0}$ is a bipartite state with local dimension $d$, so the total bond dimension is $D = d^n$.
In this case, if $\phi_0$ is not maximally entangled, the measure $\mu_e^{(D)}$ will not converge with increasing $n$ as the spectrum is not concentrated around $\frac1D$.
However, using a different measure, the entanglement spectrum of $\ket{\phi_e}$ satisfies a central limit theorem.
Namely, if~$\ket{\phi_e}$ has Schmidt coefficients $\{\lambda_i\}$,  and~$\ket{\phi_0}$ has entanglement entropy $H_0$, then the distribution of the random variable $X^{(n)}$ which takes values $\frac{1}{\sqrt{n}}(\log(\frac{1}{\lambda_i}) - nH_0)$ with probability~$\lambda_i$, converges weakly to a centered Gaussian distribution as $n\to\infty$.
Since we subtracted the entropy $nH_0$, the random variable $X^{(n)}$ has expectation zero.
Its variance can be thought of as a measure of the fluctuation of $\log(\frac{1}{\lambda_i})$ around the entropy, and is relevant for second-order asymptotic rates in quantum information processing tasks \cite{tomamichel2013hierarchy}.
%
%
We take this central limit theorem as motivation for a regularity condition on the spectra of general link and background states, and we allow the states to have varying bond dimensions $D(n)$, e.g. $D(n) \sim d^n$.
To be more precise, we define the following measure along a cut $\gamma_A$:
\begin{align*}
  \nu_{\gamma_A}^{(n)} = \sum_{\{i_e\}_{e \in \gamma_A}} \prod_{e \in\gamma_A} \lambda_{e,i_e} \delta_{\frac{1}{\sqrt n}\left( \sum_{e \in \gamma_A} \log \frac{1}{\lambda_{e,i_e}} - H(n)\right)},
\end{align*}
where $H(n)$ is a function of $n$ (which one can think of as being approximately equal to the entanglement entropy along the cut $\gamma_A$), and we assume that $\nu_{\gamma_A}^{(n)}$ converges weakly to a \emph{continuous} distribution.
Note that the distribution described above reduces to the distribution of $X^{(n)}$ if the link state is of the form in \cref{eq:many copies intro}, and is very different from the distribution $\mu_A^{(D)}$ we study for link states with bounded spectral variation.
Similarly, we let
\begin{align*}
  \nu_{A}^{(n)} = \sum_{\lambda \in \spec(\rho_A)} \lambda \, \delta_{\frac{1}{\sqrt{n}}(\log(\frac{1}{\lambda}) - H(n))}
\end{align*}
be the corresponding (random) distribution for the boundary spectrum.
Knowledge of this distribution allows computation of the entropy of $\rho_A$ (and fluctuations) as a correction to $H(n)$.
In the random tensor network setting, we find that in a situation with two competing minimal cuts, the random tensor network will `select' the minimal parts of each cut, in the following sense:

\begin{thm*}[Informal]
  Assume that we have a family of (states with) two non-intersecting $(\eps(n),K(n))$-minimal cuts $\gamma_{A,1}$ and $\gamma_{A,2}$, as defined in \cref{dfn:gen pair cut}. Suppose the entanglement spectra along the two minimal cuts are such that $\nu_{\gamma_{A,1}}^{(n)}$ and $\nu_{\gamma_{A,2}}^{(n)}$ converge weakly to continuous measures $\nu_1$ and $\nu_2$ respectively as $n \to \infty$.
  Then $\nu_{A}^{(n)}$ converges weakly, in probability, to $\minstar(\nu_1,\nu_2)$ which is the pushforward of $\nu_1$ and $\nu_2$ along the function $\min \colon \RR \times \RR \to \RR$.
  In other words, for any bounded continuous function $f \in C_b(\RR)$
  \begin{align} \label{eq:minareaspec}
    \sum_{\lambda \in \spec(\rho_A)} \lambda \, f\mleft(\frac{\log(\frac1\lambda) - H(n)}{\sqrt{n}}\mright) \rightarrow \int \int f(\min(x_1,x_2)) \, \d \nu_1(x_1) \d \nu_2(x_2)
  \end{align}
  in probability.
\end{thm*}
We also show in \cref{cor:convergence of entropy} that, up to an error of size $\bigO(\log(n))$, this allows us to compute the entropy of $\rho_A$ .

An analogous statement to \eqref{eq:minareaspec} was previously conjectured in \cite{marolf2020probing} to be valid in quantum gravity and justified at a physics level of rigour in \cite{akers2020leading}.
Our proof of \eqref{eq:minareaspec} is closely related to the arguments in \cite{akers2020leading}, but converting the physical arguments into a mathematical proof and careful controlling the relevant sources of error requires significant technical work, and constitutes the majority of \cref{sec:far from max entangled}.

In \cref{sec:gravity context}, we relate our results to the study of holographic gravity computations, particularly in situations with competing minimal surfaces.
This is not needed to understand the results of this work, but provides additional motivation for the relevance of our results and clarifies the way in which random tensor networks provide a useful toy model for holographic quantum gravity. Then, in \cref{sec:split_recovery}, we discuss the relation of our results to split transfer and its relevance in holography. This is again not needed to understand the main text, but rather pointing at topics our work is connected to.
Next, in \cref{sec:joint smoothing}, we prove two technical lemmas on joint smoothing that are necessary to analyze competing minimal cuts in the unbounded spectral variation regime.
Finally, in \cref{sec:metric}, we prove that a certain function on the symmetric group is a metric; this is not directly used in the current work, but may be of independent interest for the study of random tensor networks.

After the completion of this manuscript, we became aware of independent work by Jinzhao Wang \cite{wang2022information} on the use of free products to describe entanglement in the toy model of quantum gravity introduced in \cite{penington2019replica} that has strong overlap with the ideas in \cref{sec:bounded spectral variation and free product} and \cref{sec:gravity context}.

\subsection{Notation and conventions}
For $k \in \NN$, we let $[k] = \{1, \ldots, k\}$, and we denote by $S_k$ the group of permutations of this set.
We denote by $\norm{a}_p$ the $\ell_p$-norm of a vector~$a$, defined by $\norm a_p^p = \sum_i \abs{a_i}^p$.
If $a$ and $b$ are vectors of different dimension, we extend the shorter vector by zeros and still write $\norm{a-b}_p$ for their distance.
For example, if $a \in \CC^{d_1}$ and~$b \in \CC^{d_2}$ with $d_2 > d_1$, we write
$\norm{a - b}_1 = \sum_{i=1}^{d_1} \abs{a_i - b_i} + \sum_{i=d_1 + 1}^{d_2} \abs{b_i}$.
We also denote by $\norm{A}_p$ the Schatten $p$-norm of an operator $A$, defined as the $\ell_p$-norm of the singular values $\{s_i\}$ of~$A$; it can also be computed by $\norm{A}_p^p = \tr((A^\dagger A)^{p/2})$.
The operator norm is given by $\norm{A}_{\infty} = \max\{s_i\}$.
If~$\HH$ is a Hilbert space, we introduce the notation $\PSD(\HH)$ for the set of positive semidefinite operators on $\HH$.
We often refer to positive semidefinite operators as `density operators' or `states', \emph{without} requiring them to be normalized to unit trace.
We write $\Peq(\HH)$ for the set of $\rho \in \PSD(\HH)$ with unit trace, $\tr[\rho] = 1$, and we denote by $\Pleq(\HH)$ the set of subnormalized states, that is, $\rho \in \PSD(\HH)$ with $\tr[\rho] \leq 1$.
We use the convention that for a vector $\ket{\phi}$, we denote the corresponding pure state by $\phi$, so $\phi = \proj{\phi}$.
Given a positive semidefinite operator~$\rho$, we denote by $\spec(\rho)$ the vector containing its spectrum in non-increasing order, and we write $\spec_+(\rho)$ for the nonzero part of the spectrum.
It is a well-known fact that
\begin{align}\label{eq:spectrum bound}
  \norm{\spec(\rho) - \spec(\sigma)}_1 = \norm{\spec_+(\rho) - \spec_+(\sigma)}_1 \leq \norm{\rho - \sigma}_1
\end{align}
(in the second expression we use the convention for the distance of vectors of possibly different dimension introduced above).
If $A$ is a quantum system with Hilbert space $\HH_A$, we write $\PSD(A) = \PSD(\HH_A)$, $\Peq(A) = \Peq(\HH_A)$, and $\Pleq(A) = \Pleq(\HH_A)$, and we use subscripts, e.g.~$\rho_A \in \PSD(A)$, to indicate which system and Hilbert space a quantum state is associated with.
For a bipartite state~$\rho_{AB} \in \PSD(AB)$ defined on a tensor product of Hilbert spaces $\HH_{AB} = \HH_A \ot \HH_B$, we obtain the reduced state or reduced density operator $\rho_A \in \PSD(A)$ by taking the partial trace over the complement: $\rho_A = \tr_B[\rho_{AB}]$, and similarly in multipartite situations.
Finally, we adopt the standard notation that if $\mu_n$ is some sequence of finite measures, we write~$\mu_n \Rightarrow \mu$ if~$\mu_n$ converges weakly (or in distribution) to a finite measure $\mu$, meaning that for any bounded continuous function $f \in C_b(\RR)$,
\begin{align}\label{eq:weak convergence conventions}
  \int f(x) \d \mu_n(x) \to \int f(x) \d \mu(x).
\end{align}
If $\mu_n$ is a sequence of \emph{random} finite measures on $\RR$, we say that the sequence $\mu_n$ \emph{converges weakly, in probability,} to a finite measure $\mu$, if, for any bounded continuous function $f \in C_b(\RR)$, \cref{eq:weak convergence conventions} converges in probability, i.e., if for every~$\eps > 0$ we have that
\begin{align*}
  \lim_{n \to \infty} \Pr\left( \abs{\int f(x) \d\mu_n(x) - \int f(x) \d\mu(x) } \geq \eps\right) = 0.
\end{align*}
In this situation we will also write $\mu_n \Rightarrow \mu$, in probability.
All logarithms are to base 2.

\section{Random tensor network states}\label{sec:rtn}
We first review the random tensor network model, closely following \cite{hayden2016holographic,dong2021holographic}.
Let~$G = (V,E)$ be a connected undirected graph, and let~$V = V_\partial \sqcup V_b$ be a partition of the vertices into a set of boundary vertices~$V_\partial$ and bulk vertices~$V_b$.
If $A \subseteq V_\partial$, we write $\bar{A} = V_\partial \setminus A$.
We assign a bond dimension~$D_e$ to each edge, and we will consider families of states with increasing bond dimensions; for example, we may take~$D_e = D$ for all edges and let $D$ increase.
For each vertex $x \in V$, let $\partial \{x\}$ denote the set of edges~$e = (xy) \in E$ connecting~$x$ to some~$y \in V$.
We define Hilbert spaces~$\HH_{e,x} = \CC^{D_e}$ for $e \in \partial \{x\}$, and $\HH_x := \bigotimes_{e \in \partial \{x\}} \HH_{e,x}$.
We call the pair $(e,x)$ a \emph{half-edge}.
Moreover, we write $\HH_{e} = \HH_{e,x} \ot \HH_{e,y}$ for an edge $e = (xy) \in E$.
Let~$D_x = \dim(\HH_x)$.
For a subset $A \subseteq V$, we write $\HH_A = \bigotimes_{x \in A} \HH_x$, and similarly, for a subset~$S \subseteq E$ we write~$\HH_S = \bigotimes_{e \in S} \HH_e$.
Similarly, for a set $T$ of half-edges we write $\HH_T = \bigotimes_{(e,x) \in T} \HH_{e,x}$.
At each edge $e = (xy) \in E$, we place a pure state $\ket{\phi_e} \in \Peq(\HH_e)$
\begin{align}\label{eq:link state edge}
  \ket{\phi_{e}} = \sum_{i=1}^{D_e} \sqrt{\lambda_{e,i}} \ket{ii} \in \HH_e = \HH_{e,x} \ot \HH_{e,y},
\end{align}
that we call a \emph{link state}. Then, $\phi_{e,x} = \phi_{e,y} = \sum_{i=1}^D \lambda_{e,i} \proj{i}$ is the reduced density matrix of the link state on either of the two subsystems. We refer to the vector $\spec(\phi_{e,x}) = \spec(\phi_{e,y})$, which is ordered in non-increasing fashion, as the \textit{entanglement spectrum} of $\phi_e$. Let $\phi \in \Peq(V)$ be the full state on edges given by the tensor product of link states
\begin{align}\label{eq:link state}
  \ket{\phi} = \bigotimes_{e \in E} \ket{\phi_{e}}.
\end{align}
At every bulk vertex $x \in V_b$, we place a random vector~$\ket{\psi_x} \in \HH_x$, where the entries of $\ket{\psi_x}$ are independent standard (circularly-symmetric) complex Gaussian random variables: each entry of the tensor can be written as~$\frac{1}{\sqrt2}(x + iy)$ where $x$ and $y$ are independent real Gaussian random variables of mean 0 and unit variance.
We note that, in the model of \cite{hayden2016holographic}, the tensors $\ket{\psi_x}$ were not chosen as random Gaussian vectors, but as uniformly random vectors on the unit sphere.
However, for our choice of Gaussian $\ket{\psi_x}$, the norm $\norm{\ket{\psi_x}}$ is independent of the normalized vector $\ket{\psi_x}/\norm{\ket{\psi_x}}$, and $\ket{\psi_x}/\norm{\ket{\psi_x}}$ will be a uniformly random vectors on the unit sphere.
Therefore, these two models only differ by their normalization.
We write~$\ket{\psi} = \bigotimes_{x \in V_b} \ket{\psi_x}$.
The resulting \emph{random tensor network state} $\rho \in \PSD(V_\partial)$ is defined by
\begin{align}\label{eq:rtn state}
  \ket{\rho} = (I_{V_\partial} \otimes \bra{\psi}) \ket{\phi}.
\end{align}
The random tensor network state is obtained by projecting the link states onto random vectors, so that the final state lives in the boundary Hilbert space.
We can make this manifest by using the cyclicity of the trace to write the density matrix:
\begin{align}\label{eq:rtn state 2}
  \rho
  = \left(I_{V_{\partial}} \ot \bra{\psi} \right) \phi \left(I_{V_{\partial}} \ot \ket{\psi} \right)
  = \tr_{V_b}\mleft[\left( I_{V_{\partial}} \ot \psi \right) \phi \mright].
\end{align}
Note that this state need not be normalized, but we chose the standard deviation of the $\ket{\psi_x}$ such that $\rho$ is normalized on average, given that the link state $\phi$ is normalized:
\begin{align}\label{eq:mean normalization}
  \EE \tr[\rho] = \tr[\phi].
\end{align}
In \cref{sec:normalization}, we prove the stronger statement that $\rho$ is normalized with high probability for appropriately connected tensor networks and large bond dimension. Note also, that in \cref{eq:link state edge}, we have chosen states which have a Schmidt decomposition in a fixed basis (the standard basis).
Since we project onto uniformly random tensors, we can choose to do so without loss of generality.

\subsection{The replica trick for random tensor networks}\label{sec:replica trick}
We now consider a boundary subset $A \subseteq V_{\partial}$ and use the \emph{replica trick} to study the R\'enyi entropies of the reduced density matrix $\rho_A$.
The replica trick for random tensor network models was first studied in \cite{hayden2016holographic}, and it is the key tool we apply throughout this work.
Let $\HH$ be a Hilbert space.
The R\'enyi entropies of a (normalized) density matrix $\rho \in \Peq(\HH)$ are defined by
\begin{align*}
  H_k(\rho) = \frac{1}{1-k}\log(\tr[\rho^k])
\end{align*}
for $k \in (0,1) \cup(1,\infty)$.
For $k = 0,1,\infty$, there are well-defined limits, given by
\begin{align}\label{eq:renyi special cases}
  \begin{split}
    H_0(\rho) &:= \log(\rank(\rho))\\
    H_1(\rho) &:= -\tr[\rho\log(\rho)]\\
    H_{\infty} &:= -\log(\norm{\rho}_{\infty}).
  \end{split}
\end{align}
In particular, we see that $H(\rho) = H_1(\rho)$ is the von Neumann entropy.
We will also write $H(A)_\rho := H(\rho_A)$ and $H_k(A)_{\rho} := H_k(\rho_A)$ for reduced density matrices.
If $\rho \in \Pleq(\HH)$ is subnormalized, we let
\begin{align}\label{eq:renyi subnormalized}
  H_k(\rho) = \frac{1}{1-k}\log \frac{\tr\mleft[\rho^k\mright]}{\tr[\rho]}.
\end{align}
Denote by~$R$ the representation of~$S_k$ on~$\HH^{\ot k}$ which permutes the~$k$ copies of $\HH$ according to the action of~$S_k$.
We will write $R_x(\pi)$ when $\HH = \HH_x$ and $R_A(\pi)$ if $\HH = \HH_A$ for~$A \subseteq V$.
We let~$\tau$ denote the standard~$k$-cycle in $S_k$, i.e.,
\begin{align*}
  \tau = (1 2 \ldots k).
\end{align*}
The key idea of the replica trick is the observation that the $k$-th moment of $\rho \in \PSD(\HH)$ can be written as
\begin{align}\label{eq:replica trick}
  \tr \mleft[\rho^k\mright]  = \tr \mleft[R(\tau)\rho^{\ot k}\mright].
\end{align}
Recall the notion of the cycle type of a permutation $\pi$: if $\pi$ can be written as as a product of $m$ disjoint cycles of lengths $l_1, \ldots, l_m$, then $\pi$ has cycle type $C(\pi) = \{l_1,\ldots,l_m\}$.
Then, for an arbitrary $\pi \in S_k$,
\begin{align*}
  \tr\mleft[ R(\pi)\rho^{\ot k}\mright] = \prod_{l \in C(\pi)} \tr\mleft[ \rho^l \mright].
\end{align*}
Note that this is the generalization of the well-known \emph{swap trick} for two copies of a state $\rho$. The other crucial ingredient is a property of the Gaussian random vectors:
\begin{align}\label{eq:expectation random state}
  \EE\mleft[\psi_x^{\ot k}\mright] = \sum_{\pi \in S_k} R_x(\pi).
\end{align}
Using \cref{eq:rtn state 2}, we may then compute
\begin{align}\label{eq:replica trick computation}
  \begin{split}
    \EE\tr\mleft[\rho_A^k\mright] &= \EE\tr\mleft[R_A(\tau) \rho_{A}^{\ot k}\mright] \\
    &= \EE\tr\mleft[\left(R_A(\tau) \ot R_{V_\partial \setminus A}(\id)\right) \rho^{\ot k}\mright]\\
    &= \EE\tr\mleft[\left(R_A(\tau) \ot R_{V_\partial \setminus A}(\id)\right) \left( I_{V_{\partial}} \ot \psi \right)^{\ot k} \phi^{\ot k}\mright] \\
    &= \tr\mleft[\left(R_A(\tau) \ot R_{V \setminus A}(\id)\right) \EE \mleft[(I_{V_\partial} \ot \psi)^{\ot k}\mright] \phi^{\ot k}\mright].
  \end{split}
\end{align}
To further simplify this expression, we define the following set:
\begin{align}\label{eq:boundary conditions}
  \mathcal S_{A,\sigma} = \left\{ \{\pi_x\}_{x \in V} : \pi_x \in S_k, \text{ where } \pi_x = \sigma \text{ for } x \in A \text{ and } \pi_x = \id \text{ for } x \in \bar{A} \right\},
\end{align}
for any $\sigma \in S_k$ and $A \subseteq V_\partial$.
An element of $\mathcal S_{A,\sigma}$ assigns a permutation to each vertex in $V$ subject to a `boundary condition.'
Now, using \cref{eq:expectation random state}, we find that
\begin{align*}
  \EE\tr\mleft[\rho_A^k\mright] & = \sum_{\{\pi_x\} \in \mathcal S_{A,\tau}} \tr\mleft[\bigotimes_{x \in V} R_x(\pi_x) \phi^{\ot k}\mright].
\end{align*}
Finally, we observe that for $e = (xy)$
\begin{align*}
  \tr\mleft[R(\pi_x) \ot R(\pi_y) \phi^{\ot k}\mright] = \prod_{l \in C(\pi_x^{-1} \pi_y)} \tr\mleft[\phi_{e,x}^l\mright],
\end{align*}
where we recall that $\phi_{e,x}$ is the reduced density matrix of the link state on edge $e = (xy)$. Thus, we conclude that
\begin{align}\label{eq:replica trick for network}
  \EE\tr\mleft[\rho_A^k\mright] & = \sum_{\{\pi_x\} \in \mathcal S_{A,\tau}} \prod_{e = (xy) \in E} \prod_{l \in C(\pi_x^{-1} \pi_y)} \tr\mleft[\phi_{e,x}^l\mright].
\end{align}
We can interpret the expectation as the partition function of a classical spin model
\begin{align}\label{eq:spin model}
  \EE\tr\mleft[\rho_A^k\mright] & = \sum_{\{\pi_x\} \in \mathcal S_{A,\tau}} 2^{-\sum_{e = (xy) \in E} J_{e}(\pi_x, \pi_y)},
\end{align}
where the site variables in the spin model are permutations $\pi_x \in S_k$, the interaction at the edges between sites is given by
\begin{align*}
  J_{e}(\pi_x,\pi_y) = -\sum_{l \in C(\pi_x^{-1} \pi_y)} \log(\tr\mleft[\phi_{e,x}^l\mright]) = \sum_{l \in C(\pi_x^{-1} \pi_y)} (l-1) H_l(\phi_{e,x}),
\end{align*}
with $H_l$ the $l$-th R\'enyi entropy, and the model as boundary conditions such that the permutation must be $\tau$ on $A$ and $\id$ on $\bar{A}$.
It turns out that $J_e$ is a metric on the symmetric group $S_k$ -- see \cref{sec:metric}.
Similarly, we may place an arbitrary permutation $\pi$ on $A$ instead of $\tau$, which yields (by exactly the same reasoning)
\begin{align}\label{eq:replica trick for network general pi}
  \EE\tr \mleft[R(\pi)\rho^{\ot k}\mright] & = \sum_{\{\pi_x\} \in \mathcal S_{A,\pi}} \prod_{e = (xy) \in E} \prod_{l \in C(\pi_x^{-1} \pi_y)} \tr\mleft[\phi_{e,x}^l\mright].
\end{align}

\subsection{Maximally entangled link states and minimal cuts}\label{sec:ground state configurations}
We will now discuss the special case where all the link states are maximally entangled states of dimension $D$, which has been studied extensively in \cite{hayden2016holographic}.
We will generalize the results we discuss here to a wider class of link states in \cref{sec:almost max entangled}.
In this case, the entanglement spectra of the link states are flat: for~$e \in E$, we have $\lambda_{e,i} = \frac{1}{D}$ for~$i = 1, \ldots, D$. In particular, for all $l \in \NN$ we have~$H_l(\phi_{e}) = \log(D)$ and hence
\begin{align*}
  J_{e}(\pi_x,\pi_y) = \log(D)\sum_{l \in C(\pi_x^{-1} \pi_y)} (l-1).
\end{align*}
This leads to the so-called \textit{Cayley distance} on $S_k$:
\begin{align}\label{eq:cayley}
  d(\pi_x,\pi_y) = \sum_{l \in C(\pi_x^{-1} \pi_y)} (l-1) = k - \abs{C(\pi_x^{-1} \pi_y)},
\end{align}
where $\abs{C(\pi)}$ is the number of cycles in $\pi$.  Moreover, $d(\pi_x,\pi_y)$ is a metric and equals the minimal number of transpositions needed to transform $\pi_x$ into $\pi_y$. We say that $\pi \in S_k$ is on a geodesic between $\pi_1$ and $\pi_2$ if $d(\pi_1, \pi) + d(\pi, \pi_2) = d(\pi_1,\pi_2)$ (recall that $d$ is a metric).
We can rewrite the spin model in terms of this distance:
\begin{align}\label{eq:spin model max ent}
  \EE\tr\mleft[\rho_A^k\mright] & = \sum_{\{\pi_x\} \in \mathcal S_{A,\sigma}} 2^{-\log(D) \sum_{e = (xy) \in E} d(\pi_x, \pi_y)}.
\end{align}
The physically inclined reader may observe that the logarithm of the bond dimension has the role of an inverse temperature, and for large $D$, the dominant contribution to the partition function will be the ground state of the spin model, subject to the relevant boundary conditions.

To describe the dominant contribution to the sum in \cref{eq:spin model max ent} for large~$D$, we need the \emph{minimal cuts} for~$A$ in~$G$.
A cut for~$A$ is a subset of the vertices~$\Gamma_A \subset V$ such that $\Gamma_A \cap V_{\partial} = A$.
Throughout this work, we will denote the set of all cuts for~$A$ by~$C(A)$.
We will use the convention of denoting cuts (i.e. subsets of vertices) by capital Greek letters. Given a cut~$\Gamma_A \in C(A)$, we will denote the set of edges crossing the cut, that is, edges connecting a vertex in $\Gamma_A$ with a vertex in~$V \setminus \Gamma_A$, by lowercase Greek letters~$\gamma_A$ (and by an abuse of language, also refer to this set as a `cut'). A minimal cut for~$A$ is a cut such that the number of edges~$\abs{\gamma_A}$ is minimal.
We write $m(A) = \abs{\gamma_A}$ for a minimal cut $\gamma_A \in C(A)$.
If $\Gamma_A \in C(A)$, we write~$\Gamma_A^c = V \setminus \Gamma_A$.
Note that $\Gamma_A^c$ is a cut for $\bar{A} = V_{\partial} \setminus A$.

In the simplest case, there is a \emph{unique} minimal cut $\gamma_A$.
For this case, one can show that the dominant configuration is the one in which $\pi_x = \tau$ for $x \in \Gamma_A$ and $\pi_x = \id$ for $x \in V \setminus \Gamma_A$, see \cite{hayden2006aspects}, or \cref{prop:convergence of moments}.
That is, there are two domains in the spin model corresponding to $\tau$ and $\id$, and the minimization of the domain wall corresponds to the minimal cut in the graph.

We will also be interested in the case of exactly two non-intersecting minimal cuts~$\Gamma_{A,1}$ and $\Gamma_{A,2}$.
In this case, we have that $\Gamma_{A,1} \subset \Gamma_{A,2}$, or $\Gamma_{A,2} \subset \Gamma_{A,1}$.
After relabeling, we may assume that the first is the case, and define three domains in the graph: $V = V_1 \sqcup V_2 \sqcup V_3$ given by $V_1 = \Gamma_{A,1}$, $V_2 = \Gamma_{A,2} \setminus \Gamma_{A,1}$ and~$V \setminus \Gamma_{A,2}$.
If there are exactly two minimal cuts, then multiple dominant configurations contribute equally to the partition function \cref{eq:spin model max ent}. These dominant configurations can be constructed as follows:
for each $\pi$ on a geodesic between $\tau$ and $\id$, set $\pi_x = \tau$ for $x \in V_1$, $\pi_x = \pi$ for $x \in V_2$ and $\pi_x = \id$ for $x \in V_3$. That these are the dominant configurations follows immediately from the fact that $d(\tau, \pi) + d(\pi, \id) \geq d(\tau, \id)$, with equality if and only if $\pi$ is on a geodesic between $\tau$ and $\id$.

To understand this degeneracy, we use the following fact \cite{nica2006lectures}: \emph{the set of permutations $\pi$ on a geodesic between $\tau$ and $\id$ is in a one-to-one correspondence with the set of non-crossing partitions $NC(k)$ of $[k]$.}
See \cref{sec:free probability} for a definition and properties of $NC(k)$.
Thus, the degeneracy for the the $k$-th moment is~$\abs{NC(k)} = C_k$ where
\begin{align*}
  C_k = \frac{1}{k+1} \binom{2k}{k}
\end{align*}
is the $k$-th Catalan number.
These are the moments of the \emph{Marchenko-Pastur distribution} $\MP(t)$
\begin{align}\label{eq:mp dist}
  \begin{split}
    MP(t) &= \max(1 - t,0) \delta_0 + \nu_t \\
    \d \nu_t(x) &= \frac{\sqrt{4t - (x - 1 - t)^2}}{2\pi x} \ind_{(x - 1 - t)^2 \leq 4t} \d x.
  \end{split}
\end{align}
This allows one to show the folklore result (which we prove and extend to more general link states in \cref{thm:spectrum near max entangled}) that upon an appropriate rescaling, the empirical distribution of the spectrum of $\rho_A$ converges to a Marchenko-Pastur distribution.
This is in line with the case of a single random tensor, which precisely yields a Wishart matrix (see \cref{sec:rmt} for a brief introduction to these objects).
In the first case, where there is a unique minimal cut, the entanglement spectrum of $\rho_A$ is flat, while, as we have seen, in the second case, the degeneracy gives rise to a nontrivial spectrum in the right scaling limit.

\subsection{The replica trick for general background states}\label{sec:background states}
In \cref{eq:spin model}, we computed the result of the replica trick for the $k$-th moment for a random tensor network state.
We will also consider the more general setting where the link state is replaced by some arbitrary state~$\phi_V$.
In this setting, there need not be a graph structure, and the Hilbert space at each vertex $x \in V$ can be some arbitrary Hilbert space, rather than a tensor product of Hilbert spaces labelled by half-edges.
We will refer to $\phi_V$ as a ``background state'' instead of a ``link state'' (as the interpretation of links along the edges does not necessarily make sense in this situation).
That is, where before we had a link state
\begin{align*}
  \ket{\phi} = \bigotimes_{e \in E} \ket{\phi_e},
\end{align*}
we will now consider some \emph{arbitrary} possibly mixed and subnormalized $\phi_V \in \Pleq(V)$ in the tensor network construction. We can generalize \cref{eq:rtn state 2} to also apply for general background states to obtain a state~$\rho \in \PSD(V_\partial)$ given by
\begin{align}\label{eq:rtn background state}
  \rho
  = \tr_{V_b}\mleft[\left( I_{V_{\partial}} \ot \psi \right) \phi_V\mright]
\end{align}
where $\ket{\psi}$ is a tensor product of random states at the bulk vertices.
If $\phi$ is pure, then so is $\rho$, since in that case
\begin{align*}
  \ket{\rho} = \left(I_{V_{\partial}} \ot \bra{\psi}\right) \ket{\phi}.
\end{align*}
If $\phi$ is not pure, we can consider a purification $\phi_{VR} \in \Pleq(VR)$ and consider $R$ as an additional boundary system; this leads to a random tensor network state $\rho_{V_{\partial}R}$ which is a purification of $\rho_{V_\partial}$.
This set-up is illustrated in \cref{fig:general background}.
While formally very similar, the resulting state is no longer a PEPS tensor network state in general.

\begin{figure}
  \centering
  \centering
  \begin{overpic}[width=.5\linewidth,grid=false]{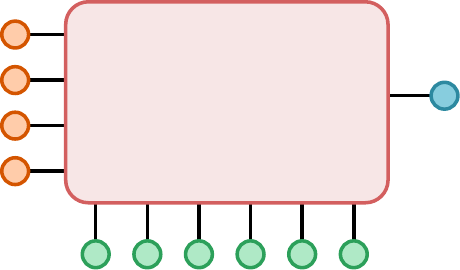}
    \put(47,35){\color{BrickRed}{\Large{$\ket{\phi}$}}}
    \put(-12,50){\color{Bittersweet}{$\bra{\psi_{x_1}}$}}
    \put(-12,40){\color{Bittersweet}{$\bra{\psi_{x_2}}$}}
    \put(-12,30){\color{Bittersweet}{$\bra{\psi_{x_3}}$}}
    \put(-12,20){\color{Bittersweet}{$\bra{\psi_{x_4}}$}}
    \put(-20,35){\color{Bittersweet}{\Large{$V_b$}}}
    \put(45,-7){\color{ForestGreen}{\Large{$V_{\partial}$}}}
    \put(103,36){\color{NavyBlue}{\Large{$R$}}}
  \end{overpic}
  \vspace*{0.5cm}
  \caption{The structure of a (purified) random tensor network with a general background state.}
  \label{fig:general background}
\end{figure}
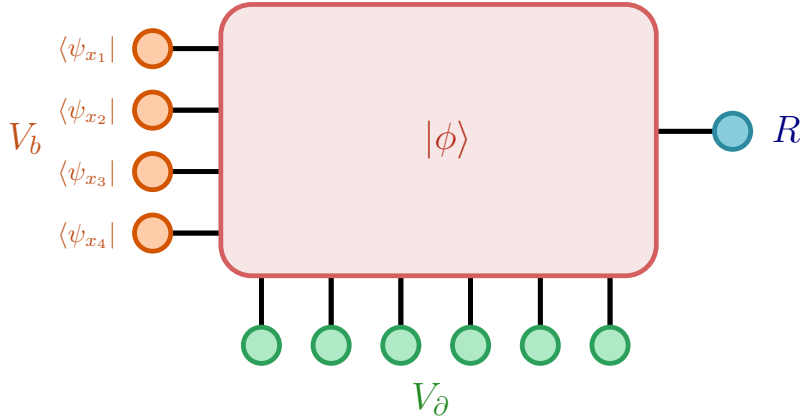

There are multiple reasons to also allow general background states.
The first reason is of a technical nature: they are useful for estimates based on \textit{smooth entropies}, which we discuss in \cref{sec:far from max entangled}.
In this application, the full state on edges is still pure, but is no longer a tensor product of link states along the edges.
A second motivation for considering general background states is that they can be used as a toy model for holographic systems where there is ``bulk entropy'' present.
Finally, these states are closely related to protocols for the quantum information processing task of split transfer \cite{dutil2010one}.
We comment on this connection in \cref{sec:split_recovery}.

Even for a general background state, a version of the replica trick still applies.
Consider a boundary subsystem $A \subseteq V_{\partial}$ with corresponding boundary state $\rho_A$.
Then, the computation in \cref{eq:replica trick computation} is still valid, and we find
\begin{align}\label{eq:replica two}
  \EE \tr[\rho_A^k] & = \sum_{\{\pi_x\} \in \mathcal S_{A,\tau}} \tr_V\mleft[\bigotimes_{x \in V} R_x(\pi_x) \phi_V^{\ot k}\mright]
\end{align}
where $\tau = (1 2 \ldots k)$.
However, \cref{eq:replica two} no longer has the interpretation of a spin model with local interactions.

For general background states, we will only need the replica trick for $k = 2$.
Since $S_2$ has only two elements, each configuration of permutations is completely characterized by the domain $\Delta_A = \{x \in V \text{ such that } \pi_x = \tau\}$.
Because of the boundary conditions in $\mathcal S_{A,\tau}$, the collection of these sets coincides with $C(A)$,  and hence
\begin{align}\label{eq:general replica trick k=2}
  \EE \tr[\rho_A^2] = \sum_{\Delta_A \in C(A)} \tr[\phi_{\Delta_A}^2] = \sum_{\Delta_A \in C(A)} \tr[\phi] 2^{-H_2(\Delta_A)_\phi} .
\end{align}
Another useful fact is that by \cref{eq:expectation random state},
\begin{align}\label{eq:expectation of rtn state}
  \EE \rho_{V_\partial} = \phi_{V_\partial}.
\end{align}

We remark that if one only uses the $k = 2$ replica trick, one could also use tensors which are drawn from a \textit{projective 2-design}, a distribution which produces tensors with the same first and second moments as uniformly random tensors of unit norm \cite{klappenecker2005mutually,gross2007evenly}. An example of a projective 2-design is the set of uniformly random stabilizer states.
For tensors $\ket{\psi_x}$ drawn from a projective 2-design of dimension $D_x$, it holds that
\begin{align*}
  \EE \psi_x^{\ot 2} = \frac{1}{D_x(D_x - 1)}I + \frac{1}{D_x(D_x - 1)}R_x(\tau),
\end{align*}
and hence
\begin{align*}
  \EE \tr[\rho_A^2] = \frac{D_x}{D_x + 1} \sum_{\Delta_A \in C(A)} \tr[\phi_{\Delta_A}^2],
\end{align*}
which is close to \cref{eq:general replica trick k=2} for large $D_x$.
Thus, it is not hard to see that all random tensor network results which only use the $k = 2$ replica trick are also valid for states with tensors drawn from projective 2-designs.
This was already observed in~\cite{hayden2016holographic} and random tensor networks with random stabilizer tensors were further studied in~\cite{nezami2020multipartite}.
The results of \cref{sec:far from max entangled} only use the $k=2$ replica trick, and thus will extend to states with tensors drawn from projective 2-designs.
This will not be true for the results in \cref{sec:almost max entangled}, which requires usage of the replica trick for all $k \in \NN$.

\subsubsection{Normalization of random tensor network states}\label{sec:normalization}

One immediate consequence of the replica trick for $k=2$ is that the random tensor network state $\rho$ will be approximately normalized with high probability, so long as a mild condition on the background state is satisfied: the bulk needs to be connected, with sufficiently entangled edges. Let
\begin{align}\label{eq:delta normalization}
  \eta = \max_{\Delta \subseteq V_b, \Delta \neq \emptyset} \tr[\phi_{\Delta}^2] = \max_{\Delta \subseteq V_b, \Delta \neq \emptyset} \tr[\phi] 2^{-H_2(\Delta)_\phi}.
\end{align}
If the state has enough correlations along each cut (or more precisely, if $H_2(\Delta)_\phi$ is large for each $\Delta$), then $\eta$ is small.
Concretely, if we consider a random tensor network state with maximally entangled link states of bond dimension $D$, we will have $\eta \leq \frac{1}{D}$.
We then have

\begin{lem}\label{lem:normalization}
  For any background state $\phi \in \Pleq(V)$, with associated $\rho \in \PSD(V_\partial)$ as in \cref{eq:rtn background state}, it holds that for any $\eps > 0$
  \begin{align*}
    \Pr\bigl(\abs{\tr[\rho] - \tr[\phi]} \geq \eps\bigr) \leq 2^{|V_b|}\frac{\eta}{\eps^2}
  \end{align*}
  where $\eta$ is defined in \cref{eq:delta normalization}.
\end{lem}
\begin{proof}
  This follows from a special case of \cref{eq:general replica trick k=2}.
  In this case, the empty cut contributes $\tr[\phi]^2,$ so we find
  \begin{align*}
    \Var(\tr[\rho]) = \EE \abs{\tr[\rho] - \tr[\phi]}^2 = \EE \abs{\tr[\rho]^2 - \tr[\phi]^2}  \leq 2^{V_b} \max_{\Delta \subseteq V_b, \Delta \neq \emptyset} \tr[\phi_{\Delta}^2],
  \end{align*}
  where we have used the normalization of $\rho$ in expectation $\EE\tr[\rho] = \tr[\phi]$, as in \cref{eq:mean normalization}.
  The result follows by an application of Chebyshev's inequality.
\end{proof}

We can improve this result by taking advantage of the fact that our random projectors are random Gaussian vectors, allowing us to use Gaussian concentration of measure rather than the Chebyshev's inequality. For instance, using a concentration bound for Gaussian polynomials (see \cite{aubrun2017alice}, Corollary 5.49) one can show that for any $\eps \geq (\sqrt{2}e)^{2V_b} \eta$:
\begin{align*}
  \Pr\bigl(\abs{\tr[\rho] - \tr[\phi]} \geq \eps\bigr) \leq \exp\left(-\frac{\abs{V_b}}{2e}\eps^{\frac{1}{\abs{V_b}}}\eta^{-\frac{1}{\abs{V_b}}}\right),
\end{align*}
where $\eta$ is defined as in \cref{eq:delta normalization}. We will not need this refinement.

\section{Link states with bounded spectral variation}\label{sec:almost max entangled}
In this section, we study random tensor network states with link states that have \textit{bounded spectral variation}, meaning that there is an effective bond dimension $D$ such that the Schmidt coefficients of the link state are of the order $\tfrac{1}{D}$.

We start by providing background material on random matrix theory and \emph{free probability}, which is a key tool in the study of products of random matrices.
In \cref{sec:bounded spectral variation and free product}, we will precisely define the notion of bounded spectral variation and generalize the results in \cref{sec:ground state configurations} for random tensor network states with maximally entangled link states to this wider class of link states.
This leads to the main result of this section, \cref{thm:spectrum near max entangled}, which shows that the asymptotic entanglement spectrum can be expressed in terms of a free product of distributions.
We will see that the results are similar to the quantum gravity set-up described in \cref{sec:jt gravity}.
Finally, in \cref{sec:negativity}, we investigate the entanglement negativity for random tensor network states with link states of bounded spectral variation.

\subsection{Random matrices, free probability and non-crossing partitions}\label{sec:free probability}

\subsubsection{Random matrix theory and Wishart matrices}\label{sec:rmt}
We start by reviewing relevant concepts from probability and random matrix theory that are relevant for our analyses. This material can be found in any introduction to random matrix theory, e.g. \cite{anderson2010introduction, bai2010spectral, potters2020first}.

A fundamental question in random matrix theory is as follows: given a family of $n \times n$ matrices with entries selected according to some distribution, what is the asymptotic distribution of the eigenvalues as~$n \to \infty$?
This question has been extensively studied, and in many cases has an elegant and concise answer.
We discuss a basic example which is closely related to our purposes: \textit{Wishart matrices}.
Consider an $n \times m$ matrix $X$ whose entries are drawn i.i.d. from a Gaussian distribution with mean zero and unit variance. The sample covariance matrix of $X$ is the $n\times n$ matrix defined as
\begin{equation}
  Y_{n,m} = \frac{1}{m}XX^T.
\end{equation}
Such random matrices are called \textit{(real) Wishart matrices}, and can be thought of as a sample second moment matrix (where one has $m$ realizations of an $n$-dimensional random variable).
One can also consider \textit{complex} Wishart matrices: in this case the entries of the $n \times m$ matrix $X$ are complex i.i.d. standard (circularly symmetric) complex Gaussian random variables.
We then let $Y_{n,m} = \frac{1}{m}XX^\dagger$.
We would like to understand the spectrum of $Y_{n,m}$, and to that end, we consider the empirical distribution of the eigenvalues.
This empirical distribution is itself random, depending on the particular realization of $Y_{n,m}$.
To characterize the convergence, we recall that if $\{\mu_n\}_{n \in \NN}$ is a sequence of random finite measures on $\RR$, we say that the sequence $\mu_n$ \emph{converges weakly, in probability,} to a finite measure $\mu$, if, for any bounded continuous function~$f \in C_b(\RR)$, it holds that for every $\eps > 0$
\begin{align*}
  \lim_{n \to \infty} \Pr\left( \abs{\int f(x) \d\mu_n(x) - \int f(x) \d\mu(x) } \geq \eps\right) = 0.
\end{align*}

The asymptotic distribution of the eigenvalues of Wishart matrices is known to obey the \textit{Marchenko-Pastur law} (see, for instance, Theorem 3.6 and Theorem 3.7 in \cite{bai2010spectral}):
\begin{thm}\label{thm:wishart}
  Consider (real or complex) Wishart matrices $Y_{n,m}$ and let
  \begin{align*}
    \mu_{n,m} = \frac{1}{n}\sum_{\lambda \in \spec(Y_{n,m})}\delta_{\lambda}
  \end{align*}
  be the empirical distribution of its eigenvalue spectrum.
  Suppose that the ratio of dimensions $n/m$ converges to a constant $t > 0$ as $n \to \infty$. Then $\mu_{n,m}$ converges weakly, in probability, to the Marchenko-Pastur distribution $\MP(t)$ with parameter $t > 0$, as defined in \cref{eq:mp dist}.
\end{thm}

Generalizations to this result are possible. For example, one still has convergence if the entries of $X$ are chosen according to non-Gaussian distributions with mean zero and unit variance.
Also, one can prove weak convergence, almost surely (rather than just in probability); see~\cite{bai2010spectral}.

If $Y_{n,m} = \frac{1}{m}XX^\dagger$ is a complex Wishart matrix, $X$ can also be interpreted as a uniformly random pure quantum state on $\CC^n \ot \CC^m$, and $Y_{n,m}$, up to normalization, as the reduced density matrix on $\CC^n$ \cite{hayden2006aspects}.
Note that $\frac{1}{n}Y_{n,m}$ is normalized in expectation in the sense that $\EE \frac{1}{n} Y_{n,m} = 1$.
So, complex Wishart matrices can be used as a model for the reduced state of a random bipartite quantum state, and this allows one to quantify the `typical entanglement' of a random state.
Equivalently, in the tensor network setting, $\frac{1}{\sqrt{n}}X$ can be thought of as a random tensor network state with a single bulk vertex, two boundary vertices, and maximally entangled link states.
We can then can interpret $\frac{1}{n}Y_{n,m}$ as the reduced density matrix on one of the boundary vertices.
We will provide a generalization of \cref{thm:wishart} for the entanglement spectrum of random tensor network states in \cref{thm:spectrum near max entangled}.

\subsubsection{Free probability}
The topic of probability distributions in random matrix theory is closely related to \emph{free probability} and, in particular, to the notion of the \emph{free product}. We provide a brief introduction here; the material in this section is very standard, and we only review a few relevant aspects. For an extensive treatment, see, for instance, Chapter 5 in \cite{anderson2010introduction} or the books \cite{nica2006lectures,mingo2017free,potters2020first}. As we will see later, the free product will allow us to concisely formulate replica trick results involving multiple minimal cuts.


A \emph{non-commutative probability space} is a pair $(\A, \omega)$, where $\A$ is a $C^*$-algebra and $\omega$ is a state on $\A$.
An element $a \in \A$ is called a non-commutative random variable.
The key example to have in mind is the space of~$n \times n$ random matrices, where the matrix entries are distributed according to some probability distribution, and $\omega(a) = \EE \frac{1}{n}\tr[a]$ defines a tracial state.
If $a \in \A$, the \emph{distribution} (or \textit{law}) $\mu_a$ of $a$ is defined as a map on polynomials, which evaluates on a polynomial $p$ as $\mu_a(p) = \omega(p(a))$.
If $a$ is self-adjoint, it has real spectrum and we can extend the domain of $\mu_a$ to all bounded continuous functions $f \in C_b(\RR)$, using the functional calculus to define $f(a)$ and letting $\mu_a(f) = \omega(f(a))$.
In this case we can identify $\mu_a$ with a distribution such that, for $f \in C_b(\RR)$, we have $\mu_a(f) = \int f(x) \d \mu_a(x)$.
In particular, if $a$ is an~$n \times n$ self-adjoint random matrix, then $\mu_a(f) = \frac{1}{n}\EE \sum_{\lambda \in \spec(a)} f(\lambda)$, and we may identify $\mu_a$ with the empirical measure of the eigenvalues of $a$.
If $\A$ is a commutative algebra, these notions reduce to the usual notions of probability theory, where $\omega$ is the expectation.

We call a set of $n$ non-commutative random variables $\{a_i\}$ on a non-commutative probability space $(\A,\omega)$ \textit{freely independent} or just \emph{free} if, for any set of $k \geq 2$ polynomials $\{p_j\}$, the variables satisfy
\begin{align*}
  \omega(p_1(a_{i_1})\ldots p_k(a_{i_k})) = 0
\end{align*}
whenever $\omega(p_m(a_{i_m})) = 0$ for all $1 \leq m \leq k$ and no two adjacent indices $i_m$ and $i_{m+1}$ for  $1 \leq m \leq k-1$ are equal. One can see that two freely independent variables $a_1,a_2$ satisfy:
\begin{equation}
  0 = \omega((a_1-\omega(a_1))(a_2-\omega(a_2))) = \omega(a_1a_2) - \omega(a_1)\omega(a_2),
\end{equation}
which, in the commutative case with random variables $x_1$, $x_2$, is the classical bivariate independence condition~$\EE[x_1x_2] = \EE[x_1]\EE[x_2]$. The definition of free independence does not specialize to independence in the commutative case (commuting independent random variables are only free when they are constant). However, the role of free independence is analogous to the role of classical independence for commuting random variables: it allows one to, in principle, compute the joint mixed moments of the variables.

We will be interested in the \textit{multiplicative free convolution} or \textit{free product} (there also exists an additive convolution or just \textit{free convolution}) of distributions.
Suppose $a,b$ are non-commutative self-adjoint free random variables on $(\A,\omega)$ with distributions $\mu_a$ and $\mu_b$.
Then we denote the distribution of $ab$ by~$\mu_{ab} = \mu_a \boxtimes \mu_b$.
Note that, generally, $ab$ need not be self-adjoint.
However, if $\omega$ is tracial (as in the random matrix case) and $a$ is positive, the distribution of $ab$ coincides with that of $\sqrt{a}b\sqrt{a}$ which is self-adjoint, and we can identify~$\mu_{ab}$ with a distribution on $\RR$.
If $\mu_a$ and $\mu_b$ are compactly supported distributions, then so is~$\mu_a \boxtimes \mu_b$.

As a concrete example of the freeness and the free product, let $X_n$ and $Y_n$ be two families of random~$n \times n$ positive diagonal matrices with uniformly bounded norm, such that their spectrum converges weakly to probability distributions $\mu$ and $\nu$ respectively. Let $U_n$ be a family of Haar random unitary $n \times n$ matrices. Then as $n$ goes to infinity, $X_n$ and $Y_n' = U_nY_nU_n^\dagger$ will be freely independent (so they are \emph{asymptotically free}), and we would like to study their product. The product of positive matrices need not be self-adjoint, so we consider $Z_n = \sqrt{X_n} Y_n' \sqrt{X_n}$ which is a positive matrix. One may then show that the distribution of the spectrum of $Z_n$ weakly converges in probability to $\mu \boxtimes \nu$. See Corollary 5.4.11 in \cite{anderson2010introduction} for a precise statement and proof.

The free product may be analyzed using generating functions: given a (non-commutative) random variable~$a$ with distribution $\mu_a$, let
\begin{align*}
  m_{a,k} = \int x^k \d \mu_a(x)
\end{align*}
be the $k$-th moment of $\mu_a$. Then the \emph{moment generating function} is the formal power series
\begin{align}\label{eq:moment generating function}
  M_{\mu_a}(z) = \sum_{k=1}^\infty m_{a,k} z^k.
\end{align}
We define the \emph{S-transform} to be the formal power series
\begin{align*}
  S_{\mu_a}(z) = \frac{1 + z}{z}M_{a}^{-1}(z),
\end{align*}
where $M_{a}^{-1}(z)$ is the power series corresponding to the formal inverse of $M_{\mu}(z)$ under composition, which is well-defined as long as $m_{a,1} \neq 0$.
For compactly supported distributions, the moment generating function, and hence the S-transform, uniquely determines the distribution.

If $a$ and $b$ are non-commutative self-adjoint free random variables, then
\begin{align}\label{eq:s transform}
  S_{\mu_{a} \boxtimes \mu_{b}}(z) = S_{\mu_{ab}}(z) = S_{\mu_a}(z)S_{\mu_b}(z).
\end{align}
This also provides a completely combinatorial interpretation of the free product, without reference to the associated non-commutative probability spaces.
That is, given compactly supported distributions~$\mu$ and~$\nu$, we can define $\mu \boxtimes \nu$ by \cref{eq:s transform}: it is the compactly supported distribution with moments prescribed by~$S_{\mu_{a} \boxtimes \mu_{b}}(z)$, and hence, $M_{\mu_{a} \boxtimes \mu_{b}}(z)$.
The free product is commutative and associative.

As an example, we compute the S-transform of the Marchenko-Pastur distribution $\mu \sim \MP(1)$.
The distribution is given by
\begin{align*}
  \d \mu(x) = \frac{1}{2\pi}\sqrt{4x^{-1} - 1} \d x.
\end{align*}
The moments can be computed directly:
\begin{equation}
  m_k = \sum_{i=0}^{k-1}\frac{1}{i+1}\binom{k}{i}\binom{k-1}{i}
\end{equation}
After some work, one can show that the moments above lead to a closed-form moment generating function
\begin{align*}
  M(z) = \frac{2z - 1 - \sqrt{1 - 4z}}{2z}.
\end{align*}
One may then invert the expression and obtain the S-transform
\begin{align*}
  S(z) = \frac{1}{1+z}.
\end{align*}
Similarly, for the Marchenko-Pastur distribution $MP(t)$ with parameter $t$, which has distribution as given in \cref{eq:mp dist}, we find that
\begin{align*}
  S(z) = \frac{1}{t+z}.
\end{align*}
See, for instance, \cite{banica2011free}.

\subsubsection{Non-crossing partitions}
Given $k \in \NN$, let $NC(k)$ denote the set of \emph{non-crossing partitions} of $[k]$.
A non-crossing partition of~$[k]$ is a partition $[k] = X_1 \sqcup \ldots \sqcup X_m$ which is such that, if $i < j \in X_\alpha$, then there are no $k, l \in X_{\beta}$ for~$\beta \neq \alpha$ with~$k < i < l < j$ or $i < k < j < l$. To any non-crossing partition, we associate a permutation $\pi \in S_k$ by mapping each subset $\{i_1, \ldots, i_l\}$ to the cycle $(i_1,\ldots, i_l)$ with $i_1 < \ldots < i_l$.
In a slight abuse of notation, we will write $\pi \in NC(k)$.
For any $\pi \in S_k$, and for a sequence of numbers $f_k$ for $k = 1,2, \ldots$, we write
\begin{align}\label{eq:permutation notation}
  f_{\pi} = \prod_{l \in C(\pi)} f_l
\end{align}
where $C(\pi)$ is the cycle type of $\pi$. We will need the following result, which is a straightforward consequence of the combinatorics of the S-transform.

\begin{thm}\label{thm:free product}
  Consider compactly supported probability distributions $\mu, \nu, \eta$.
  Suppose that the moments of $\eta$ are given by
  \begin{align*}
    m_k^\eta = \sum_{\pi \in NC(k)} m^\mu_{\pi} m^\nu_{\pi^{-1} \tau_k}
  \end{align*}
  where $\tau_k = (1 2 \ldots k)$ is the full cycle. Then
  \begin{align*}
    \eta = \MP(1) \boxtimes \mu \boxtimes \nu.
  \end{align*}
\end{thm}

\begin{proof}
  We let $\mathcal F$ be the transformation that sends a formal power series $f(z)$ to the power series $\frac{1}{z}f^{-1}(z)$.
  This is such that for some distribution $\mu$, the S-transform is given by $S_{\mu}(z) = (1 + z) \mathcal F(M_{\mu})(z)$.
  Moreover, given power series $f(z) = \sum_k f_k z^k$ and $g(z) = \sum_k g_k z^k$, define a convolution operation $\circledast$ by
  \begin{align*}
    (f \circledast g) (z) = \sum_k \left(\sum_{\pi \in NC(k)} f_{\pi} g_{\pi^{-1} \tau_k}\right)z^k
  \end{align*}
  where $\tau_k$ is the full cycle in $S_k$.
  Then Theorem 18.14 in \cite{nica2006lectures} states that for any two $f$ and $g$ with~$f_1 \neq 0$ and $g_1\neq 0$, it holds that
  \begin{align*}
    \mathcal F(f \circledast g)(z) = \mathcal F(f)(z)\mathcal F(g)(z).
  \end{align*}
  Then the S-transform of $\eta$ can be written:
  \begin{align*}
    S_\eta(z) = (1 + z) \mathcal F(M_\eta)(z) = (1 + z) \mathcal F(M_\mu)(z) \mathcal F(M_\nu)(z) = \frac{1}{1 + z} S_\mu(z)S_\nu(z).
  \end{align*}
  This implies the desired result, as the S-transform of $\MP(1)$ is given by $\frac{1}{1+z}$, and the S-transform uniquely determines a compactly supported distribution.
\end{proof}

We remark briefly that free independence can equivalently be formulated in terms of the vanishing of \textit{free cumulants}, which are themselves defined in terms of sums over non-crossing partitions. We refer the interested reader to any of the previously cited references for a more in-depth discussion on the role of non-crossing partitions in free probability. For our purposes, the fact that non-crossing partitions are intimately related to free independence will allow us to later phrase random tensor network results in terms of free probability.

\subsection{Entanglement spectrum of random tensor network states as a free product}\label{sec:bounded spectral variation and free product}
We now return to studying random tensor network states.
Consider a family of states in $\Peq(V)$ composed of the tensor product of link states $\ket{\phi_e}$ along the edges $e \in E$ as in \cref{eq:link state edge}, and assume that along each edge, the bond dimensions scale with a parameter $D$, so $D_e = \Theta(D)$.
Our key assumption is that the link states have \emph{bounded spectral variation} -- by this we mean that the empirical distribution of the rescaled entanglement spectrum of the link states
\begin{align}\label{eq:link states near max entangled}
  \mu^{(D)}_{e} := \sum_{i=1}^{D_e} \frac{1}{D_e} \delta_{D_e\lambda_{e,i}}
\end{align}
has all moments converging to the moments $m_{e,k}$ of a compactly supported probability distribution $\mu_{e}$, as $D$ goes to infinity.
We assume that the link states are normalized, so $m_{e,1} = 1$.
This condition implies that, up to a vanishing fraction as $D\rightarrow\infty$, the elements of the entanglement spectrum of the link state are of order~$D^{-1}$.

For a minimal cut $\gamma_A$, let $\mu^{(D)}_{\gamma_A}$ be the distribution for the spectrum of the tensor product of the link states in $\gamma_A$:
\begin{align*}
  \mu^{(D)}_{\gamma_A} = \bigotimes_{e \in \gamma_A} \mu^{(D)}_{e} = \frac{1}{D_{\gamma_A}}\sum_{\{i_{e}\}} \delta_{D_{\gamma_A} \prod_{e \in \gamma_A}\lambda_{e,i_e}}
\end{align*}
where $i_e = 1,\ldots,D_e$ and $D_{\gamma_A} = \prod_{e \in \gamma_A} D_e$.
We define the tensor product of distributions as follows: if $X_1$ and $X_2$ are independent real valued random variables with distributions $\mu_{X_1}$ and $\mu_{X_2}$, then $\mu_{X_1} \ot \mu_{X_2}$ is defined as the joint distribution of $(X_1, X_2)$. The distribution $\mu^{(D)}_{\gamma_A}$ has $k$-th moment given by $m^{(D)}_{\gamma_A,k} = \prod_{e \in \gamma_A} m^{(D)}_{e,k}$, and we can see that~$m^{(D)}_{\gamma_A,k}$ converges to~$m_{\gamma_A,k}$, the moments of the distribution
\begin{align*}
  \mu_{\gamma_A} := \bigotimes_{e \in \gamma_A} \mu_{e}.
\end{align*}

Let $\spec(\rho_A) = \{\lambda_{A,i}\}$ (recall that $\spec(\rho_A)$ is ordered in non-increasing order).
Let $\gamma_A$ be a cut for $A$.
By a standard argument, the number of nonzero eigenvalues of $\rho_A$ (that is, $\rank(\rho_A)$) is upper bounded by~$D_{\gamma_A}$.
If~$\gamma_A$ is the unique minimal cut, then we define
\begin{align}\label{eq:empirical measure boundary}
  \mu^{(D)}_{A} := \frac{1}{D_{\gamma_A}} \sum_{i=1}^{D_{\gamma_A}} \delta_{D_{\gamma_A} \lambda_{A,i}}.
\end{align}
If there are multiple minimal cuts, it is ambiguous which $\gamma_A$, and hence, which $D_{\gamma_A}$, we should pick; we choose the cut for which $D_{\gamma_A}$ is minimal in \cref{eq:empirical measure boundary}, and we will denote this minimal cut by $\gamma_{A,1}$.
The moments of~$\mu^{(D)}_{A}$ are given by
\begin{align*}
  m^{(D)}_{A,k} := \int z^k \d \mu_{A}(z) = D_{\gamma_A}^{k-1} \sum_{i=1}^{D_{\gamma_A}} (\lambda_{A,i})^k.
\end{align*}
Note that the distribution $\mu^{(D)}_A$ is random, and correspondingly, the moments~$m^{(D)}_{A,k}$ are random variables. In contrast, the moments~$m^{(D)}_{e,k}$ and~$m^{(D)}_{\gamma_A,k}$ are numbers depending only on the bond dimension.

The theorem we want to prove will follow straightforwardly from a key intermediate result: as $D$ goes to infinity, all the moments of the boundary distribution $\mu_A^{(D)}$ converge to the moments of $\mu_{\gamma_A}$.
We use the notation in \cref{eq:permutation notation} to write expressions like
\begin{align*}
  m_{\gamma_A,\pi} = \prod_{l \in C(\pi)} m_{\gamma_A,l}
\end{align*}
for a permutation $\pi \in S_k$. We will then apply the \textit{method of moments} to show that convergence of moments implies convergence in distribution.
As a remark on notation, in the error bounds in both the current section and \cref{sec:far from max entangled}, when we use $\bigO$-notation, the constants may depend on the graph underlying the tensor network (typically our bounds scale as $2^{\abs{V_b}}$, where $V_b$ is the set of bulk vertices).

\begin{prop}\label{prop:convergence of moments}
  If there exists a unique minimal cut $\gamma_A$ for $A$, then
  \begin{align}\label{eq:moments single surface}
    \lim_{D \rightarrow \infty} \EE m^{(D)}_{A,k} = m_{\gamma_A,k}.
  \end{align}
  If there exist exactly two minimal cuts $\gamma_{A,1}$ and $\gamma_{A,2}$, which do not intersect (so $\gamma_{A,1} \cap \gamma_{A,2} = \emptyset$) and for which $\frac{D_{\gamma_{A,1}}}{D_{\gamma_{A,2}}}$ converges to a constant $t \leq 1$, then
  \begin{align}\label{eq:moments two surfaces}
    \lim_{D \rightarrow \infty} \EE m^{(D)}_{A,k} = \sum_{\pi \in NC(k)} t^{d(\pi,\id)} m_{\gamma_{A,1},\tau^{-1}\pi}  m_{\gamma_{A,2},\pi}.
  \end{align}
  Moreover, in both cases the variance goes to zero as $D$ goes to infinity: for every $k$
  \begin{align}\label{eq:variance moments}
    \EE\mleft[\left(m^{(D)}_{A,k} - \EE \mleft[m^{(D)}_{A,k}\mright]\right)^2\mright] = \bigO\left(\frac{1}{D}\right).
  \end{align}
\end{prop}

\begin{proof}
  We first provide a sketch of the proof. It proceeds via the following steps:
  \begin{enumerate}
    \item Write the expectation of the moments of $\mu_A^{(D)}$ as the partition function for a classical spin model, as in \cref{sec:replica trick}.
    \item Show that the contributions from terms of the form given in the statement of the proposition dominate the partition function by carefully tracking the powers of $D$, and showing that all other contributions are suppressed polynomially in $D$.
    \item Show that the variance of the moments vanishes in the limit $D \to \infty$ by direct computation.
  \end{enumerate}

  We begin with Step 1.
  First, we observe that the $k$-th moment of $\mu_A^{(D)}$ is given by $m^{(D)}_{A,k} = D_{\gamma_A}^{k-1} \tr\mleft[\rho_A^k\mright]$.
  Consider the expression in \cref{eq:replica trick for network general pi} for the replica trick with permutation $\pi$ on $A$:
  \begin{align}
    Z_{k, \pi} := \EE \tr\mleft[R_A(\pi) \rho_A^{\ot k}\mright] =  \sum_{\{\pi_x\} \in \mathcal S_{A,\pi}} \prod_{e = (xy) \in E} \prod_{l \in C(\pi_x^{-1} \pi_y)} \tr\mleft[\phi_{e,x}^l\mright].
  \end{align}
  Recall that the set $\mathcal S_{A,\pi}$, as defined in \cref{eq:boundary conditions}, consists of assignments of permutations to each $x \in V$, subject to $\pi_x = \pi$ for $x \in A$ and $\pi_x = \id$ for $x \in \bar{A}$.
  As in \cref{eq:replica trick for network}, if $\pi = \tau$, then we indeed have~$Z_{k, \tau} = \EE\tr\mleft[\rho_A^k\mright]$, so
  \begin{align}\label{eq:Z is moment}
    \EE m^{(D)}_{A,k} = D_{\gamma_A}^{k-1}Z_{k, \tau}.
  \end{align}
  On the other hand, if we let $k = 2n$ and $\pi = \tilde{\tau} = (1 2 \ldots n)(n+1 \, n+2 \ldots 2n)$, then $Z_{k, \pi} = \EE\mleft[\tr\mleft[\rho_A^n\mright]^2\mright]$, and hence
  \begin{align}\label{eq:variance moments Z}
    \EE \mleft[\left(m^{(D)}_{A,n}\right)^2\mright] = D_{\gamma_A}^{2n-2}Z_{k, \tilde \tau}.
  \end{align}
  Recall that $m^{(D)}_{e,l} = D_e^{l-1}\tr\mleft[\phi_{e,x}^l\mright]$, and write
  \begin{align*}
    Z_{k, \pi} = \sum _{\{\pi_x\} \in \mathcal S_{A,\pi}} Z_k(\{\pi_x\})
  \end{align*}
  where
  \begin{align}\label{eq:Z pi}
    \begin{split}
      Z_{k}(\{\pi_x\}) &:= \prod_{e = (xy) \in E} \prod_{l \in C(\pi_x^{-1} \pi_y)} \tr\mleft[\phi_{e,x}^l\mright]\\
      &= \prod_{e = (xy) \in E} D_e^{\abs{C(\pi_x^{-1} \pi_y)} - k}\prod_{l \in C(\pi_x^{-1} \pi_y)} m^{(D)}_{e,l} \\
      &= \prod_{e = (xy) \in E} D_e^{-d(\pi_x,\pi_y)} m^{(D)}_{e,\pi_x^{-1} \pi_y}.
    \end{split}
  \end{align}
  This accomplishes Step 1: we have recast the problem of computing moments into a question of computing a partition function for a classical spin model with fixed boundary conditions.

  For Step 2, we want to show that the dominant contribution(s) to $Z_{k, \pi}$ as $D$ goes to infinity are those given in the statement of the proposition. This will simply be a matter of checking powers of $D$, and using the triangle inequality property of the Cayley distance. If $\Gamma_{A}$ is the unique minimal cut, then we let $\pi^{\min}_x = \pi$ for $x \in \Gamma_A$ and $\pi^{\min}_x = \id$ for $x \in V \setminus \Gamma_A$, and we have
  \begin{align}\label{eq:ground state unique cut}
    Z_k(\{\pi^{\min}_x\}) & = D_{\gamma_A}^{-d(\pi,\id)} \prod_{e \in \gamma_A}  m^{(D)}_{e,\pi} = D_{\gamma_A}^{-d(\pi,\id)} m^{(D)}_{\gamma_A,\pi}.
  \end{align}
  If there are exactly two minimal cuts $\Gamma_{A,1} \subset \Gamma_{A,2}$, we let~$V = V_1 \sqcup V_2 \sqcup V_3$, with~$V_1 = \Gamma_{A,1}$, $V_2 = \Gamma_{A,2} \cap \Gamma_{A,1}^c$ and~$V_3 = \Gamma_{A,2}^c$. Now consider the permutations $\sigma \in S_k$ that are on a geodesic between~$\pi$ and~$\id$ (recall this implies $d(\pi,\sigma) + d(\sigma, \id) = d(\pi,\id)$), and consider the configuration given by~$\pi^{\sigma}_x = \pi$ for~$x \in V_1$,~$\pi^{\sigma}_x = \sigma$ for~$x \in V_2$, and $\pi^{\sigma}_x = \id$ for $x \in V_3$.
  By hypothesis, $\gamma_{A,1}$ and $\gamma_{A,2}$ do not intersect, and hence, the edges in each cut are distinct. Then this configuration has weight
  \begin{align}\label{eq:ground state two cuts}
    \begin{split}
      Z_k(\{\pi^{\sigma}_x\}) &=\prod_{e_1 \in \gamma_{A,1}} D_{e_1}^{-d(\pi,\sigma)} \prod_{l_1 \in C(\pi^{-1}\sigma)} m^{(D)}_{e_1,l_1} \prod_{e_2 \in \gamma_{A,2}} D_{e_2}^{-d(\sigma,\id)} \prod_{l_2 \in C(\sigma)} m^{(D)}_{e,l}\\
      &= D_{\gamma_{A,1}}^{-d(\pi,\id)} \left(\tfrac{D_{\gamma_{A,1}}}{D_{\gamma_{A,2}}}\right)^{d(\sigma,\id)} \prod_{e_1 \in \gamma_{A,1}} m^{(D)}_{e_1,\pi^{-1}\sigma} \prod_{e_2 \in \gamma_{A,2}}  m^{(D)}_{e_2,\sigma}\\
      &= D_{\gamma_{A,1}}^{-d(\pi,\id)} \left(\tfrac{D_{\gamma_{A,1}}}{D_{\gamma_{A,2}}}\right)^{d(\sigma,\id)}  m^{(D)}_{\gamma_{A,1},\pi^{-1}\sigma}   m^{(D)}_{\gamma_{A,2},\sigma},
    \end{split}
  \end{align}
  where $D_{\gamma_A,1}/D_{\gamma_A,2}$ converges to $t$, by assumption. Now, to show that these configurations yield the dominant contributions, we will need to use that $D_e = \Theta(D)$, so let us write $\frac{D}{D_e} = C_e^{(D)} = \Theta(1)$.
  Then for general configurations labeled by $\pi_x$, we may rewrite \cref{eq:Z pi} as
  \begin{align*}
    Z_k(\{\pi_x\}) = \prod_{e = (xy) \in E} D^{-d(\pi_x,\pi_y)} (C^{(D)}_e)^{d(\pi_x,\pi_y)} m^{(D)}_{e,\pi_x^{-1} \pi_y}.
  \end{align*}
  The configurations we claimed to be dominant satisfy $Z_k(\{\pi\}) = \Theta(D^{-m(A)d(\pi,\id)})$, where we recall that~$m(A)$ is the size of a minimal cut for $A$. Now we will show that all other configurations satisfy $Z_k(\{\pi\}) = \bigO(D^{-m(A)d(\pi,\id) - 1})$. To this end, consider some arbitrary configuration $\{\pi_x\} \in \mathcal S_{A,\pi}$.
  Let $P$ be a maximal set of edge-disjoint paths in~$G$ from~$A$ to~$\bar{A}$.
  It is a well-known fact that such a set has size~$m(A)$, by the max-flow min-cut theorem.
  Let
  \begin{align*}
    C_k^{(D)} := \left(\max_{e \in E, l = 1, \ldots, k} (C^{(D)}_e)^{l-1}\right) \left(\max_{e \in E, \pi \in S_k} m^{(D)}_{e,\pi}\right).
  \end{align*}
  Then we may bound
  \begin{align}\label{eq:upper bound paths}
    Z_k(\{\pi_x\}) & \leq (C_k^{(D)})^{\abs{E}} \prod_{e = (xy) \in E} D^{-d(\pi_x, \pi_y)}
    \leq (C_k^{(D)})^{\abs{E}} \prod_{p \in P} D^{- \sum_{e = (xy) \in p}d(\pi_x, \pi_y)}.
  \end{align}
  The first inequality is clear from the definition of $C_k^{(D)}$, and in the second inequality, we simply restrict to a subset of the edges we multiply over.
  Note that $C_k^{(D)} = \bigO(1)$.
  Then, by the triangle inequality for the Cayley distance $d$, it holds that
  \begin{align*}
    \sum_{e = (xy) \in p}d(\pi_x, \pi_y) \geq d(\pi,\id)
  \end{align*}
  with equality if and only if the only edges $(xy)$ for which $\pi_x \neq \pi_y$ are on a path in $P$, and each of the paths is a geodesic. Then we conclude
  \begin{align*}
    Z_k(\{\pi_x\}) \leq C_k^{(D)} \prod_{p \in P} D^{- d(\pi,\id)} = C_k^{(D)} D^{- m(A)d(\pi,\id)},
  \end{align*}
  and we see that the weight of every configuration can be bounded by the product of a $\bigO(1)$ number and a polynomial in $D$.

  Now, as promised, we show that if $\{\pi_x\}$ is not one of the minimal configurations described above, we actually have
  \begin{align}\label{eq:non-dominant configurations}
    Z_k(\{\pi_x\}) = \bigO(D^{- m(A)d(\pi,\id) - 1}).
  \end{align}
  To see this, we rewrite the triangle inequality for the Cayley distance as:
  \begin{align}\label{eq:upper bound paths 2}
    \prod_{e = (xy) \in E} D^{-d(\pi_x, \pi_y)} \leq \prod_{p \in P} D^{- \sum_{e = (xy) \in p}d(\pi_x, \pi_y)} \leq D^{- m(A)d(\pi,\id)}
  \end{align}
  with equality if and only if the $\pi_x$ are on a geodesic path in $P$. We now show that this is satisfied only for the configurations we claimed to be minimal. Assume that $\{\pi_x\} \in \mathcal S_{A,\pi}$ is such that the inequalities in \cref{eq:upper bound paths 2} are equalities and let
  \begin{align*}
    \Delta_n = \{x \in V \text{ such that } d(\pi_x, \pi) \leq n \}.
  \end{align*}
  Then $\Delta_n \in C(A)$ for $0 \leq n < d(\pi,\id)$, and we denote by $\delta_n$ the associated set of edges crossing the cut.
  Each edge~$(xy) \in \delta_n$ must be such that $\pi_x \neq \pi_y$, so it must be on a path in $P$, and because the permutations are geodesics along the paths, they must be on different paths.
  Hence~$\abs{\delta_n} \leq \abs{P} = m(A)$, implying each~$\Delta_n$ is a minimal cut.
  This immediately implies the claim if there is a unique minimal cut, since we must have~$\Delta_{d(\pi,\id) - 1} = \Delta_{0} = \Gamma_A$. If there are exactly two minimal cuts, then we must have~$\pi_x = \pi$ for~$x \in V_1$,~$\pi_x = \id$ for~$x \in V_3$, and there must be some $l$ such that for all~$x \in V_2$ we have~$d(\pi,\pi_x) = l$ and~$d(\pi_x,\id) = d(\pi,\id) - l$.
  Then in order to have equality in \cref{eq:upper bound paths 2}, we must have that for all $x \in V_2$,~$\pi_x$ equals some fixed permutation~$\sigma$, because the assumption of having exactly two cuts implies that $V_2$ is connected, and we must have~$d(\pi_x, \pi_y) = 0$ for all $(xy) \in E$ with $x, y\in V_2$.
  This proves \cref{eq:non-dominant configurations}.

  In conclusion, if there is a unique minimal cut, then by \cref{eq:ground state unique cut} and \cref{eq:non-dominant configurations}, we find
  \begin{align}\label{eq:Zk one cut}
    Z_{k,\pi} = D_{\gamma_A}^{-d(\pi,\id)} m^{(D)}_{\gamma_A,\pi} + \bigO(D^{-m(A)d(\pi,\id) - 1}),
  \end{align}
  and if there are exactly two (non-intersecting) cuts, then by \cref{eq:ground state two cuts} and \cref{eq:non-dominant configurations}, we find
  \begin{align}\label{eq:Zk two cuts}
    Z_{k,\pi} = \sum_{\sigma, \, d(\pi,\sigma) + d(\sigma,\id) = d(\pi,\id)} D_{\gamma_{A,1}}^{-d(\pi,\id)} \left(\tfrac{D_{\gamma_{A,1}}}{D_{\gamma_{A,2}}}\right)^{d(\sigma,\id)}  m^{(D)}_{\gamma_{A,1},\pi^{-1}\sigma}   m^{(D)}_{\gamma_{A,2},\sigma} + \bigO(D^{-m(A)d(\pi,\id) - 1}).
  \end{align}

  Finally, we set $\pi = \tau$ for the full cycle $\tau$ and we use \cref{eq:Z is moment}. For a unique minimal cut $\gamma_A$, by \cref{eq:Zk one cut}
  \begin{align}\label{eq:moment for cut expansion}
    \begin{split}
      \EE m_{\gamma_A,k} &= D_{\gamma_A}^{k-1} Z_{k, \tau}\\
      &= m^{(D)}_{\gamma_A,\pi} + \bigO( D_{\gamma_A}^{k-1} D^{-m(A)(k-1) - 1})\\
      &= m^{(D)}_{\gamma_A,k} + \bigO\left(\frac{1}{D}\right).
    \end{split}
  \end{align}
  using $d(\tau,\id) = k-1$ and $D_{\gamma_A} = \Theta(D^{m(A)})$.
  This proves \cref{eq:moments single surface} as $m^{(D)}_{\gamma_A,k}$ converges to $m_{\gamma_A,k}$.

  For two non-intersecting minimal cuts, we saw that the dominant contribution is due to configurations~$\{\pi^\sigma_x\}$ for~$\sigma$ on a geodesic between $\tau$ and $\id$. Then applying the observation that $\sigma$ is on such a geodesic if and only if $\sigma$ is a non-crossing partition similarly yields that by \cref{eq:Zk two cuts}
  \begin{align}\label{eq:moment for two cut expansion}
    \begin{split}
      \EE m_{\gamma_A,k} &= D_{\gamma_A}^{k-1} Z_{k, \tau}\\
      &= D_{\gamma_A}^{k-1}\sum_{\sigma \in NC(k)} \left(\tfrac{D_{\gamma_{A,1}}}{D_{\gamma_{A,2}}}\right)^{d(\sigma,\id)}  m^{(D)}_{\gamma_{A,1},\pi^{-1}\sigma}   m^{(D)}_{\gamma_{A,2},\sigma} + \bigO( D_{\gamma_A}^{k-1} D^{-m(A)(k-1) - 1})\\
      &= \sum_{\sigma\in NC(k)}\left(\tfrac{D_{\gamma_{A,1}}}{D_{\gamma_{A,2}}}\right)^{d(\sigma,\id)} m^{(D)}_{\gamma_{A,1},\tau^{-1}\sigma} m^{(D)}_{\gamma_{A,2},\sigma}  + \bigO\left(\frac{1}{D}\right).
    \end{split}
  \end{align}
  Since $m^{(D)}_{\gamma_{A,1},\tau^{-1}\sigma} \to m_{\gamma_{A,1},\tau^{-1}\sigma}$, $m^{(D)}_{\gamma_{A,2},\sigma} \to m_{\gamma_{A,2},\sigma}$ and $D_{\gamma_{A,1}}/ D_{\gamma_{A,2}} \to t$, this proves \cref{eq:moments two surfaces}.

  This accomplishes Step 2: we have shown that the configurations $\{\pi^{\min}_x\}$ (in case of a unique minimal cut for $A$) and $\{\pi^{\sigma}_x\}$ (in case there are exactly two non-intersecting minimal cuts for $A$) dominate in the computation of the expectation of the $k$-th moment of $\mu_A$ in terms of powers of $D$.

  We complete the proof by showing that the variance of $m_{A,k}$ vanishes as $D \to \infty$.
  We use the observation in \cref{eq:variance moments Z}, applying the analysis of $Z_{2k,\pi}$ to the case where $\pi = \tilde \tau = (1 2 \ldots k)(k+1 \, k+2 \ldots 2k)$.
  If there is a unique minimal cut, then using \cref{eq:Zk one cut} and the fact that $d(\tilde \tau, \id) = 2k - 2$, we find
  \begin{align*}
    \EE \mleft[\left(m^{(D)}_{A,k}\right)^2\mright] = D_{\gamma_A}^{2k-2} Z_{2k, \tilde \tau} = \prod_{l \in C(\tilde \tau)} m^{(D)}_{\gamma_A,l} + \bigO(D_{\gamma_A}^{2k-2} D^{-m(A)(2k-2) - 1}) = (m^{(D)}_{\gamma_A,k})^2 + \bigO\left(\frac{1}{D}\right).
  \end{align*}
  By \cref{eq:moment for cut expansion} we know that $(\EE m^{(D)}_{A,k})^2 = (m^{(D)}_{\gamma_A,k} + \bigO(\frac{1}{D}))^2$, and we conclude that the variance obeys
  \begin{align*}
    \EE\mleft[\left(m^{(D)}_{A,k} - \EE \mleft[m^{(D)}_{A,k}\mright]\right)^2\mright] = \EE \mleft[(m^{(D)}_{A,k})^2\mright] - \left(\EE m^{(D)}_{A,k}\right)^2 = \bigO\left(\frac{1}{D}\right).
  \end{align*}
  For the case with exactly two minimal cuts, a similar argument holds.
  Here, the key observation is that $\sigma$ is on a geodesic between $\tilde{\tau}$ and $\id$ if and only if $\sigma = \sigma_1 \sigma_2$ where $\sigma_1$ is on a geodesic between $(1 2 \ldots k)$ and $\id$ and $\sigma_2$ is on a geodesic between $(k+1 \, k+2 \ldots 2k)$ and $\id$.
  Using \cref{eq:Zk two cuts}, this implies
  \begin{align*}
    \EE (m^{(D)}_{A,k})^2 & = D_{\gamma_A}^{2k-2} Z_{2k, \tilde \tau}                                                                                                                                                                                                                                                                         \\
                          & =\sum_{\sigma_1,\sigma_2 \in NC(k)}\left(\tfrac{D_{\gamma_{A,1}}}{D_{\gamma_{A,2}}}\right)^{d(\sigma_1,\id) +d(\sigma_2,\id) } m^{(D)}_{\gamma_{A,1},\tau^{-1}\sigma_1} m^{(D)}_{\gamma_{A,1},\tau^{-1}\sigma_2}  m^{(D)}_{\gamma_{A,2},\sigma_1} m^{(D)}_{\gamma_{A,2},\sigma_2} + \bigO\left(\frac{1}{D}\right) \\
                          & = \left(\sum_{\sigma \in NC(k)}\left(\tfrac{D_{\gamma_{A,1}}}{D_{\gamma_{A,2}}}\right)^{d(\sigma,\id)} m^{(D)}_{\gamma_{A,1},\tau^{-1}\sigma}  m^{(D)}_{\gamma_{A,2},\sigma} \right)^2 + \bigO\left(\frac{1}{D}\right).
  \end{align*}
  By \cref{eq:moment for two cut expansion}, we see that this coincides with $(\EE m^{(D)}_{A,k})^2$, up to $\bigO(\tfrac{1}{D})$, and hence \cref{eq:variance moments} holds.
\end{proof}

We now have the ingredients to prove that the entanglement spectrum of random tensor networks with link states with bounded spectral variation can be written in a simple fashion. We will use the \textit{method of moments} to translate the above result on convergence of moments to convergence in distribution. The basic statement is that, given certain conditions on the distributions in question, if the moments of a sequence of distribution $\mu_n$ converge to those of $\mu$, then $\mu_n \Rightarrow \mu$ -- see for instance Theorem 30.8 in \cite{billingsley2008probability}.

The method of moments is valid, so long as a distribution $\mu$ is completely determined by its moments. This occurs if, for all $k$, the $k$-th moment $m_{\mu,k}$ is bounded as
\begin{equation}
  m_{\mu,k} \leq AB^kk!
\end{equation}
for constants $A,B$ independent of $k$. If the distributions have compact support, as in \cref{prop:convergence of moments}, then this condition is satisfied.\footnote{A basic example of a distribution which does not have compact support, but is nevertheless uniquely determined by its moments is a standard Gaussian distribution. On the other hand, a standard example of distributions that are not determined by their moments are the densities on $\RR_{\geq 0}$ with $\d \mu_1(x) = \sqrt{2\pi}x^{-1}e^{-(\log x)^2/2} \d x$ and $\d \mu_2(x) = (1+\sin(2\pi\log x)) \d \mu_1(x)$, for which it can be verified that the $n$-th moments of both distributions are equal to $e^{-n^2/2}$, while the distributions are clearly not identical.}

Now that we have established the convergence of moments in \cref{prop:convergence of moments}, we have our main result of the (conditional) convergence in distribution.
As in \cref{prop:convergence of moments} we consider a family of random tensor network states with link states with bounded spectral variation with increasing $D$, as defined in the beginning of this section.

\begin{thm}\label{thm:spectrum near max entangled}
  If there exists a unique minimal cut $\gamma_A$ for $A$, then $\mu_A^{(D)} \Rightarrow \mu_{\gamma_A}$, in probability,  as $D \to \infty$.
  If there exist exactly two minimal cuts $\gamma_{A,1}$ and $\gamma_{A,2}$, which do not intersect and for which $\lim_{D \to \infty} \frac{D_{\gamma_{A,1}}}{D_{\gamma_{A,2}}} = t \leq 1$, then $\mu^{(D)}_A \Rightarrow \MP(1) \boxtimes \mu_{\gamma_{A,1}} \boxtimes \mu_{\gamma_{A,2}}(t)$, in probability, where $\mu_{\gamma_{A,2}}(t) = (1 - t)\delta_0 + t\mu_{\gamma_{A,2}}$
\end{thm}
\begin{proof}
  It is straightforward to see that the $k$-th moment of $\mu_{\gamma_{A,2}}(t)$ is given by $t^{k-1} m_{\gamma_{A,2},k}$, and then the result follows immediately from \cref{prop:convergence of moments}, \cref{thm:free product}, and the method of moments. Because we assumed that for any minimal cut $\gamma_A$ for $A$, the limiting distributions $\mu_{\gamma_A}$ are compactly supported, they are uniquely determined by their moments. Hence, the method of moments is valid, and the convergence of moments implies convergence in distribution.
\end{proof}

\begin{rmk}
  In \cref{thm:spectrum near max entangled}, we assumed that the two cuts were non-intersecting.
  What happens if there are still only exactly two minimal cuts, but $\gamma_{A,1} \cap \gamma_{A,2}$ is nonempty?
  This extension is straightforward.
  Let $\gamma_A^{(a)} := \gamma_{A,1} \cap \gamma_{A,2}$ and let $\gamma_{A,i}^{(b)} = \gamma_{A,i} \setminus \gamma_A^{(a)}$ for $i = 1,2$.
  In line with previous notation, let~$\mu_{\gamma_A^{(a)}}$ and~$\mu_{\gamma_{A,i}^{(b)}}$ denote the corresponding limiting distributions of the entanglement spectra along these sets, with moments $m_{\gamma_A^{(a)},k}$ and $m_{\gamma_{A,i}^{(b)},k}$.
  The only step in the proof of \cref{prop:convergence of moments} where we used that the cuts were non-intersecting is when we computed the value of $Z_k(\{\pi_x\})$ for the optimal configuration.
  If the cuts do intersect, and we consider the configuration with $\pi_x = \tau$ for $x \in V_1$ with $\tau$ the complete cycle,~$\pi_x = \sigma$ for~$x \in V_2$ and $\sigma \in NC(k)$, and $\pi_x = \id$ for $x \in V_3$, then a quick calculation shows
  \begin{align*}
    Z_{k}(\{\pi_x\}) \to D_{\gamma_{A,1}}^{-d(\pi,\id)} (D_{\gamma_A,1}/D_{\gamma_A,2})^{d(\sigma,\id)} \prod_{e \in \gamma_A^{(a)}} m_{e,k} \prod_{e_1 \in \gamma^{(b)}_{A,1}} m_{e_1,\pi^{-1}\sigma} \prod_{e_2 \in \gamma^{(b)}_{A,2}}  m_{e_2,\sigma}.
  \end{align*}
  Apart from this modification, the proof of \cref{prop:convergence of moments} is still valid, leading to
  \begin{align*}
    Z_{k,\tau} = m_{\gamma_A^{(a)},k} \sum_{\sigma \in NC(k)} t^{d(\sigma,\id)}m_{\gamma_{A,1}^{(b)},\tau^{-1}\sigma} m_{\gamma_{A,2}^{(b)},\sigma}.
  \end{align*}
  If, in \cref{thm:spectrum near max entangled}, we do not assume that the cuts are non-intersecting, then the partition function above leads to a limiting distribution given by
  \begin{align*}
    \mu_{\gamma_A^{(a)}} \otimes \left(MP(1) \boxtimes \mu_{\gamma_{A,1}^{(b)}} \boxtimes \mu_{\gamma_{A,2}^{(b)}}(t)\right).
  \end{align*}
\end{rmk}

\subsection{Nontrivial link states and entanglement negativity}\label{sec:negativity}
As another application of the theory of free probability, we will compute the entanglement negativity spectrum for random tensor network states with link states with bounded spectra.
In \cite{dong2021holographic}, it was shown how to compute the entanglement negativity spectrum for a random tensor network state with maximally entangled link states using a replica trick.
Using the methods from the previous subsection, we can analyze the negativity for entangled link states with bounded spectral variation.
We remark that similar computations have recently been performed in \cite{dong2021replica} in the context of replica wormholes, and our assumption on the link states is a generalization of the ``pairwise connected regime'' in \cite{dong2021replica}.
Another work investigating nontrivial entanglement negativity spectra in random tensor networks is \cite{kudler2021negativity}, where they focus on the effect of having multiple minimal cuts in the network.
As our analysis will be a straightforward combination of the arguments in \cite{dong2021holographic} and \cref{sec:bounded spectral variation and free product}, we will be rather concise; the main message of this section is to show that the language of free probability applies to other random tensor network computations as well.

We first recall how negativity functions as an entanglement measure for mixed states.
Let $\mathcal T$ be the superoperator which maps an operator $X$ to its transpose $X^{\tran}$, and $\mathcal I$ be the identity superoperator.
For $\rho_{AB} \in \PSD(AB)$,
\begin{align*}
  \rho_{AB}^{T_B} := (\mathcal I_A \ot \mathcal T_B)(\rho_{AB})
\end{align*}
is the partial transpose of $\rho_{AB}$ on the $B$ system.
The \emph{logarithmic} or \emph{entanglement negativity} is given by
\begin{align*}
  E_N(\rho_{AB}) = \log \frac{\norm{\rho_{AB}^{T_B}}_1}{\tr[\rho]}.
\end{align*}
It is a measure for the entanglement of the mixed state $\rho_{AB}$: if $E_N(\rho_{AB}) > 0$ the state must be entangled.
We call $\spec(\abs{\rho_{AB}^{T_B}})$ the \textit{entanglement negativity spectrum}.
In analogy to the R\'enyi entropies, we can generalize the logarithmic negativity to a one-parameter family of negativities. The $k$-th R\'{e}nyi negativity is given by
\begin{align*}
  N_k(\rho_{AB}) = \tr\mleft[(\rho_{AB}^{T_B})^k\mright].
\end{align*}
If we let $N^{(\mathrm{even})}_m(\rho_{AB}) = N_{2m}(\rho_{AB})$, then the logarithmic negativity is obtained as an analytic continuation in the R\'enyi index $m \to \frac12$ of $\log(N^{(\mathrm{even})}_m(\rho_{AB}))$.
More precisely, in the expression
\begin{align*}
  \log \sum_{\lambda \in \spec(\rho_{AB}^{T_B})} \abs{\lambda}^{\alpha},
\end{align*}
we may take $\alpha \to \frac12$ to obtain $E_N(\rho_{AB}) + \log \tr[\rho]$.

In the context of random tensor networks, we partition the boundary in three regions: $V_{\partial} = A \sqcup B \sqcup C$, and we would like to compute the R\'{e}nyi negativities of the reduced state $\rho_{AB}$.
We will then use this to determine the entanglement negativity spectrum, and compute the entanglement negativity.
The idea is that the $k$-th R\'{e}nyi negativity can be computed using a replica trick, by placing the full cycle $\tau_k = (1 2 \ldots k) \in S_k$ on $A$ and $\tau_k^{-1} = (k \, k-1 \ldots 1)$ on $B$:
\begin{align*}
  N_k(\rho_{AB}) = \tr\mleft[\rho_{AB}^{\ot k} \left(R_A(\tau) \ot R_B(\tau^{-1})\right)\mright] = \tr\mleft[\rho_{ABC}^{\ot k} \left(R_A(\tau) \ot  R_B(\tau^{-1}) \ot R_C(\id)\right)\mright].
\end{align*}

Let us first discuss the case with maximally entangled link states, following~\cite{dong2021holographic}.
The same arguments as in \cref{sec:ground state configurations} show that one can compute the expectation of $N_k(\rho_{AB})$ for a random tensor network state using a spin model, now with boundary conditions of $\tau_k$ on $A$, $\tau_k^{-1}$ on $B$ and $\id$ on $C$.
We will assume that the minimal cuts $\Gamma_A$, $\Gamma_B$ and $\Gamma_C$ are unique.
Note that the minimal cut for $AB$ is given by $\Gamma_{AB} = \Gamma_C^c$.
From the theory of multi-commodity flows, it is known that there exist sets of edge-disjoint paths $P = P_{AB} \cup P_{AC} \cup P_{BC}$, where $P_{AB}$ consists of paths from $A$ to $B$, and similarly for $P_{AC}$ and $P_{BC}$, and which are such that
\begin{align*}
  \abs{P_{AB}} + \abs{P_{AC}} = \abs{\gamma_A}, \qquad \abs{P_{AB}} + \abs{P_{BC}} = \abs{\gamma_B}, \qquad \abs{P_{AC}} + \abs{P_{BC}} = \abs{\gamma_C}.
\end{align*}
This can be used to show (in analogous fashion to the proof of \cref{prop:convergence of moments}) that, if $k = 2n$ is even, any spin model configuration contributing to $\EE N_k(\rho_{AB})$ is of order $\bigO(D^{-(n-1)(\abs{\gamma_A} + \abs{\gamma_B}) - n\abs{\gamma_C}})$.
If $k = 2n+1$ is odd, any spin model configuration contributing to $\EE N_k(\rho_{AB})$ is of order $\bigO(D^{-n(\abs{\gamma_A} + \abs{\gamma_B} + \abs{\gamma_C})})$.

In order to determine what happens as $D \to \infty$, we need to determine the dominant configurations.
Let $r$ be the number of connected components of $V \setminus (\Gamma_A \cup \Gamma_B \cup \Gamma_C)$.
There are two distinct cases. The first is when the minimal cut for $AB$ (which is the complement of the minimal cut for $C$) is the union of the minimal cuts for $A$ and $B$, so $\Gamma_{AB} = \Gamma_A \cup \Gamma_B$ and hence $\gamma_{AB} = \gamma_A \cup \gamma_B$. Then the minimal cuts naturally partition the bulk vertices into three cuts $\Gamma_A$, $\Gamma_B$ and $\Gamma_C$, and we have $r=0$. In this case, the dominant configurations in the spin model are those where the vertices in $\Gamma_A$ are assigned $\tau_k$, those in $\Gamma_B$ are assigned~$\tau_k^{-1}$ and those in $\Gamma_C$ are assigned $\id$.
This is illustrated in \cref{subfig:negativity 1}.

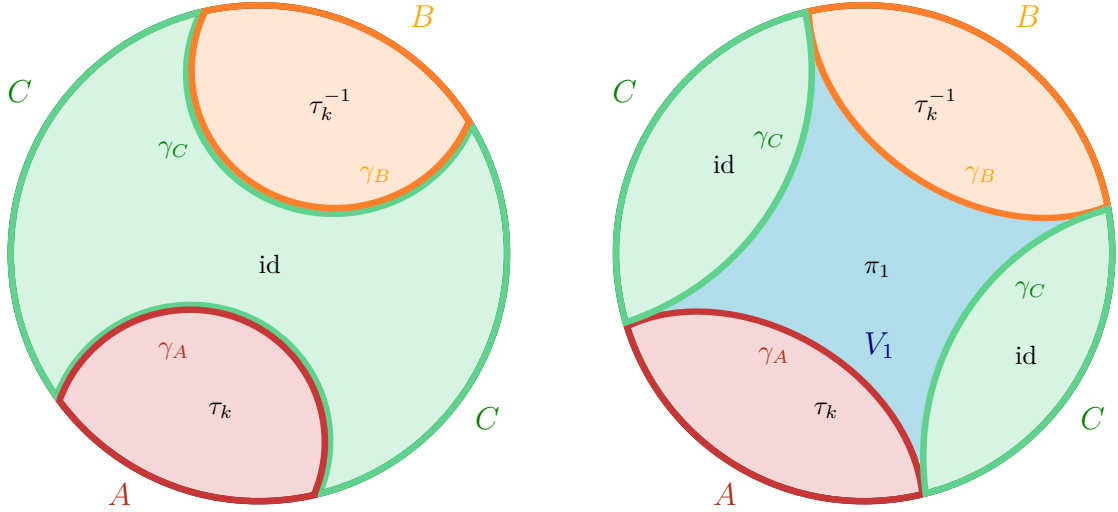
\begin{figure}
  \centering
  \begin{subfigure}[t]{.45\textwidth}
    \centering
    \begin{overpic}[width=0.9\textwidth,grid=false]{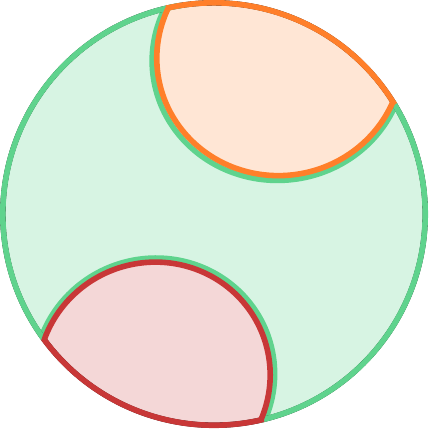}
      \put(80,95){\color{Orange}{\large{$B$}}}
      \put(0,80){\color{Green}{\large{$C$}}}
      \put(93,15){\color{Green}{\large{$C$}}}
      \put(20,0){\color{BrickRed}{\large{$A$}}}
      \put(40,18){$\tau_k$}
      \put(60,78){$\tau_k^{-1}$}
      \put(50,46){$\id$}
      \put(30,70){\color{Green}{$\gamma_C$}}
      \put(30,30){\color{BrickRed}{$\gamma_A$}}
      \put(70,65){\color{Orange}{$\gamma_B$}}
    \end{overpic}
    \caption{Illustration of the case where $\gamma_C = \gamma_A \cup \gamma_B$, where the dominant configuration is given by $\tau_k$ on $\Gamma_A$, $\tau_k^{-1}$ on $\Gamma_B$ and $\id$ on $\Gamma_C$.}
    \label{subfig:negativity 1}
  \end{subfigure}%
  \hspace*{0.5cm}
  \begin{subfigure}[t]{.45\textwidth}
    \centering
    \begin{overpic}[width=0.9\textwidth,grid=false]{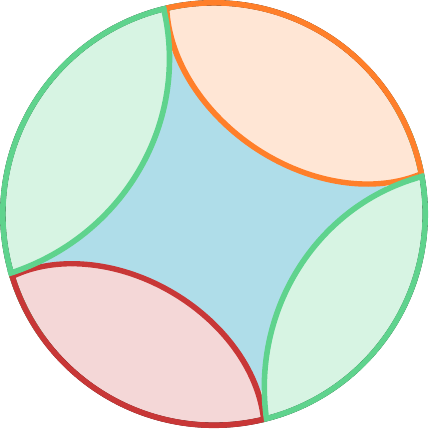}
      \put(80,95){\color{Orange}{\large{$B$}}}
      \put(0,80){\color{Green}{\large{$C$}}}
      \put(93,15){\color{Green}{\large{$C$}}}
      \put(20,0){\color{BrickRed}{\large{$A$}}}
      \put(40,18){$\tau_k$}
      \put(60,78){$\tau_k^{-1}$}
      \put(50,46){$\pi_1$}
      \put(20,66){$\id$} \put(80,28){$\id$}
      \put(50,30){\color{NavyBlue}{\large{$V_1$}}}
      \put(28,72){\color{Green}{$\gamma_C$}}
      \put(80,42){\color{Green}{$\gamma_C$}}
      \put(29,29){\color{BrickRed}{$\gamma_A$}}
      \put(70,65){\color{Orange}{$\gamma_B$}}
    \end{overpic}
    \caption{Illustration of the case where $\gamma_C \neq \gamma_A \cup \gamma_B$, where the dominant configuration is given by $\tau_k$ on $\Gamma_A$, $\tau_k^{-1}$ on $\Gamma_B$ and $\id$ on $\Gamma_C$ and some non-crossing pairing $\pi_1$ on the domain $V_1$.}
    \label{subfig:negativity 2}
  \end{subfigure}
  \caption{Tensor networks with one and two minimal cuts. The relevant ground state configuration domains are denoted by $\Gamma_A$.}
  \label{fig:negativity}
\end{figure}

The second case is when $\Gamma_A \cup \Gamma_B \subsetneq \Gamma_{AB}$ and hence $\gamma_{AB} \neq \gamma_A \cup \gamma_B$. Now, we have again the domains~$\Gamma_A$,~$\Gamma_B$ and~$\Gamma_C$, but upon removing these vertices, there may also be connected components $V_1, \ldots, V_r$ which are not connected to $A$, $B$ or $C$.
Here, the minimal configurations are those for which, again, the vertices in $\Gamma_A$ are assigned $\tau_k$, those in $\Gamma_B$ are assigned $\tau_k^{-1}$ and those in $\Gamma_C$ are assigned $\id$, and where in each component~$V_i$ the vertices are assigned a permutation $\pi_i$ which is such that it satisfies three conditions: it must be on a geodesic between $\tau_k$ and $\tau_k^{-1}$, on a geodesic between $\tau_k$ and $\id$ and on a geodesic between $\tau_k^{-1}$ and $\id$.
If $k = 2n$ is even, such permutations are given by \emph{non-crossing pairings}: permutations corresponding to non-crossing partitions in which each cycle has length 2.
The set of non-crossing pairings on $2n$ elements is in bijection with the set of non-crossing partitions on $n$ elements, so the number of non-crossing pairings on $2n$ elements is given by $\abs{NC(n)} = C_n$.
One way to obtain this correspondence is as follows.
If~$\pi$ is a non-crossing pairing,~$\tau_{2n} \pi$ will map even numbers to even numbers, and restricting to the even numbers and relabeling~$2i \mapsto i$ yields a non-crossing partition~$\sigma \in NC(n)$.
Moreover, restricting to the odd numbers and relabeling~$2i + 1 \mapsto i$ yields the non-crossing partition~$\sigma^{-1}\tau \in  NC(n)$.
This leads to~$C_n^r$ dominant contributions to~$\EE N_{2n}(\rho_{AB})$ of size~$D^{-(n-1)(\abs{\gamma_A} + \abs{\gamma_B}) - n\abs{\gamma_C}}$ since we can choose a non-crossing pairing~$\pi_i$ for each component.
Such a configuration is illustrated in \cref{subfig:negativity 2}.

For odd~$k = 2n + 1$, we similarly have permutations which correspond to a non-crossing partition, and which have a single fixed point and all other cycles with length 2.
This leads to $((2n + 1)C_n)^r$ dominant contributions to $\EE N_{2n}(\rho_{AB})$, of size $D^{-(n-1)(\abs{\gamma_A} + \abs{\gamma_B}) - n\abs{\gamma_C}}$.
We also note that $\rank(\rho_{AB}^{T_B}) \leq D^{\abs{\gamma_A} + \abs{\gamma_B}}$.
If~$\spec(\rho_{AB}^{T_B}) = \{s_i\}$, then we define the measure
\begin{align}\label{eq:mu negativity}
  \mu^{(D)}_{AB} = \frac{1}{D^{\abs{\gamma_A} + \abs{\gamma_B}}} \sum_{i=1}^{D^{\abs{\gamma_A} + \abs{\gamma_B}}} \delta_{D^{\frac12(\abs{\gamma_A} + \abs{\gamma_B} + \abs{\gamma_C})}s_i}.
\end{align}
This has moments given by
\begin{align*}
  m^{(D)}_{AB,k} = \int x^k \d \mu^{(D)}_{AB}(x) = D^{(\tfrac{k}{2}-1)(\abs{\gamma_A} + \abs{\gamma_B}) + \tfrac{k}{2}\abs{\gamma_C}} N_{k}(\rho_{AB}).
\end{align*}
If we take the expectation of the moments, we again need to distinguish the two cases.
If~$\abs{\gamma_A} + \abs{\gamma_B} = \abs{\gamma_C}$, we see that the powers of $D$ cancel for the dominant configurations, so $m^{(D)}_{AB,k} \to 1$ for all $k$.
On the other hand, for $\abs{\gamma_A} + \abs{\gamma_B} > \abs{\gamma_C}$, we see that for $D \to \infty$ with odd $k$, we have $\EE m^{(D)}_{AB,k} \to 0$. For even $k$, we recover the degeneracy of the dominant configurations, leading to
\begin{align}\label{eq:moments negativity}
  \lim_{D \rightarrow \infty} m^{(D)}_{AB,k} =
  \begin{cases}
    1         & \text{if $\abs{\gamma_A} + \abs{\gamma_B} = \abs{\gamma_C}$,}              \\
    0         & \text{if $k$ odd and $\abs{\gamma_A} + \abs{\gamma_B} > \abs{\gamma_C}$,}  \\
    C_{k/2}^r & \text{if $k$ even and $\abs{\gamma_A} + \abs{\gamma_B} > \abs{\gamma_C}$},
  \end{cases}
\end{align}
where $r$ is the number of connected components of $V \setminus (\Gamma_A \cup \Gamma_B \cup \Gamma_C)$.
In fact, one can show that, as in \cref{prop:convergence of moments}, the variance goes to zero as well, and hence the method of moments allows one to conclude that $\mu^{(D)}_{AB} \Rightarrow \mu_{AB}$, in probability, where
\begin{align}\label{eq:distribution negativity spectrum}
  \mu_{AB} =
  \begin{cases}
    \sigma^{\ot r}                       & \text{if $r > 0$,}                                                      \\
    \frac12\delta_1 + \frac12\delta_{-1} & \text{if $r=0$ and $\abs{\gamma_A} + \abs{\gamma_B} > \abs{\gamma_C}$,} \\
    \delta_1                             & \text{if $r=0$ and $\abs{\gamma_A} + \abs{\gamma_B} = \abs{\gamma_C}$},
  \end{cases}
\end{align}
where $\sigma$ is the semi-circle distribution with density
\begin{align*}
  \d \sigma(x) = \frac{1}{2\pi}{\sqrt{4-x^2}} \ind_{\abs{x}\leq 2}\,dx
\end{align*}
Alternatively, one may study the empirical distribution of the \emph{squared} entanglement negativity spectrum
\begin{align}\label{eq:nu negativity}
  \nu_{AB}^{(D)} = \frac{1}{D^{\abs{\gamma_A} + \abs{\gamma_B}}} \sum_i \delta_{D^{\abs{\gamma_A} + \abs{\gamma_B} + \abs{\gamma_C}}s_i^2}.
\end{align}
This distribution has $k$-th moment given by $m^{(D)}_{AB,2k}$, and in comparison with the limiting moments in \cref{eq:moments negativity}, one can conclude that $\nu_{AB}^{(D)} \Rightarrow \nu_{AB}$, in probability, where
\begin{align}\label{eq:nu max ent}
  \nu_{AB} = MP(1)^{\ot r}.
\end{align}
The logarithmic negativity can be computed using the distribution $\mu^{(D)}_{AB}$ or $\nu^{(D)}_{AB}$ as
\begin{align}\label{eq:negativity}
  E_N(\rho_{AB}) & = \log \int  \abs{\lambda} \d \mu^{(D)}_{AB}(\lambda)  + \frac{\log D}{2}(\abs{\gamma_A} + \abs{\gamma_B} - \abs{\gamma_C}) - \log \tr\mleft[\rho\mright]   \\
                 & = \log \int  \sqrt{\lambda} \d \nu^{(D)}_{AB}(\lambda)  + \frac{\log D}{2}(\abs{\gamma_A} + \abs{\gamma_B} - \abs{\gamma_C}) - \log \tr\mleft[\rho\mright].
\end{align}
The convergence of $\nu^{(D)}_{AB}$ to $\nu_{AB}$ implies\footnote{The function $f(\lambda) = \sqrt{\lambda}$ is not in $C_b(\RR)$, but the method of moments actually shows a stronger convergence, allowing test functions to have polynomial growth.} that $E_N(\rho_{AB}) - \frac{\log D}{2}(\abs{\gamma_A} + \abs{\gamma_B} - \abs{\gamma_C})$ converges in probability to
\begin{align*}
  \log \int \sqrt{\lambda} \d \nu_{AB}(\lambda) = r\log \frac{8}{3\pi}.
\end{align*}
See Appendix D of \cite{dong2021holographic} for details and proofs.

A straightforward combination of the arguments in \cref{sec:bounded spectral variation and free product} and
\cite{dong2021holographic} shows that the same configurations are the dominant contributions for link states with bounded spectral variation as in \cref{sec:bounded spectral variation and free product}.
To determine the limiting distribution in this case, we can generalize \cref{eq:distribution negativity spectrum} in the same fashion as in \cref{sec:bounded spectral variation and free product}.
We assume the minimal cuts $\Gamma_A$, $\Gamma_B$ and $\Gamma_C$ are unique.
We also assume that~$\gamma_A \cap \gamma_B = \emptyset$, and in the case where $\gamma_C = \gamma_{AB} \neq \gamma_A \cup \gamma_B$ (so $\abs{\gamma_A} + \abs{\gamma_B} > \abs{\gamma_C}$), all pairwise intersections between~$\gamma_A$,~$\gamma_B$ and~$\gamma_C$ are empty.
This excludes the case where $\abs{\gamma_A} + \abs{\gamma_B} > \abs{\gamma_C}$, but $r = 0$.
We let $\gamma_{A,i}$ and $\gamma_{B,i}$ denote the components of $\gamma_A$ and $\gamma_B$ which are connected to $V_i$, and we let $\mu^{(D)}_{\gamma_{A,i}}$ and $\mu^{(D)}_{\gamma_{B,i}}$ denote the distribution of the spectrum along these sets, with associated $k$-th moments $m^{(D)}_{\gamma_{A,i},k}$, $m^{(D)}_{\gamma_{B,i},k}$, which we assume to converge to the moments $m_{\gamma_{A,i},k}$, $m_{\gamma_{B,i},k}$ of compactly supported distributions $\mu_{\gamma_{A,i}}$ and $\mu_{\gamma_{B,i}}$.
For convenience, we assume $D_e = D$ for all edges $e \in E$.

We can now compute the dominant contributions to~$\EE N_{k}(\rho_{AB})$.
If~$\gamma_C = \gamma_A \cup \gamma_B$, then there is a unique dominant configuration, which contributes~$D^{-(k-1)\abs{\gamma_C}} m_{\gamma_C,k}$.
If~$\abs{\gamma_A} + \abs{\gamma_B} > \abs{\gamma_C}$ and $k =2n$ is even, consider the configuration which assigns $\pi_i$ to $V_i$, where each $\pi_i$ is a non-crossing pairing.
For each edge~$e \in \gamma_C$, we have $m^{(D)}_{e,\pi_i} = (m^{(D)}_{e,2})^{n}$, so this configuration contributes
\begin{align*}
  D^{-(n-1)(\abs{\gamma_A} + \abs{\gamma_B}) - n\abs{\gamma_C}}\left(m^{(D)}_{\gamma_C,2}\right)^{n} \prod_{i=1}^r\bigl(m^{(D)}_{\gamma_{A,i},\tau_{2n}^{-1} \pi_i}  m^{(D)}_{\gamma_{B,i},\tau_{2n} \pi_i} \bigr) .
\end{align*}
Recalling the construction of the equivalence between $NC(n)$ and non-crossing pairings on $2n$ elements, we see that
\begin{align*}
  m^{(D)}_{\gamma_{B,i},\tau_{2n} \pi_i} & = m^{(D)}_{\gamma_{B,i},\sigma_{i}} m^{(D)}_{\gamma_{B,i},\sigma_i^{-1}\tau_{n}}
\end{align*}
for some unique $\sigma_i \in NC(n)$.
Similarly, one may verify
\begin{align*}
  m^{(D)}_{\gamma_{A,i},\tau_{2n}^{-1} \pi_i} & = m^{(D)}_{\gamma_{A,i},\sigma_{i}} m^{(D)}_{\gamma_{A,i},\sigma_i^{-1}\tau_{n}}.
\end{align*}
This implies that the contribution of all dominant configurations is given by
\begin{align*}
  \left(m^{(D)}_{\gamma_C,2}\right)^{n}\prod_{i=1}^r\bigl(\sum_{\sigma \in NC(n)} m^{(D)}_{\gamma_{A,i},\sigma}  m^{(D)}_{\gamma_{B,i},\sigma} m^{(D)}_{\gamma_{A,i},\sigma^{-1}\tau_n}  m^{(D)}_{\gamma_{B,i},\sigma^{-1}\tau_n}\bigr)
\end{align*}
As in the maximally entangled case, upon rescaling, the odd moments vanish as $D \to \infty$.
In conclusion, the resulting asymptotic moments are given by
\begin{align}\label{eq:moments negativity with link states}
  \lim_{D \rightarrow \infty} \EE m^{(D)}_{AB,k} =
  \begin{cases}
    m_{\gamma_C,k} & \text{if $\abs{\gamma_A} + \abs{\gamma_B} = \abs{\gamma_C}$,}             \\
    0              & \text{if $k$ odd and $\abs{\gamma_A} + \abs{\gamma_B} > \abs{\gamma_C}$,} \\
    m_k            & \text{if $k$ even and $\abs{\gamma_A} + \abs{\gamma_B} > \abs{\gamma_C}$}
  \end{cases}
\end{align}
with
\begin{align*}
  m_{2n} = m_{\gamma_C,2}^{n}\prod_{i=1}^r\bigl(\sum_{\sigma \in NC(n)} m_{\gamma_{A,i},\sigma}  m_{\gamma_{B,i},\sigma} m_{\gamma_{A,i},\sigma^{-1}\tau_n}  m_{\gamma_{B,i},\sigma^{-1}\tau_n}\bigr).
\end{align*}
As before, one can also show, in similar fashion to the proof of \cref{prop:convergence of moments}, that the variance of the moments goes to zero as $D \to \infty$.
For the case $\abs{\gamma_A} + \abs{\gamma_B} > \abs{\gamma_C}$, we consider $\nu^{(D)}_{AB}$ similar to \cref{eq:nu negativity}, but with an additional rescaling by $m_{\gamma_C,2}$:
\begin{align*}
  \nu_{AB}^{(D)} = \frac{1}{D^{\abs{\gamma_A} + \abs{\gamma_B}}} \sum_i \delta_{D^{\abs{\gamma_A} + \abs{\gamma_B} + \abs{\gamma_C}}m_{\gamma_C,2}^{-1} s_i^2}.
\end{align*}
This has moments, which compute $N^{(\mathrm{even})}_k(\rho_{AB})$, converging to
\begin{align*}
  \lim_{D \rightarrow \infty} \EE \int x^k \d \nu_{AB}^{(D)}(x) = \prod_{i=1}^r\bigl(\sum_{\sigma \in NC(n)} m_{\gamma_{A,i},\sigma}  m_{\gamma_{B,i},\sigma} m_{\gamma_{A,i},\sigma^{-1}\tau_n}  m_{\gamma_{B,i},\sigma^{-1}\tau_n}\bigr).
\end{align*}
Thus, by the method of moments and \cref{thm:free product}, it holds that $\nu^{(D)}_{AB} \Rightarrow \nu_{AB}$, in probability, where
\begin{align}\label{eq:nu with link states}
  \nu_{AB} =
  \begin{cases}
    \bigotimes_{i = 1}^r \nu_i & \text{if $r > 0$,} \\
    \mu_{\gamma_C}             & \text{if $r=0$},
  \end{cases}
\end{align}
and where $\nu_i$ is given by
\begin{align*}
  \nu_i = (\mu_{\gamma_{A,i}} \ot \mu_{\gamma_{B,i}})^{\boxtimes 2} \boxtimes \MP(1).
\end{align*}
This reduces to \cref{eq:nu max ent} if the link states are maximally entangled.
We can use this to compute the logarithmic negativity, as we did previously.
For $r > 0$,
\begin{align*}
  E_N(\rho_{AB})
  = \log \int \sqrt{\lambda} \d \nu_{AB}^{(D)}(\lambda) + \frac{\log D}{2}(\abs{\gamma_A} + \abs{\gamma_B} - \abs{\gamma_C}) + \frac{1}{2}\log m_{\gamma_{C,2}} - \log \tr\mleft[\rho\mright],
\end{align*}
from which we find that $E_N(\rho_{AB}) - \frac{\log D}{2}(\abs{\gamma_A} + \abs{\gamma_B} - \abs{\gamma_C})$ converges in probability to
\begin{align*}
  \log \int \sqrt{\lambda} \d \nu_{AB}(\lambda) + \frac{1}{2}\log m_{\gamma_{C,2}}.
\end{align*}
For the case $\abs{\gamma_A} + \abs{\gamma_B} = \abs{\gamma_C}$, it is more elegant to use the limiting distribution of $\mu^{(D)}_{AB}$, as defined in \cref{eq:mu negativity}.
By the method of moments and \cref{eq:moments negativity with link states}, $\mu^{(D)}_{AB} \Rightarrow \mu_{\gamma_C}$, in probability.
We may then compute the entanglement negativity as
\begin{align*}
  E_N(\rho_{AB})
  = \log \int \abs{\lambda} \d \mu_{AB}^{(D)}(\lambda) + \frac{\log D}{2}(\abs{\gamma_A} + \abs{\gamma_B} - \abs{\gamma_C}) - \log \tr\mleft[\rho\mright],
\end{align*}
and hence $E_N(\rho_{AB}) - \frac{\log D}{2}(\abs{\gamma_A} + \abs{\gamma_B} - \abs{\gamma_C})$ converges in probability to $\log \int \abs{\lambda} \d \mu_{AB}(\lambda)$.

\section{Link states with unbounded spectral variation}\label{sec:far from max entangled}
We will now consider a different regime, where the link states have unbounded spectral variation.
Our methods in this section are distinct from the previous one, and the two sections can be considered separately.

\subsection{One-shot entropies} 
We begin by introducing one of our main tools for studying entanglement spectra in random tensor network states: \textit{one-shot entropies}.
In quantum information theory, the rates of certain important protocols, such as compression or state merging can be expressed as entropic quantities.
One-shot entropies are the appropriate analogs for settings where one would like to analyze a task for a single or finite number of copies of the relevant state.
Asymptotic rates in terms of ordinary von Neumann entropies are then recovered in the limit of infinitely many independent copies.
For an extensive introduction to this point of view, see \cite{tomamichel2015quantum}; here we provide the basic definitions and introduce the relevant concepts.

A random tensor network built from link states that are maximally entangled (or more generally have bounded spectral variation) can be analyzed using asymptotic tools.
Indeed, if we have a maximally entangled state of large dimension $D=2^n$, then this is equal to the $n$-th tensor power of a qubit maximally entangled state, so we are effectively in an asymptotic situation.
However, if we allow for link states with unbounded spectral variation or even completely general background states, as in \cref{sec:background states}, then it is more natural to use tools from one-shot quantum information theory.

We take the R\'enyi entropies as a starting point, which we defined in \cref{eq:renyi subnormalized} for subnormalized states.
Let~$\HH$ be some Hilbert space and for $\rho \in \Pleq(\HH)$ we define the (unconditional) \emph{min-entropy} and the \emph{max-entropy} by
\begin{align*}
  H_{\min}(\rho) & = -\log\, \norm{\rho}_{\infty}         \\
  H_{\max}(\rho) & = \log \left(\tr[\sqrt{\rho}]^2\right)
\end{align*}
which coincide with the R\'enyi entropies $H_{\infty}(\rho)$ and $H_{\frac12}(\rho)$ for $\rho \in \Peq(\HH)$.
As usual, if $\rho_A$ is the reduced density matrix on a system $A$, we write $H_{\min}(A)_{\rho} =  H_{\min}(\rho_A)$ and $H_{\max}(A)_{\rho} =  H_{\max}(\rho_A)$.

Often, when applied to study quantum information processing tasks, it is useful to allow a small error.
This leads to the introduction of smooth entropies.
To define these we use a distance measure known as the \emph{purified distance}, which is given for $\rho,\sigma \in \Pleq(\HH)$ by
\begin{align*}
  P(\rho, \sigma) = \sqrt{1 - F_*(\rho,\sigma)^2}
\end{align*}
where $F_*(\rho,\sigma)$ is the \emph{generalized fidelity} between $\rho$ and $\sigma$, which is defined by
\begin{align*}
  F_*(\rho,\sigma) = F(\rho,\sigma) + \sqrt{(1 - \tr\mleft[\rho\mright])(1 - \tr\mleft[\sigma\mright])}
\end{align*}
in terms of the ordinary \emph{fidelity} $F(\rho,\sigma) = \norm{\sqrt{\rho}\sqrt{\sigma}}_1$.
We define the \emph{smooth min- and max-entropies} of~$\rho \in \Pleq(\HH)$ as
\begin{align*}
  H^\eps_{\min}(\rho) & = \sup_{\rho^{\eps} \in \Pleq(\HH), P(\rho^\eps,\rho) \leq \eps} H_{\min}(\rho^\eps) \\
  H^\eps_{\max}(\rho) & = \inf_{\rho^{\eps}\in \Pleq(\HH), P(\rho^\eps,\rho) \leq \eps} H_{\max}(\rho^\eps).
\end{align*}
The smooth entropies are such that one recovers the usual von Neumann entropies in the limit of many independent copies.
Indeed, the following \emph{asymptotic equipartition property} holds:
\begin{align*}
  \lim_{n \to \infty} \frac{1}{n} H^{\eps}_{\min}(\rho^{\ot n}) = H(\rho) = \lim_{n \to \infty} \frac{1}{n} H^{\eps}_{\max}(\rho^{\ot n})
\end{align*}
for any $0 < \eps < 1$.
Variations on this definition are possible.
For instance, one can choose a different distance measure, which will yield different entropies.
However, for the usual choices, the differences go to zero as $\eps$ goes to zero, so the particular choice is often immaterial.
For instance, consider the \emph{trace distance} between $\rho, \sigma \in \Pleq(\HH)$, which is defined by
\begin{align*}
  T(\rho,\sigma) = \frac12\norm{\rho - \sigma}_1 + \frac12\abs{\tr\mleft[\rho - \sigma\mright]},
\end{align*}
where the last term, which is absent in usual definitions of the trace distance, accounts for subnormalized states.
It is easy to see that $T(\rho,\sigma) \leq \norm{\rho - \sigma}_1 \leq 2T(\rho,\sigma)$.
The Fuchs-van de Graaff inequalities (see Lemma 3.17 in \cite{tomamichel2015quantum}) relate the trace distance and purified distance:
\begin{align}\label{eq:fuchs vd graaff}
  T(\rho,\sigma) \leq P(\rho,\sigma) \leq \sqrt{2T(\rho,\sigma)}
\end{align}
for $\rho, \sigma \in \Pleq(\HH)$.

There are also conditional versions of the R\'enyi entropies.
Consider a bipartite quantum state $\rho_{AB} \in \Peq(AB)$.
For the von Neumann entropy, the conditional entropy can simply be defined as an entropy difference, namely~$H(A\vert B)_\rho = H(AB)_\rho - H(B)_\rho$.
However, it turns out that this is not a good definition in the R\'enyi case.
There are various ways to define a R\'enyi conditional entropy $H_k(A \vert B)$; we use a version based on the so-called sandwiched R\'enyi relative entropy.
For $k=2$, this gives a \emph{quantum conditional collision entropy}, which will be useful for defining minimal cuts and which is defined as follows.
For $\rho_{AB} \in \Pleq(AB)$, let
\begin{align}\label{eq:sandwiched renyi two}
  H_2(A\vert B)_{\rho | \rho} := -\log \tr\mleft[\left((I \ot \rho_B)^{-\frac14}\rho_{AB} (I \ot \rho_B)^{-\frac14}\right)^2\mright] + \log \tr\mleft[\rho\mright].
\end{align}
Finally, there are also conditional versions of the min- and max-entropy.
For $\rho_{AB} \in \Pleq(AB)$ and $\sigma_B \in \Pleq(B)$, we define
\begin{align*}
  H_{\min}(A\vert B)_{\rho|\sigma} & = - \inf\{ \lambda : \rho_{AB} \leq 2^\lambda I_A \ot \sigma_B\} \\
  H_{\max}(A\vert B)_{\rho|\sigma} & = \log \, \norm{\sqrt{\rho_{AB}}\sqrt{I \ot \sigma_B}}_1^2
\end{align*}
and we let
\begin{align*}
  H_{\min}(A \vert B)_{\rho} & = \sup_{\sigma \in \Pleq(B)} H_{\min}(A\vert B)_{\rho|\sigma}  \\
  H_{\max}(A \vert B)_{\rho} & = \sup_{\sigma \in \Pleq(B)} H_{\max}(A\vert B)_{\rho|\sigma}.
\end{align*}
We can also define their smoothed versions
\begin{align*}
  H^\eps_{\min}(A\vert B)_\rho & = \sup_{\rho^{\eps} \in \Pleq(AB), P(\rho^\eps,\rho) \leq \eps} H_{\min}(A\vert B)_{\rho^\eps}            \\
  H^\eps_{\max}(A\vert B)_\rho & = \underset{\rho^{\eps} \in \Pleq(AB), P(\rho^\eps,\rho) \leq \eps}{\inf} H_{\max}(A\vert B)_{\rho^\eps}.
\end{align*}
There is a duality between (smooth) max- and min-entropies.
If $\rho \in \Pleq(ABC)$ is a \emph{pure} state, it holds that
\begin{align}\label{eq:min entropy vs max entropy}
  H^{\eps}_{\min}(A\vert B)_{\rho} = - H^{\eps}_{\max}(A\vert C)_{\rho}.
\end{align}
We will use the fact that for a normalized state $\rho_{AB} \in \Peq(AB)$ (Corollary 5.10 in \cite{tomamichel2015quantum})
\begin{align}\label{eq:renyi 2 vs min entropy}
  H_{\min}(A\vert B)_{\rho} \leq H_2(A\vert B)_{\rho|\rho}.
\end{align}
A final important property of conditional smooth entropies are the \emph{data processing inequalities}.
Let $\Phi$ and $\Psi$ be completely positive and trace-preserving (CPTP) maps, mapping systems $A$ to $A'$ and $B$ to $B'$ respectively, and let $\sigma = (\Phi \ot \Psi)(\rho)$.
If $\Phi$ is also subunital, and $0 \leq \eps \leq \sqrt{\tr \rho}$, then Theorem 6.2 of \cite{tomamichel2015quantum} states
\begin{align*}
  H_{\min}^{\eps}(A' \vert B')_{\sigma} \geq H_{\min}^{\eps}(A \vert B)_{\rho}\quad \text{and}\quad H_{\max}^{\eps}(A' \vert B')_{\sigma} \geq H_{\max}^{\eps}(A \vert B)_{\rho}.
\end{align*}
In fact, for the smooth min-entropy  the data processing inequality is also valid if $\Phi$ is only trace non-increasing rather than trace-preserving, see \cite{tomamichel2012framework}.

\subsection{Recovery isometries}\label{sec:mincut}
Recall that we study random tensor network states with \emph{link states}, pure states placed on each edge whose tensor product forms the full state on edges $\phi = \bigotimes_{e \in E} \phi_e \in \Peq(V)$ for some graph $G = (V,E)$. In \cref{sec:background states}, we considered more general \emph{background states} $\phi_V \in \Pleq(V)$, where we no longer have a tensor product structure along the edges of some graph, and applying the replica trick does not yield a local spin model for the moments of the tensor network state.
This situation is of independent interest, but will also be useful as an intermediate step when applying bounds based on one-shot entropies to link states.
In \cref{sec:almost max entangled}, we studied link states for which the entanglement spectrum of the edge states $\phi_e$ had bounded variation, and we used the replica trick to compute the moments of the spectrum of $\rho_A$ for a boundary subsystem $A$.
For general background states we saw that the replica trick for $k=2$ extends as in \cref{eq:general replica trick k=2}.
What are the minimal cuts in this setting?
Based on \cref{eq:general replica trick k=2} a first guess would be that $\Gamma_A \in C(A)$ would be a minimal cut (i.e.\ correspond to the dominant term in the replica trick) if for all other cuts $\Delta_A \in C(A)$ we would have $H_2(\Gamma_A)_{\phi} \ll H_2(\Delta_A)_\phi$.
If the state is a link state, this corresponds to adding weights to the edges of the graph corresponding to the R\'enyi-2 entropies along the edges, and computing a weighted minimal cut.
Indeed, this would yield an accurate approximation of~$\tr[\rho_A^2]$ and hence of $H_2(\rho_A)$.
However, if the spectrum of $\rho_A$ is not close to a flat spectrum, this does not imply that $\spec_+(\rho_A)$ is close to $\spec_+(\phi_{\Gamma_A})$.
We would like to show that for link states with unbounded spectral variation, and an appropriate minimal cut condition for $\Gamma_A \in C(A)$, it is still true that $\spec_+(\rho_A)$ is close to $\spec_+(\phi_{\Gamma_A})$.

We will adapt the $k=2$ replica trick for general background states to get a bound on the difference in trace norm between  $\spec_+(\rho_A)$ and $\spec_+(\phi_{\Gamma_A})$ in terms of \emph{conditional} R\'enyi-2 entropies,\footnote{Note that while $H(A\vert B)_\phi = H(AB)_\phi - H(B)_\phi$, in general $H_2(A\vert B)_{\phi\vert \phi} \neq H_2(AB)_{\phi} - H_2(B)_{\phi}$.} as defined in \cref{eq:sandwiched renyi two}.
In \cref{sec:one minimal cut}, we will use this to formulate a condition for cut minimality in terms of smooth entropies for link states.


The main result of this subsection is a tensor network version of \emph{one-shot decoupling}.
Let $\phi_V \in \Pleq(V)$. We allow $\phi_V$ to be a general state, which need not be pure and also need not be a product state along the edges of some graph. Let $R$ be a purifying system and $\phi_{VR} \in \Pleq(VR)$ be a purification of $\phi_V$.
Then we can construct the random tensor network state $\rho_{V_{\partial}R}$ where the boundary systems are given by $V_\partial \cup R$, which is a purification of the random tensor network state $\rho_{V_\partial}$ as in \cref{eq:rtn background state} by
\begin{align}\label{eq:purified rtn state}
  \phi_{V_\partial R} = \tr_{V_b}[(I_{V_\partial R} \ot \psi) \phi]
\end{align}
where $\psi$ is a tensor product of random tensors.
We briefly recall our notation for boundary subsystems and cuts: for a boundary subsystem $A \subseteq V_\partial$, we denote its boundary complement by $\bar{A} = V_\partial \setminus A$, and for a cut~$\Gamma_A \in C(A)$, we let $\Gamma_A^c = V \setminus \Gamma_A$, which is a cut for $\bar{A}$.
The purifying system $R$ can be thought of as an additional boundary system in the tensor network construction.

In \cref{thm:split transfer}, we will assume that we have a cut $\Gamma_A \in C(A)$ which is such that for all cuts $\Delta_A \in C(A)$ for which $\Delta_A \subsetneq \Gamma_A$ we have $H_2(\Gamma_A \setminus \Delta_A \vert \Gamma_A^c R)_{\phi \vert \phi} \gg 1$, and similarly for all cuts~$\Delta_A \in C(A)$ for which~$\Gamma_A \subsetneq \Delta_A$ we have $H_2(\Delta_A \setminus \Gamma_A \vert \Gamma_A R)_{\phi \vert \phi} \gg 1$.
We show that this condition implies that with high probability there exist isometries $V_{A}: \HH_A \to \HH_{\Gamma_A}$ and $V_{\bar{A}} : \HH_{\bar{A}} \to \HH_{\Gamma_A^c}$ such that
\begin{align}\label{eq:recovery intro}
  (V_A \ot V_{\bar{A}} \ot I_R)\ket\rho \approx \ket\phi.
\end{align}
The approximation accuracy will be measured in trace norm.
In particular, this implies that $\spec_+(\rho_A) \approx \spec_+(\phi_{\Gamma_A})$.
If the state $\phi$ is a tensor product of link states, $\spec_+(\phi_{\Gamma_A})$ is precisely the entanglement spectrum along the cut~$\gamma_A$.
The isometries $V_A$ and $V_{\bar{A}}$ are \emph{recovery isometries}, which allow us to `recover' $\Gamma_A$ from the $A$ system, and similarly we can recover $\Gamma_A^c$ from $\bar{A}$.

The result is closely related to quantum error correction.
One way to interpret this is as follows: consider a subspace $\HH_S$ of $\HH_{V}$ and let $R$ be a reference system of dimension $\dim(\HH_S)$, and $\phi_{VR}$ a maximally entangled state between $S$ and $R$.
Then \cref{eq:recovery intro} can be interpreted as saying that if we encode the subspace $S$ by projecting onto random tensors, the information in $\Gamma_A$ is protected, after encoding, against an erasure error on $\bar{A}$.
This idea is also discussed in \cite{pastawski2015holographic} for perfect tensor network models, and in \cite{hayden2016holographic} for random tensor networks with maximally entangled link states.
In holography, the notion of local recovery isometries and their error correction interpretation goes under the name of \emph{entanglement wedge reconstruction} or \emph{subregion-subregion duality}.
See \cite{akers2020leading,akers2021quantum} for a detailed discussion of entanglement wedge reconstruction in holographic systems with bulk entropy, relating to one-shot entropies.
We provide more details in \cref{sec:split_recovery}.

Our approach to showing \cref{eq:recovery intro} is that we start by projecting only on the random tensors in $\Gamma_A$, and not on the random tensors in $\Gamma_A^c$.
This yields a random tensor network state $\sigma$ on $A\Gamma_A^c R$.

We then show that, by a version of one-shot decoupling, the reduced state on $\sigma_{\Gamma_A^c R}$ has not changed much from $\phi_{\Gamma_A^c R}$.
By Uhlmann's theorem, this implies that there exists an isometry $V_A$ such that $(V_A \ot I_{\Gamma_A^c R})\ket\sigma \approx \ket\phi$.
Combining this with a similar result for $\Gamma_A^c$ we obtain \cref{eq:recovery intro}, as will be made precise in \cref{thm:split transfer}.

In our construction of $\sigma$, we can relabel the vertices in the graph, and think of the vertices in $\Gamma_A \setminus A$ as the bulk vertices $V_b$, the boundary subsystem $A$ as the complete boundary $V_{\partial}$, and relabel all other subsystems as the reference system $R$.
Then we prove the following result, which is closely related to the one-shot decoupling results in \cite{dupuis2014one}.

\begin{prop}\label{prop:one-shot decoupling}
  Consider a random tensor network state $\rho_{V_\partial R}$ as in \cref{eq:purified rtn state} with a (purified) background state $\phi_{VR} \in \Pleq(VR)$.
  Let $A = V_\partial$ and let $\Gamma_A = V$ and suppose that for any cut $\Delta_A \in C(A)$ other than $\Gamma_A$
  \begin{align*}
    H_2(\Gamma_A \setminus \Delta_A \vert R)_{\phi \vert \phi} \geq \KK
  \end{align*}
  then
  \begin{align*}
    \EE \norm{\rho_{R} - \phi_{R}}_1 \leq 2^{\tfrac{\abs{V_b}}{2}}\sqrt{\tr[\phi]}2^{-\frac12 \KK}.
  \end{align*}
\end{prop}

Note that, since $\Gamma_A = V_b \cup A$, the sets $\Gamma_A \setminus \Delta_A$ for $\Delta_A \in C(A) \setminus \{\Gamma_A\}$ are exactly the non-empty subsets of $V_b$.
The formulation in terms of $\Delta_A \in C(A)$ will be natural when we apply this result in \cref{thm:split transfer}.

\begin{proof}
  We closely follow the strategy in \cite{dupuis2014one,dutil2010one}. We first note a basic fact (Lemma 3.7 in \cite{dupuis2014one}): for any operator $X$ and $\omega$ a subnormalized density matrix, it holds that
  \begin{equation}\label{eq:trace dist from sanwiched renyi}
    \norm{X}_1 \leq \norm{\omega^{-\frac14} X \omega^{-\frac14}}_2.
  \end{equation}
  The proof is an application of the Cauchy-Schwarz inequality.
  We use \cref{eq:trace dist from sanwiched renyi} with $\omega = \phi_{R}$ and Jensen's inequality to see that
  \begin{align*}
    \EE \norm{\rho_{R} - \phi_{R}}_1 \leq \sqrt{\EE \tr[(\tilde\rho_{R} - \tilde\phi_{R})^2]}
  \end{align*}
  where $\tilde\rho_{V_{\partial} R} = (I \ot \phi_R)^{-\frac14} \rho_{V_{\partial} R} (I \ot \phi_R)^{-\frac14}$ and $\tilde\phi_{VR} = (I \ot \phi_R)^{-\frac14} \phi_{VR} (I \ot \phi_R)^{-\frac14}$.
  Now $\EE \tilde\rho_{R} = \tilde\phi_{R}$ by \cref{eq:expectation of rtn state}, and the replica trick in \cref{eq:general replica trick k=2} yields
  \begin{align*}
    \EE \tr[(\tilde\rho_{R} - \tilde\phi_{R})^2] & = \EE \tr[\tilde\rho_{R}^2] - \tr[\tilde\phi_{R}^2]                                                                          \\
                                                 & = \sum_{\Delta_A \in C(A), \Delta_A \subsetneq \Gamma_A} \tr[\tilde\phi_{(\Gamma_A \setminus \Delta_A) R}^2]                 \\
                                                 & = \sum_{\Delta_A \in C(A), \Delta_A \subsetneq \Gamma_A} \tr[\phi] 2^{-H_2(\Gamma_A \setminus \Delta_A \vert R)_{\phi|\phi}}
  \end{align*}
  using the definition of $\tilde\phi$ and \cref{eq:sandwiched renyi two} and hence
  \begin{align*}
    \left(\EE \norm{\rho_{R} - \phi_{R}}_1\right)^2 \leq 2^{\abs{V_b}}\tr[\phi] 2^{-\KK}.
  \end{align*}
\end{proof}

Suppose that in the set-up of \cref{prop:one-shot decoupling}, we would have equality $\rho_{R} = \phi_{R}$.
Then, by Uhlmann's theorem, their purifications $\rho_{AR}$ and $\phi_{\Gamma_A R}$ are related by an isometry $V_A$ from $A$ to $\Gamma_A$.
The following lemma is useful to extend to the case where the reduced states are close in trace distance.

\begin{lem}\label{lem:unnormalized fidelity}
  Suppose $\rho_{AB} \in \PSD(AB)$ and $\sigma_{AC} \in \Pleq(AC)$ are pure states on Hilbert spaces $\HH_A \ot \HH_B$ and $\HH_A \ot \HH_C$ respectively.
  Then
  \begin{align*}
    \min_{V} \norm{(I_A \ot V) \rho_{AB} (I_A \ot V^\dagger) - \sigma_{AC}}_1 \leq 2\sqrt{2\norm{ \rho_{A} - \sigma_{A}}_1 + 2\norm{ \rho_{A} - \sigma_{A}}_1^2}.
  \end{align*}
  where the minimum is over all isometries $V : \HH_B \to \HH_C$.
\end{lem}
\begin{proof}
  Uhlmann's theorem states that if $\rho_{AB} \in \PSD(AB)$ and $\sigma_{AC} \in \PSD(AC)$ are pure quantum states with~$\dim(\HH_B) \leq \dim(\HH_C)$, then there exists an isometry $V : \HH_B \to \HH_C$ such that
  \begin{align*}
    P(\rho_A,\sigma_A) = P((I_A \ot V) \rho_{AB} (I_A \ot V^\dagger), \sigma_{AC})
  \end{align*}
  and, in particular, the isometry is the solution to an optimization problem:
  \begin{align*}
    P(\rho_A, \sigma_A) = \min_V P((I_A \ot V) \rho_{AB} (I_A \ot V^\dagger), \sigma_{AC}).
  \end{align*}
  Moreover, if both $\rho$ and $\sigma$ are subnormalized, by \cref{eq:fuchs vd graaff}, we can bound
  \begin{align*}
    \min_{V} \norm{(I_A \ot V) \rho_{AB} (I_A \ot V^\dagger) - \sigma_{AC}}_1 & \leq \min_{V} 2P((I_A \ot V) \rho_{AB} (I_A \ot V^\dagger), \sigma_{AC}) \\
                                                                              & = 2P(\rho_A,\sigma_A)                                                    \\
                                                                              & \leq 2\sqrt{2\norm{\rho_{A} - \sigma_{A}}_1}.
  \end{align*}
  From this it follows that if $\sigma$ is subnormalized and $\rho$ has $\tr[\rho] > 1$,
  \begin{align*}
    \min_{V} \norm{(I_A \ot V) \rho_{AB} (I_A \ot V^\dagger) - \sigma_{AC}}_1 \leq 2\sqrt{2\tr[\rho]\norm{ \rho_{A} - \sigma_{A}}_1}.
  \end{align*}
  Since $\tr[\rho] \leq \tr[\sigma] + \norm{\rho - \sigma}_1$ and $\tr[\sigma] \leq 1$ we conclude that
  \begin{align*}
    \min_{V} \norm{(I_A \ot V) \rho_{AB} (I_A \ot V^\dagger) - \sigma_{AC}}_1 \leq 2\sqrt{2\norm{ \rho_{A} - \sigma_{A}}_1 + 2\norm{ \rho_{A} - \sigma_{A}}_1^2}.
  \end{align*}
  for arbitrary $\rho$ and subnormalized $\sigma$.
\end{proof}

Finally, we will need a basic lemma relating tensor network states with differing background states:
\begin{lem}\label{lem:smoothing rtn state}
  Suppose we consider random tensor network states $\rho_{V_\partial R}$ and $\tilde \rho_{V_\partial R}$ with (purified) background states  $\phi_{V R}, \tilde \phi_{VR} \in \Pleq(VR)$ and projecting onto the same random tensors.
  Then
  \begin{align*}
    \EE \norm{\rho_{V_\partial R} - \tilde \rho_{V_\partial R} }_1 \leq \norm{\phi_{VR} - \tilde{\phi}_{VR}}_1.
  \end{align*}
\end{lem}

\begin{proof}
  Let $\phi - \tilde{\phi} = \Delta_+ - \Delta_-$ where both $\Delta_+$ and $\Delta_-$ are positive semidefinite and are such that $\norm{\phi - \tilde{\phi}}_1 = \tr[\Delta_+ + \Delta_-]$.
  Then we can also consider the random tensor network states $\sigma_+$ and $\sigma_-$ which take~$\Delta_+$ and~$\Delta_-$ as background states, and by the linearity of \cref{eq:purified rtn state} in the background state we have $\rho - \tilde \rho = \sigma_+ - \sigma_-$.
  By \cref{eq:expectation of rtn state}, $\EE \sigma_{\pm} = \Delta_{\pm}$.
  We then estimate
  \begin{align*}
    \EE \norm{\rho - \tilde \rho}_1 = \EE \norm{\sigma_+ - \sigma_-}_1 \leq \EE (\norm{\sigma_+}_1 + \norm{\sigma_-}_1) = \EE \tr[\sigma_+ + \sigma_-] = \tr[\Delta_+ + \Delta_-] = \norm{\phi - \tilde{\phi}}_1.
  \end{align*}
  where we have used that $\sigma_+$ and $\sigma_-$ are positive semidefinite and hence $\norm{\sigma_{\pm}}_1 = \tr[\sigma_{\pm}]$.
\end{proof}

With all our tools assembled, we are ready to prove the main result of this subsection. We again let~$\phi_{VR} \in \Pleq(VR)$ be a background state with $R$ a purifying system, and we let $\rho_{V_{\partial}R}$ be the associated random tensor network state as constructed in \cref{eq:purified rtn state}. Let $\Gamma_A$ be an arbitrary cut for the boundary region~$A$.
In \cref{thm:split transfer}, we provide a criterion to determine whether $\Gamma_A$ is a minimal cut in terms of conditional entropies.
Informally speaking, the following result shows that if $\Gamma_A$ is a minimal cut in this sense, we can recover the system $\Gamma_A$ from the boundary subsystem $A$, while we can recover $\Gamma_A^c$ from the boundary subsystem~$\bar{A}$.
For general $\phi$, \cref{thm:split transfer} is closely related to the task of \emph{split transfer}, see \cref{sec:split_recovery} for a discussion.
The following result closely follows Proposition 18 of~\cite{dutil2010one}.

\begin{thm}[Recovery isometries]\label{thm:split transfer}
  Let $\phi_{VR} \in \Pleq(VR)$ and let $\rho_{V_{\partial}R}$ be the associated random tensor network state as in \cref{eq:purified rtn state}.
  Let $\Gamma_A \in C(A)$ and suppose that
  \begin{align}\label{eq:H2 condition 1}
    H_2(\Gamma_A \setminus \Delta_A \vert \Gamma_A^c R)_{\phi|\phi} \geq \KK_1
  \end{align}
  for all cuts $\Delta_A \in C(A)$ such that $\Delta_A \subsetneq \Gamma_A$ and
  \begin{align}\label{eq:H2 condition 2}
    H_2(\Delta_A\setminus \Gamma_A \vert \Gamma_A R)_{\phi|\phi} \geq \KK_2
  \end{align}
  for all cuts $\Delta_A \in C(A)$ such that $\Gamma_A \subsetneq \Delta_A$.
  Then
  \begin{align}\label{eq:split transfer}
    \EE \min_{V_A, V_{\bar{A}}} \norm{(V_A \ot V_{\bar{A}} \ot I_R)\rho_{V_{\partial}R}(V_A^\dagger \ot V_{\bar{A}}^\dagger \ot I_R) - \phi_{VR}}_1 = \bigO(\tr[\phi]^{\frac14}(2^{-\frac14\KK_1} + 2^{-\frac14\KK_2})).
  \end{align}
  where the minimum is over isometries $V_{A}: \HH_A \to \HH_{\Gamma_A}$ and $V_{\bar{A}} : \HH_{\bar{A}} \to \HH_{\Gamma_A^c}$.
\end{thm}

\begin{proof}
  Let $\sigma_{A\Gamma_A^c R}$ be the state where we have contracted along the tensors in $\Gamma_A$ but not along those in $\Gamma_A^c$, and similarly let $\tau_{\bar{A} \Gamma_A R}$ be the state where we have contracted along the tensors in $\Gamma_A^c$ but not along those in $\Gamma_A$.
  We first use \cref{prop:one-shot decoupling} to show that $\sigma_{\Gamma_A^c R} \approx \phi_{\Gamma_A^c R}$ and $\tau_{\Gamma_A R} \approx \phi_{\Gamma_A R}$.
  Indeed, for $\sigma$ we simply apply \cref{prop:one-shot decoupling} with $V_b \cap \Gamma_A$ as the set of bulk vertices, $A$ as the set of boundary vertices and $\Gamma_A^c R$ as the reference system.
  This gives
  \begin{align*}
    \EE \norm{\sigma_{\Gamma_A^c R} - \phi_{\Gamma_A^c R}}_1 =\bigO(\sqrt{\tr[\phi]}2^{-\frac12 \KK_1}).
  \end{align*}
  A similar application of \cref{prop:one-shot decoupling}, with $V_b \cap \Gamma_A^c$ as the set of bulk vertices, $\bar{A}$ as the set of boundary vertices and $\Gamma_A R$ as the reference system, shows
  \begin{align*}
    \EE \norm{\tau_{\Gamma_A R} - \phi_{\Gamma_{A R}}}_1 = \bigO(\sqrt{\tr[\phi]}2^{-\frac12 \KK_2}).
  \end{align*}
  We note that for any isometries $V_{A}: \HH_A \to \HH_{\Gamma_A}$ and $V_{\bar{A}} : \HH_{\bar{A}} \to \HH_{\Gamma_A^c}$
  \begin{align*}
    \norm{(V_A \ot V_{\bar{A}} \ot I_R)\rho(V_A^\dagger \ot V_{\bar{A}}^\dagger \ot I_R) - \phi_{VR}}_1 & \leq \norm{(I_{\Gamma_A} \ot V_{\bar{A}} \ot I_R)\tau(I_{\Gamma_A} \ot V_{\bar{A}}^\dagger \ot I_R) - \phi_{VR}}_1 \\
                                                                                                        & \qquad + \norm{\tau - (V_A \ot I_{\bar{A} R})\rho(V_A^\dagger \ot I_{\bar{A} R})}_1,
  \end{align*}
  where we have applied the triangle inequality after adding and subtracting $(I_{\Gamma_A} \ot V_{\bar{A}} \ot I_R)\tau(I_{\Gamma_A} \ot V_{\bar{A}}^\dagger \ot I_R)$, and then using the invariance of the trace norm under isometries in the second term.
  We use this to estimate
  \begin{align}\label{eq:estimate recovery error}
    \begin{split}
      &\EE \min_{V_A, V_{\bar{A}}} \norm{(V_A \ot V_{\bar{A}} \ot I_R)\rho(V_A^\dagger \ot V_{\bar{A}}^\dagger \ot I_R) - \phi_{VR}}_1 \\
      & \qquad \leq \EE \bigl( \min_{V_{\bar{A}}} \norm{(I_{\Gamma_A} \ot V_{\bar{A}} \ot I_R)\tau(I_{\Gamma_A} \ot V_{\bar{A}}^\dagger \ot I_R) - \phi_{VR}}_1 + \min_{V_A} \norm{\tau - (V_A \ot I_{\bar{A} R})\rho(V_A^\dagger \ot I_{\bar{A} R})}_1 \bigr),
    \end{split}
  \end{align}
  where the minimum is over isometries $V_{A}: \HH_A \to \HH_{\Gamma_A}$ and $V_{\bar{A}} : \HH_{\bar{A}} \to \HH_{\Gamma_A^c}$.
  For the first term of \cref{eq:estimate recovery error}, we apply \cref{lem:unnormalized fidelity} to get
  \begin{align*}
    \min_{V_{\bar{A}}} \norm{(I_{\Gamma_A} \ot V_{\bar{A}} \ot I_R)\tau(I_{\Gamma_A} \ot V_{\bar{A}}^\dagger \ot I_R) - \phi_{VR}}_1
     & \leq 2\sqrt{2 \norm{\tau_{\Gamma_A R} - \phi_{\Gamma_{A R}}}_1 + 2\norm{\tau_{\Gamma_A R} - \phi_{\Gamma_{A R}}}_1^2}                 \\
     & \leq 2\sqrt2 \left( \sqrt{\norm{\tau_{\Gamma_A R} - \phi_{\Gamma_{A R}}}_1} + \norm{\tau_{\Gamma_A R} - \phi_{\Gamma_{A R}}}_1\right)
  \end{align*}
  and by Jensen's inequality
  \begin{align}\label{eq:tau isometry}
    \begin{split}
      &\EE \min_{V_{\bar{A}}} \norm{(I_{\Gamma_A} \ot V_{\bar{A}} \ot I_R)\tau(I_{\Gamma_A} \ot V_{\bar{A}}^\dagger \ot I_R) - \phi_{VR}}_1 \\
      & \qquad \leq 2\sqrt2 \left( \sqrt{\EE \norm{\tau_{\Gamma_A R} - \phi_{\Gamma_{A R}}}_1} + \EE\norm{\tau_{\Gamma_A R} - \phi_{\Gamma_{A R}}}_1\right) \\
      &\qquad = \bigO(\tr[\phi]^{\frac14}2^{-\frac14 \KK_2})
    \end{split}
  \end{align}
  For the second term of \cref{eq:estimate recovery error}, we can think of $\tau$ and $(V_A \ot I_{\bar{A} R})\rho(V_A^\dagger \ot I_{\bar{A} R})$ as the random tensor network states with $\phi$ and $(V_A \ot I_{\Gamma_A^c R})\sigma (V_A^\dagger \ot I_{\Gamma_A^c R})$ as the full state on edges, applying random tensors in $\Gamma_A^c$.
  Then, denoting by $\EE_{\Gamma_A^c}$ the expectation value over all random tensors in $\Gamma_A^c$, by \cref{lem:smoothing rtn state}
  \begin{align*}
    \EE_{\Gamma_A^c} \norm{\tau - (V_A \ot I_{\bar{A} R})\rho(V_A^\dagger \ot I_{\bar{A} R})}_1 = \bigO(\norm{ \phi_{VR} - (V_A \ot I_{\bar{A} R})\sigma(V_A^\dagger \ot I_{\bar{A} R})}_1).
  \end{align*}
  We thus estimate
  \begin{align*}
    \EE \min_{V_A} \norm{\tau - (V_A \ot I_{\bar{A} R})\rho(V_A^\dagger \ot I_{\bar{A} R})}_1 & = \bigO(\EE_{\Gamma_A} \min_{V_A} \norm{ \phi_{VR} - (V_A \ot I_{\bar{A} R})\sigma(V_A^\dagger \ot I_{\bar{A} R})}_1)
  \end{align*}
  for which we may argue exactly as in \cref{eq:tau isometry} and using \cref{lem:unnormalized fidelity} that
  \begin{align*}
    \EE_{\Gamma_A} \min_{V_A} \norm{ \phi_{VR} - (V_A \ot I_{\Gamma_A^c R})\sigma(V_A^\dagger \ot I_{\Gamma_A^c R})}_1 = \bigO(\tr[\phi]^{\frac14}2^{-\frac14 \KK_1}).
  \end{align*}
  We conclude that
  \begin{align*}
    \EE \min_{V_A, V_{\bar{A}}} \norm{(V_A \ot V_{\bar{A}} \ot I_R)\rho(V_A^\dagger \ot V_{\bar{A}}^\dagger \ot I_R) - \phi_{VR}}_1 & = \bigO(\tr[\phi]^{\frac14}(2^{-\frac14 \KK_1} + 2^{-\frac14 \KK_2})).
  \end{align*}
\end{proof}

We hence find that the closeness of the boundary and background state can be bounded via conditional R\'enyi-2 entropies of cuts. In particular, for large $\KK_1,\KK_2$, the recovery isometries can recover states to good accuracy, and we find that $\EE \norm{\spec_+(\rho_A) - \spec_+(\phi_{\Gamma_A})}$ is small.
However, this result is not yet completely satisfying.
The conditional R\'enyi-2 entropy is not a `robust' quantity, in the sense that a small deformation of $\phi$ can drastically change the values of the conditional R\'enyi-2 entropies in \cref{eq:H2 condition 1} and \cref{eq:H2 condition 2}.
For this reason, we would like a condition with \emph{smoothed} entropies.
We first note that one can actually show that for the condition in \cref{eq:H2 condition 1}, we can bound
\begin{align}\label{eq:min max condition 1}
  H_2(\Gamma_A \setminus \Delta_A \vert \Gamma_A^c R)_{\phi\vert \phi} \geq H_{\min}(\Gamma_A \setminus \Delta_A \vert \Gamma_A^c R)_{\phi}
\end{align}
and similarly for \cref{eq:H2 condition 2},
\begin{align}\label{eq:min max condition 2}
  H_2(\Delta_A \setminus \Gamma_A \vert \Gamma_A R)_{\phi\vert \phi} \geq H_{\min}(\Delta_A \setminus \Gamma_A \vert \Gamma_A R)_{\phi}
\end{align}
To make the condition `robust', we would like to replace these by smoothed entropies and express a condition in terms of $H^{\eps}_{\min}(\Gamma_A \setminus \Delta_A \vert \Gamma_A^c R)_{\phi}$ and $H^{\eps}_{\min}(\Delta_A \setminus \Gamma_A \vert \Gamma_A R)_{\phi}$.
This will require \emph{simultaneous smoothing}: finding a state $\phi^\eps_{VR} \in \Pleq(VR)$ which is close to $\phi$, such that $H_{\min}(\Gamma_A \setminus \Delta_A \vert \Gamma_A^c R)_{\phi^{\eps}} \geq  H^{\eps}_{\min}(\Gamma_A \setminus \Delta_A \vert \Gamma_A^c R)_{\phi}$ and $H_{\min}(\Delta_A \setminus \Gamma_A \vert \Gamma_A R)_{\phi^{\eps}} \geq H^{\eps}_{\min}(\Delta_A \setminus \Gamma_A \vert \Gamma_A R)_{\phi}$ for all relevant cuts $\Delta_A$.
If we have a general background state, it is not known how this can be done \cite{drescher2013simultaneous,dutil2011multiparty}.
However, if the background state is actually a tensor product of link states, we can perform the simultaneous smoothing.

\subsection{One minimal cut}\label{sec:one minimal cut}
The primary result in this subsection is \cref{thm:rtn with arbitrary link states}, which states that if the background state is actually a tensor product of link states as in \cref{eq:link state}, then the spectrum of the boundary state is well-approximated by the spectrum of the minimal cut link state in expectation, where the approximation accuracy is controlled by \emph{smooth} one-shot entropies. It is a straightforward application of \cref{thm:split transfer}.

For a cut $\Gamma_A \in C(A)$ we define
\begin{align*}
  \mathcal C_1(\Gamma_A) & = \{\Delta_A \in C(A) : \Delta_A \subsetneq \Gamma_A \}  \\
  \mathcal C_2(\Gamma_A) & = \{\Delta_A \in C(A) : \Gamma_A \subsetneq \Delta_A \}.
\end{align*}

The key result we need is the following lemma, which we prove in \cref{sec:joint smoothing}, which shows that if we have a link state, we can perform the desired joint smoothing.

\begin{restatable}{lem}{jointrelativeminsmoothing}\label{lem:joint relative min smoothing}
  Let $\phi \in \Peq(V)$ be a link state, $A \subseteq V_{\partial}$ a boundary subsystem and $\Gamma_A \in C(A)$ a cut for $A$.
  Then there exists a pure state $\phi^{\eps} \in \Pleq(V)$ which is such that
  \begin{align*}
    P(\phi,\phi^\eps) \leq 2\left(\sqrt{\abs{\mathcal C_1(\Gamma_A)}} + \sqrt{\abs{\mathcal C_2(\Gamma_A)}}\right) \sqrt{\eps}
  \end{align*}
  and it holds that for any $\Delta_A \in\mathcal C_1(\Gamma_A)$
  \begin{align*}
    H_{\min}(\Gamma_A \setminus \Delta_A \vert \Gamma_A^c)_{\phi^\eps} \geq H_{\min}^{\eps}(\Gamma_A \setminus \Delta_A \vert \Gamma_A^c)_{\phi}
  \end{align*}
  and for any $\Delta_A \in\mathcal C_2(\Gamma_A)$
  \begin{align*}
    H_{\min}(\Delta_A \setminus \Gamma_A \vert \Gamma_A)_{\phi^\eps} \geq H_{\min}^{\eps}(\Delta_A \setminus \Gamma_A \vert \Gamma_A)_{\phi}.
  \end{align*}
\end{restatable}

We now define the notion of a minimal cut for an arbitrary link state.
\begin{dfn}[Generalized minimal cut]\label{dfn:gen mincut}
  A cut $\Gamma_{A}$ is an \emph{$(\eps,\KK)$-minimal cut} if for all $\Delta_A \in \mathcal C_1(\Gamma_A)$
  \begin{align*}
    H_{\min}^{\eps}(\Gamma_A \setminus \Delta_A \vert \Gamma_A^c)_{\phi} \geq K
  \end{align*}
  and for all $\Delta_A \in\mathcal C_2(\Gamma_A)$
  \begin{align*}
    H_{\min}^{\eps}(\Delta_A \setminus \Gamma_A \vert \Gamma_A)_{\phi} \geq K.
  \end{align*}
\end{dfn}

This definition is consistent with the smooth entropy conditions for minimal surfaces in holography from \cite{akers2020leading}.
The following is now a straightforward consequence of \cref{thm:split transfer} and \cref{lem:joint relative min smoothing}.
It justifies our notion of a generalized minimal cut, as it controls the degree to which the spectrum of the corresponding cut link state $\phi_{\Gamma_A}$ is close to the boundary state $\rho_A$.

\begin{thm}\label{thm:rtn with arbitrary link states}
  Consider a random tensor network state $\rho$ constructed with $\phi \in \Peq(V)$ a tensor product of link states as in \cref{eq:link state}. Let $A$ be a boundary region of the network, $\rho_A$ the corresponding boundary state, and~$\Gamma_A$ an~$(\eps,\KK)$-minimal cut.
  Then the spectra of $\rho_A$ and the state $\phi_{\Gamma_A}$ on $A$ are related as:
  \begin{align*}
    \EE \norm{\spec_+(\rho_A) - \spec_+(\phi_{\Gamma_A})}_1 = \mathcal O(2^{-\frac14 K} + \sqrt{\eps}).
  \end{align*}
\end{thm}

\begin{proof}
  Let $\phi^{\eps}$ be a state as constructed in \cref{lem:joint relative min smoothing}, and let $\rho^{\eps}$ be the random tensor network state using this background state.
  Then, by \cref{thm:split transfer}
  \begin{align*}
    \EE \norm{\spec_+(\rho_A^{\eps}) - \spec_+(\phi^{\eps}_{\Gamma_A})}_1 & \leq \EE \min_{V_A, V_{\bar{A}}} \norm{(V_A \ot V_{\bar{A}} \ot I_R)\rho^{\eps}(V_A^\dagger \ot V_{\bar{A}}^\dagger \ot I_R) - \phi^{\eps}}_1 \\
                                                                          & = \bigO(\tr[ \phi]^{\frac14}(2^{-\frac14\KK_1} + 2^{-\frac14\KK_2}))
  \end{align*}
  where $\KK_1$ is the minimal value over $\Delta_A \in \mathcal C_1(\Gamma_A)$ of
  \begin{align*}
    H_2(\Gamma_A \setminus \Delta_A \vert \Gamma_A^c)_{\phi^{\eps}|\phi^{\eps}} \geq H_{\min}(\Gamma_A \setminus \Delta_A \vert \Gamma_A^c)_{\phi^{\eps}} - \log \tr[\phi^\eps] \geq H_{\min}^{\eps}(\Gamma_A \setminus \Delta_A \vert \Gamma_A^c)_{\phi} - \log \tr[\phi^\eps]
  \end{align*}
  using \cref{eq:renyi 2 vs min entropy}, the defining property of $\phi^{\eps}$ from \cref{lem:joint relative min smoothing} and accounting for the normalization.
  Similarly $\KK_2$ is the minimal value over $\Delta_A \in \mathcal C_2(\Gamma_A)$ of
  \begin{align*}
    H_2(\Delta_A\setminus \Gamma_A \vert \Gamma_A)_{\phi^{\eps}|\phi^{\eps}} \geq H_{\min}(\Delta_A\setminus \Gamma_A \vert \Gamma_A)_{\phi^{\eps}} - \log \tr[\phi^\eps] \geq H_{\min}^{\eps}(\Delta_A\setminus \Gamma_A \vert \Gamma_A)_{\phi} - \log \tr[\phi^\eps]
  \end{align*}
  and hence $K_1 \geq \KK$ and $K_2 \geq \KK$, so
  \begin{align*}
    \EE \norm{\spec_+(\rho_A^{\eps}) - \spec_+(\tilde \phi^{\eps}_{\Gamma_A})}_1 = \bigO(2^{-\frac14\KK}).
  \end{align*}
  Moreover, $\norm{\phi - \phi^{\eps}}_1 \leq 2T(\phi,\phi^{\eps}) \leq 2P(\phi,\phi^{\eps}) = \bigO(\sqrt{\eps})$ by \cref{eq:fuchs vd graaff} and hence by \cref{lem:smoothing rtn state}
  \begin{align*}
    \EE \norm{\rho - \rho^{\eps}}_1 = \bigO(\sqrt{\eps}).
  \end{align*}
  We conclude that
  \begin{align*}
    \EE \norm{\spec_+(\rho_A) - \spec_+(\phi_{\Gamma_A})}_1 & \leq \EE \norm{\spec_+(\rho_A^{\eps}) - \spec_+(\phi^{\eps}_{\Gamma_A})}_1  + \EE \norm{\rho_A - \rho^{\eps}_A}_1 +  \norm{\phi_{\Gamma_A} - \phi^{\eps}_{\Gamma_A}}_1 \\
                                                            & = \bigO(2^{-\frac14\KK} + \sqrt{\eps}).
  \end{align*}
\end{proof}

\subsection{Two competing minimal cuts}\label{sec:two minimal cuts}
We now consider the case of two minimal cuts for $A$, where the link states do not have bounded spectral variation as in \cref{sec:almost max entangled}. We cannot directly apply \cref{thm:rtn with arbitrary link states} if there are two competing minimal cuts (note that we will need to define what this means exactly), nor can we assume that the empirical measure, or even a rescaling of it, will converge.
Instead, we will see that the spectrum can be approximated in trace distance in a way that allows one to compute entropies, as well as showing convergence of certain measures depending on the spectrum of the reduced state. In particular, under these regularity conditions, we will prove the main result of this subsection: the distribution defined by the boundary state converges to the measure defined by the pushforward along the $\min$-function acting on the distributions of the two competing minimal cuts.
The intuition behind this is that the state on the edges can be approximated by a superposition of two states which both do have a unique minimal cut.
As a corollary of this result, we will then show how to approximate the von Neumann entropy of the boundary state in such a situation.

We begin by introducing the relevant family of probability distributions for our purposes. Given a probability distribution vector $p \in \RR^d$ of $d$ outcomes and $f(d) > 0$ and $g(d) \in \RR$, consider the random variable which takes value $f(d)[-\log(p_i) - g(d)]$ with probability $p_i$.
It has distribution
\begin{align}\label{eq:central limit distribution}
  \nu_p = \sum_{i = 1}^d p_i \delta_{f(d)[-\log(p_i) - g(d)]}.
\end{align}
Typically, we will have families of probability distributions with increasing $d$, let $f(d) = (\log d)^{-1}$ and choose $g(d)$ such that it corresponds to the entropy of $p$.
We will also write $\nu_\rho = \nu_{\spec(\rho)}$ for a quantum state $\rho$.
We may also consider \cref{eq:central limit distribution} for the case where $p \in \RR^d$ is positive but normalized, in which case $\nu_p$ is a finite measure.
To motivate the study of the distribution in \cref{eq:central limit distribution}, we note that it is closely related to the central limit theorem.
Let $d = d_0^n$ and $\rho = \rho_0^{\ot n}$ for some state $\rho_0$ on $\CC^{d_0}$. Then
\begin{align*}
  \nu^{(n)}_\rho = \sum_{\lambda \in \spec(\rho_0^{\ot n})} \lambda\delta_{\frac{1}{\sqrt n}[-\log(\lambda) - nH(\rho_0)]}
\end{align*}
converges in distribution to a normal distribution by the central limit theorem.
We study this particular measure because the empirical measures we investigated in \cref{sec:almost max entangled} may not have good convergence properties if, for example, the full state on edges consists of many copies of a single non-maximally entangled link state.
Moreover, the measure in \cref{eq:central limit distribution} clearly captures information about the entropy of the probability distribution, and turns out to capture information about second order asymptotics of information processing tasks \cite{hayashi2008second, tomamichel2013hierarchy}. In fact, the example of many independent copies of a single link state is not the only relevant situation. In \cite{czech2015information,bao2019beyond}, it was argued that the entanglement spectrum of conformal field theories with large central charge (which are the motivating example to study random tensor networks) have similar behavior.

Before diving into more technical details, we can build some intuition for how distributions of the form \cref{eq:central limit distribution} will behave in the two-cut setting. Consider the situation with a \textit{single} random tensor.
Let us take a single bulk vertex $x$ and two boundary vertices $a$ and $b$, with edges $(xa)$ and $(xb)$. There is a single random tensor at $x$, and we take link states $\ket{\phi_1}^{\ot n}$ along $(xa)$ and $\ket{\phi_2}^{\ot n}$ along $(xb)$.
We denote
\begin{align*}
  \ket{\phi_j} = \sum_i \sqrt{\lambda_{i,j}} \ket{ii}
\end{align*}
for $j = 1,2$ and we let $h_j := H(\{\lambda_{i,j}\})$ be the entanglement entropies along the two minimal link states.
There are two separating cuts for the boundary subsystem $\{a\}$:
if $h_1 < h_2$ then $(xa)$ is the minimal cut, and if $h_1 > h_2$, the minimal cut is given by $(xb)$.
In these cases, for large $n$, the entanglement spectrum of $\rho_{\{a\}}$ can be approximated by the entanglement spectrum along the minimal cut.
What happens at the `phase transition' where $h_1 = h_2 = h$?
The intuition is that we can split $\ket{\phi_1}^{\ot n} \ot \ket{\phi_2}^{\ot n}$ into a superposition of the two states, one for which the minimal surface is at $(xa)$ and one for which the minimal surface is at $(xb)$.
We find that in this case, if we let $\sigma_j^2 = \sum_i \lambda_{i,j}(\log(\frac{1}{\lambda_{i,j}}) - h)^2$, then for any bounded continuous function $f$, the quantity
\begin{align*}
  \sum_{\lambda \in \spec(\rho_{\{a\}})} \lambda f\left(\frac{-\log(\lambda) - nh}{\sqrt{n}}\right)
\end{align*}
converges to
\begin{align*}
  \frac{1}{2\pi \sigma_1\sigma_2} \int \int e^{-\frac{(x_1 - h)^2}{2\sigma_1^2} -\frac{(x_2 - h)^2}{2\sigma_2^2} } f(\min(x_1,x_2))\d x_1 \d x_2
\end{align*}
in probability as $n$ goes to infinity. Our goal in this section will be to prove a general version of this result for full random tensor networks.

In \cref{sec:almost max entangled}, we used the method of moments (if all moments of a distribution converge, and the distribution is uniquely determined by its moments, then we have weak convergence).
However, convergence of moments is a stronger condition than weak convergence, and it requires computation of all moments.
In this section, we will use that for distributions of the form \cref{eq:central limit distribution}, convergence in distribution follows from convergence in trace norm:

\begin{lem}\label{lem:trace distance convergence}
  Consider a sequence of increasing integers $\{d_n\}_{n \in \NN}$ and for each $n$, postive vectors $p^{(n)}, q^{(n)} \in \RR^{d_n}$. Let $\sigma, h: \NN \rightarrow \RR$ be such that $\sigma(d) > 0$ for $d \in \NN$ and $\lim_{d \to \infty} \sigma(d)^{-1} = 0$.
  Let
  \begin{align*}
    \nu^{(n)}_p := \sum_{i=1}^{d_n} p^{(n)}_i \delta_{\frac{1}{\sigma(d_n)}[-\log(p^{(n)}_i)-h(d_n))]}
  \end{align*}
  and similarly let
  \begin{align*}
    \nu^{(n)}_{q} = \sum_{i=1}^{d_n} q^{(n)}_i \delta_{\frac{1}{\sigma(d_n)}[-\log(q^{(n)}_i)-h(d_n)]}.
  \end{align*}
  \begin{enumerate}
    \item\label{it:weak convergence from norm} Suppose $\norm{p^{(n)} - q^{(n)}}_1$ goes to zero as $n$ goes to infinity, and $\nu^{(n)}_p \Rightarrow \nu$ for some probability distribution $\nu$.
    Then $\nu^{(n)}_{q} \Rightarrow \nu$.
    \item\label{it:weak convergence from norm expectation} Suppose the vectors $q^{(n)}$ and $p^{(n)}$ are random, $\EE \norm{p^{(n)} - q^{(n)}}_1$ goes to zero as $n$ goes to infinity, and $\nu^{(n)}_p \Rightarrow \nu$ in probability.
    Then $\nu^{(n)}_{q} \Rightarrow \nu$ in probability.
  \end{enumerate}
\end{lem}

\begin{proof}
  We need to show that
  \begin{align*}
    \abs{\int f(x) \d \nu_q^{(n)}(x) - \int f(x) \d \nu(x)} \to 0
  \end{align*}
  for any $f \in C_b(\RR)$.
  In fact, it suffices to show this for $f$ uniformly continuous (see for instance Theorem C.10 in \cite{anderson2010introduction}).
  Let $\eps > 0$, then if we assume $f$ to be uniformly continuous, there exists $\delta > 0$ such that for any $x,y$ for which $\abs{x - y} \leq \delta$, it holds that $\abs{f(x) - f(y)} \leq \eps$.
  We use a triangle inequality to bound
  \begin{align}\label{eq:split q vs p}
    \begin{split}
      \abs{\int f(x) \d \nu_q^{(n)}(x) - \int f(x) \d \nu(x)} \leq \abs{\int f(x) \d \nu_q^{(n)}(x) - \int f(x) \d \nu_p^{(n)}(x)}\\
      \ \,+ \abs{\int f(x) \d \nu_p^{(n)}(x) - \int f(x) \d \nu(x)}.
    \end{split}
  \end{align}
  Since $\nu_p^{(n)} \Rightarrow \nu$, the second term on the right hand side vanishes as $n$ goes to infinity.
  If we write $g_n(x) = \frac{1}{\sigma(d_n)}(\log \frac{1}{x} - h(d_n))$, then the first term on the right hand side of \cref{eq:split q vs p} is given by
  \begin{align*}
    \abs{\sum_{i = 1}^{d_n} q_i^{(n)} f(g_n(q_i^{(n)})) - p_i^{(n)} f(g_n(p_i^{(n)}))}.
  \end{align*}
  This can be bounded by
  \begin{align*}
    \sum_{i = 1}^{d_n} \abs{q_i^{(n)} f(g_n(q_i^{(n)})) - p_i^{(n)} f(g_n(p_i^{(n)}))} \leq \sum_{i = 1}^{d_n} \abs{q_i^{(n)} - p_i^{(n)}} \abs{f(g_n(q_i^{(n)}))} + \sum_{i = 1}^{d_n} p_i^{(n)} \abs{f(g_n(q_i^{(n)})) - f(g_n(p_i^{(n)}))}
  \end{align*}
  In this expression, the first term is bounded by
  \begin{align}\label{eq:bounded function}
    \sum_{i = 1}^{d_n} \abs{q_i^{(n)} - p_i^{(n)}} \abs{f(g_n(q_i^{(n)}))} \leq C\norm{q^{(n)} - p^{(n)}}_1
  \end{align}
  where $C$ is a uniform upper bound for $f$ (using $f \in C_b(\RR)$).
  For the second term, we partition the sum over $[d_n]$ into three sets
  \begin{align*}
    I_1 & = \{i \in [d_n] \text{ such that } 2^{-\delta\sigma(d_n)} \leq \frac{p_i^{(n)}}{q_i^{(n)}} \leq 2^{\delta\sigma(d_n)} \} \\
    I_2 & = \{i \in [d_n] \text{ such that } \frac{p_i^{(n)}}{q_i^{(n)}} > 2^{\delta\sigma(d_n)} \}                                \\
    I_3 & = \{i \in [d_n] \text{ such that } \frac{p_i^{(n)}}{q_i^{(n)}} < 2^{-\delta\sigma(d_n)} \}.
  \end{align*}
  The idea is that for $i \in I_1$, $p_i$ and $q_i$ are sufficiently close that we may use the continuity of $f$, whereas $I_2$ and $I_3$ cannot have too much weight (using that $p^{(n)}$ and $q^{(n)}$ are close in trace distance).
  Let us now make this precise.
  For the sum over the elements in $I_1$, we have
  \begin{align}\label{eq:I 1}
    \sum_{i \in I_1} p_i^{(n)} \abs{f(g_n(q_i^{(n)})) - f(g_n(p_i^{(n)}))} \leq \sum_{i \in I_1} p_i^{(n)} \eps \leq \norm{p^{(n)}}_1\eps,
  \end{align}
  using the uniform continuity of $f$ and the fact that for $i \in I_1$, we have $g(q_i^{(n)}) - g(p_i^{(n)}) = \frac{1}{\sigma(d_n)}\log\left(\frac{p_i^{(n)}}{q_i^{(n)}}\right)$, implying $\abs{g(q_i^{(n)}) - g(p_i^{(n)})} \leq \delta$ by the definition of $I_1$.
  Next, we observe that
  \begin{align*}
    \norm{p^{(n)} - q^{(n)}}_1 & \geq \sum_{i \in I_2} \abs{p_i^{(n)} - q_i^{(n)}} = \sum_{i \in I_2} p_i^{(n)}\left(1 - \frac{q_i^{(n)}}{p_i^{(n)}}\right) \geq \sum_{i \in I_2} p_i^{(n)}\left(1 - 2^{-\delta\sigma(d_n)}\right),
  \end{align*}
  and hence for $\sigma(d_n) \geq \delta^{-1}$, we have
  \begin{align}\label{eq:I 2}
    \sum_{i \in I_2} p_i^{(n)} \leq 2\norm{p^{(n)} - q^{(n)}}_1.
  \end{align}
  In analogous fashion we see that
  \begin{align*}
    \norm{p^{(n)} - q^{(n)}}_1 & \geq \sum_{i \in I_3, p_i^{(n)} > 0} \abs{p_i^{(n)} - q_i^{(n)}} = \sum_{i \in I_3, p_i^{(n)} > 0} p_i^{(n)}\left(\frac{q_i^{(n)}}{p_i^{(n)}} - 1\right) \geq \sum_{i \in I_3} p_i^{(n)}\left(2^{\delta\sigma(d_n)} - 1\right),
  \end{align*}
  so for $\sigma(d_n) \geq \delta^{-1}$ we have
  \begin{align}\label{eq:I 3}
    \sum_{i \in I_3} p_i^{(n)} \leq \norm{p^{(n)} - q^{(n)}}_1.
  \end{align}
  In conclusion, for $\sigma(d_n) \geq \delta^{-1}$, collecting \cref{eq:I 1}, \cref{eq:I 2}, \cref{eq:I 3} and using that $f$ is uniformly bounded by $C$, we can bound
  \begin{align*}
    \sum_{i = 1}^{d_n} p_i^{(n)} \abs{f(g_n(q_i^{(n)})) - f(g_n(p_i^{(n)}))} \leq \eps\norm{p^{(n)}}_1 + 3C\norm{p^{(n)} - q^{(n)}}_1.
  \end{align*}
  Together with \cref{eq:bounded function}, this implies that for $\sigma(d_n) \geq \delta$, we obtain the bound:
  \begin{align}\label{eq:bound error p^n vs q^n}
    \abs{\int f(x) \d \nu_q^{(n)}(x) - \int f(x) \d \nu_p^{(n)}(x)} \leq 4C\norm{q^{(n)} - p^{(n)}}_1 + \eps\norm{p^{(n)}}_1.
  \end{align}
  Since $\nu_{p}^{(n)} \Rightarrow \nu$, $\norm{p^{(n)}}_1 \to 1$, and since $\eps > 0$ was arbitrary, we conclude that as $n$ goes to infinity
  \begin{align*}
    \abs{\int f(x) \d \nu_q^{(n)}(x) - \int f(x) \d \nu(x)} \to 0
  \end{align*}
  proving~\ref{it:weak convergence from norm}.

  To prove~\ref{it:weak convergence from norm expectation}, we note that by Markov's inequality, it suffices to show that
  \begin{align}\label{eq:expectation value weak convergence}
    \EE \abs{\int f(x) \d \nu_q^{(n)}(x) - \int f(x) \d \nu(x)} \to 0,
  \end{align}
  where $f \in C_b(\RR)$ is a uniformly continuous function.
  By \cref{eq:split q vs p}, it suffices to show
  \begin{align*}
    \EE \abs{\int f(x) \d \nu_q^{(n)}(x) - \int f(x) \d \nu_p^{(n)}(x)} \to 0.
  \end{align*}
  Let $\eps > 0$, and let $\delta$ such that if $\abs{x - y} \leq \delta$, $\abs{f(x) - f(y)} \leq \eps$.
  Then by \cref{eq:bound error p^n vs q^n}, for $\sigma(d_n) \geq \delta^{-1}$,
  \begin{align*}
    \EE \abs{\int f(x) \d \nu_q^{(n)}(x) - \int f(x) \d \nu_p^{(n)}(x)} \leq 4C \EE \norm{q^{(n)} - p^{(n)}}_1 + \eps\EE \norm{p^{(n)}}_1.
  \end{align*}
  Since $\eps > 0$ was arbitrary, $\EE \norm{q^{(n)} - p^{(n)}}_1$ goes to zero, and $\EE \norm{p^{(n)}}_1$ goes to one, we conclude that \cref{eq:expectation value weak convergence} holds, proving~\ref{it:weak convergence from norm expectation}.
\end{proof}

The value of this result is clear: so long as we restrict ourselves to measures of the form \cref{eq:central limit distribution}, then we can prove convergence results by proving convergence of trace norms.

We will also need a basic lemma to help estimate overlaps of states.

\begin{lem}\label{lem:tracenorm off-diagonal}
  Suppose we have two bipartite pure states $\psi_{AB}, \phi_{AB} \in \PSD(AB)$.
  Then
  \begin{align*}
    \norm{\tr_{B}[\ket{\psi}\bra{\phi}]}_1 = F(\psi_B, \phi_B).
  \end{align*}
\end{lem}
\begin{proof}
  For any operator $M_A$ on $\HH_A$ it holds that
  \begin{align*}
    \norm{M_A}_1 = \sup_{U_A} \, \abs{\tr[U_A M_A]}
  \end{align*}
  where the supremum is over unitary operators $U_A$ on $\HH_A$, so
  \begin{align*}
    \norm{\tr_{B}[\ket{\psi}\bra{\phi}]}_1 = \sup_{U_A} \, \abs{\tr[U_A\tr_{B}[\ket{\psi}\bra{\phi}]]} = \sup_{U_A} \, \abs{\tr[U_A \ot I_B \ket{\psi}\bra{\phi}]} = \sup_{U_A} \, \abs{\bra{\phi} U_A \ot I_B \ket{\psi}}
  \end{align*}
  which equals $F(\psi_B, \phi_B)$ by Uhlmann's theorem.
\end{proof}

With these tools in hand, let us now return to the random tensor network setting. Let $\Gamma_{A,1}$ and $\Gamma_{A,2}$ be two cuts with associated sets of edges $\gamma_{A,1}$ and $\gamma_{A,2}$. We will assume these two cuts are non-intersecting: $\gamma_{A,1} \cap \gamma_{A,2} = \emptyset$, and, without loss of generality $\Gamma_{A,1} \subset \Gamma_{A,2}$.
We will assume that the full state on edges is a tensor product of link states along the edges, and that the associated central limit measure for the spectrum along the cuts $\gamma_{A,1}$ and $\gamma_{A,2}$ converges to a continuous probability distribution.
The most obvious application is where the link states on each edge $\phi_e$ are many copies of some single state $\phi_{e,0}$, so $\phi_e = (\phi_{e,0})^{\ot n}$, in which case the spectrum is subject to a central limit theorem and the assumptions of \cref{thm:measure convergence surface transition} are satisfied.

We will now formalize what it means for these two cuts to be competing minimal cuts.
We have fixed $\Gamma_{A,1}$ and $\Gamma_{A,2}$ with $\Gamma_{A,1} \subsetneq \Gamma_{A,2}$.
Then we consider three sets of cuts: $\mathcal C_1(\Gamma_{A,1}, \Gamma_{A,2})$ is the collection of cuts strictly contained in $\Gamma_{A,1}$, $\mathcal C_2(\Gamma_{A,1}, \Gamma_{A,2})$ is the collection of cuts which strictly contain $\Gamma_{A,2}$, and $\mathcal C_3(\Gamma_{A,1}, \Gamma_{A,2})$ is the collection of cuts which are `in between' $\Gamma_{A,1}$ and $\Gamma_{A,2}$.
That is,
\begin{align*}
  \mathcal C_1(\Gamma_{A,1}, \Gamma_{A,2}) & = \{ \Delta_A \in C(A) : \Delta_A \subsetneq \Gamma_{A,1}\}                         \\
  \mathcal C_2(\Gamma_{A,1}, \Gamma_{A,2}) & = \{ \Delta_A \in C(A) : \Gamma_{A,2} \subsetneq \Delta_A\}                         \\
  \mathcal C_3(\Gamma_{A,1}, \Gamma_{A,2}) & = \{ \Delta_A \in C(A) : \Gamma_{A,1} \subsetneq \Delta_A \subsetneq \Gamma_{A,2}\}
\end{align*}
and we let $\mathcal C = \mathcal C_1(\Gamma_{A,1}, \Gamma_{A,2}) \cup \mathcal C_2(\Gamma_{A,1}, \Gamma_{A,2}) \cup \mathcal C_3(\Gamma_{A,1}, \Gamma_{A,2})$.
As in the single cut case, we would like to say that the surfaces $\Gamma_{A,1}$ and $\Gamma_{A,2}$ are minimal cuts if an appropriate set of conditional entropies are sufficiently large for all $\Delta_A \in \mathcal C$. We formalize this as follows:
\begin{dfn}\label{dfn:gen pair cut}
  A pair of cuts $\Gamma_{A,1}$ and $\Gamma_{A,2}$ for $A$ is a \emph{pair of $(\eps,\KK)$-minimal cuts} if
  \begin{align*}
    H_{\min}^\eps(\Gamma_{A,1} \setminus \Delta_A \vert \Gamma_{A,1}^c) \geq K
  \end{align*}
  for all cuts $\Delta_A \in \mathcal C_1(\Gamma_{A,1}, \Gamma_{A,2})$
  \begin{align*}
    H_{\min}^\eps(\Delta_A \setminus \Gamma_{A,2}\vert \Gamma_{A,2}) \geq K
  \end{align*}
  for all cuts $\Delta_A \in \mathcal C_2(\Gamma_{A,1}, \Gamma_{A,2})$ and
  \begin{align*}
    H_{\min}^\eps(\Delta_A \setminus \Gamma_{A,1}\vert \Gamma_{A,1})    & \geq K \\
    H_{\min}^\eps(\Gamma_{A,2} \setminus \Delta_A \vert \Gamma_{A,2}^c) & \geq K
  \end{align*}
  for all cuts $\Delta_A \in \mathcal C_3(\Gamma_{A,1}, \Gamma_{A,2})$.
\end{dfn}

We will need a joint smoothing result similar to \cref{lem:joint relative min smoothing} for the particular case where $\Gamma_{A,1} = A$ and $\Gamma_{A,2}^c = \bar A$.
To this end we consider a graph $G = (V,E)$ with a set of boundary vertices $V_\partial$, and a boundary subsystem $A \subseteq V_\partial$.
We denote by $\gamma_{A,1}$ the edges incident to $A$, and $\gamma_{A,2}$ the edges incident to $\bar A$, and assume these sets do not intersect.
We let $E_b = E \setminus (\gamma_{A,1} \cup \gamma_{A,2})$.
For a cut $\Delta_A \in C(A)$ we let $Y^{\Delta_A}$ be the set of half-edges
\begin{align*}
  Y^{\Delta_A} = \{(e,x) : e = (xy), x \in \Delta_A^c, y \in A \}.
\end{align*}

\begin{restatable}{lem}{jointsmoothingtransition}\label{lem:joint smoothing at transition}
  Let $\phi \in \Pleq(V)$ be a pure background state which is of the form $\phi = \phi_{\gamma_{A,1}} \ot \phi_{\gamma_{A,2}} \ot \phi_{E_b}$, where
  \begin{align*}
    \phi_{E_b} & = \bigotimes_{e \in E_b} \phi_e
  \end{align*}
  is a product state and the $\phi_{\gamma_{A,i}}$ have a Schmidt decomposition in the standard basis along the half-edges.
  Then there exists a state $\phi^{\eps}$ which is such that
  \begin{align*}
    P(\phi,\phi^{\eps}) \leq 2\sqrt{2^{V_b}\eps}
  \end{align*}
  and for all cuts $\Delta_A \in C(A)$ it holds that
  \begin{align*}
    H_{\min}(\Delta_A \setminus A \vert A)_{\phi^\eps} \geq H_{\min}^{\eps}(\Delta_A \setminus A \vert A Y^{\Delta_A})_{\phi}.
  \end{align*}
\end{restatable}

To state the main result of this section, recall that given a measurable function $f : \mathcal X \rightarrow \mathcal Y$ between measure spaces, the \emph{pushforward} of the function $f$ on a measure $\mu$ on $\mathcal X$ is defined by $(f_*(\mu))(A) = \mu(f^{-1})(A)$ for any measureable set $A \subseteq \mathcal Y$. We apply this to the function $\min : \RR \times \RR \rightarrow \RR$ in the result:
\begin{thm}\label{thm:measure convergence surface transition}
  Consider a family of random tensor network states on a graph $G$ with pure state on edges $\phi$, indexed by an increasing sequence of positive integers $n$.
  We assume that $\Gamma_{A,1}, \Gamma_{A,2} \in C(A)$ are a pair of nonintersecting $(\eps(n),\KK(n))$-minimal cuts for all $n$ where $\eps(n) = \bigO(n^{-\gamma})$ for $\gamma > 4$ and $\KK(n) = \Omega(\log(n)^2)$ as $n$ goes to infinity.
  We let $H: \NN \to \RR$, and we assume that
  \begin{align*}
    \nu_{\phi_{\Gamma_{A,i}}} = \sum_{\lambda \in \spec(\phi_{\Gamma_{A,i}})} \lambda \delta_{\frac{1}{\sqrt{n}}[\log(\frac{1}{\lambda}) - H(n)]}
  \end{align*}
  for $i = 1,2,$ is such that $\nu_{\phi_{\Gamma_{A,i}}} \Rightarrow \nu_i$, where $\nu_i$ is a probability distribution with a continuous cumulative distribution function.
  Then $\nu_{\rho_A} \Rightarrow \minstar(\nu_1,\nu_2)$ in probability.
\end{thm}

\begin{figure}[!t]
  \centering
  \begin{subfigure}{.45\textwidth}
    \centering
    \begin{overpic}[width=0.9\textwidth,grid=false]{two-surfaces}
      \put(51,66){\color{WildStrawberry}{\small{$\gamma_{A,1}$}}} \put(27,55){\color{RoyalBlue}{\small{$\gamma_{A,2}$}}}
      \put(90,90){\color{BrickRed}{$A$}} \put(90,8){\color{SeaGreen}{$\bar{A} $}}
    \end{overpic}
    \caption{Consider a tensor network with a boundary region $A$ with two competing minimal cuts.}
  \end{subfigure}%
  \hspace*{1cm}
  \begin{subfigure}{.45\textwidth}
    \centering
    \begin{overpic}[width=0.9\textwidth,grid=false]{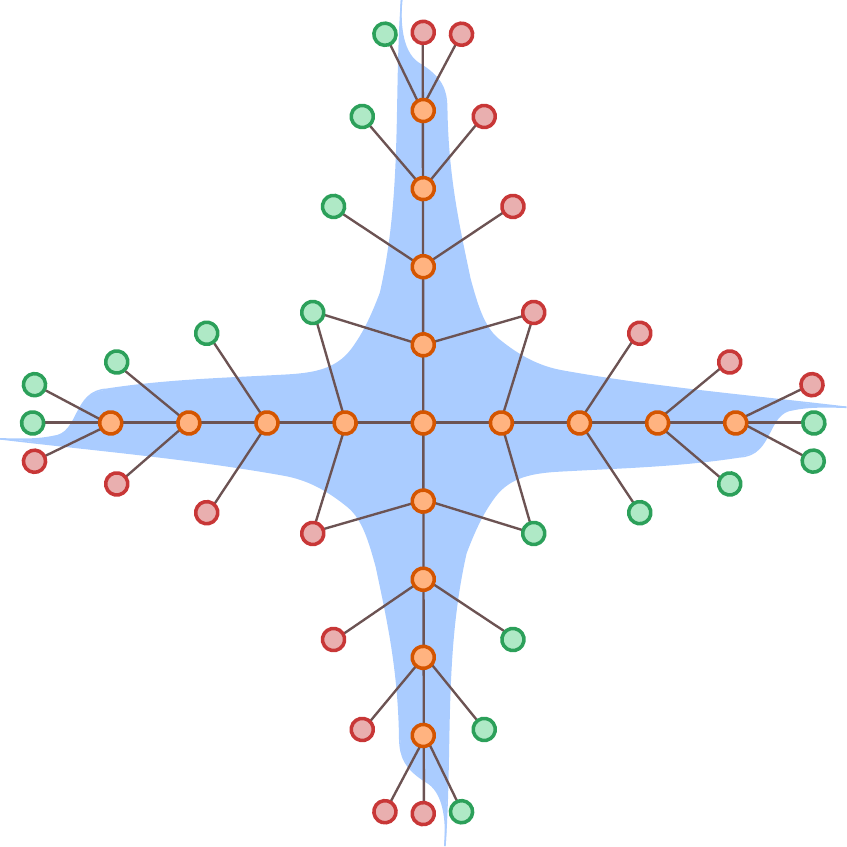}
      \put(70,70){\color{BrickRed}{$B$}} \put(70,24){\color{SeaGreen}{$\bar{B}$}}
    \end{overpic}
    \caption{We show that the spectrum is almost preserved when removing the boundary regions; that is, the spectrum of $\rho_A$ is approximated by that of $\sigma_B$.}
    \label{fig:reduce to sigma}
  \end{subfigure}
  \bigskip
  \begin{subfigure}{.9\textwidth}
    \centering
    \begin{overpic}[width=0.9\textwidth,grid=false]{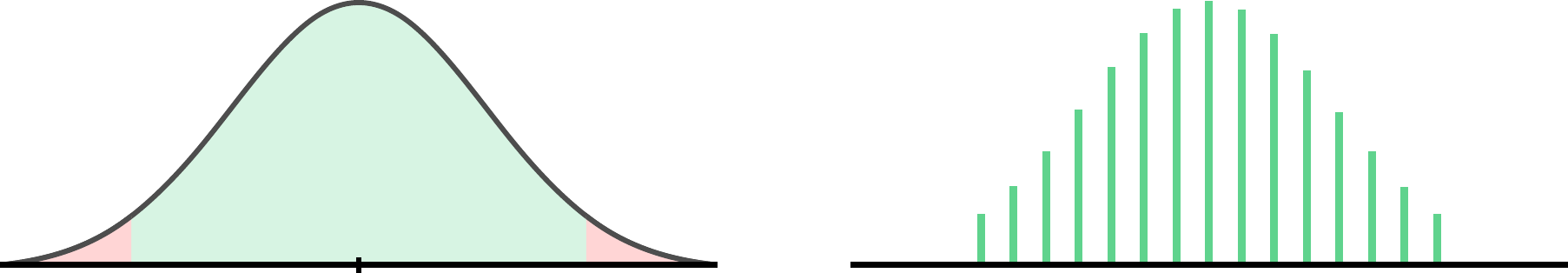}
      \put(21,-5){\footnotesize{$H(n)$}}
      \put(10,0.5){\footnotesize{$\underbrace{\framebox(12.5,0){}}_{\sqrt{n}}$}}
      \put(30,-3){\footnotesize{$\rightarrow \log(\frac{1}{\lambda_{i,j}})$}}
      \put(-11,5){\footnotesize{$ \nu_{\gamma_{A,i}} \uparrow$}}
      \put(46,5){\footnotesize{$ q_{i,j} \uparrow$}}
      \put(70,-3){\footnotesize{$\rightarrow \log(\frac{1}{\tilde\lambda_{i,j}})$}}
    \end{overpic}
    \bigskip
    \caption{The distribution of the logarithms of the eigenvalues has width of order $\sqrt{n}$. We approximate the spectrum of the state along the two minimal cuts by a superposition of maximally entangled states and discard the tails of the spectrum.}
    \label{fig:cut tails and bin}
  \end{subfigure}
  \begin{subfigure}{.45\textwidth}
    \centering
    \begin{overpic}[width=0.9\textwidth,grid=false]{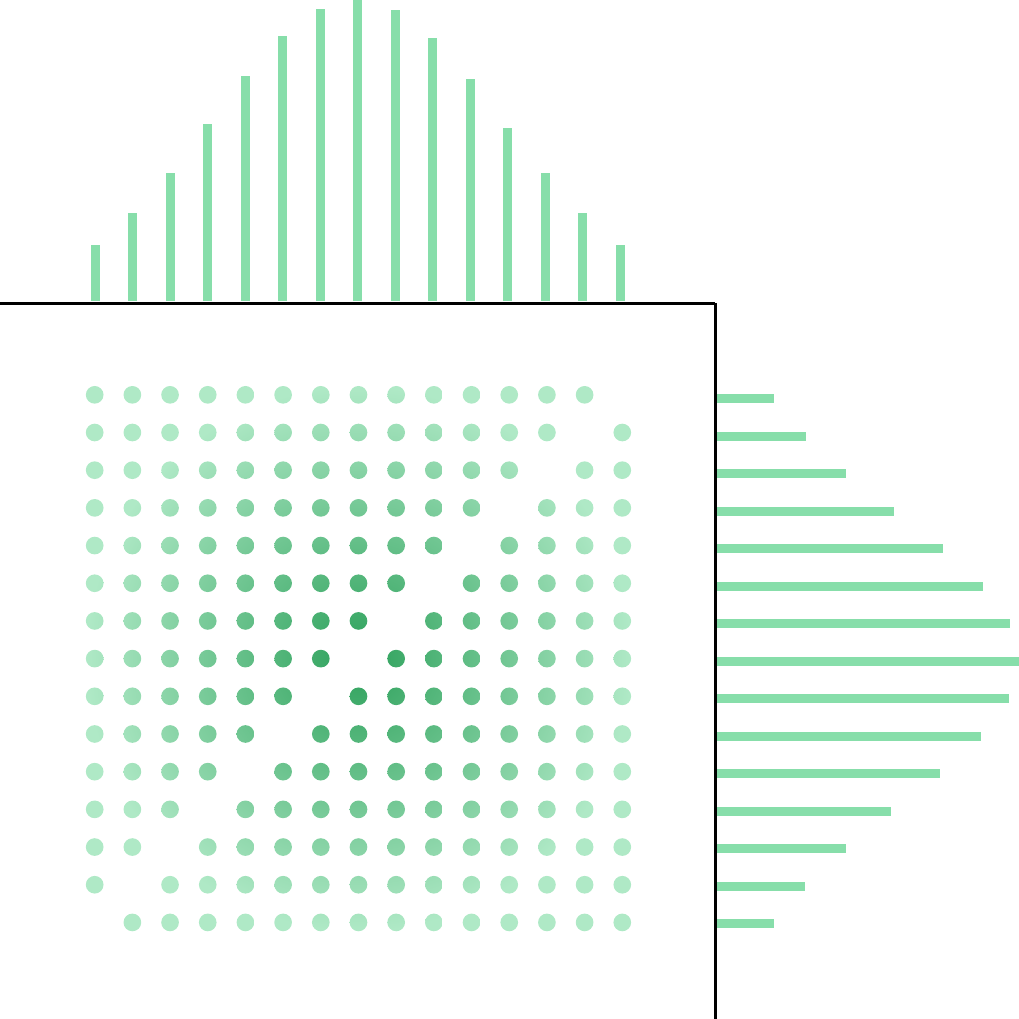}
      \put(0,90){\footnotesize{$\ket{\bar\phi_{\gamma_{A,2}}}$}}
      \put(85,60){\footnotesize{$\ket{\bar\phi_{\gamma_{A,1}}}$}}
    \end{overpic}
    \caption{We then discard the part of the state where the maximally entangled states along the two cuts are of almost equal dimension to get the approximate state on edges $\ket{\bar\phi}$.}
    \label{fig:remove middle}
  \end{subfigure}
  \hspace*{1cm}
  \begin{subfigure}{.45\textwidth}
    \centering
    \begin{overpic}[width=0.9\textwidth,grid=false]{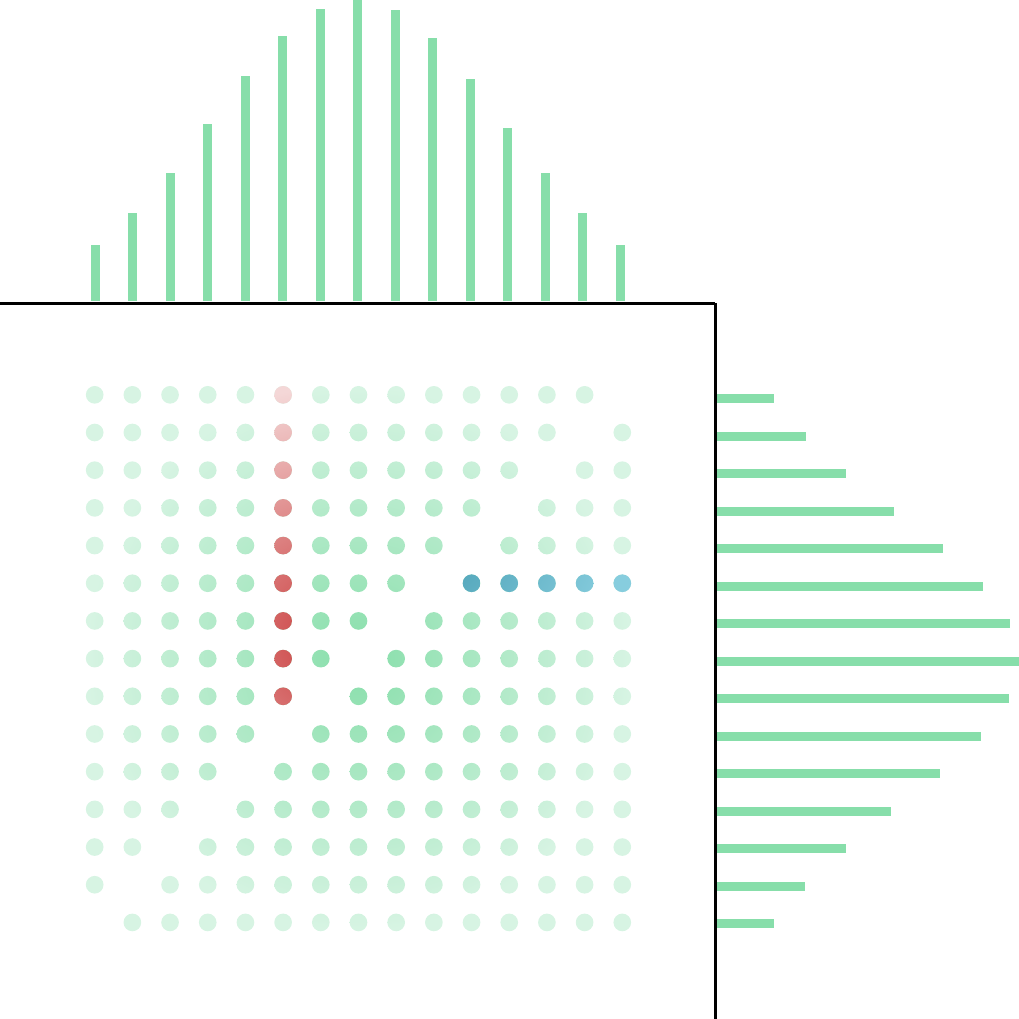}
      \put(0,90){\footnotesize{$\ket{\bar\phi_{\gamma_{A,2}}}$}}
      \put(85,60){\footnotesize{$\ket{\bar\phi_{\gamma_{A,1}}}$}}
      \put(5,40){\color{BrickRed}{\footnotesize{$\ket{\phi^{(j,a)}}$}}}
      \put(43,33){\color{RoyalBlue}{\footnotesize{$\ket{\phi^{(k,b)}}$}}}
    \end{overpic}
    \caption{We decompose the background state as a superposition over $\ket{\phi^{(j,a)}}$ and $\ket{\phi^{(j,b)}}$. We denote the corresponding tensor network states by $\sigma^{(j,a)}$ and $\sigma^{(j,b)}$.}
    \label{fig:decompose}
  \end{subfigure}
\end{figure}
\begin{figure}[!ht]\ContinuedFloat
  \begin{subfigure}{.9\textwidth}
    \centering
    \begin{overpic}[width=0.9\textwidth,grid=false]{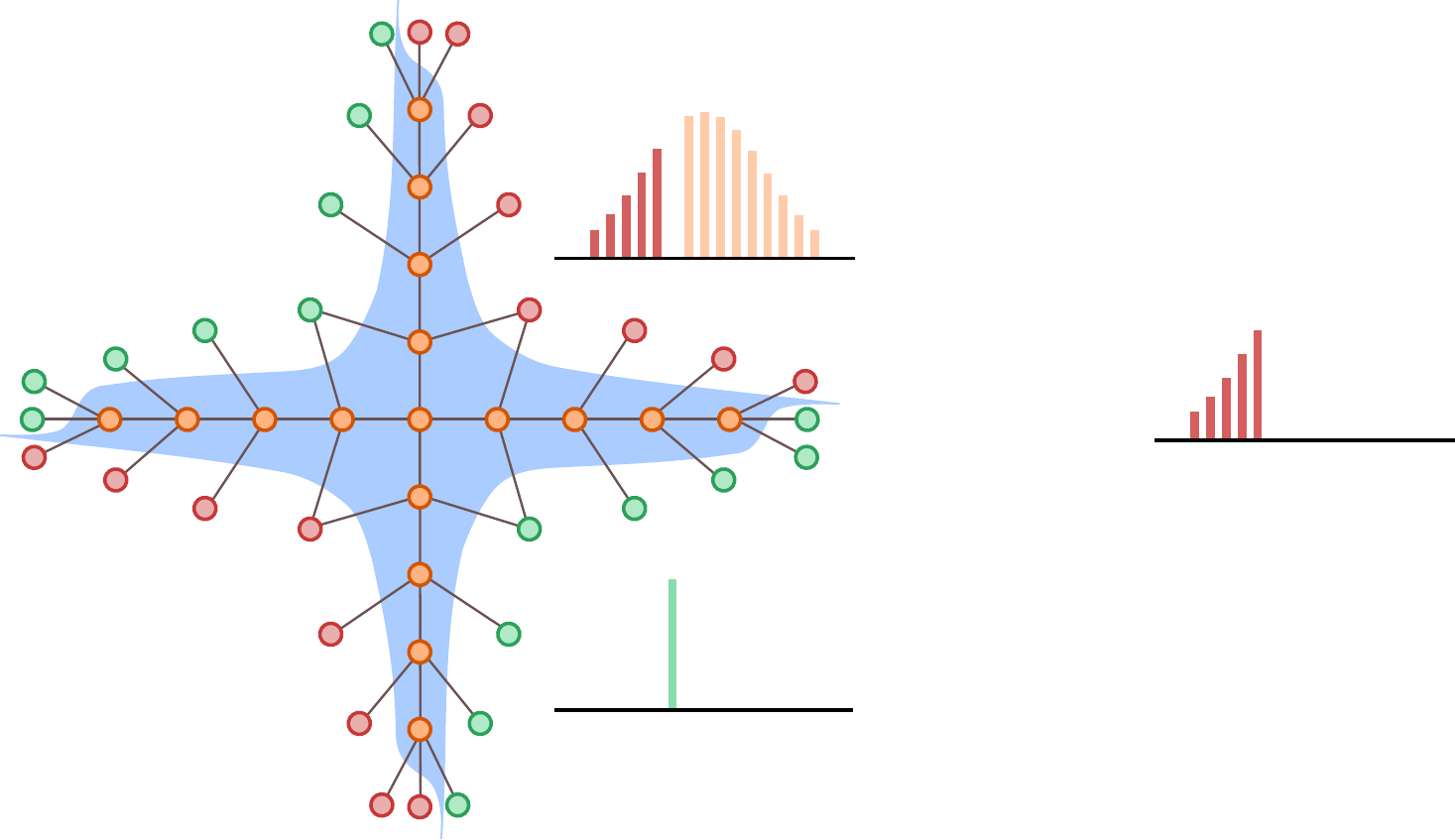}
      \put(35,48){\color{BrickRed}{\footnotesize{$\ket{\phi^{(j,a)}_{\gamma_{A,1}}}$}}}
      \put(87,22){\color{BrickRed}{$\sigma^{(j,a)}$}}
      \put(65,28){$\Rightarrow$}
      \put(50,13){\color{SeaGreen}{\footnotesize{$\ket{\Phi^+_{2,j}}$}}}
    \end{overpic}
    \vspace*{0.2cm}
    \caption{For the state $\ket{\phi^{(j,a)}}$, the minimal cut is $\gamma_{A,1}$, and along $\gamma_{A,2}$ we have a maximally entangled state, which gives $\sigma^{(j,a)}_B \approx \phi^{(j,a)}_B$.
    Similarly, for $\ket{\phi^{(j,b)}}$, the minimal cut is $\gamma_{A,2}$ and we get $\sigma^{(j,b)}_{\bar{B}} \approx \phi^{(j,b)}_{\bar{B}}$.}
    \label{fig:sigma min cut}
  \end{subfigure}
  \bigskip
  \vspace*{0.5cm}
  \begin{subfigure}{.9\textwidth}
    \centering
    \begin{overpic}[width=0.9\textwidth,grid=false]{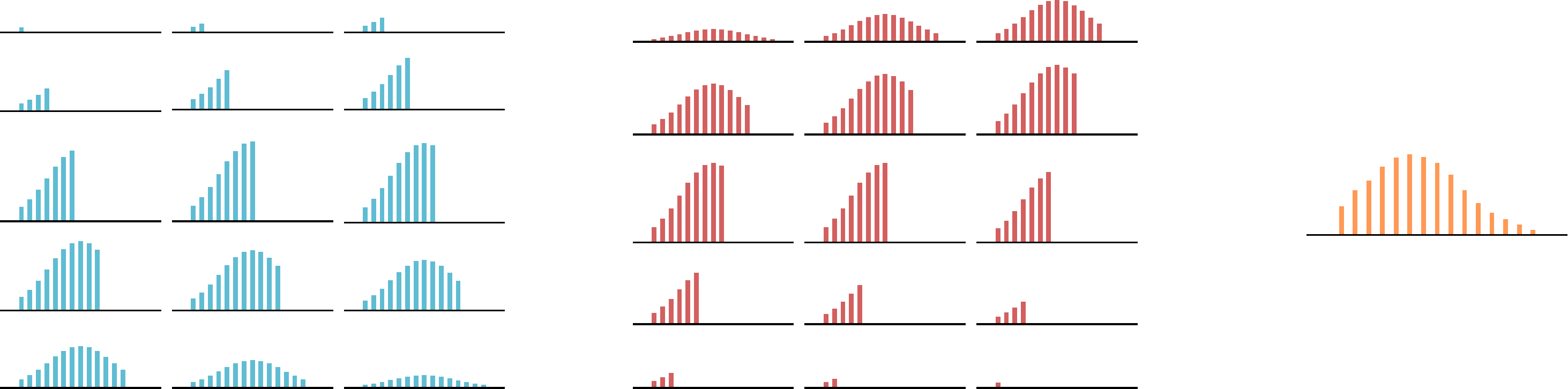}
      \put(35,10.5){\small{$+$}} \put(77.5,10.5){\footnotesize{$\Rightarrow$}}
      \put(12,-5){\color{RoyalBlue}{\footnotesize{$\sigma_B^{(a)}$}}}
      \put(58,-5){\color{BrickRed}{\footnotesize{$\sigma_B^{(b)}$}}}
      \put(89,6){\color{Bittersweet}{\footnotesize{$\sigma_B$}}}
    \end{overpic}
    \vspace*{0.5cm}
    \caption{All the $\sigma_B^{(j,a)}$ have the same eigenbasis, so we can add the spectra. We can do the same for the $\sigma_{\bar{B}}^{(j,b)}$, yielding $\sigma_B^{(a)}$ and $\sigma_{\bar{B}}^{(b)}$, which has the same spectrum as $\sigma_{B}^{(b)}$.  Finally we argue that the $\sigma_B^{(a)}$ and $\sigma_B^{(b)}$ are nearly orthogonal, so we can combine the spectra, as in \cref{eq:perturb sigma}. We complete the proof by showing that the resulting spectrum converges to the pushforward measure of the min function applied to the spectra of the two cuts.}
    \label{fig:sigma orth}
  \end{subfigure}

  \caption{Illustration of the proof of \cref{thm:measure convergence surface transition}.}
  \label{fig:phase transition}
\end{figure}

\begin{proof}
  Proving this result will require several intermediate results. We provide a very high-level, enumerated sketch of our proof here, involving the following steps:
  \begin{enumerate}
    \item Study a reduced problem on a subnetwork; this subnetwork is such that the minimal cuts $\gamma_{A,i}$ are incident to the boundary.
    \item Approximate the background state with superpositions of maximally entangled states by binning eigenvalues, and construct approximate tensor network states with the approximate background states.
    \item Prove that the spectrum of the boundary state converges to the spectrum of the approximate tensor network state, which in turn converges to the spectrum of the approximate background states. We write the background state as a superposition of two states, one of which has $\gamma_{A,1}$ as its minimal cut, whereas the other state has $\gamma_{A,2}$ as its minimal cut. We show that the resulting states are approximately orthogonal.
    \item Show that the distribution of the approximate background states converges to the min-distribution, and hence conclude that the spectrum of the boundary state converges to the min-distribution.
  \end{enumerate}
  \cref{fig:phase transition} provides a more detailed visual sketch of the intuition behind our proof strategy. \cref{lem:trace distance convergence} will be a key tool, as it implies that it will suffice to show convergence in trace norm.

  We assume without loss of generality that $\Gamma_{A,1} \subset \Gamma_{A,2}$.
  We now define the subnetwork that we will analyze in our proof. Let $G' = (V',E')$ be the induced graph on $V' = V_b' \cup V_\partial'$, where $V_b' = \Gamma_{A,2} \setminus \Gamma_{A,1}$, and $V_\partial' = B \cup \bar{B}$, with $B$ the set of vertices in $\Gamma_{A,1}$ which are incident to $V_b'$ and $\bar{B}$ the set of vertices in $V \setminus \Gamma_{A,2}$ which are incident to $V_b'$. The subgraph $G'$ is depicted in \cref{fig:reduce to sigma}.

  We also define the random tensor network state $\tau$ as the state obtained by applying random tensors only on bulk vertices in the complement of $V_b'$.
  Then for this state, by a slight variation on \cref{thm:rtn with arbitrary link states}, it holds that
  \begin{align*}
    \EE \min_{V_A, V_{\bar{A}}} \norm{(V_A \ot V_{\bar{A}} \ot I)\tau(V_A^\dagger \ot V_{\bar{A}}^\dagger \ot I) - \phi}_1 = \bigO(2^{-\KK(n)} + \sqrt{\eps(n)}),
  \end{align*}
  where the minimum is over isometries $V_A : \HH_A \rightarrow \HH_{\Gamma_{A,1}}$ and $V_{\bar{A}} : \HH_{\bar{A}} \rightarrow \HH_{\Gamma_{A,2}^c}$.
  Let $\tilde \tau = (V_A \ot V_{\bar{A}} \ot I)\tau(V_A^\dagger \ot V_{\bar{A}}^\dagger \ot I)$ where the $V_A$ and $V_{\bar{A}}$ are the isometries that realize the minimum.

  Now, consider two random tensor network states on $G'$ obtained by applying the same random tensors on $V_b'$ with background states $\tilde \tau$ and $\phi$ respectively.
  The state where we take $\tilde \tau$ as the background state yields $(V_A \ot V_{\bar{A}} \ot I)\rho(V_A^\dagger \ot V_{\bar{A}}^\dagger \ot I)$.
  Denote the state where we take $\phi$ as the background state by $\sigma$. Then by \cref{lem:smoothing rtn state}:
  \begin{align*}
    \EE \norm{(V_A \ot V_{\bar{A}})\rho(V_A^\dagger \ot V_{\bar{A}}^\dagger) - \sigma}_1 = \bigO(2^{-\KK(n)} + \sqrt{\eps(n)}).
  \end{align*}
  If we let $E_1$ denote the set of edges $(xy)$ for which both $x,y \in \Gamma_{A,1}$, and $E_2$ the set of edges $(xy)$ for which both $x,y \in \Gamma_{A,2}^c$, then $\sigma_{\Gamma_{A,1}} = \phi_{E_1} \ot \sigma_B$ and $\sigma_{\Gamma_{A,2}} = \phi_{E_2} \ot \sigma_{\bar{B}}$.
  This shows that
  \begin{align}\label{eq:tilde rho vs sigma}
    \EE \norm{\spec_+(\rho_A) - \spec_+(\sigma_{B})}_1 = \bigO(2^{-\KK(n)} + \sqrt{\eps(n)})
  \end{align}
  and $\spec_+(\sigma_B) = \spec_+(\sigma_{\bar{B}})$, because $\sigma$ is pure, which follows because its background state $\phi$ is pure.
  We will continue to study $\sigma$ on the reduced graph $G'$, as in \cref{fig:reduce to sigma}, and at the end of the proof, we will see that \cref{eq:tilde rho vs sigma} will be sufficient to prove the desired properties of $\spec(\rho_A)$.

  We have accomplished Step 1 by reducing the full network to a subnetwork, and we now try to construct an approximation to $\sigma$. We will do so by coarse-graining the spectrum of the background states along the two minimal cuts: slicing the tails of the distribution, binning the remaining eigenvalues, throwing away the smallest bins of the binned distribution, and then approximating the states as a superposition of maximally entangled states defined on the bins. Consider the background states along the two minimal cuts in the Schmidt basis:
  \begin{align}\label{eq: J basis}
    \ket{\phi_{\gamma_{A,i}}} = \sum_{J} \sqrt{\lambda_{i,J}}\ket{JJ}.
  \end{align}
  Here, the $\ket{J}$ are a tensor product basis along the half-edges (which we again may take to be the standard basis), so that
  \begin{align*}
    \ket{JJ} = \bigotimes_{e = (xy) \in \gamma_{A,i}} \ket{i_{e,x}} \ot \ket{i_{e,y}},
  \end{align*}
  where the $\ket{i_{e,x}}$ and $\ket{i_{e,y}}$ are a basis for $\HH_{e,x}$ and $\HH_{e,y}$.
  First, we truncate the allowed range of eigenvalues by slicing off the tails of the distribution, as in the left side of \cref{fig:cut tails and bin}. Let
  \begin{align*}
    I_n = [2^{-H(n)-C\sqrt{n}\log(n)}, 2^{-H(n)+C\sqrt{n}\log(n)}]
  \end{align*}
  for some constant $C$, and let
  \begin{align*}
    \ket{\phi_{\gamma_{A,i}}^{(1)}} = \sum_{\lambda_{i,J} \in I_n} \sqrt{\lambda_{i,J}}\ket{JJ}.
  \end{align*}
  Note that the entanglement spectrum of $\phi_{\gamma_{A,i}}$ is given by the spectrum of $\phi_{\Gamma_{A,i}}$.
  This implies that
  \begin{align*}
    \sum_{\lambda_{i,J} \notin I_n} \lambda_{i,J} = \nu_{\phi_{\Gamma_{A,i}}}((-\infty,-C\log(n))\cup(C\log(n),\infty)) \rightarrow 0
  \end{align*}
  and hence
  \begin{align*}
    \norm{\phi_{\gamma_{A,i}} - \phi_{\gamma_{A,1}}^{(1)}}_1 \rightarrow 0.
  \end{align*}
  Next, for $\lambda_{i,J} \in I_n$, we define new $\tilde{\lambda}_{i,J}$ according to $\log(\tilde{\lambda}_{i,J}) = \frac{1}{n^{\alpha}}\floor{n^{\alpha}\log(\lambda_{i,J})}$, effectively binning the values of $\log(\lambda_{i,J})$ into intervals of size $n^{-\alpha}$, as in the right side of \cref{fig:cut tails and bin}.
  We will choose $\alpha > 0$ small (but other choices will be useful for \cref{cor:convergence of entropy}), to be precise, we choose $\alpha < \frac{1}{8}(\gamma - 4)$.
  We can now define the background state
  \begin{align*}
    \ket{\phi_{\gamma_{A,i}}^{(2)}} = \sum_{\lambda_{i,J} \in I_n} \sqrt{\tilde\lambda_{i,J}}\ket{JJ}.
  \end{align*}
  Then by \cref{lem:consequence fvdg}
  \begin{align*}
    \norm{\phi_{\gamma_{A,i}}^{(1)}- \phi_{\gamma_{A,i}}^{(2)}}_1 & \leq 2\sqrt{\sum_{J} \abs{\lambda_{i,J} - \tilde\lambda_{i,J}}},
  \end{align*}
  which we may bound using the fact that $2^{-\frac{1}{n^{\alpha}}} \leq \frac{\tilde \lambda_{i,J}}{\lambda_{i,J}} \leq 1$, so
  \begin{align*}
    \sum_{J} \abs{\lambda_{i,J} - \tilde\lambda_{i,J}} & = \sum_J \lambda_{i,J}\left(1 - \frac{\tilde \lambda_{i,J}}{\lambda_{iJ}}\right) \leq \max_{J} \left(1 - \frac{\tilde \lambda_{i,J}}{\lambda_{i,J}}\right) \leq 1 - 2^{-\frac{1}{n^{\alpha}}} = \bigO\left(\frac{1}{n^{\alpha}}\right).
  \end{align*}
  Thus, $\norm{\phi_{\gamma_{A,i}}^{(1)}- \phi_{\gamma_{A,i}}^{(2)}}_1 = \bigO(n^{-\frac12 \alpha})$.
  Since the interval $I_n$ has length $\bigO(\sqrt{n}\log(n))$ and the distinct values $\tilde\lambda_{i,J}$ are $n^{-{\alpha}}$ apart, the $\tilde{\lambda}_{i,J}$ take $\bigO(n^{{\alpha} + \frac12}\log(n))$ different values.
  Denote these values by $p_{i,j}$, and let $d_{i,j}$ be their multiplicities.
  Setting
  \begin{align*}
    q_{i,j} = d_{i,j}p_{i,j},
  \end{align*}
  we can rewrite the state $\ket{\phi_{\gamma_{A,i}}^{(2)}}$ using the collected eigenvalues:
  \begin{align*}
    \ket{\phi_{\gamma_{A,i}}^{(2)}} = \sum_{j} \sqrt{q_{i,j}}\ket{\Phi^+_{i,j}},
  \end{align*}
  where the $\ket{\Phi^+_{i,j}}$ are maximally entangled states of dimension $d_{i,j}$ which are orthogonal, i.e. $\braket{\Phi^+_{i,j}|\Phi^+_{i,k}} = \delta_{j,k}$.
  Note that $0 \leq q_{i,j} \leq 1$ and $\sum_j q_{i,j} \leq 1$.
  Now, we discard any bins that are too small: let $\beta > \alpha + \frac12$ and consider the state
  \begin{align}\label{eq:decompose link state into max entangled}
    \ket{\phi_{\gamma_{A,i}}^{(3)}} = \sum_{q_{i,j} > n^{-\beta}} \sqrt{q_{i,j}}\ket{\Phi^+_{i,j}}
  \end{align}
  where we restricted the sum to terms for which $q_{i,j}$ is sufficiently large (and hence $p_{i,j}$ will be sufficiently close to $d_{i,j}^{-1}$).
  Then, since the number of terms is $\bigO(n^{\alpha + \frac12}\log(n))$ and using \cref{lem:consequence fvdg}, we have
  \begin{align*}
    \norm{\phi_{\gamma_{A,i}}^{(2)} - \phi_{\gamma_{A,i}}^{(3)}}_1 =\bigO\left(n^{\frac12\alpha + \frac14 - \frac12\beta}\sqrt{\log(n)}\right)
  \end{align*}
  To summarize, $\phi_{\gamma_{A,i}}^{(3)}$ is the state obtained from the original background state $\phi_{\gamma_{A,i}}$ by 1) removing the tails of the spectrum, 2) binning the eigenvalues, and finally 3) dropping any of the bins that are too small. The first approximation incurs an error $\norm{\phi_{\gamma_{A,i}} - \phi_{\gamma_{A,i}}^{(1)}}_1$, which converges to zero, and the second and third approximations incur errors of order $\bigO(n^{-\frac12\alpha})$ and $\bigO\left(n^{\frac12\alpha + \frac14 - \frac12\beta}\sqrt{\log(n)}\right)$ respectively.

  Now, let $\tilde\phi$ be the background state for which we have replaced $\phi_{\gamma_{A,i}}$ by $\phi_{\gamma_{A,i}}^{(3)}$ for $i=1,2$. Recall that $\sigma$ is the state on the random tensor network state constructed with background state $\phi$ on the subgraph $G'$, as in \cref{fig:reduce to sigma}. We then define the approximation $\tilde\sigma$ to be the tensor network state on $G'$ which uses $\tilde\phi$ as its background state instead of $\phi$. With these background states, we see
  \begin{align}\label{eq:phi vs tilde phi}
    \norm{\phi - \tilde\phi}_1 = \bigO\left(\norm{\phi_{\gamma_{A,i}} - \phi_{\gamma_{A,i}}^{(1)}}_1 + n^{-\frac12\alpha} + n^{\frac12\alpha + \frac14 - \frac12\beta}\sqrt{\log(n)} \right) \to 0.
  \end{align}
  We then apply \cref{lem:smoothing rtn state} to find:
  \begin{align}\label{eq:rho vs tilde rho}
    \EE \norm{\sigma - \tilde\sigma}_1 = \bigO\left(\norm{\phi - \tilde\phi}_1\right) \to 0.
  \end{align}
  We now make one final approximation to $\sigma$, in which we discard the parts of the background state where the maximally entangled states along each cut are close in dimension, as in \cref{fig:remove middle}. Consider the state
  \begin{align*}
    \ket{\bar\phi_{\gamma_{A,1},\gamma_{A,2}}} = \sum_{2^{n^{\nicefrac{1}{4}}}p_{2,k} \leq p_{2,j} \text{ or } 2^{-n^{\nicefrac{1}{4}}}p_{2,k} \geq p_{2,j}}\sqrt{q_{1,k}q_{2,j}} \ket{\Phi_{1,k}}\ot \ket{\Phi_{2,j}},
  \end{align*}
  where the sum is still only over those $j$ and $k$ for which $q_{1,k} > n^{-\beta}$ and $q_{2,j} > n^{-\beta}$.
  If we let
  \begin{align*}
    D_n = \{(x,y) \in \RR^2 : x - n^{-\frac14} - n^{-\frac12} \leq y \leq x + n^{-\frac14} + n^{-\frac12}, \abs{x} \leq \log(n) + 1\},
  \end{align*}
  then as $n$ increases, the measure of $D_n$ converges to zero and since the measures $\nu_{\phi_{\gamma_{A,i}}} \Rightarrow \nu_i$ and $\nu_i$ has continuous cumulative distribution function
  \begin{align}\label{eq:middle measure zero}
    \sum_{2^{-n^{\nicefrac{1}{4}}}p_{1,j} \leq p_{2,k} \leq 2^{n^{\nicefrac{1}{4}}}p_{1,j}} q_{1,j}q_{2,k} \leq (\nu_{\phi_{\Gamma_{A,1}}} \times\nu_{\phi_{\Gamma_{A,2}}}) (D_n)\rightarrow 0.
  \end{align}
  Hence, if we let $\bar \phi$ denote the background state with $\tilde{\phi}_{\gamma_{A,1}}\ot \tilde{\phi}_{\gamma_{A,2}}$ replaced by $\bar\phi_{\gamma_{A,1},\gamma_{A,2}}$, we get that $\norm{\bar \phi - \tilde \phi}_1 \rightarrow 0$.
  By \cref{lem:smoothing rtn state}, if we denote by $\bar\sigma$ the state we obtain on $G'$ by using $\bar\phi$ rather than $\tilde\phi$ as the background state, we get
  \begin{align}\label{eq:bar sigma vs sigma}
    \EE\norm{\bar\sigma - \sigma}_1 \rightarrow 0.
  \end{align}

  We pause here to note that we have accomplished Step 2: we have constructed an approximation $\bar\sigma$ to $\sigma$ by approximating the background state along the cuts as superpositions of maximally entangled states. The utility in doing so is that working with this approximated tensor network state allows us to reduce to calculations where we restrict to a maximally entangled state along one of the two surfaces. Our next major step is to then show that the spectrum of $\bar\sigma$ can be used to approximate the spectrum of $\rho$, which will then allow us to analyze the convergence of the corresponding distribution. As a first intermediate step, we will show that $\bar{\sigma}$ can be approximated as a superposition of a state with minimal cut $\gamma_{A,1}$ and a state with minimal cut $\gamma_{A,2}$. This will allow us to more easily reason about how the background states are related to the spectrum of $\bar{\sigma}$, in turn, $\sigma$, and in turn, $\rho_A$.

  Let us write
  \begin{align*}
    \ket{\phi_{\gamma_{A,1},\gamma_{A,2}}^{(j,a)}} & = \sum_{k \text{ s.t. } p_{1,k} \geq p_{2,j}2^{n^{\nicefrac{1}{4}}}} \sqrt{q_{1,k} q_{2,j}} \ket{\Phi^+_{1,k}} \ot \ket{\Phi^+_{2,j}}  \\
    \ket{\phi_{\gamma_{A,1},\gamma_{A,2}}^{(j,b)}} & = \sum_{k \text{ s.t. } p_{2,k} \geq p_{1,j}2^{n^{\nicefrac{1}{4}}}} \sqrt{q_{1,j} q_{2,k}} \ket{\Phi^+_{1,j}} \ot \ket{\Phi^+_{2,k}},
  \end{align*}
  allowing us to write the background state $\bar\phi_{\gamma_{A,1},\gamma_{A,2}}$ as a different superposition, as depicted in \cref{fig:decompose}:
  \begin{align*}
    \ket{\bar\phi_{\gamma_{A,1},\gamma_{A,2}}} & = \sum_{j} \ket{\phi_{\gamma_{A,1},\gamma_{A,2}}^{(j,a)}} + \sum_j \ket{\phi_{\gamma_{A,1},\gamma_{A,2}}^{(j,b)}}.
  \end{align*}
  Let $\ket{\phi^{(j,a)}}$ and $\ket{\phi^{(j,b)}}$ denote the background states on $G'$ where we have replaced $\ket{\bar\phi_{\gamma_{A,1},\gamma_{A,2}}}$ by $\ket{\phi_{\gamma_{A,1},\gamma_{A,2}}^{(j,a)}}$ and $\ket{\phi_{\gamma_{A,1},\gamma_{A,2}}^{(j,b)}}$ respectively, and let
  \begin{align*}
    \ket{\phi^{a}} = \sum_j \ket{\phi^{(j,a)}}, \qquad \ket{\phi^{b}} = \sum_j \ket{\phi^{(j,b)}}.
  \end{align*}
  Denote by $\ket{\sigma^{(j,a)}}$, $\ket{\sigma^{(j,b)}}$ the random tensor network states on $G'$ with background states $\ket{\phi^{(j,a)}}$ and $\ket{\phi^{(j,b)}}$, respectively. Similarly, denote by
  $\ket{\sigma^{(a)}}$ and $\ket{\sigma^{(b)}}$ the random tensor network states on $G'$ with background states $\ket{\phi^{(a)}}$ and $\ket{\phi^{(b)}}$, respectively.

  We start with the following bound on the rank of $\sigma^{(j,a)}_B$:
  \begin{align}\label{eq:rank sigma a}
    \begin{split}
      \rank(\sigma^{(j,a)}_B) &\leq \rank(\phi^{(j,a)}_B) \\
      &\leq \sum_{p_{1,k} \geq p_{2,j}2^{n^{\nicefrac{1}{4}}}} d_{1,k} \\
      &\leq \sum_{p_{1,k} \geq p_{2,j}2^{n^{\nicefrac{1}{4}}}} \frac{q_{1,k}}{p_{1,k}}\\
      &\leq \sum_{k} \frac{q_{1,k}2^{-n^{1/4}}}{p_{2,j}} \\
      &\leq \frac{2^{-n^{1/4}}}{p_{2,j}}
    \end{split}
  \end{align}
  By the same reasoning we may bound the rank of $\sigma^{(j,b)}_{\bar{B}}$ as
  \begin{align}\label{eq:rank sigma b}
    \rank(\sigma^{(j,b)}_{\bar{B}}) \leq \frac{2^{-n^{1/4}}}{p_{1,j}}.
  \end{align}

  Now, we will argue that the state $\sigma^{(j,a)}$ has minimal cut $\gamma_{A,1}$, and $\sigma^{(j,b)}$ has minimal cut $\gamma_{A,2}$, as in \cref{fig:sigma min cut}. Intuitively, it is sensible that for $\sigma^{(j,a)}$ the unique minimal cut is along $\gamma_{A,1}$, as for $\phi^{(j,a)}$, we have a fixed maximally entangled state along $\gamma_{A,2}$, and a superposition of maximally entangled states of lower dimension along $\gamma_{A,1}$.
  Similarly, for $\sigma^{(j,b)}$ the minimal cut is along $\gamma_{A,2}$.
  To confirm this intuition, we will now show that $\sigma^{(j,a)}_B \approx \phi^{(j,a)}_B$ and $\sigma^{(j,b)}_{\bar{B}} \approx \phi^{(j,b)}_{\bar{B}}$.

  We show this for $\sigma^{(j,a)}_B$. In this case, the `minimal' cut is simply $B$.
  Let $\Delta_B$ be a cut for $B$ unequal to $B$ or to $V_b' \cup B$.
  We denote by $\Delta_A$ the associated minimal cut for $A$ on the original graph $G$ given by $\Delta_A = \Delta_B \cup \Gamma_{A,1}$.
  We let
  \begin{align*}
    Y^{\Delta_B} = \{(e,x) : e  = (xy), x \in \Delta_B^c, y \in B \}.
  \end{align*}
  Note that as $\phi$ is a product state we have $H_{\min}^{\eps}(\Delta_B \setminus B \vert B)_{\phi} = H_{\min}^{\eps}(\Delta_B \setminus B \vert B Y^{\Delta_B})_{\phi}$.
  We can obtain $\phi^{(j,a)}$ from $\phi$ by acting with subunital CP maps on $B$ and $\bar{B}$ and therefore (by two applications of data processing)
  \begin{align*}
    H_{\min}^{\eps}(\Delta_B \setminus B \vert B Y^{\Delta_B})_{\phi^{(j,a)}} \geq H_{\min}^{\eps}(\Delta_B \setminus B \vert B Y^{\Delta_B})_{\phi} = H_{\min}^{\eps}(\Delta_B \setminus B \vert B)_{\phi} \geq H_{\min}^{\eps}(\Delta_A \setminus \Gamma_A \vert \Gamma_A)_{\phi}
  \end{align*}
  We now consider $\Delta_B = V_b' \cup B$ (in which case $Y^{\Delta_B}$ is empty, since we assume $\gamma_{A,1} \cap \gamma_{A,2} = \emptyset$).
  Then
  \begin{align*}
    H_{\min}^{\eps}(\Delta_B \setminus B \vert B)_{\phi^{(j,a)}} \geq H_{\min}(\Delta_B \setminus B \vert B)_{\phi^{(j,a)}} = H_{\min}(\bar B)_{\phi^{(j,a)}} - H_{\max}(B)_{\phi^{(j,a)}}
  \end{align*}
  using the product structure of $\phi^{(j,a)}$.
  Now, we compute
  \begin{align*}
    H_{\min}(\bar B)_{\phi^{(j,a)}} =  - \log(\norm{\phi^{(j,a)}_{\bar{B}}}_{\infty}) \geq -\log(p_{2,j}),
  \end{align*}
  where we recall that $p_{2,j}$ is a binned eigenvalue.
  By \cref{eq:rank sigma a}, we can bound
  \begin{align*}
    H_{\max}(B)_{\phi^{(j,a)}} & \leq  \log(\rank(\phi^{(j,a)}_B)) \leq \log\left(\frac{2^{-n^{1/4}}}{p_{2,j}}\right) = -\log(p_{2,j}) - n^{\frac14}.
  \end{align*}
  In conclusion,
  \begin{align*}
    H_{\min}^{\eps}(\Delta_B \setminus B \vert B)_{\phi^{(j,a)}} \geq n^{\frac14}.
  \end{align*}

  By \cref{lem:joint smoothing at transition} we can find $\phi^{(\eps,j,a)}$ such that $P(\phi^{(j,a)}, \phi^{(\eps,j,a)}) = \bigO(\sqrt{\eps})$ and for all cuts $\Delta_B \neq B$ we have
  \begin{align*}
    H_{\min}(\Delta_B \setminus B \vert B)_{\phi^{(\eps,j,a)}} \geq H_{\min}^{\eps}(\Delta_B \setminus B \vert B Y^{\Delta_B})_{\phi^{(j,a)}}
  \end{align*}
  and hence for $\Delta_B \neq V_b' \cup B$
  \begin{align*}
    H_{\min}(\Delta_B \setminus B \vert B)_{\phi^{(\eps,j,a)}} \geq H_{\min}^{\eps}(\Delta_A \setminus \Gamma_A \vert \Gamma_A)_{\phi}
  \end{align*}
  and for $\Delta_B = V_b' \cup B$
  \begin{align*}
    H_{\min}(\Delta_B \setminus B \vert B)_{\phi^{(\eps,j,a)}} \geq n^{\frac14}.
  \end{align*}

  Therefore, \cref{prop:one-shot decoupling} allows us to conclude that if we let $\sigma^{(\eps,j,a)}$ denote the random tensor network state with background state $\phi^{(\eps,j,a)}$, then
  \begin{align*}
    \EE \norm{\sigma^{(\eps,j,a)}_B - \phi^{(\eps,j,a)}_B}_1 = \bigO\left(2^{-\frac12 \KK(n)} + 2^{-\frac12 n^{\frac14}}\right),
  \end{align*}
  and by \cref{lem:smoothing rtn state} and the fact that $T(\phi^{(j,a)}, \phi^{(\eps,j,a)}) \leq P(\phi^{(j,a)}, \phi^{(\eps,j,a)})$ this implies
  \begin{align*}
    \begin{split}
      \EE \norm{\sigma^{(j,a)}_B - \phi^{(j,a)}_B}_1 &\leq \EE \norm{\sigma^{(\eps,j,a)}_B - \phi^{(\eps,j,a)}_B}_1 + \EE \norm{\sigma^{(\eps,j,a)}_B - \sigma^{(j,a)}_B}_1   +\norm{\phi^{(\eps,j,a)}_B - \phi^{(j,a)}_B}_1 \\
      &= \bigO\left(2^{-\frac12 \KK(n)} + 2^{-\frac12 n^{\frac14}} + \sqrt{\eps(n)}\right).
    \end{split}
  \end{align*}
  It follows that
  \begin{align}\label{eq:B spectrum sigma a}
    \EE \norm{\sigma^{(a)}_B - \phi^{(a)}_B}_1 \leq \sum_j \EE \norm{\sigma^{(j,a)}_B - \phi^{(j,a)}_B}_1 = \bigO\left(n^{\alpha + \frac12}\log(n)(2^{-\frac12 \KK(n)} + 2^{-\frac12 n^{\frac14}} + \sqrt{\eps(n)})\right) \to 0
  \end{align}
  using that the number of terms is $\bigO(n^{\alpha + \frac12}\log(n))$.
  This expression goes to zero since we assume $\KK(n) = \Omega(\log(n)^2)$ and $\eps(n) = \bigO(n^{-\gamma})$.
  A completely symmetric argument shows that
  \begin{align}\label{eq:sigma j b}
    \EE \norm{\sigma^{(j,b)}_{\bar B} - \phi^{(j,b)}_{\bar B}}_1 = \bigO\left(2^{-\frac12 \KK(n)} + 2^{-\frac12 n^{\frac14}} + \sqrt{\eps(n)}\right)
  \end{align}
  and hence
  \begin{align}\label{eq:B complement spectrum}
    \begin{split}
      \EE \norm{\sigma^{(b)}_{\bar{B}} - \phi^{(b)}_{\bar{B}}}_1 = \bigO(n^{\alpha + \frac12}\log(n)(2^{-\frac12 \KK(n)} + 2^{-\frac12 n^{\frac14}} + \eps(n))) \to 0.
    \end{split}
  \end{align}
  We have hence shown that $\sigma^{(j,a)}_B \approx \phi_B^{(j,a)}$ and $\sigma^{(j,b)}_{\bar B} \approx \phi_{\bar B}^{(j,b)}$ in expectation.

  We now show that we can approximate the state $\bar{\sigma}_B$ on $G'$ by the sum of these states. We first claim that we can sum over the $j$ index to obtain the $\sigma^{(a)}_B$ and $\sigma^{(b)}_B$. To see this, we introduce some notation: given a self-adjoint matrix $X$, let us write $\supp(X)$ for the image, or support, of $X$. In other words, $\supp(X)$ is the space spanned by the eigenvectors of $X$ with nonzero eigenvalue. We let $\HH_{B,j}$ be the support of the reduced state of $\ket{\Phi^+_{1,j}}$ on the $B$ system, and we similarly let $\HH_{\bar{B},j}$ be the support of the reduced state of $\ket{\Phi^+_{2,j}}$ on the $\bar{B}$ system.
  Then for $j \neq k$, the subspace $\HH_{B,j}$ is orthogonal to $\HH_{B,k}$ and similarly $\HH_{\bar{B},j}$ is orthogonal to $\HH_{\bar{B},k}$.
  By construction, it is clear that $\supp(\sigma^{(j,a)}_{\bar{B}}) \subseteq \HH_{\bar{B},j}$ and $\supp(\sigma^{(j,b)}_{B}) \subseteq \HH_{B,j}$.
  This orthogonality for indices $j\neq k$ then makes it clear that we can sum the reduced states, as in \cref{fig:sigma orth}:
  \begin{align*}
    \sigma^{(a)}_B = \sum_{j} \sigma^{(j,a)}_B, \qquad \sigma^{(b)}_{\bar{B}} = \sum_{j} \sigma^{(j,b)}_{\bar{B}},
  \end{align*}
  and by the purity of $\sigma^{(b)}$, we have that $\sigma_B^{(b)}$ has the same spectrum as $\sigma_{\bar{B}}^{(b)}$, We remind the reader that the $(a)$ and $(b)$ indices indicate that the minimal cut for the state is given by $\gamma_{A,1}$ and $\gamma_{A,2}$, respectively. Naturally, this also holds for the summed states, by orthogonality of the summands.

  We now claim that $\bar\sigma_B \approx \sigma^{(a)}_B +  \sigma^{(b)}_B$, i.e. the spectrum of the approximate state on $G'$, as in \cref{eq:bar sigma vs sigma}, is well-approximated by the sum of two states with differing minimal cuts. We estimate their difference by
  \begin{align*}
    \EE \norm{\sigma^{(a)}_B + \sigma^{(b)}_B - \bar\sigma_B}_1 & = \EE \norm{\sum_{j,k} \tr_{\bar{B}}[\ket{\sigma^{(j,a)}}\bra{\sigma^{(k,b)}} + \ket{\sigma^{(k,b)}}\bra{\sigma^{(j,a)}}]}_1 \\
                                                                & \leq 2 \sum_{j,k}\EE \norm{\tr_{\bar{B}}[\ket{\sigma^{(j,a)}}\bra{\sigma^{(k,b)}}]}_1                                        \\
                                                                & = 2\sum_{j,k} \EE F(\bar\sigma^{(j,a)}_{\bar{B}}, \bar\sigma^{(k,b)}_{\bar{B}})
  \end{align*}
  by \cref{lem:tracenorm off-diagonal}.
  Now, if $\rho$, $\sigma_1$ and $\sigma_2$ are positive operators, then by Lemma B.9 in \cite{fawzi2015quantum}, we have
  \begin{align*}
    \abs{F(\rho,\sigma_1) - F(\rho,\sigma_2)} \leq \sqrt{\norm{\sigma_1 - \sigma_2}_1 \tr[\rho]}.
  \end{align*}
  So, we find that
  \begin{align*}
    \EE F(\sigma^{(j,a)}_{\bar{B}}, \sigma^{(k,b)}_{\bar{B}}) & \leq \EE \abs{F(\sigma^{(j,a)}_{\bar{B}}, \sigma^{(k,b)}_{\bar{B}}) - F(\sigma^{(j,a)}_{\bar{B}}, \phi^{(k,b)}_{\bar{B}})} + \EE F(\sigma^{(j,a)}_{\bar{B}}, \phi^{(k,b)}_{\bar{B}})) \\
                                                              & \leq \EE \sqrt{\tr[\sigma^{(j,a)}_{\bar{B}}] \norm{\sigma^{(k,b)}_{\bar{B}} - \phi^{(k,b)}_{\bar{B}})}_1} + \EE F(\sigma^{(j,a)}_{\bar{B}}, \phi^{(k,b)}_{\bar{B}})                   \\
                                                              & \leq \sqrt{\EE \tr[\sigma^{(j,a)}_{\bar{B}}]} \sqrt{\EE \norm{\sigma^{(k,b)}_{\bar{B}} - \phi^{(k,b)}_{\bar{B}})}_1} + \EE F(\sigma^{(j,a)}_{\bar{B}}, \phi^{(k,b)}_{\bar{B}}),
  \end{align*}
  where we used Cauchy-Schwarz in the last step.
  For the first term, we may use \cref{eq:sigma j b} and $\EE \tr[\bar\sigma^{(j,a)}_{\bar{B}}] \leq 1$ to see
  \begin{align*}
    \sqrt{\EE \tr[\sigma^{(j,a)}_{\bar{B}}]} \sqrt{\EE \norm{\sigma^{(k,b)}_{\bar{B}} - \phi^{(k,b)}_{\bar{B}})}_1} = \bigO(2^{-\frac14 \KK(n)} + 2^{-\frac14 n^{\frac14}} + (\eps(n))^{\frac14}).
  \end{align*}
  For the second term we use a basic estimate on the fidelity: if $\rho$ and $\sigma$ are positive operators, then
  \begin{align*}
    F(\rho,\sigma) & = \norm{\sqrt{\rho}\sqrt{\sigma}}_1 \leq \norm{\sqrt{\rho}}_1 \norm{\sqrt{\sigma}}_{\infty} \leq \sqrt{\rank(\rho)} \norm{\sqrt{\rho}}_2 \norm{\sqrt{\sigma}}_{\infty} = \sqrt{\rank(\rho)\tr[\rho] \norm{\sigma}_{\infty}},
  \end{align*}
  which follows by H\"older's inequality and the standard relation between Schatten $1$ and $2$-norms. Then the second term can be estimated as follows: write $P_{\bar{B},j}$ for the projection onto $\HH_{\bar{B}, j}$, then
  \begin{align*}
    F(\sigma^{(j,a)}_{\bar{B}}, \phi^{(k,b)}_{\bar{B}}) & = F(\sigma^{(j,a)}_{\bar{B}}, P_{\bar{B},j}\phi^{(k,b)}_{\bar{B}}P_{\bar{B},j})                                                            \\
                                                        & \leq \sqrt{\tr[\sigma^{(j,a)}_{\bar{B}}] \rank(\sigma^{(j,a)}_{\bar{B}}) \norm{P_{\bar{B},j}\phi^{(k,b)}_{\bar{B}}P_{\bar{B},j}}_{\infty}} \\
                                                        & \leq \sqrt{\tr[\sigma^{(j,a)}_{\bar{B}}] \frac{2^{-n^{1/4}}}{p_{2,j}} p_{2,j}}                                                             \\ & \leq \sqrt{\tr[\sigma^{(j,a)}_{\bar{B}}] 2^{-n^{1/4}}}
  \end{align*}
  using \cref{eq:rank sigma a} and $\norm{P_{\bar{B},j}\phi^{(k,b)}_{\bar{B}}P_{\bar{B},j}}_{\infty} \leq p_{2,j}$.
  Thus
  \begin{align*}
    \EE F(\sigma^{(j,a)}_{\bar{B}}, \phi^{(k,b)}_{\bar{B}}) \leq \sqrt{ \EE \tr[\sigma^{(j,a)}_{\bar{B}}] 2^{-n^{1/4}}} \leq 2^{-\frac12 n^{1/4}}
  \end{align*}
  and we may estimate
  \begin{align*}
    \EE F(\sigma^{(j,a)}_{\bar{B}}, \sigma^{(k,b)}_{\bar{B}}) = \bigO(2^{-\frac14 \KK(n)} + 2^{-\frac14 n^{\frac14}} + (\eps(n))^{\frac14}).
  \end{align*}
  and therefore
  \begin{align}\label{eq:split sigma}
    \begin{split}
      \EE \norm{\sigma^{(a)}_B + \sigma^{(b)}_B - \bar\sigma_B}_1 &\leq 2\sum_{j,k} \EE F(\bar\sigma^{(j,a)}_{\bar{B}}, \bar\sigma^{(k,b)}_{\bar{B}}) = \bigO(n^{2\alpha + 1} \log(n)^2(2^{-\frac14 \KK(n)} + 2^{-\frac14 n^{\frac14}} + (\eps(n))^{\frac14})) \to 0
    \end{split}
  \end{align}
  using that the number of terms is $\bigO(n^{2\alpha + 1}\log(n)^2)$, $\KK(n) = \Omega(\log(n)^2)$ and $\eps(n) = \bigO(n^{-\gamma})$, and our choice of $\alpha$ is such that $\bigO(n^{2\alpha + 1} \log(n)^2 n^{-\frac14 \gamma})$ goes to zero.

  To summarize, we have shown that we can approximate the random tensor network state $\bar \sigma_B$ on $G'$ by $\sigma^{(a)}_B + \sigma^{(b)}_B$, and we can approximate $\sigma^{(a)}_B$ by $\phi^{(a)}_B$. Moreover, we can approximate $\sigma^{(b)}_{\bar{B}}$ by $\phi^{(b)}_{\bar{B}}$, so the spectrum of $\sigma^{(b)}_B$ can be approximated by the spectrum of $\phi^{(b)}_{\bar{B}}$. Recall that $\sigma^{(a)}$ and $\sigma^{(b)}$ are the random tensor network states on $G'$ which take $\phi^{(a)}$ and $\phi^{(b)}$ as background states, respectively. Furthermore, recall that our larger goal is to show that the spectra of the reduced background states $\phi_B^{(a)}$ and $\phi_B^{(b)}$ are, in some sense, close to the spectrum of the approximate random tensor network state $\sigma_B$ on the subgraph $G'$. This is a two-step process:
  \begin{enumerate}
    \item Show that $\phi^{(a)}_B$ and $\sigma_B^{(a)}$ can be slightly deformed to $\phi^{(a,\perp)}_B$ and $\sigma_B^{(a,\perp)}$, states with support orthogonal to $\supp(\sigma_B^{(j,b)})$ for all $j$.

    \item Use the slightly deformed states to show that $\spec_+(\phi^{(a)}_B\oplus\phi^{(b)}_B)$ is close to $\spec_+(\sigma^{(a)}_B+\sigma^{(b)}_B)$, and follow the chain of approximations: $\sigma^{(a)}_B+\sigma^{(b)}_B \approx \bar\sigma_B \approx \sigma_B$ along with $\spec_+(\sigma_B) \approx \spec_+(\rho_A)$ to conclude that $\spec_+(\phi^{(a)}_B\oplus\phi^{(b)}_B) \approx \spec_+(\rho_A)$
  \end{enumerate}

  Now, consider the state $\sigma^{(j,b)}_B$. It has support contained in $\HH_{B,j}$ and by \cref{eq:rank sigma b}, it has $\rank(\sigma^{(j,b)}_B) \leq 2^{-n^{1/4}} p_{1,j}^{-1}$.
  Let $Q_j$ be the projection onto $\supp(\sigma^{(j,b)}_B)$, so $\supp(Q_j) \subseteq \HH_{B,j}$ and $\rank(Q_j) \leq 2^{-n^{1/4}} p_{1,j}^{-1}$.
  We find that
  \begin{align*}
    \norm{Q_j \phi^{(a)}_B}_1 \leq \rank(Q_j) \norm{P_{B,j} \phi^{(a)}_B P_{B,j}}_{\infty} \leq \frac{2^{-n^{1/4}}}{p_{1,j}} p_{1,j} = 2^{-n^{1/4}}.
  \end{align*}
  If we let $Q = \sum_j Q_j$, then we see that $(I - Q)\phi^{(a)}(I - Q)$ is a small deformation of $\phi_B^{(a)}$:
  \begin{align}\label{eq:project phi}
    \norm{\phi^{(a)}_B - (I - Q)\phi^{(a)}(I - Q)}_1 \leq 2\norm{Q \phi^{(a)}}_1 \leq 2\sum_j \norm{Q_j \phi^{(a)}_B}_1= \bigO(n^{\nicefrac92}\log(n)2^{-n^{1/4}}) \to 0.
  \end{align}
  With this property in hand, we can show that $\sigma^{(a)}_B$ can similarly be deformed. Let $\sigma^{(a,\perp)}_B = (I - Q)\sigma^{(a)}(I - Q)$. By construction, $\sigma^{(a,\perp)}_B$ has support orthogonal to that of $\sigma^{(b)}_B$, and hence, $\spec_+(\sigma^{(a,\perp)}_B + \sigma^{(b)}_B) = \spec_+(\sigma^{(a,\perp)}_B \oplus \sigma^{(b)}_{\bar{B}})$.
  We observe that $\sigma^{(a,\perp)}_B$ is a small deformation of $\sigma^{(a)}_B$:
  \begin{align}\label{eq:perturb sigma}
    \begin{split}
      \EE\norm{\sigma_{B}^{(a)} - \sigma^{(a,\perp)}_B}_1 &\leq \EE \norm{\sigma_{B}^{(a)} - \phi_{B}^{(a)}}_1\\
      &\qquad + \EE \norm{(I- Q)\sigma_{B}^{(a)}(I- Q) - (I- Q)\phi_{B}^{(a)} (I - Q)}_1\\
      &\qquad + \norm{\phi^{(a)}_B - (I - Q)\phi^{(a)}(I - Q)}_1\\
      &\leq 2\EE \norm{\sigma_{B}^{(a)} - \phi_{B}^{(a)}}_1 + \norm{\phi^{(a)}_B - (I - Q)\phi^{(a)}(I - Q)}_1
    \end{split}
  \end{align}
  which goes to zero, using \cref{eq:B spectrum sigma a} for the first term, and \cref{eq:project phi} for the second term.
  By construction, $(I - Q)\phi^{(a)}(I - Q)$ has support orthogonal to $\supp(\sigma^{(b)}_{B})$.
  This allows us to bound
  \begin{align*}
     & \EE \norm{\spec_+(\phi^{(a)}_B \oplus \phi^{(a)}_{\bar{B}}) - \spec_+(\sigma^{(a)}_B + \sigma^{(b)}_B)}_1                                                                                    \\
     & \qquad \leq \EE \norm{\sigma_{B}^{(a)} - \sigma^{(a,\perp)}_B}_1 + \EE \norm{\spec_+(\phi^{(a)}_B \oplus \phi^{(a)}_{\bar{B}}) - \spec_+(\sigma^{(a,\perp)}_B \oplus \sigma^{(b)}_B)}_1      \\
     & \qquad = \EE \norm{\sigma_{B}^{(a)} - \sigma^{(a,\perp)}_B}_1 + \EE \norm{\spec_+(\phi^{(a)}_B)- \spec_+(\sigma^{(a,\perp)}_B) + \spec_+(\phi^{(a)}_{\bar{B}})  - \spec_+(\sigma^{(b)}_B)}_1 \\
     & \qquad \leq \EE \norm{\sigma_{B}^{(a)} - \sigma^{(a,\perp)}_B}_1 + \EE \norm{\phi^{(a)}_B - \sigma^{(a,\perp)}_B}_1 + \EE \norm{\spec_+(\phi^{(a)}_{\bar{B}}) - \spec_+(\sigma^{(b)}_B)}_1,
  \end{align*}
  where we have used the bound on the difference of spectra by the trace distance of the corresponding states \cref{eq:spectrum bound}.
  Each term on the RHS can be shown to converge to zero. The first term goes to zero by \cref{eq:perturb sigma}.
  Similarly, the second term can be bounded by
  \begin{align*}
    \EE \norm{\phi^{(a)}_B - \sigma^{(a,\perp)}_B}_1 \leq \EE \norm{\phi^{(a)}_B - \sigma^{(a)}_B}_1 + \EE \norm{\sigma^{(a)}_B - \sigma^{(a,\perp)}_B}_1
  \end{align*}
  which goes to zero by \cref{eq:perturb sigma} and \cref{eq:B spectrum sigma a}.
  The third term can be estimated by observing that $\spec_+(\sigma^{(b)}_B) = \spec_+(\sigma^{(b)}_{\bar{B}})$ and
  \begin{align*}
    \EE \norm{\spec_+(\phi^{(a)}_{\bar{B}}) - \spec_+(\sigma^{(b)}_{B})}_1 \leq \EE \norm{\phi^{(a)}_{\bar{B}} - \sigma^{(b)}_{\bar{B}}}_1
  \end{align*}
  which goes to zero by \cref{eq:B complement spectrum}. We conclude that $\EE \norm{\spec_+(\phi^{(a)}_B \oplus \phi^{(a)}_{\bar{B}}) - \spec_+(\sigma^{(a)}_B + \sigma^{(b)}_B)}_1 \to 0$.

  We now follow a chain of approximations to get the desired closeness between $\spec_+(\phi^{(a)}_B \oplus \phi^{(a)}_{\bar{B}})$ and $\spec_+(\rho_A)$. Using \cref{eq:bar sigma vs sigma} and \cref{eq:split sigma}, we can see that
  \begin{align*}
    \EE \norm{\sigma_B - (\sigma^{(a)}_B + \sigma^{(b)}_B)}_1 \leq \EE \norm{\sigma_B - \bar\sigma_B}_1 + \EE\norm{\sigma^{(a)}_B + \sigma^{(b)}_B - \bar\sigma_B}_1 \to 0,
  \end{align*}
  so $\EE \norm{\sigma_B - (\sigma^{(a)}_B + \sigma^{(b)}_B)}_1$ will converge to 0. This then implies
  \begin{align*}
    \EE \norm{\spec_+(\sigma_B) - \spec_+(\phi^{(a)}_B \oplus \phi^{(a)}_{\bar{B}})}_1 & \leq \EE \norm{\sigma_B - (\sigma^{(a)}_B + \sigma^{(b)}_B)}_1                                                      \\
                                                                                       & \qquad + \EE \norm{\spec_+(\phi^{(a)}_B \oplus \phi^{(a)}_{\bar{B}}) - \spec_+(\sigma^{(a)}_B + \sigma^{(b)}_B)}_1,
  \end{align*}
  and as we have just shown, both terms on the RHS converge to zero.
  Together with \cref{eq:tilde rho vs sigma}, this yields the desired relationship to the spectrum of $\rho$:
  \begin{align}\label{eq:rho vs phi a plus phi b}
    \begin{split}
      \EE \norm{\spec_+(\phi^{(a)}_B \oplus \phi^{(b)}_{\bar{B}}) - \spec_+(\rho_A)}_1 &\leq \EE \norm{\spec_+(\phi^{(a)}_B \oplus \phi^{(b)}_{\bar{B}}) - \spec_+(\sigma_{B})}_1\\
      &\qquad + \EE \norm{\spec_+(\rho_A) - \spec_+(\sigma_{B})}_1 \rightarrow 0.
    \end{split}
  \end{align}

  We pause here again to note that we have accomplished Step 3: approximating the spectrum of $\rho_A$ by the spectra of the (approximate) background states. The final step will then be to consider the convergence properties of the distributions on the background states, which will then translate to convergence for the distribution on $\rho_A$.

  Explicitly, we want to relate the above result to $\minstar(\nu_1,\nu_2)$.
  First, we observe that by \cref{lem:trace distance convergence} and \cref{eq:phi vs tilde phi}, $\nu_{\tilde \phi_B} \Rightarrow \nu_1$ and $\nu_{\tilde \phi_{\bar{B}}} \Rightarrow \nu_2$, and since $\min$ is a continuous function, convergence holds for the pushforward measure:
  \begin{align}\label{eq:min phi tilde converges}
    \minstar(\nu_{\tilde \phi_{B}},\nu_{\tilde \phi_{\bar{B}}}) \Rightarrow \minstar(\nu_1,\nu_2).
  \end{align}
  We compute
  \begin{align}\label{eq:min mu phi tilde}
    \begin{split}
      \minstar(\nu_{\tilde \phi_{B}},\nu_{\tilde \phi_{\bar{B}}}) &= \sum_{j,k} q_{1,k}q_{2,j}\delta_{\frac{1}{\sqrt{n}}[\min(\log(\frac{1}{p_{1,k}}),\log(\frac{1}{p_{2,j}})) - H(n)]}\\
      &= \sum_{j}\bigl( \sum_{p_{1,k} \geq 2^{n^{1/4}}p_{2,j}} q_{1,k}q_{2,j}\delta_{\frac{1}{\sqrt{n}}[\log(\frac{1}{p_{1,k}}) - H(n)]} \\
      &\quad+ \sum_{p_{1,k} \leq 2^{-n^{\nicefrac{1}{4}}}p_{2,j}} q_{1,k}q_{2,j}\delta_{\frac{1}{\sqrt{n}}[\log(\frac{1}{p_{2,j}}) - H(n)]} \bigr)\\
      &\quad+ \sum_{2^{-n^{\nicefrac{1}{4}}}p_{2,j} \leq p_{1,k} \leq 2^{n^{\nicefrac{1}{4}}}p_{2,j}} q_{1,k}q_{2,j}\delta_{\frac{1}{\sqrt{n}}[\min(\log(\frac{1}{p_{1,k}}),\log(\frac{1}{p_{2,j}})) - H(n)]}\\
      &= \sum_{k} q_{1,k}q_{2,>k}\delta_{\frac{1}{\sqrt{n}}[\log(\frac{1}{p_{1,k}}) - H(n)]} + \sum_{j} q_{1,> j}q_{2,j}\delta_{\frac{1}{\sqrt{n}}[\log(\frac{1}{p_{2,j}}) - H(n)]}  + \nu_0
    \end{split}
  \end{align}
  where we have written
  \begin{align*}
    q_{1,>j} & = \sum_{p_{1,k} \geq 2^{n^{1/4}}p_{2,j}} q_{1,k}              \\
    q_{2,>k} & = \sum_{p_{1,k} \leq 2^{-n^{\nicefrac{1}{4}}}p_{2,j}} q_{2,j}
  \end{align*}
  and
  \begin{align*}
    \nu_0 = \sum_{2^{-n^{\nicefrac{1}{4}}}p_{2,j} \leq p_{1,k} \leq 2^{n^{\nicefrac{1}{4}}}p_{2,j}} q_{1,k}q_{2,j}\delta_{\frac{1}{\sqrt{n}}[\min(\log(\frac{1}{p_{1,k}}),\log(\frac{1}{p_{2,j}})) - H(n)]}.
  \end{align*}
  Then by \cref{eq:middle measure zero}, we see that $\nu_0 \Rightarrow 0$.
  Using \cref{eq:min phi tilde converges}, we conclude
  \begin{align}\label{eq:min phi tilde min nu converges}
    \minstar(\nu_{\tilde \phi_{B}},\nu_{\tilde \phi_{\bar{B}}}) - \nu_0 \Rightarrow \minstar(\nu_1,\nu_2).
  \end{align}
  On the other hand, by construction, we have
  \begin{align*}
    \nu_{\spec(\phi^{(a)}_B \oplus \phi^{(b)}_{\bar{B}})} = \nu_{\phi^{(a)}_B} + \nu_{\phi^{(b)}_{\bar{B}}}.
  \end{align*}
  We can explicitly write down
  \begin{align*}
    \nu_{\phi^{(a)}_B} + \nu_{\phi^{(b)}_{\bar{B}}} = \sum_k q_{1,k}q_{2,>k}\delta_{\frac{1}{\sqrt{n}}[\log(\frac{1}{p_{1,k}q_{2,>k}}) - H(n)]} + \sum_{j} q_{1,> k}q_{2,j}\delta_{\frac{1}{\sqrt{n}}[\log(\frac{1}{p_{2,j}q_{1,> k}}) - H(n)]}.
  \end{align*}
  Since the $q_{1,k}$ and $q_{2,j}$ are at least $n^{-\beta}$ (if the corresponding term is nonzero), $\abs{\log(q_{1,>j})} \leq \beta\log(n)$ and $\abs{\log(q_{2,>k})} \leq \beta\log(n)$. With these bounds in mind, we can compare to the last line of \cref{eq:min mu phi tilde}. More precisely, if we let $\nu'_{\bar \phi} = \minstar(\nu_{\tilde \phi_{B}},\nu_{\tilde \phi_{\bar{B}}}) - \nu_0$, then for uniformly continuous $f \in C_b(\RR)$, we have
  \begin{align*}
    \lim_{n \to \infty} \abs{\int f(x) \d \nu'_{\bar \phi}(x) - \int f(x) \d \nu_{\phi^{(a)}_B}(x) - \int f(x) \d \nu_{\phi^{(b)}_{\bar{B}}}(x)} = 0
  \end{align*}
  We conclude that
  \begin{align}\label{eq:split phi converges to min}
    \nu_{\phi^{(a)}_B} + \nu_{\phi^{(b)}_{\bar{B}}} \Rightarrow \minstar(\nu_1,\nu_2).
  \end{align}
  Now, we can finally put all of our ingredients together. Recall the statement of convergence in probability implied by the convergence of spectra, in expectation, as in the second part of \cref{lem:trace distance convergence}. Then using \cref{eq:rho vs phi a plus phi b} as the vectors $p^{(n)}$ and $q^{(n)}$ in the statement of \cref{lem:trace distance convergence} and the convergence of distributions in \cref{eq:split phi converges to min}, we conclude
  \begin{align}\label{eq:smoothed state converges}
    \nu_{\rho_A} \Rightarrow \minstar(\nu_1,\nu_2)
  \end{align}

\end{proof}

\subsection{Computing entropies with two minimal cuts}
Ideally, we would like to use \cref{thm:measure convergence surface transition} to compute von Neumann entropies of random tensor network states. However, \cref{thm:measure convergence surface transition} alone is too weak to  allow us to directly compute entropies up to $o(\sqrt{n})$ corrections, as weak convergence of the spectrum does not directly imply convergence of the mean.
However, we kept track of various approximation errors in the proof of \cref{thm:measure convergence surface transition}, and we will use these to show that  with slightly stronger assumptions, these errors allow to compute the entropy up to $\bigO(\log(n))$ corrections.

To begin, we first bound the difference in the entropy of a sum of density matrices and the sum of the entropies of the individual density matrices:
\begin{lem}\label{lem:entropy convexity unnormalized}
  Suppose $\{p_i\}_{i=1}^m$ is a subnormalized distribution and $\{\rho_i\in \PSD(\HH) \}_{i=1}^m$ is a collection of (unnormalized) density matrices and let $\rho = \sum_i p_i \rho_i \in \PSD(\HH)$.
  Suppose that $C^{-1} \leq \tr[\rho_i] \leq C$, and $C^{-1} \leq \tr[\rho] \leq C$, then
  \begin{align*}
    \abs{H(\rho) - \sum_{i=1}^m p_i H(\rho_i)} \leq C(\log(m) + 2\log(C)).
  \end{align*}
\end{lem}
\begin{proof}
  If the $\rho_i$ are normalized and $\sum_i p_i = 1$, then by the Holevo bound we have
  \begin{align*}
    \sum_i p_i H(\rho_i) \leq H(\sum_i p_i \rho_i) \leq \sum_i p_i H(\rho_i) + H(\{p_i\}) \leq \sum_i p_i H(\rho_i) + \log(m).
  \end{align*}
  Now let $q_i = \tr[\rho_i]$, $P = \sum_i {q_i p_i}$ and let $\sigma_i = \frac{\rho_i}{q_i}$, $r_i = \frac{p_i q_i}{P}$.
  Then $H(\sum_i p_i \rho_i) = H(P\sum_i r_i \sigma_i) = P H(\sum_i r_i \sigma_i) - P\log(P)$.
  On the other hand
  \begin{align*}
    P\sum_i r_i H(\sigma_i) = \sum_i p_i H(\rho_i) + \sum_i p_i\log(q_i)
  \end{align*}
  so it follows that
  \begin{align*}
    \abs{H(\sum_i p_i \rho_i) - \sum_i p_iH(\rho_i)} \leq P\log(m) + \abs{P\log(P)} + \max \abs{\log(q_i)}\leq C(\log(m) + 2\log(C)).
  \end{align*}
\end{proof}

The idea is that for a random tensor network state, we will split up the background state as a superposition of states which are maximally entangled along the two minimal cuts and then use \cref{lem:entropy convexity unnormalized}.
This approach formalizes an argument sketched in \cite{akers2020leading}.

To state our result, we introduce a new function: for vectors of positive numbers $p \in \RR^{d_1}$, $q \in \RR^{d_2}$, we let
\begin{align*}
  H^*(p,q) := \sum_{i,j} p_i q_j \min\mleft(\log\left(\frac{1}{p_i}\right),\log\left(\frac{1}{q_j}\right) \mright).
\end{align*}
A key tool we will need is the continuity of the entropy: if $\rho, \sigma \in \Pleq(\HH)$ are quantum states on a Hilbert space $\HH$ of dimension $d$ with $T(\rho,\sigma) \leq \frac1e$, then the Fannes-Audenaert inequality states that
\begin{align}\label{eq:fannes}
  \abs{H(\rho) - H(\sigma)} \leq T(\rho,\sigma)\log d - T(\rho,\sigma)\log(T(\rho,\sigma)).
\end{align}
Let $\eta$ be the function defined by $\eta(x) = x\log x$.
Then, if we write $\spec(\rho) = \{p_j\}_{j=1}^d$ and $\spec(\sigma) = \{q_j\}_{j=1}^d$ for \cref{eq:fannes} one can actually show
\begin{align*}
  \abs{H(\rho) - H(\sigma)} \leq \sum_{j=1}^d \, \abs{\eta(p_j) - \eta(q_j)} \leq T(\rho,\sigma)\log d - T(\rho,\sigma)\log(T(\rho,\sigma)).
\end{align*}
We will use this to show continuity of $H^*$ as well.
Consider $\rho_i, \sigma_i \in \Pleq(\HH_i)$ with $\dim(\HH_i) \leq d$ for $i = 1,2$.
We let $\spec(\rho_i) = \{p_{i,j}\}_{j=1}^{d_i}$ and $\spec(\sigma_i) = \{q_{i,j}\}_{j=1}^{d_i}$.
For real numbers $x_i, y_i$ we have
\begin{align*}
  \abs{\min(x_1,y_1) - \min(x_2,y_2)} \leq \abs{x_1 - x_2} + \abs{y_1 - y_2}.
\end{align*}
Then we see that
\begin{align*}
   & \abs{H^*(\spec(\rho_1),\spec(\rho_2)) - H^*(\spec(\sigma_1),\spec(\sigma_2))}                                                                                                                                                                                     \\
   & \qquad \leq \sum_{j,k} \, \abs{ p_{1,j} p_{2,k} \min \left( \log \mleft(\frac{1}{p_{1,j}}\mright) ,\log \mleft(\frac{1}{p_{2,k}}\mright)\right) - q_{1,j} q_{2,k} \min \left(\log \mleft(\frac{1}{q_{1,j}}\mright), \log \mleft(\frac{1}{q_{2,k}}\mright)\right)} \\
   & \qquad \leq \sum_{j,k} \, \abs{ p_{1,j} p_{2,k} \log p_{1,j} - q_{1,j} q_{2,k}  \log q_{1,j}} + \abs{ p_{1,j} p_{2,k} \log p_{2,k} - q_{1,j} q_{2,k} \log q_{2,k}}.
\end{align*}
We may estimate the first term by
\begin{align*}
  \sum_{j,k} \, \abs{ p_{1,j} p_{2,k} \log p_{1,j} - q_{1,j} q_{2,k}  \log q_{1,j}} & \leq \sum_{j,k} p_{2,k}\abs{ p_{1,j} \log p_{1,j} - q_{1,j} \log q_{1,j}} - \sum_{j,k} \, \abs{ p_{2,k} - q_{2,k}} \, q_{1,j}  \log q_{1,j} \\
                                                                                    & \leq \sum_j \, \abs{\eta(p_{1,j}) - \eta(p_{2,j})} + \norm{\rho_2 - \sigma_2}_1 H(\sigma_1)                                                 \\
                                                                                    & \leq T(\rho_1,\sigma_1)\log d - T(\rho_1,\sigma_1)\log(T(\rho_1,\sigma_1)) + 2T(\rho_2,\sigma_2)\log d
\end{align*}
The second term may be estimated in similar fashion.
We conclude that if $T(\rho_1, \sigma_1)\leq \eps \leq e^{-1}$ and $T(\rho_2, \sigma_2)\leq \eps \leq e^{-1}$
\begin{align}\label{eq:fannes h star}
  \abs{H^*(\spec(\rho_1),\spec(\rho_2)) - H^*(\spec(\sigma_1),\spec(\sigma_2))} \leq 6\eps\log d - 2\eps\log(\eps).
\end{align}

We will now show that we can approximate the entropy as would be expected from \cref{thm:measure convergence surface transition}, if we make some additional assumptions (which in particular are satisfied if the state along each edge is a copy of $n$ states, $\phi_e = \phi_{e,0}^{\ot n}$).

\begin{cor}\label{cor:convergence of entropy}
  Let $\rho$ be a random tensor network state satisfying the same assumptions as in \cref{thm:measure convergence surface transition}, and assume additionally that the $\nu_{\phi_{\gamma_{A,i}}}$ have uniformly exponentially decaying tail probabilities and $\eps(n) = \bigO(n^{-\gamma})$ for $\gamma > 10$.
  Assume that for each edge $e = (xy)$, the bond dimension is $D_e = 2^{\bigO(n)}$ and $\tr[\phi_{e,x}^2] \leq 2^{-\Omega(n)}$.
  Then with high probability
  \begin{align*}
    \abs{H(\rho_A) - H^*(\spec(\phi_{\gamma_{A,1}}),\spec(\phi_{\gamma_{A,2}}))} = \bigO(\log(n)).
  \end{align*}
\end{cor}

\begin{proof}
  The basic proof strategy will be that of \cref{thm:measure convergence surface transition}: we study a slightly reduced problem on the approximated tensor network states $\tilde{\sigma}$ with approximate background states $\tilde{\phi}$, work out the entropies for $\tilde{\sigma}$ and $\tilde{\phi}$, and then argue that the closeness of the resulting entropies will continue to hold for the original tensor network state and background state, up to errors that we carefully keep track of.

  First of all, we note that we can reduce to the tensor network state $\sigma$ on the reduced graph $G'$, with error as in \cref{eq:tilde rho vs sigma}; in particular
  \begin{align}\label{eq:rho vs sigma}
    \EE \norm{\spec_+(\rho_A) - \spec_+(\sigma_{B})}_1 = \bigO(n^{-\frac{\gamma}{2}}).
  \end{align}

  Next we adapt the part of the proof of \cref{thm:measure convergence surface transition} where we modify the state along the minimal cuts.
  In the proof of \cref{thm:measure convergence surface transition}, we first observed that if the (regularized) spectrum along each cut $\nu_{\phi_{\Gamma_{A,i}}}$ has uniformly exponentially decaying tail probabilities, then for sufficiently large $C$, we can slice off the tails with vanishing probability mass:
  \begin{align*}
    \sum_{\lambda_{i,j} \notin I_n} \lambda_{i,j} = \nu_{\phi_{\Gamma_{A,i}}}((-\infty,-C\log(n))\cup(C\log(n),\infty)) = O\left(\frac{1}{n^4}\right).
  \end{align*}
  We also binned the eigenvalues of $\phi$ in the reduced spectrum:
  \begin{align*}
    q_{i,j} = d_{i,j}p_{i,j},
  \end{align*}
  where $p_{i,j}$ are the bins and $d_{i,j}$ is the multiplicity of each bin.
  We then removed bins that were too small, leading to a state $\tilde \phi$.
  In the notation of the proof of \cref{thm:measure convergence surface transition} we choose $\alpha $ and $\beta$ such that $\alpha > 2$ and $\frac{5}{2} + \alpha < \beta \leq \frac12(\gamma - 1)$.
  By \cref{eq:phi vs tilde phi} this yields an error
  \begin{align}\label{eq:phi tilde phi entropy}
    \norm{\tilde \phi - \phi}_1 = o(n^{-1}).
  \end{align}

  Note that in the proof of \cref{thm:measure convergence surface transition}, we performed one more approximation of removing the ``middle'' or ``diagonal'' part of the spectrum and obtained a state $\bar{\phi}$; we will not need to do this here. Now, we recall that the binning of eigenvalues allows us to write $\tilde{\phi}$ as a superposition over maximally-entangled states along each cut:
  \begin{align*}
    \ket{\tilde \phi_{\gamma_{A,1}, \gamma_{A,2}}} & = \sum_{j,k} \sqrt{q_{1,j} q_{2,k}} \ket{\Phi^+_{1,j}} \ot \ket{\Phi^+_{2,k}},
  \end{align*}
  and use this to decompose $\tilde \phi$ as
  \begin{align*}
    \ket{\tilde \phi} & = \sum_{j,k} \sqrt{q_{1,j} q_{2,k}} \ket{\psi^{(j,k)}},
  \end{align*}
  where in $\ket{\psi^{(j,k)}}$, we have replaced $\ket{\tilde \phi_{\gamma_{A,1}, \gamma_{A,2}}}$ with the maximally entangled state $\ket{\Phi^+_{1,j}} \ot \ket{\Phi^+_{2,k}}$.
  The states $\ket{\psi^{(j,k)}}$ are normalized background states on the graph $G'$.
  We let
  \begin{align*}
    \ket{\phi^{(j,k)}} = \sqrt{q_{1,j} q_{2,k}}\ket{\psi^{(j,k)}}
  \end{align*}
  which are subnormalized states.
  At this point, the key idea of the argument is straightforward.
  We will consider the random tensor network states which have (a smoothed version) of the background states $\ket{\psi^{(j,k)}}$.
  These will have entropy close to $\min(\log(d_{1,j}), \log(d_{2,k}))$.
  Then we will use \cref{lem:entropy convexity unnormalized} to argue that the entropy of $\rho_A$ is approximated up to $\bigO(\log(n))$ terms by the convex combination $\sum_{j,k} q_{1,j} q_{2,k} \min(\log(d_{1,j}), \log(d_{2,k}))$, which we can relate to the desired result.
  To make this easy argument precise, we will need to take care of smoothing the background state appropriately and ensure that the relevant states are close to normalized.

  We will now argue that we can smoothen the states $\phi^{(j,k)}$.
  We may assume without loss of generality that the states $\phi^{(j,k)}$ have nonnegative coefficients in the standard basis $\ket{I} = \prod_{e \in E'} \ket{i_e i_e}$ where $I = \{i_e\}_{e \in E'}$ runs over all possible basis elements over each edge, so we can write
  \begin{align*}
    \ket{\phi^{(j,k)}} = \sum_I \sqrt{\lambda^{(j,k)}_I}.
  \end{align*}
  By the same argument as in \cref{thm:measure convergence surface transition} we find a state
  \begin{align*}
    \phi^{(j,k,\eps)} =  \sum_I \sqrt{\lambda^{(j,k,\eps)}_I}
  \end{align*}
  which is such that for any $\Delta_B \in C(B)$ not equal to $B$ or $V_b' \cup B$ and $\Delta_A = \Gamma_{A,1} \cup \Delta_B$
  \begin{align*}
    H_{\min}(\Delta_B \setminus B \vert B)_{\phi^{(j,k,\eps)}} \geq H^\eps_{\min}(\Delta_A \setminus \Gamma_A \vert \Gamma_A)_{\phi},
  \end{align*}
  while for $\Delta_B = B$ and $\Delta_B = B \cup V_b'$ we have
  \begin{align*}
    H_{\min}(B)_{\phi^{(j,k,\eps)}} \geq H_{\min}(B)_{\phi^{(j,k)}} = \log \frac{1}{q_{2,k} p_{1,j}} \\
    H_{\min}(\bar B)_{\phi^{(j,k,\eps)}} \geq H_{\min}(\bar B)_{\phi^{(j,k)}} = \log \frac{1}{q_{1,j}p_{2,k}}.
  \end{align*}
  Moreover, $\phi^{(j,k,\eps)}$ is close to $\phi^{(j,k)}$ in the sense that
  \begin{align*}
    \norm{\phi^{(j,k,\eps)} - \phi^{(j,k)}}_1 & = \bigO\mleft(\sqrt{\sum_I \abs{\lambda^{(j,k,\eps)} - \lambda^{(j,k)}}}\mright) = \bigO(\sqrt{\eps}).
  \end{align*}
  By the remark after \cref{lem:joint smoothing at transition} we may assume that $\lambda^{(j,k,\eps)} \leq \lambda^{(j,k)}$.

  By a chain rule (e.g. Theorem 5.13 in \cite{tomamichel2015quantum}, proven in \cite{dupuis2015chain}) for any $\Delta_B \in C(B)$ it holds that
  \begin{align*}
    H_2(\Delta_B)_{\phi^{(j,k,\eps)}} \geq H_{\min}(\Delta_B \setminus B \vert B)_{\phi^{(j,k,\eps)}} + H_{\min}(B)_{\phi^{(j,k,\eps)}}.
  \end{align*}
  We now let
  \begin{align*}
    \ket{\phi^{\eps}} = \sum_{j,k} \ket{\phi^{(j,k,\eps)}}.
  \end{align*}
  Then, by \cref{lem:consequence fvdg}
  \begin{align}\label{eq:phi eps vs tilde phi}
    \norm{\tilde \phi - \phi^{\eps}}_1 \leq \sqrt{2 \sum_{j,k} \sum_I \abs{\lambda_I^{(j,k)} - \lambda_I^{(j,k,\eps)}}} = \bigO(\sqrt{n^{2\alpha}\log(n)^2\eps(n)}) = \bigO(n^{\alpha - \frac12\gamma}\log(n)) = \bigO(n^{-3}).
  \end{align}
  Therefore, if we let $\sigma^{\eps}_B$ be the random tensor network state with background state $\phi^{\eps}$
  \begin{align}\label{eq:sigma vs sigma eps}
    \EE \norm{\sigma - \sigma^{\eps}}_1 \leq \norm{\phi - \tilde \phi}_1 + \norm{\tilde \phi - \phi^{\eps}}_1 = o(\frac{1}{n}).
  \end{align}
  Finally, let
  \begin{align*}
    \ket{\phi^{(j,k,\eps)}} = \sqrt{q_{1,j} q_{2,k}}\ket{\psi^{(j,k,\eps)}}
  \end{align*}
  By \cref{eq:trace difference} in the remark after the proof of \cref{lem:joint smoothing at transition} we have
  \begin{align*}
    \abs{\tr[\psi^{(j,k,\eps)}] - 1} & = \frac{1}{q_{1,j}q_{2,k}}\abs{\tr[\phi^{(j,k,\eps)}] - \tr[\phi^{(j,k)}]}         \\
                                     & = \bigO(\frac{\eps}{q_{1,j}q_{2,k}}) = \bigO(n^{2\beta - \gamma}) = \bigO(n^{-1}).
  \end{align*}
  using that $q_{1,j}, q_{2,k} > n^{-\beta}$ and $\beta \leq \frac12(\gamma - 1)$.
  We denote the random tensor network states with background states $\psi^{(j,k,\eps)}$ by $\sigma^{(j,k)}$, and the random tensor network states with background states $\sum_{j} \sqrt{q_{1,j}}\psi^{(j,k,\eps)}$ by $\sigma^{(k)}$.

  We introduce the event $N$, which entails that $\rho, \sigma$ and $\sigma^{(j,k)}$ for all $j,k$ are close to normalized, that is
  \begin{align*}
    \abs{\tr[\rho] - 1} \leq \frac{1}{n^2} \qquad \text{ and } \qquad \abs{\tr[\sigma] - 1} \leq \frac{1}{n^2} \qquad \text{ and } \qquad \abs{\tr[\sigma^{(j,k)}] - \tr[\psi^{(j,k,\eps)}]} \leq \frac{1}{n^2}.
  \end{align*}
  By assumption, for each edge $e = (xy)$ we have $\tr[\phi_{e,x}^2] \leq 2^{-\Omega(n)}$.
  Moreover, using that $q_{1,j}, q_{2,k} \geq n^{-\beta}$
  \begin{align*}
    \tr[(\phi^{(j,k,\eps)}_{\Delta})^2] \leq q_{1,j}^{-2} q_{2,k}^{-2} \tr[\phi^2_{\Delta}] = \bigO(\poly(n)2^{-\Omega(n)})
  \end{align*}
  and therefore the quantity $\eta$ in \cref{lem:normalization} is $\bigO(2^{-\Omega(n)})$ for $\rho$, $\sigma$ and $\sigma^{(j,k)}$.
  So, by \cref{lem:normalization} and the union bound the event $N$ has probability
  \begin{align*}
    p_N & \geq 1 - \left( \Pr(\abs{\tr[\rho] - 1} \geq \frac{1}{n^2}) + \Pr(\abs{\tr[\sigma] - 1} \geq \frac{1}{n^2}) + \sum_{j,k} \Pr(\abs{\tr[\sigma^{(j,k)}] - \tr[\psi^{(j,k,\eps)}]} \geq \frac{1}{n^2}) \right) \\
        & = 1 - \bigO(\poly(n)2^{-\Omega(n)}) = 1 - \bigO(2^{-\Omega(n)}).
  \end{align*}
  We denote by $\EE_N$ the expectation value over the random tensors conditioned on this event.

  We now use \cref{lem:entropy convexity unnormalized} to approximate the entropy of $\sigma^{\eps}_B$ conditioned on $N$:
  \begin{align*}
    \abs{H(\sigma^{\eps}_B) - \sum_k q_{2,k} H(\sigma_B^{(k)})} = \bigO(\log(n)).
  \end{align*}
  Because $\sigma^{(k)}$ is pure, we have $H(\sigma_B^{(k)}) = H(\sigma_{\bar B}^{(k)})$. We apply \cref{lem:entropy convexity unnormalized} again, this time to the decomposition of $\sigma_{\bar B}^{(k)}$ to see:
  \begin{align*}
    \abs{H(\sigma_{\bar B}^{(k)}) - \sum_j q_{1,j} H(\sigma_{\bar B}^{(j,k)})} = \bigO(\log(n)).
  \end{align*}
  Thus,
  \begin{align*}
    \abs{H(\sigma^{\eps}_B) - \sum_{j,k} q_{1,j} q_{2,k} H(\sigma_{B}^{(j,k)})} = \bigO(\log(n)).
  \end{align*}
  Since $\rank(\sigma_{B}^{(j,k)}) \leq \min\{d_{1,j}, d_{2,k} \}$ (and taking into account the normalization of $\sigma^{(j,k)}$)
  \begin{align*}
    H(\sigma_{B}^{(j,k)}) \leq \tr[\sigma^{(j,k)}]\min\mleft( \log d_{1,j}, \log d_{2,k} \mright) - \tr[\sigma^{(j,k)}] \log \tr[\sigma^{(j,k)}].
  \end{align*}
  Conditioned on $N$ and using $\abs{\tr[\psi^{(j,k,\eps)}] - 1} = \bigO(n^{-1})$ this implies
  \begin{align*}
    H(\sigma_{B}^{(j,k)}) \leq \min\mleft( \log d_{1,j}, \log d_{2,k} \mright) + \bigO(1).
  \end{align*}
  For a lower bound we use that
  \begin{align*}
    \EE_N H(\sigma_{B}^{(j,k)}) & \geq \EE_N \tr[\sigma^{(j,k)}]\left( -\log \tr\mleft[(\sigma_{B}^{(j,k)})^2\mright] + \log \tr[\sigma_{B}^{(j,k)}]\right)                                         \\
                                & \geq (\tr[\psi^{(j,k,\eps)}] - \frac{1}{n^2})\left( -\log \EE_N \tr\mleft[ (\sigma_{B}^{(j,k)})^2\mright] + \log(\tr[\psi^{(j,k,\eps)}] - \frac{1}{n^2})] \right) \\
                                & \geq -\log \EE_N \tr\mleft[ (\sigma_{B}^{(j,k)})^2\mright] - \bigO(1)
  \end{align*}
  where in the first and second inequality we have used Jensen's inequality, and in the second and third inequality we have used $\tr[\sigma_{B}^{(j,k)}] \geq \tr[\psi^{(j,k,\eps)}] - \frac{1}{n^2} \geq 1 - \bigO(\frac{1}{n})$.
  We use the replica trick to estimate
  \begin{align*}
    \EE \tr\mleft[ (\sigma^{(j,k)}_B)^2 \mright] = \sum_{\Delta_B \in C(B)} \tr[(\psi^{(j,k,\eps)}_{\Delta_B})^2] 
  \end{align*}
  In this expression, we have contributions from $\Delta_B = B$ and $\Delta_B = B \cup V_b'$, which yield contributions
  \begin{align*}
    \tr[(\psi^{(j,k,\eps)}_{B})^2] \leq d_{1,j}^{-1} \\
    \tr[(\psi^{(j,k,\eps)}_{\bar B})^2] \leq d_{2,k}^{-1}.
  \end{align*}
  For any other cut $\Delta_B$, we have a contribution at most
  \begin{align*}
    \tr[(\psi^{(j,k,\eps)}_{\Delta_B})^2] & = q_{1,j}^{-2}q_{2,k}^{-2} \tr[\phi^{(j,k,\eps)}] 2^{-H_{\min}(B)_{\phi^{(j,k,\eps)}} - H^\eps_{\min}(\Delta_A \setminus \Gamma_A \vert \Gamma_A)_{\phi}} \\
                                          & \leq q_{1,j}^{-1}q_{2,k}^{-1}2^{-H_{\min}(B)_{\phi^{(j,k,\eps)}}} \leq d_{1,j}^{-1}
  \end{align*}
  using that $\tr[\phi^{(j,k,\eps)}] \leq q_{1,j}q_{2,k}$ and $2^{-H_{\min}(B)_{\phi^{(j,k,\eps)}}} \leq q_{2,k}p_{1,j} = q_{1,j}q_{2,k} d_{1,j}^{-1}$.
  Therefore
  \begin{align*}
    \EE_N \tr\mleft[ (\sigma^{(j,k,\eps)}_B)^2 \mright] & \leq p_N^{-1} \EE \tr\mleft[ (\sigma^{(j,k,\eps)}_B)^2 \mright] \\
                                                        & = d_{1,j}^{-1} + d_{2,k}^{-1} + \bigO(d_{1,j}^{-1}).
  \end{align*}
  so
  \begin{align*}
    -\log \EE_N \tr\mleft[ (\sigma^{(j,k,\eps)}_B)^2 \mright] & \geq -\log\left( \max\mleft(d_{1,j}^{-1}, d_{2,k}^{-1}\mright)(2 + \bigO(1)) \right) \\
                                                              & \geq \min( \log d_{1,j}, \log d_{2,k} ) - \bigO(1).
  \end{align*}
  We find that
  \begin{align*}
    \EE_N \abs{H(\sigma_{B}^{(j,k)}) -  \min( \log d_{1,j}, \log d_{2,k} )} = \bigO(1)
  \end{align*}
  and hence
  \begin{align}\label{eq:entropy sigma vs min}
    \EE_N \abs{\sum_{j,k} q_{1,j} q_{2,k} H(\sigma_{B}^{(j,k)}) - \sum_{j,k} q_{1,j} q_{2,k} \min( \log d_{1,j}, \log d_{2,k} )} = \bigO(1).
  \end{align}

  We collect the various estimates we have found:
  \begin{align*}
    \EE_N \norm{\spec_+(\rho_A) - \spec_+(\sigma_{B})}_1 \leq p_N^{-1}\EE \norm{\spec_+(\rho_A) - \spec_+(\sigma_{B})}_1         & = o(\frac{1}{n}) \\
    \EE_N \norm{\sigma_B - \sigma^{\eps}_B}_1 \leq p_N^{-1}\EE \norm{\sigma_B - \sigma^{\eps}_B}_1                               & = o(\frac{1}{n}) \\
    \EE_N \abs{\sum_{j,k} q_{1,j} q_{2,k} H(\sigma_{B}^{(j,k)}) - \sum_{j,k} q_{1,j} q_{2,k} \min( \log d_{1,j}, \log d_{2,k} )} & = \bigO(1).
  \end{align*}
  Let $M$ be the event that $N$ holds and moreover
  \begin{align*}
    \norm{\spec_+(\rho_A) - \spec_+(\sigma_{B})}_1                                                                         & \leq \frac{1}{n} \\
    \norm{\sigma_B - \sigma^{\eps}_B}_1                                                                                    & \leq \frac{1}{n} \\
    \abs{\sum_{j,k} q_{1,j} q_{2,k} H(\sigma_{B}^{(j,k)}) - \sum_{j,k} q_{1,j} q_{2,k} \min( \log d_{1,j}, \log d_{2,k} )} & \leq \log(n).
  \end{align*}
  By Markov's inequality and the union bound, the probability that $M$ holds goes to one as $n$ goes to infinity.
  Moreover, if $M$ holds, it is easy to verify that by the Fannes-Audenaert inequality \cref{eq:fannes} and the fact that $\rho$, $\sigma$ and $\sigma^{\eps}$ are close to normalized
  \begin{align*}
    \abs{H(\frac{\rho_A}{\tr[\rho]}) - H(\sigma^{\eps}_B)} & = \bigO(1).
  \end{align*}
  Also, if $M$ holds, by \cref{eq:entropy sigma vs min}
  \begin{align*}
    \abs{H(\sigma^{\eps}_B) - \sum_{j,k} q_{1,j} q_{2,k} \min( \log d_{1,j}, \log d_{2,k} )} & = \abs{\sum_{j,k} q_{1,j} q_{2,k} H(\sigma_{B}^{(j,k)}) - \sum_{j,k} q_{1,j} q_{2,k} \min( \log d_{1,j}, \log d_{2,k} )} \\
                                                                                             & \qquad +\bigO(\log(n))                                                                                                   \\
                                                                                             & = \bigO(\log(n))
  \end{align*}
  so we conclude that
  \begin{align*}
    \abs{H(\frac{\rho_A}{\tr[\rho]}) - \sum_{j,k} q_{1,j} q_{2,k} \min( \log d_{1,j}, \log d_{2,k} )} & = \bigO(\log(n)).
  \end{align*}
  Since $q_{i,j} = d_{i,j}p_{i,j} \geq n^{-\beta}$ from the binning procedure, we have the simple observation
  \begin{align*}
    \frac{1}{n^{\beta} p_{i,j}} \leq \frac{q_{i,j}}{p_{i,j}} = d_{i,j} \leq \frac{1}{p_{i,j}},
  \end{align*}
  and hence $\abs{\min(\log d_{1,j},\log d_{2,k}) - \log(\min(\frac{1}{p_{1,k}},\frac{1}{p_{2,j}}))} = \bigO(\log(n))$. We conclude that
  \begin{align*}
    \abs{H(\frac{\rho_A}{\tr[\rho]}) - \sum_{j,k} q_{1,k}q_{2,j}\log\left(\min\mleft(\frac{1}{p_{1,k}},\frac{1}{p_{2,j}}\mright)\right)} = \bigO(\log(n)).
  \end{align*}
  Finally, we need to relate the result back to the original background state.
  In the above approximation to $H(\tilde{\sigma}_B)$, we see that
  \begin{align*}
    \sum_{j,k} q_{1,k}q_{2,j}\log\left(\min\mleft(\frac{1}{p_{1,k}},\frac{1}{p_{2,j}}\mright)\right) = H^*(\spec(\tilde\phi_{\gamma_{A,1}}),\spec(\tilde\phi_{\gamma_{A,2}})).
  \end{align*}
  Then by \cref{eq:phi tilde phi entropy} and \cref{eq:fannes h star}, this will converge to the appropriate quantity on the non-approximated background state:
  \begin{align*}
    \abs{H^*(\spec(\tilde\phi_{\gamma_{A,1}}),\spec(\tilde\phi_{\gamma_{A,2}})) - H^*(\spec(\phi_{\gamma_{A,1}}),\spec(\phi_{\gamma_{A,2}}))} \rightarrow 0,
  \end{align*}
  proving the desired result.

\end{proof}

\addcontentsline{toc}{section}{Acknowledgments}
\section*{Acknowledgments}
NC is supported in part by the Department of Energy via the GeoFlow consortium (QuantISED Award DE-SC0019380).
CL acknowledges support from the projects ESQuisses (ANR-20-CE47-0014-01), STARS (ANR-20-CE40-0008), Qtraj (ANR-20-CE40-0024-01) and Random Tensors (ANR-11-LABX-0040) of the French National Research Agency (ANR).
GP is supported by the UC Berkeley Physics Department, the Simons Foundation through the ``It from Qubit'' program, the Department of Energy via the GeoFlow consortium (QuantISED Award DE-SC0019380), and AFOSR award FA9550-22-1-0098.
He also acknowledges support from an IBM Einstein Fellowship at the Institute for Advanced Study.
MW acknowledges support by the NWO through grant OCENW.KLEIN.267, by the Deutsche Forschungsgemeinschaft (DFG, German Research Foundation) under Germany's Excellence Strategy - EXC\ 2092\ CASA - 390781972, by the BMBF through project Quantum Methods and Benchmarks for Resource Allocation (QuBRA), and by the European Research Council~(ERC).
Funded by the European Union.
Views and opinions expressed are however those of the authors only and do not necessarily reflect those of the European Union or the European Research Council.
Neither the European Union nor the granting authority can be held responsible for them.

\begin{appendix}
  \section{Euclidean gravity path integrals and entanglement spectra}\label{sec:gravity context}
  In this appendix, we give a heuristic description of certain Euclidean gravity path integrals in holography whose descriptions are in close analogy with random tensor network models.
  This section serves as a motivation for the random tensor network models we study, but is not needed to understand the random tensor network results.
  In \cref{sec:euclidean}, we review the replica trick in quantum field theory and the role of Euclidean path integrals. We review an application of such tools in \cref{sec:surface phase transition}, we discuss the problem of studying the entanglement entropy near a phase transition between two minimal surfaces, which agrees with results in \cref{sec:far from max entangled}. Then, in \cref{sec:jt gravity}, we discuss a simplified model of quantum gravity which is in very close correspondence to a random tensor network model with link states with bounded spectral variation, as in \cref{sec:almost max entangled}.

  \subsection{The replica trick and Euclidean path integrals}\label{sec:euclidean}
  In order to compute entropies in quantum field theory, one often uses a version of the \emph{replica trick} to compute the R\'enyi entropies, which can then be analytically continued to deduce von Neumann entropies.
  We consider a pure quantum field theory state $\ket{\rho}$ on a space $M$, which is prepared by a Euclidean path integral on $M \times (-\infty,0]$.
    Correspondingly, $\bra{\rho}$ is prepared by the time-reflected path integral on $M \times [0,\infty)$.
  Let $A$ be a subregion of $M$. The reduced density matrix on $A$ is given by taking $\ket{\rho}$ and $\bra{\rho}$, then integrating over the field configurations on the complement of $A$ (the equivalent of the partial trace for field theories), as shown in \cref{fig:path-integral2}.
  Analogous to \cref{eq:replica trick}, we may now compute $\tr[\rho_A^k]$ by taking $k$ copies of this path integral, and gluing the boundaries at the $A$ system cyclically, then integrating over the field configurations at each boundary. This operation is manifestly invariant under cyclic permutations, a symmetry called \emph{replica symmetry}.
  This is illustrated in \cref{fig:branched-cover}.
  We conclude that $\tr[\rho_A^k]$ is computed by a path integral $Z_{A,k}$ on a manifold $M_{A,k}$, allowing us to compute $H_k(\rho_A)$. The space $M_{A,k}$ is a $k$-fold cover of $M \times \RR$, branching at the boundary $\partial A$ of the subregion $A$.

  Of course, the path integral $Z_{A,k}$ is formally infinite. One `normalizes' the path integral by normalizing by $Z_1$ which can be thought of as $\tr[\rho]$. We have the expression for the $k$-th R\'enyi entropy:
  \begin{align}\label{eq:path integral renyi}
    H_k(\rho_A) = \log \frac{Z_{A,k}}{Z_1^k}.
  \end{align}
  Note that this is a slightly different normalization convention (the denominator is $Z_1^k$ instead of $Z_1$) to match the standard convention in the quantum gravity literature. To get a finite result, one has to impose a UV cut-off of size $\eps$, an aspect we will ignore in our discussion. By employing such a cutoff, we will pretend that the relevant Hilbert spaces are finite-dimensional Hilbert spaces that factorize with respect to spatial decompositions of $M$.

  \begin{figure}[t]
    \centering
    \begin{subfigure}[t]{.27\textwidth}
      \centering
      \begin{overpic}[width=.8\linewidth,grid=false]{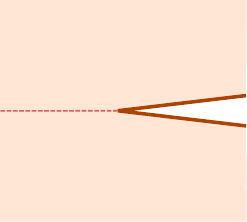}
        \put(-20,25){\footnotesize{$\uparrow \tau$}}
        \put(20,50){\color{BrickRed}{$\bar{A}$}}
        \put(70,55){\color{Bittersweet}{$A$}}
      \end{overpic}
      \caption{Path integral representation of the reduced density matrix on a subsystem $A$.}
      \label{fig:path-integral2}
    \end{subfigure}%
    \hspace*{0.3cm}
    \begin{subfigure}[t]{.3\textwidth}
      \centering
      \begin{overpic}[width=.75\linewidth,grid=false]{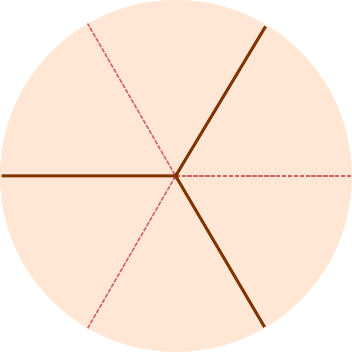}
        \put(20,55){\footnotesize{\color{Bittersweet}{$A$}}}
        \put(70,72){\footnotesize{\color{Bittersweet}{$A$}}}
        \put(70,25){\footnotesize{\color{Bittersweet}{$A$}}}
      \end{overpic}
      \caption{For the replica trick the path integral is glued cyclically along $A$ to obtain $M_{A,k}$. The space $M_{A,k}$ is a $k$-fold cover of $M$, branching at the boundary $\partial A$. In this case $k =3$.}
      \label{fig:branched-cover}
    \end{subfigure}%
    \hspace*{0.3cm}
    \begin{subfigure}[t]{.3\textwidth}
      \centering
      \begin{overpic}[width=.75\linewidth,grid=false]{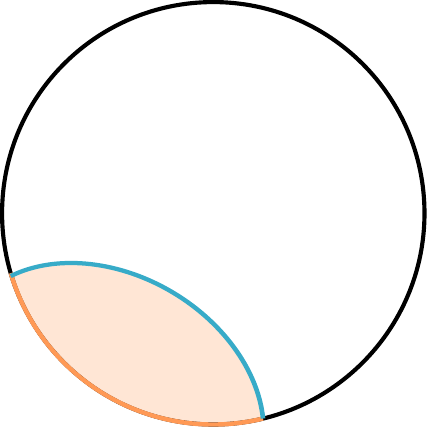}
        \put(10,5){\footnotesize{\color{Bittersweet}{$A$}}}
        \put(30,15){\footnotesize{$\tau$}}
        \put(60,60){\footnotesize{$\id$}}
        \put(40,38){\footnotesize{\color{NavyBlue}{$\gamma_{A,k}$}}}
      \end{overpic}
      \caption{The holographic version of the replica trick. The boundary manifold is glued as $M_{A,k}$, and the bulk manifold is glued along a cyclic permutation $\tau$ adjacent to $A$, and along the identity permutation along a region adjacent to the complement $\bar{A}$. These two regions are separated by a brane $\gamma_{A,k}$.}
      \label{fig:holographic rep trick}
    \end{subfigure}%
    \caption{The path integral replica trick to compute $\tr[\rho_A^k]$.}
    \label{fig:path integral}
  \end{figure}
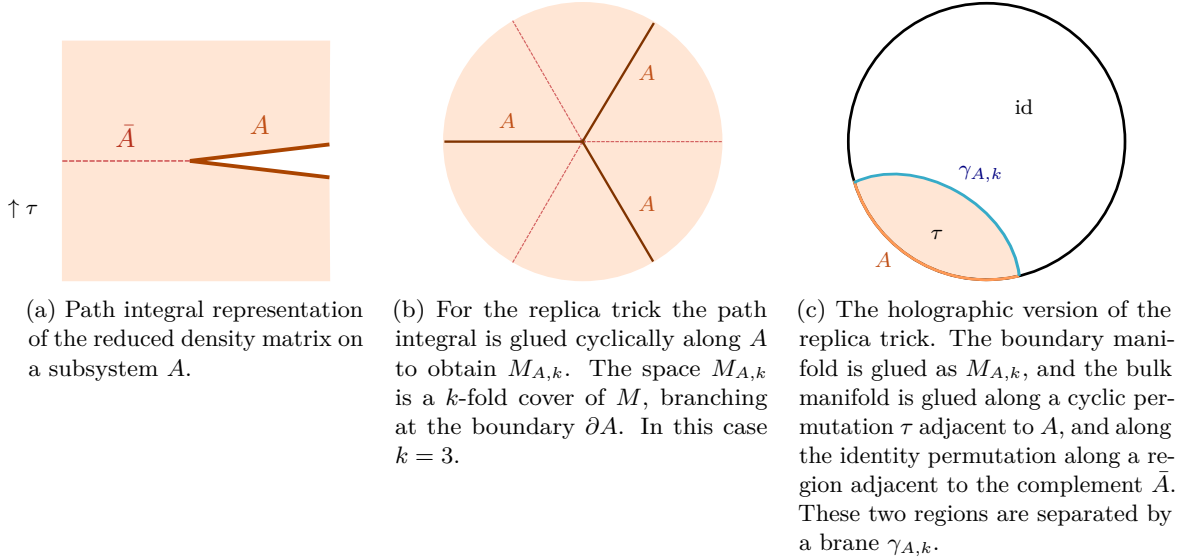

  What happens if the quantum field theory is a \emph{holographic CFT}?
  In this case, we have a correspondence between the path integral of the CFT on the one hand, and a bulk quantum gravity path integral on the other hand, which, for large effective central charge, we may approximate by its semiclassical saddle point geometry.
  This can be used to derive the \emph{Ryu-Takayanagi formula} \cite{lewkowycz2013generalized}, by finding a bulk manifold $B_{A,k}$ which has $M_{A,k}$ as its boundary. In other words, one obtains the gravitational dual by setting $M_{A,k}$ as a boundary condition, then allowing the system to evolve according to the gravitational theory.

  The saddle point solution will be such that the bulk copies are glued along a cyclic permutation in an area adjacent to $A$ and along the identity permutation in an area adjacent to $\bar{A}$, the boundary complement of $A$.
  The boundary between these two regions is a surface $\gamma_{A,k}$ anchored at the boundary at $\partial A$.
  After orbifolding with respect to the replica symmetry, this leads to a bulk spacetime which has $M$ as its boundary and a `cosmic brane' at the surface $\gamma_{A,k}$. The orbifold procedure gives rise to a conical deficit, corresponding to the fixed points under the action of the replica symmetry.
  Comparing the action of $B_{A,k}$ and $B_1$, we find that in the computation of \cref{eq:path integral renyi}, only the conical deficit at $\gamma_{A,k}$ does not cancel, yielding
  \begin{align}\label{eq:holographic renyi entropy}
    H_k(\rho_A) \approx \frac{\abs{\gamma_{A,k}}}{4G_N}.
  \end{align}
  where $\abs{\gamma_{A,k}}$ is the area of the brane in the saddle-point solution.
  This can be analytically continued to non-integer values, and in particular, continuation to $k = 1$ yields the RT formula:
  \begin{align}\label{eq:rt formula}
    H(\rho_A) \approx \frac{\abs{\gamma_{A}}}{4G_N}.
  \end{align}
  where $\gamma_A$ is the surface of minimal area out of all bulk surface homologous to $A$ and anchored at $\partial A$.

  The spectrum of $\rho_A$ can be recovered from the R\'enyi entropies.
  In \cref{eq:holographic renyi entropy}, we see that $\abs{\gamma_{A,k}}$ depends on $k$, but is otherwise fixed as we let $G_N$ go to zero in the classical limit.
  It was argued in \cite{bao2019beyond} that this behavior implies a spectrum that is flat at leading order:
  \begin{align*}
    H_{\min}^{\eps}(\rho_A) & = H(\rho_A) - \bigO(G_N^{-\frac12}) = H(\rho_A) - \bigO(\sqrt{H(\rho_A)}) \\
    H_{\max}^{\eps}(\rho_A) & = H(\rho_A) + \bigO(G_N^{-\frac12}) = H(\rho_A) + \bigO(\sqrt{H(\rho_A)})
  \end{align*}
  This is precisely the link state regime we investigated in \cref{sec:two minimal cuts}, and shows that the large $c$ limit in a holographic CFT is similar to considering the many-copy limit in quantum information theory.

  \subsubsection{Fixed area states}
  A useful variation on the derivation of the RT formula is to consider \emph{fixed-area states} \cite{dong2019flat}.
  We consider the `area operator' $\hat{\gamma}_A$ for a subsystem $A$, which measures the area of a minimal surface.
  The operator $\hat{\gamma}_A$ actually has fluctuations when we consider a state with a semiclassical gravity dual, and we may write
  \begin{align*}
    \ket{\rho} = \int \d \alpha \ket{\psi_{\alpha}}
  \end{align*}
  where $\psi_{\alpha}$ is an eigenvector of $\hat A$ with eigenvalue $\alpha$.
  The state $\ket{\psi_\alpha}$ is a fixed-area state: while it is not itself a physical state, they form a basis with which to construct physical states. These fixed-area states can be thought of as prepared by a bulk path integral where we have restricted to bulk geometries for which $\abs{\gamma_A} = \alpha$.
  The same derivation as above now leads to
  \begin{align*}
    H_k(\rho_A) \approx \frac{\alpha}{4G_N},
  \end{align*}
  since the area of the minimal surface is fixed to be $\alpha$.
  Thus, for fixed-area states, all R\'enyi entropies are (to good approximation) equal, which implies that the state has flat entanglement spectrum.
  This corresponds to a random tensor network state with maximally entangled link states.

  \subsection{States at a minimal surface phase transition}\label{sec:surface phase transition}
  Consider, again, a holographic CFT state $\rho_A$ on a boundary subregion $A$. In the case where there is a unique minimal surface, the previous subsection showed how the replica trick leads to the RT formula.

  We would now like to investigate the entanglement spectrum of $\rho_A$ when there are two RT surfaces for $A$ that have area of the same order of magnitude. Let us denote the two competing minimal surfaces in the bulk by $\gamma_{A,1}$ and $\gamma_{A,2}$. In \cite{marolf2020probing}, it was shown how the entanglement entropy should behave at this phase transition between the two minimal surfaces. A similar computation was performed in \cite{akers2020leading} for the setting with two competing minimal surfaces and bulk matter. We will briefly sketch their argument, referring the interested reader to \cite{akers2020leading} for more details.

  \subsubsection{Fixed area states with two minimal surfaces}
  We begin by considering bipartite fixed-area states $\rho_{A\bar{A}}$ prepared by a Euclidean path integral, in which we have fixed the size of the two competing surfaces $\gamma_{A,1}$ and $\gamma_{A,2}$. In this case, saddle points of the path integral have to satisfy the equations of motion everywhere except at the surfaces $\gamma_{A,i}$, where there could be conical singularities.
  The two surfaces divide the bulk into 3 regions: $a_1$, $a_2$ and $a_3$.

  The saddle points are states with smooth geometries in the regions where the copies of regions $a_i$ are glued to each other -- the $k$ copies of region $a_1$ are glued cyclically, while the copies of region $a_3$ are glued along the identity permutation. On the middle region $a_2$, we are free to glue along an arbitrary permutation $\pi$, giving solutions that can break the replica symmetry. However, it is often the case that the dominant solutions of the path integral are those that respect the replica symmetry, and we assume this to be the case for all solutions we consider.

  Let us denote a saddle-point solution of the path integral with permutation $\pi$ by $B_{A,\pi}$, and let us write $\phi_i$ for the conical singularity angle at $\gamma_{A,i}$.
  It turns out that these saddle points lead to an action of the form
  \begin{align*}
    I(B_{A,\pi}) = k I_{\text{away}}(B_{A,\pi}) + (k\phi_1 - 2\pi\abs{C(\pi)})\frac{\abs{\gamma_{A,1}}}{8\pi G_N} + (k\phi_2 - 2\pi\abs{C(\tau^{-1}\pi)})\frac{\abs{\gamma_{A,2}}}{8\pi G_N},
  \end{align*}
  where $I_{\text{away}}$ is the action away from the surfaces, $\abs{\gamma_{A,i}}$ is the area of the surface $\gamma_{A,i}$, and we recall that $\abs{C(\pi)}$ is the number of cycles of $\pi$, and $\tau$ is the full cycle $(1 2 \ldots k)$.
  In particular, for $k=1$, we have
  \begin{align*}
    I(B_{A,\pi}) = I_{\text{away}}(B_{A,\pi}) + (\phi_2 - 2\pi)\frac{\abs{\gamma_{A,1}}}{8\pi G_N} + (\phi_2 - 2\pi)\frac{\abs{\gamma_{A,2}}}{8\pi G_N},
  \end{align*}
  so when we look at the normalized path integral, and sum over all permutations
  \begin{align*}
    \frac{Z_{A,k}}{(Z_{A,1})^k} \approx \sum_{\pi \in S_k} e^{(\abs{C(\pi)} - k)\frac{\abs{\gamma_{A,1}}}{4G_N} + (\abs{C(\tau^{-1}\pi)} - k)\frac{\abs{\gamma_{A,2}}}{4G_N}} = \sum_{\pi \in S_k} e^{d(\id,\pi)\frac{\abs{\gamma_{A,1}}}{4G_N} + d(\pi,\tau)\frac{\abs{\gamma_{A,2}}}{4G_N}}.
  \end{align*}
  Note that the areas $\gamma_{A,i}$ are of the same order of magnitude, and are divergent. As a result, only the permutations for which $d(\id,\pi) + d(\pi,\tau)$ is minimal will contribute.
  In other words, only the configurations where $\pi \in NC(k)$ contribute, as all other permutations are suppressed by at least a factor of the area of $\gamma_{A,i}$ in the action.
  We conclude that
  \begin{align}\label{eq:fixed area two surfaces}
    \frac{Z_{A,k}}{(Z_{A,1})^k} \approx \sum_{\pi \in NC(k)} e^{d(\id,\pi)\frac{\abs{\gamma_{A,1}}}{4G_N} + d(\pi,\tau)\frac{\abs{\gamma_{A,2}}}{4G_N}}.
  \end{align}

  This computation is in one-to-one correspondence with the computation of the moments for a subsystem of a single random tensor, as observed in \cite{penington2019replica}.
  It also corresponds more generally to a random tensor network computation with two minimal cuts and maximally entangled link states, as is clear from the computations in \cref{sec:rtn} and \cref{sec:almost max entangled}.
  This is in agreement with the claim that random tensor network states are a model for fixed area states.
  One can also add bulk matter in this path integral computation, which will again be in correspondence to a similar computation in a random tensor network \cite{akers2020leading} with a background state, as in \cref{sec:background states}.
  From the moment computation in \cref{eq:fixed area two surfaces} and applying the results for the entanglement of a single random tensor, we observe that for two fixed surfaces of exactly equal size, the (appropriately scaled) entanglement spectrum is a Marchenko-Pastur distribution, giving an $\mathcal O(1)$ correction to the entanglement entropy, agreeing with the gravitational replica trick computation in \cite{marolf2020probing}.

  \subsubsection{General states at the minimal surface phase transition}
  We now relax the fixed-area restriction, and study similar calculations performed in \cite{dong2019flat}, \cite{marolf2020probing}, and \cite{akers2020leading}. Denote by $Z_{A,k}(\alpha_1, \alpha_2)$ the path integral where we have fixed the areas of $\gamma_{A,i}$ to be $\alpha_i$.
  Then, following Section 2.3 in \cite{dong2019flat}, the full path integral is given by
  \begin{align*}
    Z_{A,k} = \int \d \alpha_1 \d \alpha_2\ Z_k(\alpha_1, \alpha_2).
  \end{align*}
  Again, we consider the semiclassical limit, so we take our saddle-point approximation of $Z_k(\alpha_1, \alpha_2)$ in \cref{eq:fixed area two surfaces}, and we also take a saddle-point approximation for the integral over $\alpha_1$ and $\alpha_2$. This saddle point will be at the values for $\alpha_i$ where the deficit angles are $\phi_i = \frac{2\pi}{n}$ (since then the saddle point geometry is smooth), which leads to
  \begin{align}
    \frac{Z_{A,k}}{(Z_{A,1})^k} \approx \sum_{\pi \in NC(k)} e^{(\abs{C(\pi)} - k)\frac{\abs{\gamma^{(k)}_{A,1}}}{4G_N} + (\abs{C(\tau^{-1}\pi)} - k)\frac{\abs{\gamma^{(k)}_{A,2}}}{4G_N}}
  \end{align}
  where $\gamma^{(k)}_{A,i}$ are now minimal surfaces, with a dependence on $k$. Analytic continuation to $k = 1$ yields the usual surface prescription. In particular, if there are two surfaces that are of almost equal area, the contribution of the larger term is exponentially suppressed for any $\mathcal{O}(1)$ or larger difference in areas.

  To zoom in on the region where the two surfaces are nearly equal, we write the state as a superposition of fixed area states, as in Section 3 of \cite{marolf2020probing}.
  We discretize the area size $\alpha_1$ and $\alpha_2$ over $\poly(\frac{1}{G_N})$ values and approximate the full holographic pure state with boundary regions $A$ and $B = A^c$ as a (finite) sum
  \begin{align*}
    \ket{\psi}_{AB} = \sum_{\alpha_1,\alpha_2} \sqrt{p(\alpha_1,\alpha_2)} \ket{\psi_{\alpha_1,\alpha_2}}
  \end{align*}
  where $\ket{\psi_{\alpha_1,\alpha_2}}$ is the state where the areas are fixed as $\abs{\gamma_{A,i}} = \alpha_i$, and $p$ is a probability distribution over the possible areas.
  Then a straightforward calculation of the reduced density matrix $\rho_A$ yields a state of the form:
  \begin{align*}
    \rho_A & = \sum_{\alpha_1,\alpha_2}p(\alpha_1,\alpha_2)\rho_{A,\alpha_1,\alpha_2} + \sum_{\alpha_1\neq \alpha_1', \alpha_2\neq \alpha_2'}\sqrt{p(\alpha_1,\alpha_2)p(\alpha_1',\alpha_2')}\tr_B\left(\ket{\psi_{\alpha_1,\alpha_2}}\bra{\psi_{\alpha_1',\alpha_2'}}\right), \\
           & = \sum_{\alpha_1,\alpha_2}p(\alpha_1,\alpha_2)\rho_{A,\alpha_1,\alpha_2} + OD_A,
  \end{align*}
  where $OD_A$ are the off-diagonal elements of $\rho_A$. One can argue that the states $\rho_{A,\alpha_1,\alpha_2}$ are all mutually orthogonal by entanglement wedge reconstruction -- the area operator can be reconstructed on $A$, and hence, each $\rho_{A,\alpha_1,\alpha_2}$ is perfectly distinguishable from each other by measuring the area operator. Then the entropy of the diagonal part of the state is easily computed as
  \begin{equation}
    H\left(\sum_{\alpha_1,\alpha_2}p(\alpha_1,\alpha_2)\rho_{A,\alpha_1,\alpha_2}\right) = \sum_{\alpha_1,\alpha_2}p(\alpha_1,\alpha_2)H(\rho_{A,\alpha_1,\alpha_2}) - \sum_{\alpha_1,\alpha_2}p(\alpha_1,\alpha_2)\log p(\alpha_1,\alpha_2).
  \end{equation}
  The second term is the so-called entropy of mixing, and it is a standard argument that this term is suppressed relative to the first term \cite{marolf2020probing} as $\mathcal{O}(\ln G_N)$ or smaller. The entropies appearing in the first term can be computed using the methods in the previous subsection, for which one finds that $H(\rho_{A,\alpha_1,\alpha_2}) = \frac{\min\{\alpha_1, \alpha_2\}}{4G_N}$.

  Returning to the off-diagonal terms $OD_A$, \cite{marolf2020probing} argued that such terms should be subleading in the analytic continuation due to the relevant surfaces breaking replica symmetry. At the same time, \cite{akers2020leading} argued that such terms should be subleading due to reasons similar to those for the orthogonality of the diagonal elements: complementary entanglement wedge reconstruction implies one may reconstruct the bulk area operator on $B$, and hence, such states are perfectly distinguishable on $B$. Therefore, the partial trace over $B$ vanishes for $\alpha_1\neq \alpha_1'$, $\alpha_2\neq \alpha_2'$.

  One reaches the conclusion:
  \begin{align}
    H(\rho_A) = \sum_{\alpha_1,\alpha_2} p(\alpha_1, \alpha_2)\frac{\min\{\alpha_1, \alpha_2\}}{4G_N}  + \mathcal{O}(\ln G_N),
  \end{align}
  In this computation the $\bigO(1)$ corrections due to the Marchenko-Pastur distribution along each pair of minimal cuts of equal size (or equivalently, the degeneracy in the contributions to the saddle point approximation) is irrelevant, as the entropy of mixing already leads to $\bigO(\ln G_N)$ deviations.

  Our results in \cref{sec:far from max entangled} can be seen as a rigorous version of the above result for random tensor networks.

  \subsection{Replica wormholes and JT gravity}\label{sec:jt gravity}

  One of the most basic holographic models of quantum gravity is \emph{JT gravity}, a $1+1$-dimensional model of gravity; see \cite{sarosi2017ads} for a review. JT gravity is also a useful model for the near-horizon dynamics of extremal black holes in any dimension.
  In this case, the dual theory should be $0+1$-dimensional. In other words, it should be regular quantum mechanics rather than a quantum field theory.
  Indeed, in \cite{saad2019jt}, it was shown that JT gravity theory is dual to a random matrix model, where the Hamiltonian is a random self-adjoint matrix according to some distribution, providing another strong connection between quantum gravity and random matrix theory. It also appears that such gravitational systems may be dual to an ensemble of boundary theories \cite{bousso2020gravity}, rather than a single one.
  Whether this is fundamental, a special feature of 1+1-dimensional models, or due to averaging over microscopic features of the gravity theory, is a line of active research~\cite{saad2021wormholes}.

  We now sketch a variation on a calculation in \cite{penington2019replica}, providing a proof-of-principle that the free probability techniques used in \cref{sec:almost max entangled} provide an elegant framework to understand such results. We refer the interested reader to \cite{penington2019replica} for more in-depth motivation and detailed computations.

  We consider JT gravity with an end of the world (EOW) brane containing a large number $n$ of internal states. This model has action
  \begin{align*}
    I = I_{\JT} + \mu\int_{\text{brane}} \d s,
  \end{align*}
  where the action of a manifold $M$ with metric $g$, induced boundary metric $h$, (trace of) extrinsic curvature $K$, and dilaton $\phi$ is given by
  \begin{align*}
    I_{\JT}[M, g] = -\frac{S_0}{2\pi}\left[\frac12\int_{M} \sqrt{g} R + \int_{\partial M} \sqrt{h} K \right] - \left[\frac12\int_{M} \sqrt{g}\phi(R + 2) + \int_{\partial M} \sqrt{h}\phi K \right].
  \end{align*}
  The details of this action are not very important for us; we just note that we will take the $S_0$ parameter to be large, and that this suppresses contributions where the manifold $M$ has genus $\gamma > 0$ in the Euclidean path integral.

  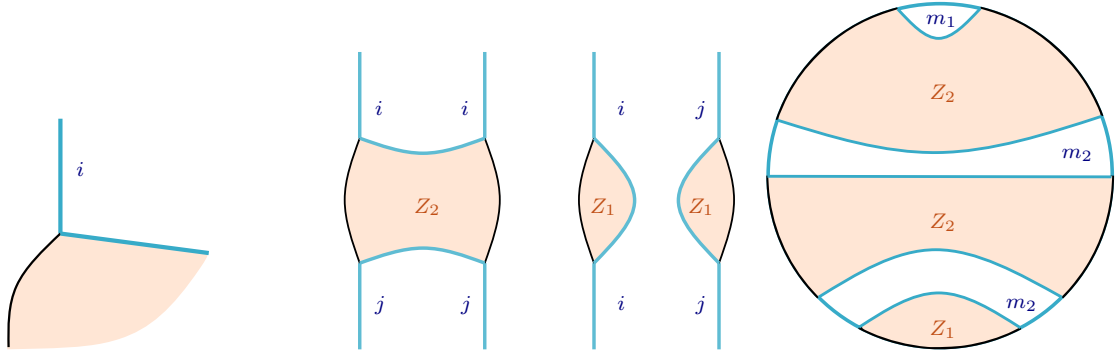
\begin{figure}[t]
    \centering
    \begin{subfigure}[t]{.25\textwidth}
      \centering
      \begin{overpic}[width=.65\linewidth,grid=false]{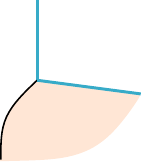}
        \put(30,75){\footnotesize{\color{NavyBlue}{$i$}}}
      \end{overpic}
      \caption{Path integral which prepares the state $\ket{\psi_i}$.}
      \label{fig:path integral jt}
    \end{subfigure}%
    \hspace*{0.3cm}
    \begin{subfigure}[t]{.7\textwidth}
      \centering
      \begin{overpic}[width=.45\linewidth,grid=false]{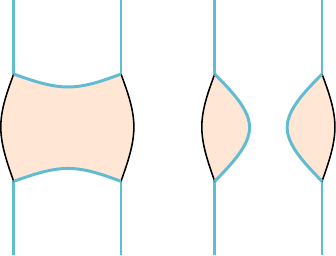}
        \put(8,60){\footnotesize{\color{NavyBlue}{$i$}}}
        \put(30,60){\footnotesize{\color{NavyBlue}{$i$}}}
        \put(8,10){\footnotesize{\color{NavyBlue}{$j$}}}
        \put(30,10){\footnotesize{\color{NavyBlue}{$j$}}}
        \put(70,60){\footnotesize{\color{NavyBlue}{$i$}}}
        \put(90,60){\footnotesize{\color{NavyBlue}{$j$}}}
        \put(70,10){\footnotesize{\color{NavyBlue}{$i$}}}
        \put(90,10){\footnotesize{\color{NavyBlue}{$j$}}}
        \put(18,35){\footnotesize{\color{Bittersweet}{$Z_2$}}}
        \put(63,35){\footnotesize{\color{Bittersweet}{$Z_1$}}}
        \put(88,35){\footnotesize{\color{Bittersweet}{$Z_1$}}}
      \end{overpic}%
      \hspace*{0.3cm}
      \begin{overpic}[width=.4\linewidth,grid=false]{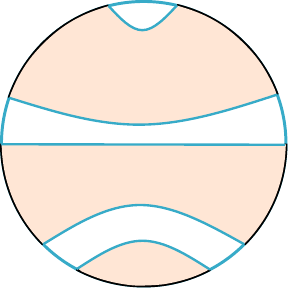}
        \put(47,72){\footnotesize{\color{Bittersweet}{$Z_2$}}}
        \put(47,35){\footnotesize{\color{Bittersweet}{$Z_2$}}}
        \put(47,4){\footnotesize{\color{Bittersweet}{$Z_1$}}}
        \put(85,55){\footnotesize{\color{NavyBlue}{$m_2$}}}
        \put(69,12){\footnotesize{\color{NavyBlue}{$m_2$}}}
        \put(46,93.5){\footnotesize{\color{NavyBlue}{$m_1$}}}
      \end{overpic}
      \caption{Replica trick for JT gravity. On the left side two diagrams that show up when computing matrix elements of $\psi_{\rad}^{\ot 2}$, on the right a diagram contributing a factor $Z_2^2 Z_1 m_2^2 m_1$ to $\tr[\psi_{\rad}^5]$.}
      \label{fig:jt gravity replica}
    \end{subfigure}
    \caption{Path integrals and the replica trick for JT gravity. See \cite{penington2019replica} for a detailed explanation of the diagrammatic notation.}
    \label{fig:jt gravity}
  \end{figure}

  Such systems are of interest when studying a simple version of an evaporating black hole.
  Let
  \begin{align*}
    \ket{\psi} = \frac{1}{\sqrt{n}} \sum_{i = 1}^n \ket{\psi_{B,i}}\ket{i_{\rad}}
  \end{align*}
  where $\ket{\psi_{B,i}}$ is the state of the black hole with the EOW brane in state $i$, and $\ket{i_{\rad}}$ is a reference state, which can be thought of as the radiation system.
  Notice that the entanglement spectrum of this state is flat. We generalize this to
  \begin{align*}
    \ket{\psi} = \sum_{i = 1}^n \sqrt{p_i} \ket{\psi_{B,i}}\ket{i_{\rad}},
  \end{align*}
  where the entanglement between the black hole and the radiation has some nontrivial spectrum, which we will assume to be close to uniform, so that $\tfrac{p_i}{n}$ is bounded by a constant for all $i$. In other words, we assume this distribution satisfies the bounded spectral variation assumption from \cref{sec:almost max entangled}.
  We let
  \begin{align*}
    m_k = \sum_{i=1}^n n^{k-1} p_i^k
  \end{align*}
  be the (appropriately scaled) moments of the entanglement spectrum of the EOW brane.
  Moreover, we write the path integral on a disc geometry with $k$ boundary components and $k$ EOW branes as $e^{-S_0}Z_k$.
  Then following the arguments of \cite{penington2019replica}, one can compute the $k$-th moment of the radiation system for large $n$ and large $e^{S_0}$ (large $n$ enforces a planar limit with only non-crossing partitions, while large $e^{S_0}$ ensures that only genus $\gamma=0$ geometries contribute), as illustrated in \cref{fig:jt gravity replica}.
  The contributions of path integral configurations connecting different replicas are called \emph{replica wormholes}.
  This diagrammatic computation shows that
  \begin{align}\label{eq:moments jt gravity}
    \tr[\psi_{\rad}^k] = \sum_{\pi \in NC(k)} m_{\pi} \frac{Z_{\pi^{-1}\tau}}{Z_1^k} n^{-d(\pi,\id)}e^{-S_0 d(\pi,\tau)}.
  \end{align}
  In this expression, recall that $Z_{\sigma} = \prod_{l \in C(\sigma)} Z_{l}$, where $C(\sigma)$ is the cycle type of $\sigma$, and $l \in C(\sigma)$ are the lengths of the cycles of $\sigma$.
  This expression implies that if $n \gg e^{S_0}$, the dominant contribution in \cref{eq:moments jt gravity} is given by $\pi = \tau$.
  On the other hand, if $n \ll e^{S_0}$, the dominant contribution in \cref{eq:moments jt gravity} is given by $\pi = \id$.
  This corresponds to the situation where there is a unique minimal surface (more precisely, a unique \textit{quantum extremal surface}).
  We are interested in the regime at the phase transition, which is analogous to the Page time of an evaporating black hole, so we assume $ne^{-S_0} \to 1$.
  The coefficients $m_{\sigma}$ correspond to the weight of the $\sigma$ configuration, as determined by the number and length of the cycles in $\sigma$, and the probability distribution of eigenstates $p_i$. In the case of the flat entanglement spectrum, this number equals the number of closed loops between the connected components. This will also be the case for the non-trivial entanglement spectrum, but each loop will have a different weight that depend on the $p_i$'s.

  The $m_k$ are the (scaled) moments of a probability distribution.  While the explicit expression itself is not important for our purposes, the $Z_l$ can be written as the $l$-th moments of a probability distribution \cite{penington2019replica}. Hence, \cref{eq:moments jt gravity} is a product of moments, summed over all non-crossing partitions of length $k$. As a result, $\tr[\psi_{\rad}^k]$ can be calculated in the planar limit very simply by way of free probability theory.

  More precisely, we may define moment-generating functions for $\EOW$, $\JT$, and $\rad$:
  \begin{align*}
    M_{\JT}(z)  & = \sum_{k=1}^\infty \frac{Z_{k}}{Z_1^k} z^k,             \\
    M_{\EOW}(z) & = \sum_{k = 1}^{\infty} m_k z^k,                         \\
    M_{\rad}    & = \sum_{k=1}^\infty e^{S_0(k-1)} \tr[\psi_{\rad}^k] z^k.
  \end{align*}
  Given a moment generating function $M(z)$, which is a formal power series, recall that its S-transform is given by
  \begin{align*}
    S(z) = \frac{1+z}{z}M^{-1}(z).
  \end{align*}
  We use this to define the S-transforms $S_{\JT}$, $S_{\EOW}$ and $S_{\rad}$.
  From \cref{thm:free product}, we see that the relation between the moments in \cref{eq:moments jt gravity} implies that these are related as
  \begin{align}\label{eq:s-transform jt}
    S_{\rad}(z) = \frac{1}{1 + z}S_{\JT}(z) S_{\EOW}(z).
  \end{align}
  This means that, at the phase transition where $n \approx e^{S_0}$, the spectrum of $\psi_{\rad}$ can be described as a free product of the spectra of the end of the world brane, the JT gravity spectrum and a Marchenko-Pastur distribution.

  \subsubsection{A recursion relation for the resolvent}
  Given a moment generating function $M(z)$, we may also define the \emph{resolvent function} $R(z)$ by
  \begin{align*}
    R(z) = \frac{1}{z}\left(1 + M\left(\frac{1}{z}\right)\right).
  \end{align*}
  To relate to previous results, we consider the case where the entanglement with the radiation is maximally entangled. In this case, $S_{\EOW}(z) = 1$ and $S_{\rad}(z) = \frac{1}{1 + z}S_{\JT}(z)$.
  By definition of the S-transform and setting $z \to M_{\rad}(z)$, this implies
  \begin{align*}
    \frac{1}{1 + M_{\rad}(z)}S_{\JT}(M_{\rad}(z)) = \frac{1 + M_{\rad}(z)}{M_{\rad}(z)}z,
  \end{align*}
  which we may rewrite as (again using the definition of the S-transform):
  \begin{align*}
    M_{\rad}(z) = M_{\JT}[z(1 + M_{\rad}(z))].
  \end{align*}
  In terms of the resolvent, this becomes
  \begin{align*}
    R(z) & = \frac{1}{z} + \frac{1}{z}M_{R}\left(\frac{1}{z}\right)               \\
         & = \frac{1}{z} + \frac{1}{z}M_{\JT}(R(z))                               \\
         & = \frac{1}{z} + \sum_{k=1}^\infty \frac{Z_{k}}{Z_1^k}\frac{R(z)^k}{z},
  \end{align*}
  which is a recursion relation previously derived in \cite{penington2019replica} by a diagrammatic argument. More generally, we can interpret \cref{eq:s-transform jt} as a (complicated) recursion relation that directly generalizes the above recursion relation.

  \section{Random tensor network states and split transfer protocols}\label{sec:split_recovery}
  Our results involving general background states are closely related to the quantum-information-theoretic task of \textit{split transfer} introduced in \cite{dutil2010one}, which can be understood as a variant of quantum state merging. The standard setup is as follows: two parties, Alice and Bob, share a state $\phi_{ABC_1\ldots C_m}$, with Alice controlling $A$, Bob controlling $B$, and the systems $C_1,\ldots, C_m$ being $m$ ``helpers''.
  Let $R$ be a reference system and $\phi_{ABRC_1\ldots C_m}$ a purification of $\phi$.
  Initially, the state is shared not only by Alice and Bob, but also with all the helper systems $C_i$. The goal of split transfer is to try to redistribute the state to Alice, Bob, and $R$, using local quantum operations and classical communication (LOCC) between Alice, Bob and the helper systems, and possibly with the assistance of additional maximally entangled states.

  A split transfer protocol consists of
  \begin{enumerate}
    \item A partitioning of the set of the helper systems: $T_A \sqcup T_B = \{C_1,\ldots,C_m\}$.
    \item For each $C_i \in T_A$, a number $K_{A,i}$ of shared maximally entangled qubits between Alice and $C_i$, and for each $C_i \in T_B$, a number $K_{B,i}$ of shared maximally entangled qubits between Bob and $C_i$.
    \item An LOCC operation between Alice, Bob and the helper systems, such that after applying the protocol, Alice and Bob share a state $\psi_{ABRC_1\ldots C_m}$, which is close to $\phi_{ABRC_1\ldots C_m}$, and is such that Alice possesses systems $A$ and $T_A$, while Bob controls $B$ and $T_B$. Moreover, after applying the protocol they may be in possession of a number $L_{A,i}$ or $L_{B,i}$ of (approximately) maximally entangled qubits between $C_i$, and respectively $A$ or $B$.
  \end{enumerate}
  In this case, we say that the split transfer protocol has \emph{entanglement costs} $K_{A,i} - L_{A,i}$ for all $C_i \in T_A$ and $K_{B,i} - L_{B,i}$ for all $C_i \in T_B$.
  A precise definition can be found as Definition 14 in \cite{dutil2010one}.


  Intuitively, the helper systems need to transfer their correlations with $R$ to Alice and Bob, but without touching $R$.
  For instance, a simple protocol would be that the helpers simply teleport their full system to either Alice or Bob, consuming EPR pairs, leading to large entanglement costs.
  We can construct a potentially much more efficient protocol by way of random measurements, as detailed in Proposition 16 of \cite{dutil2010one}. Roughly speaking, such a protocol functions because random measurements have the effect of \textit{decoupling} the helper systems from $R$. The helpers perform simultaneous random measurements on their systems, and send the results of their measurements to Alice and Bob. Then, Alice and Bob can use their share of the global state and their portions of the maximally-entangled states to perform a decoding operation conditioned on the results of the random measurements. The state they receive will be a purification of $\psi_{ABC_1\ldots C_m}$, which will then be equivalent to the original $\psi_{ABRC_1\ldots C_m}$ up to local isometries.
  The way we set up the split transfer protocol above was in a one-shot fashion: we get a single copy of $\phi$ and need to determine the optimal entanglement cost for the protocol.

  One can also consider asymptotic variants of split transfer, where one has many copies available and aims to achieve an optimal transfer rate.
  An example application is the \emph{entanglement of assistance}.
  Suppose that Alice, Bob and the helper systems $C_i$ get many copies of a pure state $\phi$.
  At what rate can they distill maximally entangled pairs between Alice and Bob, if Alice and Bob are allowed to perform LOCC operations with all the helper systems?
  In this case, the answer is that the rate is given by
  \begin{align*}
    \min_{T_A} S(AT_A)_{\phi},
  \end{align*}
  that is, by minimizing the entanglement entropy over all bipartitions.
  This rate is reminiscent of the importance of minimal cuts in a random tensor network, and the connection was explored in \cite{hayden2016holographic}.

  To see how the task of split transfer relates to random tensor networks, we consider three boundary regions, $A$ and $B$, under the control of Alice and Bob, and the purifying system $R$.
  Each of the bulk vertices on the network correspond to a `helper' party.
  We would like to know whether there exists a protocol in which the assisting parties are allowed to perform local operations and classical communication (LOCC) such that the state $\phi_{VR}$ is redistributed into a state $\rho_{ABR}$ held by Alice and Bob that can be transformed by local isometries, acting only on $A$ and $B$, to a state close to $\phi_{VR}$. The protocol given in \cite{dutil2010one} consists of simultaneous random measurements by each of the helpers.
  This precisely corresponds to the random projections performed in constructing the random tensor network state with this background state!

  In this light, we can interpret \cref{thm:split transfer} as a result on split transfer.
  Let us assume that $\phi \in \Peq(A B R C_1 \ldots C_n)$, denote the associated random tensor network state by $\rho_{ABR}$, and choose a partitioning $T_A \sqcup T_B$ of the assisting (bulk) parties.
  Since $H_2(A \vert B)_{\phi \vert \phi} \geq H_{\min}(A \vert B)_{\phi}$, \cref{thm:split transfer} directly yields
  \begin{thm*}
    Suppose that
    \begin{align}\label{eq:min entropy condition 1}
      H_{\min}(S_A \vert B R T_B)_{\phi} \geq \KK_1
    \end{align}
    for all non-empty subsets $S_A \subseteq T_A$ and
    \begin{align}\label{eq:min entropy condition 2}
      H_{\min}(S_B \vert A R T_A)_{\phi} \geq \KK_2
    \end{align}
    for all non-empty subsets $S_B \subseteq T_B$.
    Then
    \begin{align}\label{eq:recovery split transfer}
      \EE \min_{V_A, V_{B}} \norm{(V_A \ot V_{B} \ot I_R)\rho(V_A^\dagger \ot V_{B}^\dagger \ot I_R) - \phi_{ABR C_1 \ldots C_n}}_1 = \bigO((2^{-\frac14\KK_1} + 2^{-\frac14\KK_2}).
    \end{align}
    where the minimum is over isometries $V_{A}: \HH_A \to \HH_{A T_A}$ and $V_{B} : \HH_{B} \to \HH_{B T_B}$.
  \end{thm*}
  This result shows that if $K_1$ and $K_2$ are sufficiently large, then after measurement in a random basis, the state possessed by Alice and Bob can, with high probability, be used to approximately reconstruct $\phi$ by acting with local isometries on the systems of Alice and Bob.

  Suppose we fix values for $K_1$ and $K_2$. If the conditions in \cref{eq:min entropy condition 1} and \cref{eq:min entropy condition 1} are not satisfied for the initial state $\phi$, we can use another state where we have added an appropriate number of maximally entangled Bell pairs between the assisting parties, increasing the entanglement cost of the protocol.\footnote{In fact, the protocol in \cite{dutil2010one} is slightly more general than what we describe; rather than measuring a random state one could also measure a random projection of rank greater than 1. This can be used to obtain EPR pairs between the helpers and and Alice and Bob to get nonzero $L_{A,i}$ and $L_{B,i}$.}
  An interesting open question in this context is whether one can generalize this result using \emph{smooth} entropies in \cref{eq:min entropy condition 1} and \cref{eq:min entropy condition 1}. As alluded to in \cref{sec:far from max entangled}, the problem is that for a general state $\phi$, one would need to perform simultaneous smoothing for all the relevant subsystems, which remains an open problem.

  There is an alternative approach, in which it is straightforwardly possible to use smooth entropies \cite{dutil2010one}.
  In this approach, one merges each party in $T_A$ one by one, and similarly for $T_B$.
  That is, we choose some ordering $T_A = \{1, \ldots, m\} = [m]$ and we apply a sequence of state merging protocols where we merge the state in $m$ steps, where a single step merges $A \cup T_A \setminus [i-1]$ into $A \cup T_A \setminus [i]$.
  In this case it is not hard to see that, if we allow some error the entanglement cost is determined by the smooth conditional entropies $H_{\min}^\eps(\{i+1\} \vert  B R [i] T_B)_{\phi}$.
  We perform a similar protocol for $B$ and the assisting $T_B$ systems.

  \subsection{Split transfer and recovery in holography}

  Split transfer is closely related to \emph{subregion-subregion duality}, or \emph{entanglement wedge reconstruction}, in holography.
  Consider an asymptotically AdS, stationary semiclassical geometry which is dual to a boundary CFT state $\rho$.
  We fix a time-reversal invariant spatial slice and partition the boundary into $A$ and $\bar{A}$.
  Let $\gamma_A$ be the minimal surface for $A$, and recall that $\partial \gamma_A = \partial A$.
  The region enclosed by $A$ and $\gamma_A$ is the \emph{entanglement wedge} of $A$, which we denote by $\Gamma_A$.
  The claim of entanglement wedge reconstruction is that if we act with a low-energy local bulk operator in the entanglement wedge of $A$, we can reconstruct the action of this operator as a corresponding operator acting on the boundary system $A$.

  One way to make this more precise is by considering a \emph{code subspace} of bulk states $S$, which can be thought of as a set of states obtained by acting with low-energy operators on a fixed semiclassical space-time. The action of the operators is small, in the sense that they do not create a significant backreaction that changes the geometry.
  We then introduce a reference system $R$ of the same dimension as $S$. In this framework, the AdS/CFT correspondence describes an encoding of the bulk into the boundary, taking a bulk state $\phi_{\Gamma_A\Gamma_A^cR}$ to a boundary state $\rho_{A \bar A R}$.
  In this set-up, the claim of entanglement wedge reconstruction is that we can act with an isometry on $A$ to recover $\phi_{\Gamma_A R}$ (and $\Gamma_A$ is actually the maximal such region). This situation corresponds to quantum state merging -- there is a single decoder $A$. However, in AdS/CFT, one usually requires the stronger condition of \emph{complementary entanglement wedge reconstruction}, in which the entanglement wedge for $\bar{A}$ is also the complement of $\Gamma_A$, so that $\phi_{\Gamma_A^c R}$ is recoverable from $\bar{A}$. This stronger requirement with two decoders is closely related to split transfer.
  These ideas and the precise relation to quantum information theory remain an active area of research, e.g. \cite{hayden2016holographic, harlow2016qec,akers2020leading, akers2021quantum}.

  In particular, \cite{bao2019beyond, akers2020leading, akers2021quantum} have argued that one-shot quantum information is the correct framework to understand entanglement wedge reconstruction, as holography is fundamentally a one-shot setting -- we are provided with a single copy of a gravitational or CFT state, rather than asymptotically-many copies.

  The limit of large effective central charge, and hence small $G_N$, does reproduce certain aspects of the many-copy limit, as we reviewed in \cref{sec:euclidean}. However, the distinction between the one-shot and asymptotic regimes becomes apparent in the presence of large bulk entropy, where the bulk entropy, and hence $R$, is large. Such situations arise when studying, for example, the black hole information paradox.
  In this case, where we assume we have a bulk state $\phi$, the minimal surface prescription in \cref{eq:rt formula} is replaced by \emph{quantum extremal surface} prescription:
  \begin{align*}
    H(\rho_A) = \min \text{ext}_{\gamma_A} \left\{ \frac{\abs{\gamma_A}}{4 G_N} + H(\Gamma_A)_{\phi} \right\},
  \end{align*}
  where we minimize over extremal surfaces $\gamma_A$, and we minimize the joint contribution of the area of $\gamma_A$ and the bulk entropy contained in the associated entanglement wedge $\Gamma_A$.
  This formula has a natural tensor network interpretation: consider a random tensor network state with background state $\phi = \phi_{V^{(b)}R} \ot \phi_{V^{(l)}}$, where $\phi_{V^{(b)}R}$ is a general background state (accounting for bulk entropy) and $\phi_{V^{(l)}}$ is a tensor product of link states on a graph $G = (V,E)$.
  Let us take maximally entangled link states with dimension $D$.
  Then, for some cut $\Gamma_A$ with edge set $\gamma_A$, we have
  \begin{align*}
    H(\phi_{\Gamma_A}) = \log(D)\abs{\gamma_A} + H(\phi^{(b)}_{\Gamma_A}),
  \end{align*}
  and we may hope that minimization over this quantity along the cuts gives a good approximation to the entropy.
  Whether such a prescription is valid depends on the structure of the background state $\phi$.
  A proposal put forth in \cite{akers2020leading} is that the surface $\gamma_A$ with entanglement wedge $\Gamma_A$ gives the \emph{$\max$-entanglement wedge}, if $\Gamma_A$ is the largest region such that, for any other surface $\delta_A$ homologous to $A$, with $\Delta_A$ the region enclosed by $A$ and $\delta_A$, and where $\Delta_A$ is contained in $\Gamma_A$, it holds that
  \begin{align*}
    H_{\min}^{\eps}(\Gamma_A \setminus \Delta_A \vert \Gamma_A^c R) \gg \frac{\abs{\gamma_A} - \abs{\delta_A}}{4G_N}.
  \end{align*}
  In this case, $\Gamma_A$ should be the largest region which can be (approximately) reconstructed from $A$.
  Again, one can think of a random tensor network where the background state is a tensor product of a bulk state and a maximally entangled link state of dimension $D$ with $\log(D) = \Theta(G_N^{-1})$ along the discretization of the space.
  Then this condition is (apart from the simultaneous smoothing problem) equivalent to \cref{eq:min max condition 1}.
  Enforcing complementary reconstruction, in which $\Gamma_A$ is the $\max$-entanglement wedge for $A$, and its complement $\Gamma_A^c$ is the $\max$-entanglement wedge for $\bar{A}$, allows us to interpret the holographic encoding of the bulk state into the boundary as a version of one-shot split transfer.
  See \cite{akers2020leading} for a detailed discussion of this proposal for holographic systems.

  \section{Joint smoothing of link states}\label{sec:joint smoothing}
  In this section, we provide a proof of \cref{lem:joint relative min smoothing}. In our application of one-shot entropy estimates in \cref{sec:one minimal cut} and \cref{sec:two minimal cuts},
  we needed to jointly smooth over different subsystems for the background state $\phi$.
  For general background states, joint smoothing is an open problem \cite{drescher2013simultaneous,dutil2011multiparty}; however, if the background state is a product of link states as in \cref{eq:link state}, we can straightforwardly do so.
  In order to prove \cref{lem:joint relative min smoothing}, we first need some intermediate lemmas.
  These lemmas involve quantum states which can be expanded in some preferential basis with positive coefficients.
  Therefore, for \cref{lem:smaller max entropy 1,lem:smaller max entropy 2,lem:basis min smoothing,lem:bound fidelity minimum}, we assume that each quantum system has a preferential basis, which we assume without loss of generality to be the standard basis.

  \begin{lem}\label{lem:smaller max entropy 1}
    Suppose that $\phi \in \Pleq(X S T)$ can be written as
    \begin{align*}
      \phi_{XST} = \sum_{i} \proj{i} \ot \phi_{ST,i}
    \end{align*}
    where each $\phi_{ST,i}$ is a pure state such that $\ket{\phi_{ST,i}}$ can be written with positive coefficients in the standard basis.
    If
    \begin{align*}
      \tilde\phi_{XST} = \sum_{i} \proj{i} \ot \tilde\phi_{ST,i}
    \end{align*}
    is such that each $\tilde\phi_{ST,i}$ is a pure state satisfying $\ket{\tilde\phi_{ST,i}} \leq \ket{\phi_{ST,i}}$ elementwise, then
    \begin{align*}
      H_{\max}(XS \vert T)_{\tilde \phi} \leq H_{\max}(XS \vert T)_{\phi}.
    \end{align*}
  \end{lem}

  \begin{proof}
    We start by arguing that for the max-entropy
    \begin{align}\label{eq:max entropy optimization}
      H_{\max}(XS \vert T)_{\tilde \phi} = \max_{\sigma_T \in \Pleq(T)} \log F(\tilde \phi_{XST}, I_{XS} \ot \sigma_{T})^2
    \end{align}
    we may choose $\sigma_T$ in the optimization problem with nonnegative matrix elements in the standard basis.
    Indeed, suppose that $\sigma_T$ realizes the maximum in \cref{eq:max entropy optimization}.
    Then we can write a spectral decomposition $\sigma_T = \sum_{j} p_j \proj{e_j}$.
    Now let $\ket{\tilde e_j} = \sum_k \abs{\braket{e_j | k}} \ket{k}$.
    Then $\tilde \sigma_T = \sum_{j} p_j \proj{\tilde e_j} \in \Pleq(T)$ has nonnegative matrix elements in the standard basis.
    Moreover,
    \begin{align*}
      F(\tilde \phi_{XST}, I_{XS} \ot \sigma_{T})^{2} & = \sum_i F(\tilde \phi_{ST,i}, I_{S} \ot \sigma_{T})^{2}                                      \\
                                                      & \leq \sum_{i,j} p_j \abs{\bra{\tilde\phi_{ST,i}} I_S \ot \proj{e_j}} \ket{\tilde\phi_{ST,i}}  \\
                                                      & \leq \sum_{i,j} p_j \bra{\tilde\phi_{ST,i}} I_S \ot \proj{\tilde e_j} \ket{\tilde\phi_{ST,i}} \\
                                                      & = F(\tilde \phi_{XST}, I_{XS} \ot \tilde \sigma_{T})^{2}.
    \end{align*}
    Let $\sigma_T$ be optimal for \cref{eq:max entropy optimization} with nonnegative matrix elements.
    It is clear that if $\ket{\tilde\phi_{ST,i}} \leq \ket{\phi_{ST,i}}$, then
    \begin{align*}
      F(\tilde \phi_{XST}, I_{XS} \ot \sigma_{T})^{2} & = \sum_i \bra{\tilde\phi_{ST,i}} I_{S} \ot \sigma_{T} \ket{\tilde\phi_{ST,i}} \\
                                                      & \leq \sum_i \bra{\phi_{ST,i}} I_{S} \ot \sigma_{T} \ket{\phi_{ST,i}}          \\
                                                      & = F(\phi_{XST}, I_{XS} \ot \sigma_{T})^{2}
    \end{align*}
    and hence
    \begin{align*}
      H_{\max}(XS \vert T)_{\tilde \phi} = \log F(\tilde \phi_{XST}, I_{XS} \ot \sigma_{T})^2 \leq \log F(\phi_{XST}, I_{XS} \ot \sigma_{T})^2 \leq H_{\max}(XS \vert T)_{\phi}.
    \end{align*}
  \end{proof}

  \begin{lem}\label{lem:smaller max entropy 2}
    Suppose that $\phi \in \Pleq(XYST)$ can be written as
    \begin{align*}
      \phi_{XYST} = \sum_{i} \proj{i} \ot \proj{j} \ot \phi_{ST,ij}
    \end{align*}
    where each $\phi_{ST,ij}$ is a pure state such that $\ket{\phi_{ST,ij}}$ can be written with positive coefficients in the standard basis.
    If
    \begin{align*}
      \tilde\phi_{XYST} = \sum_{i} \proj{i} \ot \proj{j} \ot \tilde\phi_{ST,ij}
    \end{align*}
    is such that each $\tilde\phi_{ST,ij}$ is a pure state satisfying $\ket{\tilde\phi_{ST,ij}} \leq \ket{\phi_{ST,ij}}$ elementwise, then
    \begin{align*}
      H_{\max}(XS \vert YT)_{\tilde \phi} \leq H_{\max}(XS \vert YT)_{\phi}.
    \end{align*}
  \end{lem}

  \begin{proof}
    If we have a state $\rho \in \Pleq(XYST)$ which is of the form
    \begin{align*}
      \sum_{j} \proj{j} \ot \rho_j
    \end{align*}
    where each $\rho_j \in \Pleq(XST)$, then (Proposition 4.6 in \cite{tomamichel2012framework})
    \begin{align*}
      H_{\max}(XS \vert YT) = \log\mleft(\sum_{j} 2^{H_{\max}(XS \vert T)_{\rho_j}} \mright).
    \end{align*}
    Together with \cref{lem:smaller max entropy 1} this implies the result.
  \end{proof}

  The background states we consider are tensor products of link states, and they can be expressed with positive coefficients in a product basis along the half-edges.
  The following lemma will help us to show that when we perform smoothing we can retain some of this structure.

  \begin{lem}\label{lem:basis min smoothing}
    Suppose that $\phi \in \Pleq(X S T)$ can be written as
    \begin{align*}
      \phi_{X S T} = \sum_{i} \proj{i} \ot \phi_{S T, i}
    \end{align*}
    where each $\phi_{S T, i}$ is pure and has a Schmidt decomposition
    \begin{align*}
      \ket{\phi_{S T,i}} = \sum_j \sqrt{\lambda_{i,j}} \ket{jj}
    \end{align*}
    in the standard basis.
    Then
    \begin{align*}
      H_{\min}^{\eps}(X S \vert T)_\phi = H_{\min}(X S \vert T)_{\tilde \phi}
    \end{align*}
    for a state $\tilde \phi \in \Pleq(XST)$ which is such that $P(\phi, \tilde \phi) \leq \eps$ and
    \begin{align*}
      \tilde \phi_{X S T} = \sum_{i} \proj{i} \ot \tilde \phi_{S T,i}
    \end{align*}
    with
    \begin{align*}
      \ket{\tilde \phi_{S T,i}} = \sum_j \sqrt{\tilde\lambda_{i,j}} \ket{jj}.
    \end{align*}
    for some $0 \leq \tilde \lambda_{i,j} \leq 1$.
  \end{lem}

  \begin{proof}
    Let $Y$ be a copy of $X$ and let $R$ be an additional reference system of sufficiently large dimension.
    Then we let
    \begin{align*}
      \ket{\phi_{XYSTR}} = \sum_{i,j} \sqrt{\lambda_{i,j}} \ket{ii} \ot \ket{jj} \ot \ket{0}
    \end{align*}
    be a purification of $\phi_{XST}$.
    Now, by duality of smooth entropies
    \begin{align*}
      H_{\min}^{\eps}(X S \vert T)_\phi = - H_{\max}^{\eps}(X S \vert Y)_\phi
    \end{align*}
    By Lemma 6.13 of \cite{tomamichel2015quantum} we can find a state
    \begin{align*}
      \bar{\phi}_{S XY} = \sum_{j} \proj{j} \ot \bar{\phi}_{XY,j}
    \end{align*}
    such that $P(\phi, \bar \phi) \leq \eps$ and $H_{\max}^{\eps}(X S \vert Y)_\phi = H_{\max}(X S \vert Y)_{\bar \phi}$.
    Now, an arbitrary purification of $\bar{\phi}_{S XY}$ will be of the form
    \begin{align*}
      \ket{\bar \phi_{STXYR}} = \sum_{j,k} \ket{jk} \ot \ket{\bar \phi_{XYR,jk}}.
    \end{align*}
    We see that
    \begin{align*}
      \abs{\braket{\phi | \bar \phi}} = \sum_{i,j,k} \delta_{j,k} \sqrt{\lambda_{i,j}}\abs{\left(\bra{ii} \ot \bra{0} \right)\ket{\phi_{XYR,jk}}}
    \end{align*}
    so, optimizing over the choice of purification $\ket{\bar \phi}$, by Uhlmann's theorem we find a purification of the form
    \begin{align*}
      \ket{\bar \phi_{STXYR}} = \sum_{j} \ket{jj} \ot \ket{\bar \phi_{XYR,j}}
    \end{align*}
    such that $F(\phi_{XYS}, \bar \phi_{XYS}) = \abs{\braket{\phi | \bar \phi}}$.
    We define the projector $\Pi_{XY} = \sum_{j} \proj{jj}$.
    Then for any $\sigma_Y \in \Pleq(Y)$ we note that $(I_S \ot \Pi_{XY}) I_{XS} \ot \sigma_Y (I_S \ot \Pi_{XY}) \leq I_{XS} \ot \sigma_{Y}'$ where $\sigma_{Y}' = \sum_j \bra{j} \sigma_Y \ket{j} \proj{j}$ and hence
    \begin{align*}
      F((I_S \ot \Pi_{XY} )\bar \phi (I_S \ot \Pi_{XY} ), I_{XS} \ot \sigma_Y) = F(\bar \phi, (I_S \ot \Pi_{XY} ) I_{XS} \ot \sigma_Y (I_S \ot \Pi_{XY} )) \leq F(\bar{\phi}, I_{XS} \ot \sigma_{Y}').
    \end{align*}
    If we let $\ket{\phi_{XYSTR}'} = (\Pi_{XY} \ot I_{STR})\ket{\bar\phi_{XYSTR}}$ this implies that
    \begin{align*}
      H_{\max}(XS \vert Y)_{\phi'} \leq H_{\max}(XS \vert Y)_{\bar \phi}.
    \end{align*}
    Moreover it is easy to see that $\abs{\braket{\phi' | \phi}} = \abs{\braket{\bar \phi | \phi}}$.
    Again using duality we find that
    \begin{align*}
      H_{\min}^{\eps}(X S \vert T)_\phi \leq H_{\min}(X S \vert T R)_{\phi'}.
    \end{align*}
    Finally, let
    \begin{align*}
      \ket{\tilde \phi} = (I_{XYST} \ot \bra{0}) \ket{\phi'}
    \end{align*}
    then by data processing (Theorem 5.7 of \cite{tomamichel2012framework}, note that data processing for the conditional min-entropy is valid for \emph{trace non-increasing} completely postive maps) it holds that
    \begin{align*}
      H_{\min}(X S \vert T R)_{\phi'} \leq H_{\min}(X S \vert T)_{\tilde \phi}.
    \end{align*}
    By construction $\abs{\braket{\tilde \phi | \phi}} = \abs{\braket{\phi' | \phi}}$ and hence $P(\tilde \phi,\phi) \leq \eps$.
    Finally, by construction $\tilde \phi$ is of the desired form
    \begin{align*}
      \ket{\tilde \phi_{STXY}} = \sum_{i,j} \sqrt{\tilde \lambda_{i,j}}\ket{ii} \ot \ket{jj}.
    \end{align*}
  \end{proof}

  The following lemmas will be used in our joint smoothing construction to bound the purified distance of the smoothed state.

  \begin{lem}\label{lem:consequence fvdg}
    Suppose $\phi, \psi \in \Pleq(\HH)$ are pure states, and suppose
    \begin{align*}
      \ket{\phi} = \sum_i \sqrt{\lambda_i} \ket{i} \qquad \ket{\psi} = \sum_i \sqrt{\mu_i} \ket{i}
    \end{align*}
    for some $\lambda_i, \mu_{i} \geq 0$.
    Then
    \begin{align*}
      T(\phi, \psi) & \leq P(\phi,\psi) \leq \sqrt{2\sum_i \abs{\lambda_i - \mu_i}}.
    \end{align*}
  \end{lem}

  \begin{proof}
    Let
    \begin{align*}
      \rho = \Delta(\phi), \qquad \sigma = \Delta(\psi)
    \end{align*}
    where $\Delta$ is the completely dephasing channel, so
    \begin{align*}
      \rho = \sum_i \lambda_i \proj{i} \qquad \sigma = \sum_i \mu_i \proj{i}.
    \end{align*}
    Then since $\rho$ and $\sigma$ are diagonal in the same basis
    \begin{align*}
      P(\phi, \psi) = P(\rho, \sigma).
    \end{align*}
    We then estimate the trace distance by
    \begin{align*}
      T(\rho, \sigma) \leq \norm{\rho - \sigma}_1 = \sum_i \abs{\lambda_i - \mu_i}
    \end{align*}
    and we apply the Fuchs-Van de Graaff inequalities, \cref{eq:fuchs vd graaff}, to estimate
    \begin{align*}
      T(\phi, \psi) \leq P(\phi,\psi) = P(\rho, \sigma) \leq \sqrt{2T(\rho, \sigma)} \leq \sqrt{2\sum_i \abs{\lambda_i - \mu_i}}.
    \end{align*}
  \end{proof}

  \begin{lem}\label{lem:bound fidelity minimum}
    Suppose $\phi, \phi_j \in \Pleq(\HH)$ are pure states for $j \in [n]$, and suppose
    \begin{align*}
      \ket{\phi} = \sum_i \sqrt{\lambda_i} \ket{i} \qquad \ket{\phi_j} = \sum_i \sqrt{\lambda_{j,i}} \ket{i}
    \end{align*}
    for some $\lambda, \lambda_{i,j} \geq 0$ and suppose $P(\phi, \phi_j) \leq \eps$ for all $j$.
    Then if we let
    \begin{align*}
      \ket{\tilde \phi} = \sum_i \min_j \sqrt{\lambda_{j,i}} \ket{i}
    \end{align*}
    it holds that $P(\phi, \tilde \phi) \leq 2\sqrt{n\eps}$.
  \end{lem}

  \begin{proof}
    Let
    \begin{align*}
      \rho = \sum_i \lambda_i \proj{i} \qquad \rho_j = \sum_i \lambda_{j,i} \proj{i}.
    \end{align*}
    Then, $P(\rho,\rho_j) = P(\phi, \phi_j)$ and using the Fuchs-Van de Graaff inequalities
    \begin{align*}
      \sum_i \abs{\lambda_i - \min_j \lambda_{j,i}} \leq \sum_j \sum_i \abs{\lambda_i - \lambda_{j,i}} = \sum_j \norm{\rho - \rho_j}_1 \leq 2\sum_{j} T(\rho, \rho_j) \leq 2\sum_j P(\phi, \phi_j) \leq 2n\eps.
    \end{align*}
    The result now follows from \cref{lem:consequence fvdg}.
  \end{proof}

  Finally, we prove our main joint smoothing result.
  We again consider the setting of a tensor product of link states on a graph $G = (V,E)$, where we consider some boundary subsystem $A \subseteq V_\partial$.
  Recall that for a cut $\Gamma_A \in C(A)$ we define
  \begin{align*}
    \mathcal C_1(\Gamma_A) & = \{\Delta_A \in C(A) : \Delta_A \subsetneq \Gamma_A \}  \\
    \mathcal C_2(\Gamma_A) & = \{\Delta_A \in C(A) : \Gamma_A \subsetneq \Delta_A \}.
  \end{align*}

  \jointrelativeminsmoothing*
  \begin{proof}
    We may assume without loss of generality that $\phi = \bigotimes_{e \in E} \phi_e$ is such that each $\phi_e$ has Schmidt decomposition in the standard basis,
    \begin{align*}
      \ket{\phi_e} = \sum_{i = 1}^{D_e} \sqrt{\lambda_{e,i}} \ket{i i}
    \end{align*}
    This means we may write
    \begin{align*}
      \ket{\phi} = \sum_I \sqrt{\lambda_I} \ket{I}
    \end{align*}
    where $I$ runs over all possible basis elements along each edge $I = \{i_e\}_{e \in E}$ and
    \begin{align*}
      \lambda_I = \prod_{e \in E} \lambda_{e,i_e} \qquad \ket{I} = \bigotimes_{e \in E} \ket{i_e i_e}.
    \end{align*}
    We also let
    \begin{align*}
      E_1 & = \{e = (xy) \in E, x,y \in \Gamma_A\}    \\
      E_2 & = \{e = (xy) \in E, x,y \in \Gamma_A^c\}.
    \end{align*}
    so $E = E_1 \sqcup E_2 \sqcup \gamma_A$.

    Consider a cut $\Delta_A \in \mathcal C_1(\Gamma_A)$.
    Let
    \begin{align*}
      X & = \{(e,x) : e = (xy), x \in \Gamma_A \setminus \Delta_A, y \in \Delta_A \}   \\
      S & = \{(e,x) : e = (xy), x \in \Gamma_A \setminus \Delta_A, y \in \Delta_A^c \} \\
      T & = \{(e,x) : e = (xy), x \in \Gamma_A^c, y \in \Delta_A^c \}                  \\
    \end{align*}
    Then we see that $H_{\min}^{\eps}(\Gamma_A \setminus \Delta_A \vert \Gamma_A^c)_{\phi} = H_{\min}^{\eps}(XS \vert T)_{\phi}$ and we can write (as $\phi$ is a product state)
    \begin{align*}
      \phi_{XST} = \phi_X \ot \phi_{ST}.
    \end{align*}
    Here $\phi_{ST}$ is pure, whereas $\phi_X$ is diagonal in the standard basis along the half-edges in $X$.
    From \cref{lem:basis min smoothing} (applied with $\phi_{ST,i} = \phi_{ST}$ for each $i$) we find that we obtain a state $\phi^{\Delta_A, \eps}$ which is such that $P(\phi^{\Delta_A, \eps}, \phi) \leq \eps$ and
    \begin{align*}
      H_{\min}(\Gamma_A \setminus \Delta_A \vert \Gamma_A^c)_{\phi^{\Delta_A, \eps}} \geq H_{\min}^{\eps}(\Gamma_A \setminus \Delta_A \vert \Gamma_A^c)_{\phi}
    \end{align*}
    and which can be written as
    \begin{align*}
      \ket{\phi^{\Delta_A, \eps}} = \sum_I \sqrt{\lambda_I^{\Delta_A, \eps}} \ket{I}.
    \end{align*}
    for some coefficients $\lambda_I^{\Delta_A, \eps}$.
    Moreover, we can write
    \begin{align*}
      \ket{\phi^{\Delta_A, \eps}} = \sum_i \ket{ii} \ot \ket{\phi_{E_1, i}^{\Delta_A, \eps}} \ot \ket{\phi_{E_2}}
    \end{align*}
    where $i$ runs over all possible basis elements along the cut $\gamma_A$, $i = \{i_e\}_{e \in \gamma_A}$ and $\ket{i} = \bigotimes_{e \in \gamma_A} \ket{i_e}$, and the $\phi_{E_1, i}^{\Delta_A, \eps}$ are pure states with positive coefficients in the standard basis.
    We now let
    \begin{align*}
      \ket{\phi_1^{\eps}} & = \sum_{I} \sqrt{\lambda_I^\eps} \ket{I}                                 \\
      \lambda_I^\eps      & = \min_{\Delta_A \in \mathcal C_1(\Gamma_A)} \lambda_I^{\Delta_A, \eps}.
    \end{align*}
    By construction, this state can be written as
    \begin{align*}
      \ket{\phi_1^{\eps}} = \sum_i \ket{ii} \ot \ket{\phi_{E_1, i}^{\eps}} \ot \ket{\phi_{E_2}}
    \end{align*}
    By \cref{lem:bound fidelity minimum}, it holds that
    \begin{align*}
      P(\phi, \phi_1^{\eps}) \leq 2\sqrt{\abs{\mathcal C_1(\Gamma_A)}\eps}.
    \end{align*}

    An analogous construction can be used to construct $\phi^{\Delta_A, \eps}$ for $\Delta_A \in \mathcal C_2(\Gamma_A)$, and taking the minimum over all such cuts, we get a state $\phi_2^{\eps}$ of the form
    \begin{align*}
      \ket{\phi_2^{\eps}} = \sum_i \ket{ii} \ot \ket{\phi_{E_1}} \ot \ket{\phi_{E_2, i}^{\eps}}
    \end{align*}
    which satisfies
    \begin{align*}
      P(\phi, \phi_2^{\eps}) \leq 2\sqrt{\abs{\mathcal C_2(\Gamma_A)}\eps}.
    \end{align*}
    We now define
    \begin{align*}
      \ket{\phi^{\eps}} = \sum_i \ket{ii} \ot \ket{\phi_{E_1, i}^{\eps}} \ot \ket{\phi_{E_2, i}^{\eps}}.
    \end{align*}

    We will now show that $\phi^{\eps}$ has the desired properties.
    First of all, since $\phi$ is normalized, $F_*(\phi, \phi^\eps) = F(\phi, \phi^\eps)$ and $F(\phi, \phi^\eps) = F(\phi_1^\eps, \phi_2^\eps)$ so
    \begin{align*}
      P(\phi, \phi^\eps) = P(\phi_1^\eps, \phi_2^\eps) \leq P(\phi,\phi_1^\eps) + P(\phi, \phi_2^\eps) \leq 2\left(\sqrt{\abs{\mathcal C_1(\Gamma_A)}} + \sqrt{\abs{\mathcal C_2(\Gamma_A)}}\right) \sqrt{\eps}
    \end{align*}
    using the fact that the purified distance is a metric.

    Next, consider $\Delta_A \in C_1(A)$.
    We note that by duality $H_{\min}(\Gamma_A \setminus \Delta_A \vert \Gamma_A^c)_{\phi^{\eps}} = - H_{\max}(\Gamma_A \setminus \Delta_A \vert \Delta_A)_{\phi^{\eps}}$.
    We define the following subsystems of half-edges,
    \begin{align*}
      X' & = \{(e,x) : e = (xy), x \in \Gamma_A \setminus \Delta_A, y \in \Gamma_A^c\} \\
      Y' & = \{(e,x) : e = (xy), x \in \Delta_A, y \in \Gamma_A^c \}                   \\
      S' & = \{(e,x) : e = (xy), x \in \Gamma_A \setminus \Delta_A, y \in \Gamma_A \}  \\
      T' & = \{(e,x) : e = (xy), x \in \Delta_A, y \in \Gamma_A \}.
    \end{align*}
    so $H_{\max}(\Gamma_A \setminus \Delta_A \vert \Delta_A)_{\phi^{\eps}} = H_{\max}(X'S' \vert Y'T')_{\phi^{\eps}}$.
    Next, we observe that by construction, $\phi^{\eps}$ is of the form
    \begin{align*}
      \phi^{\eps}_{X'Y'S'T'} = \sum_{i,j} \proj{i} \ot \proj{j} \ot \phi^{\eps}_{S'T',ij}
    \end{align*}
    where $i$ runs over all possible basis elements along $\gamma_A \setminus \delta_A$, $i = \{i_e\}_{e \in \gamma_A\setminus \delta_A}$ and $\ket{i} = \bigotimes_{e \in \gamma_A \setminus \delta_A} \ket{i_e}$, which forms a basis for $\HH_X$; and $j$ runs over all possible basis elements along $\gamma_A \cap \delta_A$, $j = \{j_e\}_{e \in \gamma_A\cap \delta_A}$ and $\ket{j} = \bigotimes_{e \in \gamma_A \cap \delta_A} \ket{j_e}$, which forms a basis for $\HH_Y$.
    Each $\phi^{\eps}_{S'T',ij}$ is a pure state.
    Similarly,
    \begin{align*}
      \phi^{\Delta_A,\eps}_{X'Y'S'T'} = \sum_{i,j} \proj{i} \ot \proj{j} \ot \phi^{\Delta_A,\eps}_{S'T',ij}
    \end{align*}
    where the $\phi^{\Delta_A,\eps}_{S'T',ij}$ are pure states.
    Moreover, by construction $\ket{\phi^{\eps}_{S'T',ij}} \leq \ket{\phi^{\Delta_A,\eps}_{S'T',ij}}$ elementwise in the standard basis.
    Therefore, by \cref{lem:smaller max entropy 2} it holds that
    \begin{align*}
      H_{\max}(X'S' \vert Y'T')_{\phi^{\eps}} \leq H_{\max}(X'S' \vert Y'T')_{\phi^{\Delta_A,\eps}}.
    \end{align*}
    Finally, by duality $H_{\min}(\Gamma_A \setminus \Delta_A \vert \Gamma_A^c)_{\phi^{\Delta_A,\eps}} = -H_{\max}(X'S' \vert Y'T')_{\phi^{\Delta_A,\eps}}$ and we conclude
    \begin{align*}
      H_{\min}(\Gamma_A \setminus \Delta_A \vert \Gamma_A^c)_{\phi^{\eps}} \geq H_{\min}(\Gamma_A \setminus \Delta_A \vert \Gamma_A^c)_{\phi^{\Delta_A,\eps}} \geq H_{\min}^{\eps}(\Gamma_A \setminus \Delta_A \vert \Gamma_A^c)_{\phi}.
    \end{align*}
    A similar argument is valid for $\Delta_A \in \mathcal C_2(\Gamma_A)$, showing that
    \begin{align*}
      H_{\min}(\Delta_A \setminus \Gamma_A \vert \Gamma_A)_{\phi^{\eps}} \geq H_{\min}(\Delta_A \setminus \Gamma_A \vert \Gamma_A)_{\phi^{\Delta_A,\eps}} \geq H_{\min}^{\eps}(\Delta_A \setminus \Gamma_A \vert \Gamma_A)_{\phi}.
    \end{align*}
  \end{proof}

  In the proof of \cref{thm:measure convergence surface transition} we need a similar smoothing lemma for a slightly different situation as stated in \cref{lem:joint smoothing at transition}.

  Recall the set-up: we consider a graph $G = (V,E)$ and we have a set of boundary vertices $V_\partial$, and $A \subseteq V_\partial$.
  We denote by $\gamma_{A,1}$ the edges incident to $A$, and $\gamma_{A,2}$ the edges incident to $\bar A$, and we assume these sets do not intersect.
  We let $E_b = E \setminus (\gamma_{A,1} \cup \gamma_{A,2})$.
  For a cut $\Delta_A \in C(A)$ we let $Y^{\Delta_A}$ be the set of half-edges
  \begin{align*}
    Y^{\Delta_A} = \{(e,x) : e = (xy), x \in \Delta_A^c, y \in A \}.
  \end{align*}

  \jointsmoothingtransition*

  \begin{proof}
    The argument is much the same as the proof of \cref{lem:joint relative min smoothing}.
    We again assume without loss of generality that each $\phi_e$ has Schmidt decomposition in the standard basis, and we we may write
    \begin{align*}
      \ket{\phi} = \sum_I \sqrt{\lambda_I} \ket{I}
    \end{align*}
    as in the proof of \cref{lem:joint relative min smoothing}.
    Consider a cut $\Delta_A \in C(A)$.
    Then we let
    \begin{align*}
      X & = \{(e,x) : e = (xy), x \in \Delta_A, y \in \Delta_A^c \}          \\
      S & = \{(e,x) : e = (xy), x \in \Delta_A, y \in \Delta_A \setminus A\} \\
      T & = \{(e,x) : e = (xy), x \in A, y \notin A \} \cup Y^{\Delta_A}
    \end{align*}
    so $H^{\eps}_{\min}(\Delta_A \setminus A \vert A Y^{\Delta_A})_{\phi} = H^{\eps}_{\min}(XS \vert T)_{\phi}$.
    Moreover, by the structure of the state, $\phi_{XST} = \phi_X \ot \phi_{ST}$, where $\phi_{ST}$ is pure and $\phi_X$ is diagonal in the standard basis.
    So, from \cref{lem:basis min smoothing} we obtain a state
    \begin{align*}
      \ket{\phi^{\Delta_A, \eps}} = \sum_I \sqrt{\lambda_I^{\Delta_A, \eps}} \ket{I}
    \end{align*}
    for some real coefficients $\lambda_I^{\Delta_A, \eps} \geq 0$ which is such that $P(\phi, \phi^{\Delta_A, \eps}) \leq \eps$ and $H_{\min}(XS \vert T)_{\phi^{\eps}} \geq H^{\eps}_{\min}(XS \vert T)_{\phi}$.
    As before, we define
    \begin{align*}
      \ket{\phi^{\eps}} & = \sum_I \sqrt{\lambda_{I}^\eps} \ket{I}               \\
      \lambda_I^\eps    & = \min_{\Delta_A \in C(A)} \lambda_I^{\Delta_A, \eps}.
    \end{align*}
    By \cref{lem:bound fidelity minimum} $P(\phi, \phi^\eps) \leq 2\sqrt{\abs{C(A)}\eps} = 2\sqrt{2^{\abs{V_b}}\eps}$.
    If we define the following subsystems
    \begin{align*}
      X' & = \{(e,x) : e = (xy), x \in \Delta_A, y A \}           \\
      Y' & = \{(e,x) : e = (xy), x \in \Delta_A^c, y A \}         \\
      S' & = \{(e,x) : e = (xy), x \in \Delta_A, y \notin A\}     \\
      T' & = \{(e,x) : e = (xy),  x \in \Delta_A^c, y \notin A \}
    \end{align*}
    then $\phi^{\eps}_{X'S'T'}$ can be written as
    \begin{align*}
      \phi^{\eps}_{X'Y'S'T'} = \sum_{i, j} \proj{i} \ot \proj{j} \ot \phi^{\eps}_{S'T',ij}
    \end{align*}
    where $i$ runs over all possible basis elements along $\gamma_{A,1} \setminus \delta_A$, $i = \{i_e\}_{e \in \gamma_{A,1}\setminus \delta_A}$ and $\ket{i} = \bigotimes_{e \in \gamma_{A,1} \setminus \delta_A} \ket{i_e}$, which forms a basis for $\HH_{X'}$; and $j$ runs over all possible basis elements along $\gamma_{A,1} \cap \delta_A$, $j = \{j_e\}_{e \in \gamma_{A,1}\cap \delta_A}$ and $\ket{j} = \bigotimes_{e \in \gamma_{A,1} \cap \delta_A} \ket{j_e}$, which forms a basis for $\HH_{Y'}$.
    The $\phi^{\eps}_{S'T',ij}$ are pure states.
    We can therefore apply \cref{lem:smaller max entropy 1}, and
    \begin{align*}
      H_{\max}(X'S' \vert T'Y')_{\phi^{\eps}} \leq H_{\max}(X'S' \vert T'Y')_{\phi^{\Delta_A,\eps}}
    \end{align*}
    and hence, by applying duality and data processing
    \begin{align*}
      H_{\min}(\Delta_A \setminus A \vert A)_{\phi^\eps} & = - H_{\max}(\Delta_A \setminus A \vert T'Y')_{\phi^{\eps}} = - H_{\max}(X'S' \vert T'Y')_{\phi^{\eps}}                                                  \\
                                                         & \geq  H_{\max}(X'S' \vert T'Y')_{\phi^{\Delta_A, \eps}} = H_{\min}(\Delta_A \setminus A \vert A)_{\phi^{\Delta_A, \eps}}                                 \\
                                                         & \geq H_{\min}(\Delta_A \setminus A \vert A Y^{\Delta_A})_{\phi^{\Delta_A, \eps}} \geq H^{\eps}_{\min}(\Delta_A \setminus A \vert A Y^{\Delta_A})_{\phi}.
    \end{align*}
  \end{proof}

  \begin{rmk}
    We may make two observations on \cref{lem:joint relative min smoothing} and \cref{lem:joint smoothing at transition}. In both cases, if we construct the coefficients as $\min \lambda_{I}^{\eps}, \lambda_I$ we find that we may assume the resulting state has coefficients in the basis $\ket{I}$ which are upper bounded by $\lambda_I$ (affecting only the constant factor in the upper bound for $P(\phi, \phi^{\eps})$).
    Secondly, we note that
    \begin{align}\label{eq:trace difference}
      \abs{\tr[\phi] - \tr[\phi^{\eps}]} \leq \sum_{I}\, \abs{\lambda_I - \lambda_I^{\eps}} = \bigO(\eps)
    \end{align}
    rather than the naive $\bigO(\sqrt{\eps})$ estimate.
  \end{rmk}

  \section{A family of metrics on the symmetric group}\label{sec:metric}

  In analogy to the Cayley distance on the symmetric group~$S_k$, we consider the following function $d_{\rho} \colon S_k \times S_k \rightarrow \RR_{\geq0}$ given some arbitrary density matrix $\rho \in \Peq(\HH)$:
  \begin{align*}
    d_{\rho}(\pi_1,\pi_2) & = \sum_{l \in C(\pi_1^{-1} \pi_2)} (l-1) H_l(\rho) = -\sum_{l \in C(\pi_1^{-1} \pi_2)} \log \tr[\rho^l].
  \end{align*}
  As we saw in \cref{sec:replica trick}, this function is closely related to the replica trick for random tensor networks with nontrivial link states.
  In the case where $\rho$ is a maximally mixed state, this function is precisely proportional to the Cayley distance~\eqref{eq:cayley}.
  In this appendix we show that $d_\rho$ is a metric for \emph{any} non-pure quantum state~$\rho$.
  (If it is pure, then $d_\rho$ vanished identically.)

  We say that two permutations $\alpha, \beta \in S_k$ are \emph{disjoint} if any point not fixed by $\alpha$ is fixed by $\beta$ and vice versa.

  \begin{lem}
    The function $d_\rho$ defines a metric on $S_k$ for any state $\rho$ which is not pure.
    Moreover, if $\rho$ does not have flat spectrum the following holds for all $\pi_1,\pi_2,\pi_3\in S_k$:
    \begin{align*}
      d_{\rho}(\pi_1,\pi_2) + d_{\rho}(\pi_2,\pi_3) = d_{\rho}(\pi_1,\pi_3)
    \end{align*}
    if and only if $\pi_1^{-1}\pi_2$ and $\pi_2^{-1}\pi_3$ are disjoint permutations.
  \end{lem}

  \begin{proof}
    The fact that $d_{\rho}(\pi_1, \pi_2) = 0$ if and only if $\pi_1 = \pi_2$ follows from the assumption that $\rho$ is not pure.
    The symmetry $d_{\rho}(\pi_1,\pi_2) = d_{\rho}(\pi_2,\pi_1)$ is clear from $C(\pi_1^{-1}\pi_2) = C(\pi_2^{-1}\pi_1)$.
    Thus, the only nontrivial property we have to show in order for $d_\rho$ to be a metric is the triangle inequality.
    We let $d = \dim(\HH)$ and we write $\spec(\rho) = \{\lambda_i\}_{i=1}^d$.
    Then, by letting $\alpha = \pi_1^{-1}\pi_2$ and $\beta = \pi_2^{-1}\pi_3$ the triangle inequality
    \begin{align*}
      d_{\rho}(\pi_1,\pi_2) + d_{\rho}(\pi_2,\pi_3) \geq d_{\rho}(\pi_1,\pi_3)
    \end{align*}
    is equivalent to
    \begin{align}\label{eq:to prove permtation distance}
      \prod_{l \in C(\alpha)}\left(\sum_{i=1}^d \lambda_i^l\right) \prod_{m \in C(\beta)} \left(\sum_{i=1}^d \lambda_i^m\right) \leq \prod_{n \in C(\alpha \beta)} \left(\sum_{i=1}^d \lambda_i^n\right).
    \end{align}
    Moreover, we need to show that if the spectrum is not flat, we have equality if and only if $\alpha$ and $\beta$ are disjoint permutations.
    It suffices to show this for the case where $\beta$ is a cycle, as we can write $\beta$ as a product of disjoint cycles in general and iteratively apply the result for the case where $\beta$ is a cycle.
    We write $\beta = (i_1 \ldots i_m)$ and $\beta_q = (i_1 \ldots i_q)$ for $q \leq m$ and we let $\alpha_{q} = \alpha \beta_q$ and $\alpha_0 = \alpha$.

    We can then find a unique sequence of numbers $q_0 = 0 < q_1 < q_2 < \ldots< m$ such that, if one compares $\alpha_{q_j}$ and $\alpha_{q_{j+1}}$, then either $q_{j+1} - q_j + 1$ cycles have merged into a single cycle, or a single cycle has split into $q_{j+1} - q_j + 1$ cycles.
    Moreover, these two operations are alternating in the sense that if $\alpha_{q_j}$ to $\alpha_{q_{j+1}}$ is a merge then $\alpha_{q_{j+1}}$ to $\alpha_{q_{j+2}}$ is a split and vice versa. Indeed, $\beta_{q+1} = \beta_q (i_{q} \, i_{q+1})$, so $\alpha_{q+1} = \alpha_q (i_{q} \, i_{q+1})$. If $i_q$ and $i_{q+1}$ are in different cycles in $\alpha_q$, then applying $(i_q \, i_{q+1})$ merges these two cycles, and if $i_q$ and $i_{q+1}$ are in the same cycle in $\alpha_q$, then applying $(i_q \, i_{q+1})$ splits this cycle into two cycles.
    Let
    \begin{align*}
      J = \{j : \text{$\alpha_{q_j}$ to $\alpha_{q_{j+1}}$ is a merge}\}
    \end{align*}
    and for $j \in J$ let $m_j = q_{j+1} - q_j + 1$, then $\sum_{j \in J} m_j \leq m$.
    Now, it is clear that for any collection $n_1, \ldots, n_r$ of numbers with $\sum_{j=1}^r n_j \leq n$ it holds that
    \begin{align}\label{eq:simple inequality}
      \prod_{j =1}^r\left(\sum_{i=1}^d \lambda_i^{n_j}\right) \geq \sum_{i=1}^d \lambda_i^n
    \end{align}
    where the inequality is strict unless $r =1$ and $n_1 = n$ (since we assume that $\rho$ is not pure).
    In particular (recall that we assumed $\beta$ to be a cycle of length $m$) we may estimate the left hand side of \cref{eq:to prove permtation distance} by
    \begin{align}\label{eq:split beta}
      \prod_{l \in C(\alpha)}\left(\sum_{i=1}^d \lambda_i^l\right) \left(\sum_{i=1}^d \lambda_i^m\right) \leq \prod_{l \in C(\alpha)}\left(\sum_{i=1}^d \lambda_i^l\right) \prod_{j \in J} \left(\sum_{i=1}^d \lambda_i^{m_j}\right).
    \end{align}
    with equality if and only if $\{m_j\}_{j \in J} = \{m\}$ (so $\beta$ only merges cycles).
    We will next argue that for $j \notin J$
    \begin{align}\label{eq:split}
      \prod_{l \in C(\alpha_{q_j})}\left(\sum_{i=1}^d \lambda_i^l\right) \leq \prod_{l \in C(\alpha_{q_{j+1}})}\left(\sum_{i=1}^d \lambda_i^l\right).
    \end{align}
    and for $j \in J$
    \begin{align}\label{eq:merge}
      \prod_{l \in C(\alpha_{q_j})}\left(\sum_{i=1}^d \lambda_i^l \right) \left(\sum_{i=1}^d \lambda_i^{m_j} \right) \leq \prod_{l \in C(\alpha_{q_{j+1}})}\left(\sum_{i=1}^d \lambda_i^l\right)
    \end{align}
    with equality if and only if the spectrum is flat or all the cycles that are merged are $1$-cycles.
    Then, combining \cref{eq:split beta}, \cref{eq:split} and \cref{eq:merge} we may conclude that \cref{eq:to prove permtation distance} holds, with equality if and only if $\beta$ is a disjoint cycle from $\alpha$.

    \cref{eq:split} follows immediately from \cref{eq:simple inequality}, so it remains to show \cref{eq:merge}.
    To this end we will apply Jensen's inequality.
    In \cref{eq:merge} let us assume that $s$ cycles of lengths $l_1,\ldots,l_s$ are merged into a cycle of length $l$, so $l = \sum_{p=1}^s l_p$.
    We let $l_{s+1} = s$.
    Let $f_p : [d] \to \RR_{\geq0}$ be the function defined by $i\mapsto\lambda_{i}^{l_p - 1}$, for~$p \in [s + 1]$.
    Moreover, let
    \begin{align*}
      \eta_p = \frac{l_p-1}{l - 1} \quad \text{ for $p \in [s]$.}
    \end{align*}
    We consider expectation values over the probability measure on $[d]$ where $i$ has probability $\lambda_{i}$.
    Since $\eta_p \leq 1$ for all $p$, by Jensen's inequality for $\eta_p \neq 0$
    \begin{align}
      \left(\EE f_p^{\frac{1}{\eta_p}}\right)^{\eta_p} \geq \EE f_p \label{eq:expineq}
    \end{align}
    with equality if and only if the spectrum is flat or $\eta_p = 1$.
    Note that for $\eta_p = 0$ (equivalently $l_p = 1$), $\EE f_p = 1$.
    Using that
    \begin{align*}
      \sum_{p=1}^{s+1} \eta_p & = \sum_{p=1}^{s} \frac{l_p - 1}{l-1} + \frac{s-1}{l-1} = \frac{l - s + s-1}{l-1} = 1 \\
    \end{align*}
    and for $\eta_p \neq 0$
    \begin{align*}
      \EE f_p^{\frac{1}{\eta_p}} & = \sum_{i=1}^d \lambda_i^{l}
    \end{align*}
    we find that
    \begin{align*}
      \prod_{p=1}^s \left(\sum_{i=1}^d \lambda_i^{l_p}\right) \sum_{i=1}^d \lambda_i^{s} = \prod_{p=1}^{s+1}\EE f_p \leq \prod_{p, \eta_p \neq 0} \left(\EE f_p^{\frac{1}{\eta_p}}\right)^{\eta_p} = \sum_{i=1}^d \lambda_i^{l}.
    \end{align*}
    Here we have equality if and only if the spectrum is flat or for all $p \in [s+1]$ it holds that $\eta_p = 1$ or $\eta_p = 0$.
    Since $\eta_p < 1$ for $p \in [s]$ and $l_{s+1} > 1$ this only happens if $l_p = 1$ for $p \in [s]$ and $s = l_{s+1} = l$ (in other words, if all the merged cycles are 1-cycles).
    Applying this with $s=m_j$ for each $j\in J$ proves \cref{eq:merge}.
  \end{proof}
\end{appendix}

\addcontentsline{toc}{section}{References}
\bibliographystyle{alpha}
\bibliography{references.bib}

\newcommand{\etalchar}[1]{$^{#1}$}
\begin{thebibliography}{CGGPG13}

\bibitem[AEMM19]{almheiri2019entropy}
Ahmed Almheiri, Netta Engelhardt, Donald Marolf, and Henry Maxfield.
\newblock The entropy of bulk quantum fields and the entanglement wedge of an
  evaporating black hole.
\newblock {\em Journal of High Energy Physics}, 2019(12):1--47, 2019.

\bibitem[AFLR21]{akers2021reflected}
Chris Akers, Thomas Faulkner, Simon Lin, and Pratik Rath.
\newblock Reflected entropy in random tensor networks.
\newblock {\em arXiv preprint arXiv:2112.09122}, 2021.

\bibitem[AGZ10]{anderson2010introduction}
Greg~W Anderson, Alice Guionnet, and Ofer Zeitouni.
\newblock {\em An Introduction to Random Matrices}.
\newblock Cambridge University Press, 2010.

\bibitem[AHM{\etalchar{+}}20]{almheiri2020replica}
Ahmed Almheiri, Thomas Hartman, Juan Maldacena, Edgar Shaghoulian, and
  Amirhossein Tajdini.
\newblock Replica wormholes and the entropy of {Hawking} radiation.
\newblock {\em Journal of High Energy Physics}, 2020(5):1--42, 2020.

\bibitem[AN12]{aubrun2012realigning}
Guillaume Aubrun and Ion Nechita.
\newblock Realigning random states.
\newblock {\em Journal of Mathematical Physics}, 53(10):102210, 2012.

\bibitem[AP20]{akers2020leading}
Chris Akers and Geoff Penington.
\newblock Leading order corrections to the quantum extremal surface
  prescription.
\newblock {\em arXiv preprint arXiv:2008.03319}, 2020.

\bibitem[AP22]{akers2021quantum}
Chris Akers and Geoff Penington.
\newblock Quantum minimal surfaces from quantum error correction, 2022.

\bibitem[AR19]{akers2019holographic}
Chris Akers and Pratik Rath.
\newblock Holographic {R}\'enyi entropy from quantum error correction.
\newblock {\em Journal of High Energy Physics}, 2019(5):1--24, 2019.

\bibitem[AS17]{aubrun2017alice}
Guillaume Aubrun and Stanis{\l}aw~J Szarek.
\newblock {\em Alice and Bob meet Banach}, volume 223.
\newblock American Mathematical Society, 2017.

\bibitem[ASY12]{aubrun2012phase}
Guillaume Aubrun, Stanis{\l}aw~J Szarek, and Deping Ye.
\newblock Phase transitions for random states and a semicircle law for the
  partial transpose.
\newblock {\em Physical Review A}, 85(3):030302, 2012.

\bibitem[Aub12]{aubrun2012partial}
Guillaume Aubrun.
\newblock Partial transposition of random states and non-centered semicircular
  distributions.
\newblock {\em Random Matrices: Theory and Applications}, 1(02):1250001, 2012.

\bibitem[BBCC11]{banica2011free}
Teodor Banica, Serban~Teodor Belinschi, Mireille Capitaine, and Benoit Collins.
\newblock Free {B}essel laws.
\newblock {\em Canadian Journal of Mathematics}, 63(1):3--37, 2011.

\bibitem[Bil08]{billingsley2008probability}
Patrick Billingsley.
\newblock {\em Probability and {M}easure}.
\newblock John Wiley \& Sons, 2008.

\bibitem[BPSW19]{bao2019beyond}
Ning Bao, Geoffrey Penington, Jonathan Sorce, and Aron~C Wall.
\newblock Beyond toy models: distilling tensor networks in full {AdS/CFT}.
\newblock {\em Journal of High Energy Physics}, 2019(11):1--63, 2019.

\bibitem[BS10]{bai2010spectral}
Zhidong Bai and Jack~W Silverstein.
\newblock {\em Spectral {A}nalysis of {L}arge {D}imensional {R}andom
  {M}atrices}, volume~20.
\newblock Springer, 2010.

\bibitem[BW20]{bousso2020gravity}
Raphael Bousso and Elizabeth Wildenhain.
\newblock Gravity/ensemble duality.
\newblock {\em Physical Review D}, 102(6):066005, 2020.

\bibitem[CDKW14]{christandl2014eigenvalue}
Matthias Christandl, Brent Doran, Stavros Kousidis, and Michael Walter.
\newblock Eigenvalue distributions of reduced density matrices.
\newblock {\em Communications in Mathematical Physics}, 332(1):1--52, 2014.

\bibitem[CGGPG13]{collins2012matrix}
Beno{\^\i}t Collins, Carlos~E Gonz{\'a}lez-Guill{\'e}n, and David
  P{\'e}rez-Garc{\'\i}a.
\newblock Matrix product states, random matrix theory and the principle of
  maximum entropy.
\newblock {\em Communications in Mathematical Physics}, 320(3):663--677, 2013.

\bibitem[CHLS15]{czech2015information}
Bartlomiej Czech, Patrick Hayden, Nima Lashkari, and Brian Swingle.
\newblock The information theoretic interpretation of the length of a curve.
\newblock {\em Journal of High Energy Physics}, 2015(6):1--40, 2015.

\bibitem[CN16]{collins2016random}
Benoit Collins and Ion Nechita.
\newblock Random matrix techniques in quantum information theory.
\newblock {\em Journal of Mathematical Physics}, 57(1):015215, 2016.

\bibitem[CN{\.Z}10]{collins2010random}
Beno{\^\i}t Collins, Ion Nechita, and Karol {\.Z}yczkowski.
\newblock Random graph states, maximal flow and {Fuss--Catalan} distributions.
\newblock {\em Journal of Physics A: Mathematical and Theoretical},
  43(27):275--303, 2010.

\bibitem[CN{\.Z}13]{collins2013area}
Beno{\^\i}t Collins, Ion Nechita, and Karol {\.Z}yczkowski.
\newblock Area law for random graph states.
\newblock {\em Journal of Physics A: Mathematical and Theoretical},
  46(30):305302, 2013.

\bibitem[CPGSV21]{cirac2021matrix}
J~Ignacio Cirac, David Perez-Garcia, Norbert Schuch, and Frank Verstraete.
\newblock Matrix product states and projected entangled pair states: Concepts,
  symmetries, theorems.
\newblock {\em Reviews of Modern Physics}, 93(4):045003, 2021.

\bibitem[DBWR14]{dupuis2014one}
Fr{\'e}d{\'e}ric Dupuis, Mario Berta, J{\"u}rg Wullschleger, and Renato Renner.
\newblock One-shot decoupling.
\newblock {\em Communications in Mathematical Physics}, 328(1):251--284, 2014.

\bibitem[DF13]{drescher2013simultaneous}
Lukas Drescher and Omar Fawzi.
\newblock On simultaneous min-entropy smoothing.
\newblock In {\em 2013 IEEE International Symposium on Information Theory},
  pages 161--165. IEEE, 2013.

\bibitem[DH10]{dutil2010one}
Nicolas Dutil and Patrick Hayden.
\newblock One-shot multiparty state merging.
\newblock {\em arXiv preprint arXiv:1011.1974}, 2010.

\bibitem[DHM19]{dong2019flat}
Xi~Dong, Daniel Harlow, and Donald Marolf.
\newblock Flat entanglement spectra in fixed-area states of quantum gravity.
\newblock {\em Journal of High Energy Physics}, 2019(10):1--25, 2019.

\bibitem[DMW21]{dong2021replica}
Xi~Dong, Sean McBride, and Wayne~W Weng.
\newblock Replica wormholes and holographic entanglement negativity.
\newblock {\em arXiv preprint arXiv:2110.11947}, 2021.

\bibitem[DQW21]{dong2021holographic}
Xi~Dong, Xiao-Liang Qi, and Michael Walter.
\newblock Holographic entanglement negativity and replica symmetry breaking.
\newblock {\em arXiv preprint arXiv:2101.11029}, 2021.

\bibitem[Dup15]{dupuis2015chain}
Fr{\'e}d{\'e}ric Dupuis.
\newblock Chain rules for quantum {R\'enyi} entropies.
\newblock {\em Journal of Mathematical Physics}, 56(2):022203, 2015.

\bibitem[Dut11]{dutil2011multiparty}
Nicolas Dutil.
\newblock Multiparty quantum protocols for assisted entanglement distillation.
\newblock {\em arXiv preprint arXiv:1105.4657}, 2011.

\bibitem[EW15]{engelhardt2015quantum}
Netta Engelhardt and Aron~C Wall.
\newblock Quantum extremal surfaces: holographic entanglement entropy beyond
  the classical regime.
\newblock {\em Journal of High Energy Physics}, 2015(1):1--27, 2015.

\bibitem[FLM13]{faulkner2013quantum}
Thomas Faulkner, Aitor Lewkowycz, and Juan Maldacena.
\newblock Quantum corrections to holographic entanglement entropy.
\newblock {\em Journal of High Energy Physics}, 2013(11):1--18, 2013.

\bibitem[FR15]{fawzi2015quantum}
Omar Fawzi and Renato Renner.
\newblock Quantum conditional mutual information and approximate {Markov}
  chains.
\newblock {\em Communications in Mathematical Physics}, 340(2):575--611, 2015.

\bibitem[GAE07]{gross2007evenly}
David Gross, Koenraad Audenaert, and Jens Eisert.
\newblock Evenly distributed unitaries: On the structure of unitary designs.
\newblock {\em Journal of Mathematical Physics}, 48(5):052104, 2007.

\bibitem[Har17]{harlow2016qec}
Daniel Harlow.
\newblock The {Ryu-Takayanagi} formula from quantum error correction.
\newblock {\em Communications in Mathematical Physics}, 354(3):865--912, 2017.

\bibitem[Has17]{hastings2017asymptotics}
Matthew~B Hastings.
\newblock The asymptotics of quantum max-flow min-cut.
\newblock {\em Communications in Mathematical Physics}, 351(1):387--418, 2017.

\bibitem[Hay08]{hayashi2008second}
Masahito Hayashi.
\newblock Second-order asymptotics in fixed-length source coding and intrinsic
  randomness.
\newblock {\em IEEE Transactions on Information Theory}, 54(10):4619--4637,
  2008.

\bibitem[HLW06]{hayden2006aspects}
Patrick Hayden, Debbie~W Leung, and Andreas Winter.
\newblock Aspects of generic entanglement.
\newblock {\em Communications in Mathematical Physics}, 265(1):95--117, 2006.

\bibitem[HNQ{\etalchar{+}}16]{hayden2016holographic}
Patrick Hayden, Sepehr Nezami, Xiao-Liang Qi, Nathaniel Thomas, Michael Walter,
  and Zhao Yang.
\newblock Holographic duality from random tensor networks.
\newblock {\em Journal of High Energy Physics}, 2016(11):1--56, 2016.

\bibitem[HRT07]{hubeny2007covariant}
Veronika~E Hubeny, Mukund Rangamani, and Tadashi Takayanagi.
\newblock A covariant holographic entanglement entropy proposal.
\newblock {\em Journal of High Energy Physics}, 2007(07):062, 2007.

\bibitem[KFNR21]{kudler2021negativity}
Jonah Kudler-Flam, Vladimir Narovlansky, and Shinsei Ryu.
\newblock Negativity spectra in random tensor networks and holography.
\newblock {\em arXiv preprint arXiv:2109.02649}, 2021.

\bibitem[KR05]{klappenecker2005mutually}
Andreas Klappenecker and Martin Rotteler.
\newblock Mutually unbiased bases are complex projective 2-designs.
\newblock In {\em Proceedings. International Symposium on Information Theory,
  2005. ISIT 2005.}, pages 1740--1744. IEEE, 2005.

\bibitem[LC21]{levy2021entanglement}
Ryan Levy and Bryan~K Clark.
\newblock Entanglement entropy transitions with random tensor networks.
\newblock {\em arXiv preprint arXiv:2108.02225}, 2021.

\bibitem[LM13]{lewkowycz2013generalized}
Aitor Lewkowycz and Juan Maldacena.
\newblock Generalized gravitational entropy.
\newblock {\em Journal of High Energy Physics}, 2013(8):1--29, 2013.

\bibitem[LPG21]{lancien2021correlation}
C{\'e}cilia Lancien and David P{\'e}rez-Garc{\'\i}a.
\newblock Correlation length in random {MPS and PEPS}.
\newblock {\em Annales Henri Poincar{\'e}}, pages 1--82, 2021.

\bibitem[LPWV20]{lopez2020mean}
Javier Lopez-Piqueres, Brayden Ware, and Romain Vasseur.
\newblock Mean-field entanglement transitions in random tree tensor networks.
\newblock {\em Physical Review B}, 102(6):064202, 2020.

\bibitem[LVFL21]{li2021statistical}
Yaodong Li, Romain Vasseur, Matthew Fisher, and Andreas~WW Ludwig.
\newblock Statistical mechanics model for {Clifford} random tensor networks and
  monitored quantum circuits.
\newblock {\em arXiv preprint arXiv:2110.02988}, 2021.

\bibitem[Mal99]{maldacena1999large}
Juan Maldacena.
\newblock The large-{$N$} limit of superconformal field theories and
  supergravity.
\newblock {\em International Journal of Theoretical Physics}, 38(4):1113--1133,
  1999.

\bibitem[MB21]{morgan2021classical}
Erica Morgan and Fernando~GSL Brand{\~a}o.
\newblock A classical model correspondence for {$G$-symmetric} random tensor
  networks.
\newblock {\em Journal of Physics Communications}, 2021.

\bibitem[MM20]{marolf2020transcending}
Donald Marolf and Henry Maxfield.
\newblock Transcending the ensemble: baby universes, spacetime wormholes, and
  the order and disorder of black hole information.
\newblock {\em Journal of High Energy Physics}, 2020(8):1--72, 2020.

\bibitem[MS17]{mingo2017free}
James~A Mingo and Roland Speicher.
\newblock {\em Free Probability and Random Matrices}, volume~35.
\newblock Springer, 2017.

\bibitem[MVS21]{medina2021entanglement}
Raimel Medina, Romain Vasseur, and Maksym Serbyn.
\newblock Entanglement transitions from restricted {Boltzmann} machines.
\newblock {\em Physical Review B}, 104(10):104205, 2021.

\bibitem[MWW20]{marolf2020probing}
Donald Marolf, Shannon Wang, and Zhencheng Wang.
\newblock Probing phase transitions of holographic entanglement entropy with
  fixed area states.
\newblock {\em Journal of High Energy Physics}, 2020(12):1--41, 2020.

\bibitem[NRSR21]{nahum2021measurement}
Adam Nahum, Sthitadhi Roy, Brian Skinner, and Jonathan Ruhman.
\newblock Measurement and entanglement phase transitions in all-to-all quantum
  circuits, on quantum trees, and in {Landau-Ginsburg} theory.
\newblock {\em PRX Quantum}, 2(1):010352, 2021.

\bibitem[NS06]{nica2006lectures}
Alexandru Nica and Roland Speicher.
\newblock {\em Lectures on the combinatorics of free probability}, volume~13.
\newblock Cambridge University Press, 2006.

\bibitem[NW20]{nezami2020multipartite}
Sepehr Nezami and Michael Walter.
\newblock Multipartite entanglement in stabilizer tensor networks.
\newblock {\em Physical Review Letters}, 125(24):241602, 2020.

\bibitem[PB20]{potters2020first}
Marc Potters and Jean-Philippe Bouchaud.
\newblock {\em A First Course in Random Matrix Theory: For Physicists,
  Engineers and Data Scientists}.
\newblock Cambridge University Press, 2020.

\bibitem[Pen20]{penington2020entanglement}
Geoffrey Penington.
\newblock Entanglement wedge reconstruction and the information paradox.
\newblock {\em Journal of High Energy Physics}, 2020(9):1--84, 2020.

\bibitem[PSSY19]{penington2019replica}
Geoff Penington, Stephen~H Shenker, Douglas Stanford, and Zhenbin Yang.
\newblock Replica wormholes and the black hole interior.
\newblock {\em arXiv preprint arXiv:1911.11977}, 2019.

\bibitem[PYHP15]{pastawski2015holographic}
Fernando Pastawski, Beni Yoshida, Daniel Harlow, and John Preskill.
\newblock Holographic quantum error-correcting codes: Toy models for the
  bulk/boundary correspondence.
\newblock {\em Journal of High Energy Physics}, 2015(6):1--55, 2015.

\bibitem[QSY21]{qi2021holevo}
Xiao-Liang Qi, Zhou Shangnan, and Zhenbin Yang.
\newblock Holevo information and ensemble theory of gravity.
\newblock {\em arXiv preprint arXiv:2111.05355}, 2021.

\bibitem[QY18]{qi2018spacetime}
Xiao-Liang Qi and Zhao Yang.
\newblock Space-time random tensor networks and holographic duality.
\newblock {\em arXiv preprint arXiv:1801.05289}, 2018.

\bibitem[QYY17]{qi2017holographic}
Xiao-Liang Qi, Zhao Yang, and Yi-Zhuang You.
\newblock Holographic coherent states from random tensor networks.
\newblock {\em Journal of High Energy Physics}, 2017(8):1--29, 2017.

\bibitem[RT06a]{ryu2006aspects}
Shinsei Ryu and Tadashi Takayanagi.
\newblock Aspects of holographic entanglement entropy.
\newblock {\em Journal of High Energy Physics}, 2006(08):045, 2006.

\bibitem[RT06b]{ryu2006holographic}
Shinsei Ryu and Tadashi Takayanagi.
\newblock Holographic derivation of entanglement entropy from the anti--de
  {Sitter} space/conformal field theory correspondence.
\newblock {\em Physical Review Letters}, 96(18):181602, 2006.

\bibitem[S{\'a}r17]{sarosi2017ads}
G{\'a}bor S{\'a}rosi.
\newblock {$\text{AdS}_2$} holography and the {SYK} model.
\newblock {\em arXiv preprint arXiv:1711.08482}, 2017.

\bibitem[SSS19]{saad2019jt}
Phil Saad, Stephen~H Shenker, and Douglas Stanford.
\newblock {JT} gravity as a matrix integral.
\newblock {\em arXiv preprint arXiv:1903.11115}, 2019.

\bibitem[SSSY21]{saad2021wormholes}
Phil Saad, Stephen~H Shenker, Douglas Stanford, and Shunyu Yao.
\newblock Wormholes without averaging.
\newblock {\em arXiv preprint arXiv:2103.16754}, 2021.

\bibitem[Swi12a]{swingle2012constructing}
Brian Swingle.
\newblock Constructing holographic spacetimes using entanglement
  renormalization.
\newblock {\em arXiv preprint arXiv:1209.3304}, 2012.

\bibitem[Swi12b]{swingle2012entanglement}
Brian Swingle.
\newblock Entanglement renormalization and holography.
\newblock {\em Physical Review D}, 86(6):065007, 2012.

\bibitem[TH13]{tomamichel2013hierarchy}
Marco Tomamichel and Masahito Hayashi.
\newblock A hierarchy of information quantities for finite block length
  analysis of quantum tasks.
\newblock {\em IEEE Transactions on Information Theory}, 59(11):7693--7710,
  2013.

\bibitem[Tom12]{tomamichel2012framework}
Marco Tomamichel.
\newblock {\em A framework for non-asymptotic quantum information theory}.
\newblock PhD thesis, ETH Zurich, 2012.

\bibitem[Tom15]{tomamichel2015quantum}
Marco Tomamichel.
\newblock {\em Quantum Information Processing with Finite Resources:
  Mathematical Foundations}, volume~5.
\newblock Springer, 2015.

\bibitem[VC04]{verstraete2004valence}
Frank Verstraete and J~Ignacio Cirac.
\newblock Valence-bond states for quantum computation.
\newblock {\em Physical Review A}, 70(6):060302, 2004.

\bibitem[VMC08]{verstraete2008matrix}
Frank Verstraete, Valentin Murg, and J~Ignacio Cirac.
\newblock Matrix product states, projected entangled pair states, and
  variational renormalization group methods for quantum spin systems.
\newblock {\em Advances in Physics}, 57(2):143--224, 2008.

\bibitem[VPYL19]{vasseur2019entanglement}
Romain Vasseur, Andrew~C Potter, Yi-Zhuang You, and Andreas~WW Ludwig.
\newblock Entanglement transitions from holographic random tensor networks.
\newblock {\em Physical Review B}, 100(13):134203, 2019.

\bibitem[Wan22]{wang2022information}
Jinzhao Wang.
\newblock {\em ON THE INFORMATION-THEORETIC ASPECTS OF BLACK HOLES}.
\newblock PhD thesis, ETH Zurich, 2022.

\bibitem[WW21]{walter2021hypergraph}
Michael Walter and Freek Witteveen.
\newblock Hypergraph min-cuts from quantum entropies.
\newblock {\em Journal of Mathematical Physics}, 62(9):092203, 2021.

\bibitem[YHQ16]{yang2016bidirectional}
Zhao Yang, Patrick Hayden, and Xiao-Liang Qi.
\newblock Bidirectional holographic codes and sub-{AdS} locality.
\newblock {\em Journal of High Energy Physics}, 2016(1):1--24, 2016.

\bibitem[YLFC21]{yang2021entanglement}
Zhi-Cheng Yang, Yaodong Li, Matthew Fisher, and Xiao Chen.
\newblock Entanglement phase transitions in random stabilizer tensor networks.
\newblock {\em arXiv preprint arXiv:2107.12376}, 2021.

\bibitem[YYQ18]{you2018machine}
Yi-Zhuang You, Zhao Yang, and Xiao-Liang Qi.
\newblock Machine learning spatial geometry from entanglement features.
\newblock {\em Physical Review B}, 97(4):045153, 2018.

\end{thebibliography}

\end{document}